\documentclass[12pt,a4paper,twoside,openany]{book}
\usepackage[
  top=2cm,
  bottom=2cm,
  left=2cm,
  right=2cm,
  headheight=17pt, % as per the warning by fancyhdr
  includehead,includefoot,
  heightrounded, % to avoid spurious underfull messages
]{geometry}

\usepackage[utf8]{inputenc}
\usepackage[hidelinks]{hyperref}
\usepackage{amsmath}
\usepackage{cite}
\usepackage{amsthm}
\usepackage{amssymb}
\usepackage{amsfonts}
\usepackage{mathpartir}
\usepackage{scalerel}
\usepackage{stmaryrd}
\usepackage{authblk}
\usepackage{todonotes}
\presetkeys{todonotes}{inline}{}
\usepackage{bm}
\usepackage{titlesec}
\usepackage{graphicx}
\usepackage{afterpage}

\newcommand\blankpage{%
    \null
    \thispagestyle{empty}%
    \addtocounter{page}{-1}%
    \newpage}
\usepackage{multirow}
\usepackage{tikz}
\usetikzlibrary{arrows,
                calc,
                decorations.pathreplacing,
                calligraphy,% had to be after decorations.pathreplacing
                matrix,
                positioning
}

\makeatletter
\DeclareRobustCommand{\sqcdot}{\mathbin{\mathpalette\morphic@sqcdot\relax}}
\newcommand{\morphic@sqcdot}[2]{%
  \sbox\z@{$\m@th#1\centerdot$}%
  \ht\z@=.33333\ht\z@
  \vcenter{\box\z@}%
}
\makeatother

\usepackage{fancyhdr}
\theoremstyle{remark}
\newtheorem{notation}{Notation}

\theoremstyle{definition}
\newtheorem{mydefinition}{Definition}
\newtheorem{myexample}{Example}
\newtheorem{mylemma}{Lemma}

\theoremstyle{theorem}
\newtheorem{theorem}{Theorem}

\titleformat{\chapter}[display]
  {\bfseries\Large}
  {\filright\MakeUppercase{\chaptertitlename} \Huge\thechapter}
  {1ex}
  {\titlerule\vspace{1ex}\filleft}
  [\vspace{1ex}\titlerule]

\newcommand{\mi}[1]{\mathit{#1}}
\newcommand{\ms}[1]{\mathsf{#1}}
\newcommand{\mbb}[1]{\mathbb{#1}}
\newcommand{\mbf}[1]{\mathbf{#1}}
\newcommand{\bs}[1]{\boldsymbol{#1}}
\newcommand{\ap}{\ms{ap}}
\newcommand{\apd}{\ms{apd}}
\newcommand{\tr}{\ms{tr}}
\newcommand{\happly}{\ms{happly}}
\newcommand{\funext}{\ms{funext}}
\newcommand{\toind}{\to^{\ms{int}}}

\newcommand{\Tys}{\ms{Ty_{sig}}}
\newcommand{\Tms}{\ms{Tm_{sig}}}

\newcommand{\Ix}{\mi{Ix}}

\newcommand{\zero}{\ms{zero}}
\newcommand{\suc}{\ms{suc}}
\newcommand{\J}{\ms{J}}
\newcommand{\UIP}{\ms{UIP}}

\newcommand{\refl}{\mathsf{refl}}
\newcommand{\reflect}{\mathsf{reflect}}

\newcommand{\id}{\mathsf{id}}
\newcommand{\Con}{\mathsf{Con}}
\newcommand{\Sub}{\mathsf{Sub}}
\newcommand{\Tm}{\mathsf{Tm}}
\newcommand{\Ty}{\mathsf{Ty}}
\newcommand{\U}{\mathsf{U}}
\newcommand{\El}{\mathsf{El}}
\newcommand{\Id}{\mathsf{Id}}
\newcommand{\ID}{\mathsf{ID}}
\newcommand{\proj}{\mathsf{proj}}
\renewcommand{\tt}{\mathsf{tt}}
\newcommand{\blank}{\mathord{\hspace{1pt}\text{--}\hspace{1pt}}}
\newcommand{\ra}{\rightarrow}
\newcommand{\Set}{\mathsf{Set}}
\newcommand{\Lift}{\Uparrow}
\newcommand{\ToS}{\mathsf{ToS}}
\newcommand{\ext}{\triangleright}
\newcommand{\emptycon}{\scaleobj{.75}\bullet}

\newcommand{\Pii}{\Pi}
\newcommand{\funi}{\Rightarrow}
\newcommand{\appi}{\mathsf{app}}
\newcommand{\lami}{\mathsf{lam}}

\newcommand{\fune}{\Rightarrow^{\ms{Ext}}}
\newcommand{\Pie}{\Pi^{\mathsf{Ext}}}
\newcommand{\appe}{\mathsf{app^{Ext}}}
\newcommand{\lame}{\mathsf{lam^{Ext}}}
\newcommand{\toe}{\to^{\ms{Ext}}}

\newcommand{\lambdae}{\lambda^{\ms{Ext}}}

\newcommand{\Piinf}{\Pi^{\mathsf{ext}}}
\newcommand{\appinf}{\mathsf{app^{ext}}}
\newcommand{\laminf}{\mathsf{lam^{ext}}}
\newcommand{\laminfprime}{\mathsf{lam^{ext'}}}
\newcommand{\toinf}{\to^{\ms{ext}}}
\newcommand{\lambdainf}{\lambda^{\ms{ext}}}

\newcommand{\bPiinf}{\bs{\Piinf}}

\newcommand{\Sig}{\mathsf{Sig}}
\newcommand{\ToSSig}{\mathsf{ToSSig}}

\newcommand{\NatSig}{\mathsf{NatSig}}

\newcommand{\flcwf}{\mathsf{flcwf}}
\newcommand{\SigTy}{\mathsf{SigTy}}
\newcommand{\SigTm}{\mathsf{SigTm}}

\newcommand{\lamK}{\mathsf{lam}_{\K}}
\newcommand{\appK}{\mathsf{app}_{\K}}

\newcommand{\p}{\mathsf{p}}
\newcommand{\q}{\mathsf{q}}
\newcommand{\K}{\mathsf{K}}
\newcommand{\A}{\mathsf{A}}

\newcommand{\arri}{\Rightarrow}
\newcommand{\syn}{\mathsf{syn}}

\newcommand{\bCon}{\bs{\Con}}
\newcommand{\bTy}{\bs{\Ty}}
\newcommand{\bSub}{\bs{\Sub}}
\newcommand{\bTm}{\bs{\Tm}}
\newcommand{\bGamma}{\bs{\Gamma}}
\newcommand{\bDelta}{\bs{\Delta}}
\newcommand{\bsigma}{\bs{\sigma}}
\newcommand{\bdelta}{\bs{\delta}}

\newcommand{\bt}{\bs{t}}
\newcommand{\bu}{\bs{u}}
\newcommand{\bA}{\bs{A}}
\newcommand{\ba}{\bs{a}}
\newcommand{\bb}{\bs{b}}
\newcommand{\bB}{\bs{B}}
\newcommand{\bid}{\bs{\id}}
\newcommand{\bemptycon}{\scaleobj{.75}{\bs{\bullet}}}

\newcommand{\bU}{\bs{\U}}
\newcommand{\bEl}{\bs{\El}}

\newcommand{\bPie}{\bs{\Pie}}

\newcommand{\bId}{\bs{\Id}}
\newcommand{\bM}{\bs{\mathsf{M}}}

\newcommand{\bK}{\bs{\mathsf{K}}}

\newcommand{\ul}[1]{\underline{#1}}
\newcommand{\ulGamma}{\ul{\Gamma}}
\newcommand{\ulDelta}{\ul{\Delta}}

\newcommand{\uldelta}{\ul{\delta}}
\newcommand{\ulsigma}{\ul{\sigma}}

\newcommand{\ulemptycon}{\ul{\emptycon}}
\newcommand{\ult}{\ul{t}}
\newcommand{\ulu}{\ul{u}}
\newcommand{\ulA}{\ul{A}}

\newcommand{\ulB}{\ul{B}}

\newcommand{\coe}{\mathsf{coe}}
\newcommand{\coh}{\mathsf{coh}}
\newcommand{\llb}{\llbracket}
\newcommand{\rrb}{\rrbracket}
\newcommand{\sem}[1]{\llb#1\rrb}

\newcommand{\Var}{\ms{Var}}
\newcommand{\var}{\ms{var}}
\newcommand{\app}{\ms{app}}
\newcommand{\vz}{\ms{vz}}
\newcommand{\vs}{\ms{vs}}
\newcommand{\Alg}{\ms{Alg}}
\newcommand{\Mor}{\ms{Mor}}
\newcommand{\DispAlg}{\ms{DispAlg}}
\newcommand{\Section}{\ms{Section}}
\newcommand{\Initial}{\ms{Initial}}
\newcommand{\Inductive}{\ms{Inductive}}
\newcommand{\TmAlg}{\ms{TmAlg}}
\newcommand{\Rec}{\ms{Rec}}
\newcommand{\Ind}{\ms{Ind}}
\newcommand{\Obj}{\ms{Obj}}
\newcommand{\Nat}{\ms{Nat}}
\newcommand{\Bool}{\ms{Bool}}
\newcommand{\mbbC}{\mbb{C}}
\newcommand{\hmbbC}{\hat{\mbb{C}}}

\newcommand{\lam}{\ms{lam}}

\newcommand{\true}{\ms{true}}
\newcommand{\false}{\ms{false}}
\newcommand{\up}{\uparrow}
\newcommand{\down}{\downarrow}

\newcommand{\lab}{\langle}
\newcommand{\rab}{\rangle}
\newcommand{\defn}{:\equiv}
\newcommand{\yon}{\ms{y}}
\newcommand{\lub}{\,\sqcup\,}

\begin{document}
\clearpage
%% \maketitle

\begin{titlepage}
    \begin{center}
        \vspace*{1cm}

        {\LARGE \textbf{Type-Theoretic Signatures for Algebraic Theories and Inductive Types}}

        %% \vspace{0.5cm}
        %% \LARGE
        %% Thesis Subtitle
        \vspace{2em}

        {\Large \textsc{ András Kovács}}\\
        %% \vspace{2em}
        %% {\large \textit{Supervisor: Ambrus Kaposi}}

        \vfill

        {\large \textsc{Submitted in fulfillment of the requirements
        for the degree of Doctor of Philosophy}}\\

        \vspace{0.8cm}

        \includegraphics[width=0.35\textwidth]{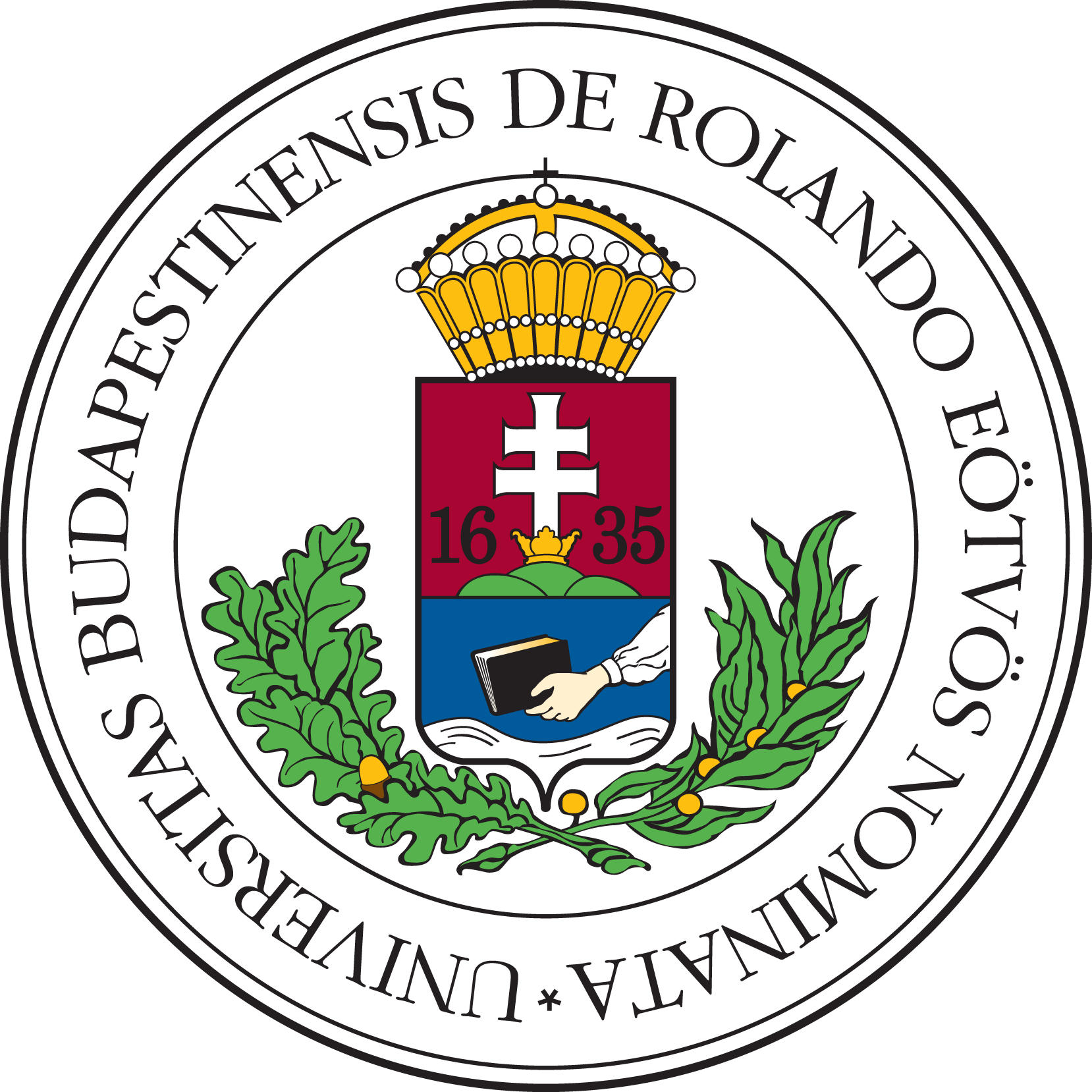}

        \large \textsc{
        Eötvös Loránd University\\
        Doctoral School of Informatics\\
        March 2022}
        \\
        \vspace{2em}
        \small{\textsc{DOI: 10.15476/ELTE.2022.070}}

    \end{center}
\end{titlepage}

\afterpage{\blankpage}

\chapter*{Acknowledgements}
\thispagestyle{empty}

I would like to thank my supervisor Ambrus Kaposi for his support and the
remarkably low-stress experience throughout my studies, and also for the
research collaboration which provided ample content for this thesis. Many thanks
to Tamás Kozsik for being a likewise supportive supervisor in my first year of
studies, and to Zoltán Horváth for doing a great job providing financial and
administrative support, and I also want to thank the secretaries at department
and at the faculty.

I would like to thank my opponents Fredrik Forsberg and Christian Sattler for
their careful reading and comments on my thesis.

I would like to thank attendees of the Budapest type theory seminar for the many
discussions over the years, in particular Balázs Kőműves, Péter Diviánszky,
Gábor Lehel, Gergő Érdi and my fellow PhD students Rafaël Bocquet and István
Donkó.

I would like to thank everyone in the research community who educated and
inspired me through various discussions and interactions. I would like to
express special thanks to my coauthors Thorsten Altenkirch, Nicolai Kraus and
Ambroise Lafont.

Finally, I would like to thank my parents for their love and support.

\frontmatter
\tableofcontents{}

\mainmatter

%% \chapter*{Acknowledgments}
%% \addcontentsline{toc}{chapter}{Acknowledgments}

\chapter{Introduction}

This thesis develops the usage of certain type theories as specification
languages for algebraic theories and inductive types.

\emph{Type theories} have emerged as popular metatheoretic settings for
mechanized mathematics. One reason is that the field of type theory is generally
aware of the issue of overheads in representation, and it is a common endeavor
to search for concise ``synthetic'' ways to talk about different areas of
mathematics. In type theory, it is a virtue to be able to directly say what we
mean, and in a way such that simple-minded computers are able to verify it.

\emph{Algebraic theories} are certain mathematical structures which are
especially well-behaved, and which are ubiquitous in mathematics, such as groups
or categories. In type theories, \emph{inductive types} are certain freely
generated (initial) models of algebraic theories. Inductive types are a core
feature in implementations of type theories, widely used in mathematical
formalization, but also as the primary way to define the data structures which
are used in programming.

This thesis observes that if we are to specify more complicated algebraic
theories, dependent type theories provide the natural tool to manage
complexities. The expressive power of type theory which makes it suitable as a
foundation for mechanized mathematics, also proves useful for the more
specialized task of specifying algebraic signatures.

There is a trade-off between the complexity of a mathematical language and the
ease of usage of the language. Minimal languages are convenient to reason about
and develop metatheory for, but they often require an excessive amount of
boilerplate to work in. However, it is a worthwhile effort to try to move
towards the Pareto frontier of this trade-off. We believe that the current
thesis makes progress in this respect.

Our signatures are useful in broader mathematical contexts, but we are also
concerned with potential implementation in proof assistants. Although it is
unlikely that our signatures can be deployed in practice exactly as they are, they
should be still helpful as formal bases of practical implementations.

\section{Overview}

In \textbf{Chapter \ref{chap:simple-inductive-signatures}}, we present a minimal
example of a type theory of signatures. This allows specifying single-sorted
signatures without equations. The purpose of the chapter is didactic. We develop
just enough semantics to get notions of initiality and induction for
algebras. We also present a \emph{term algebra construction}: this shows that
the initial algebra for each signature can be constructed from the syntax of
signatures itself.

In \textbf{Chapter \ref{chap:2ltt}} we describe a metatheoretic setting which is
often used in the thesis. This is \emph{two-level type theory}
\cite{twolevel}. It allows us to develop general semantics for signatures, while
still working inside a convenient type theory. As a demonstration, we generalize
the semantics from Chapter \ref{chap:simple-inductive-signatures} so that it is
given internally to arbitrary categories-with-families. As a special case,
signatures can be interpreted in arbitrary categories with finite products.

In \textbf{Chapter \ref{chap:fqiit}} we describe finitary quotient
inductive-inductive signatures. These are close to generalized algebraic
theories \cite{gat} in expressive power. In particular, most type theories
themselves can be specified with finitary quotient inductive-inductive
signatures. We significantly expand the semantics of signatures, now for each
signature we provide a category of algebras with certain extra structure, which
is equivalent to having finite limits. This allows us to prove for each
signature the equivalence of initiality and induction. Also, owing to two-level
type theory, signatures can be interpreted internally to any category with
finite limits. Additionally
\begin{itemize}
  \item We present a term algebra construction.
  \item We show that morphisms of signatures are interpreted as right adjoint functors
        in the semantics.
  \item We present how self-description of signatures can be exploited to minimize metatheoretic
        assumptions.
  \item For certain fragments of the theory of signatures, we describe ways to construct
        initial algebras from simpler type formers.
\end{itemize}

In \textbf{Chapter \ref{chap:iqiit}}, we describe infinitary quotient
inductive-inductive signatures. These allow specification of infinitely
branching trees (as initial algebras). We adapt the semantics from the previous
chapter. We also revisit term models, left adjoints of signature morphisms and
self-description of signatures. Self-description in particular is significantly
strengthened, since the full theories of signatures in Chapters
\ref{chap:fqiit}-\ref{chap:iqiit} can be now described using infinitary quotient
inductive-inductive signatures. We also describe how to build semantics of
signatures internally to the theory of signatures itself. For example, this
means that for each signature, algebra morphisms are also specified with a
signature. The full semantics can be internalized in the theory of signatures
in this manner; this is useful for building new signatures in a generic way.

In \textbf{Chapter \ref{chap:hiit}}, we describe higher inductive-inductive
signatures. These differ from the previous signatures mostly in their intended
semantics, whose context is now homotopy type theory \cite{hottbook}, and which
allows specified equalities to be proof-relevant. The higher-dimensional
generalization of types and equalities makes semantics more complicated, so we
only present enough semantics to specify notions of initiality and induction for
each signature. Additionally, we consider two different notions of algebra
morphisms: one preserves structure strictly (up to definitional equality), while
the other preserves structure up to paths.

\section{How to Read This Thesis}

We list several general references which could be helpful for readers.
\begin{itemize}
\item It is useful to have some user experience with a type-theory-based proof
      assistant or programming language, such as Agda, Coq, Lean or Idris. In the
      author's view, mechanized formalization is the most effective way to build
      intuition about working in type theories.
\item We often use categories-with-families \cite{cwfs,Hofmann97syntaxand,Dybjer96internaltype}
      throughout the thesis.
\item We use a modest amount of category theory, for which
      \cite{awodey2010category} should be a sufficient reference.
\item For Chapter \ref{chap:hiit}, the Homotopy Type Theory book \cite{hottbook}
      provides context and motivation.
\end{itemize}
This thesis is mostly written in a linear fashion, as later chapters often
revisit or generalize earlier concepts. There are some breaks from linearity
though, so we summarize dependencies between chapters as follows:
\begin{itemize}
  \item Chapter \ref{chap:2ltt} depends on Chapter \ref{chap:simple-inductive-signatures}.
  \item All chapters after Chapter \ref{chap:2ltt} depend on it.
  \item Chapter \ref{chap:iqiit} depends on Chapter \ref{chap:fqiit} as it
        revisits most constructions from there.
  \item Chapter \ref{chap:hiit} only depends on Chapter \ref{chap:2ltt}.
\end{itemize}

\section{Formalization}

Chapter \ref{chap:simple-inductive-signatures} is fully formalized in Agda, and
the semantics of weak signatures in Chapter \ref{chap:hiit} is mostly formalized,
with some omissions and shortcuts. The formalization can be found in \cite{ak-thesis-agda}.

\section{Notation and Conventions}

Throughout this thesis, we always work in some sort of type theory, although the
exact flavor of the type theory will vary. We summarize here the notations and
conventions that will stay consistent. Our style of notation is a mostly a mix of
the homotopy type theory book \cite{hottbook} and the syntax of the proof
assistant Agda.

\subsubsection{$\Sigma$-types}

We write a dependent pair type as $(a : A) \times B$, where $B$ may refer to
$a$.  Pairing is $(t,\,u)$, and projections are $\proj_1$ and
$\proj_2$. Iterated $\Sigma$-types can written as $(a : A) \times (b : B) \times
C$, for example. We often silently re-associate left-nested $\Sigma$-types to
the right, e.g.\ write $(a : A) \times (b : B) \times C$ instead of
$(\mi{ab} : (a : A) \times B) \times C$.

Field projection notation: we reuse binder names in $\Sigma$-types as field
projections. For example, if we have $t : (\ms{foo} : A) \times B$, then
$\ms{foo}_t$ projects the first component from $t$. To make this a bit more
convenient, we also allow to name the last components, for example if $t :
(\ms{foo} : A) \times (\ms{bar} : B)$, then we have $\ms{foo}_t : A$, and
$\ms{bar}_t : B[\ms{foo} \mapsto \ms{foo}_t]$. This notation is useful when we
handle components of more complicated algebraic structures.

\subsubsection{Unit type}

Whenever the unit type is available, we name it $\top$, and its inhabitant $\tt$.

\subsubsection{$\Pi$-types}
Dependent function types are written as $(a : A) \to B$, where $B$ may depend on
the $a$ variable. It is possible to group multiple binders with the same type,
as in $(x\,y : A) \to B$. For non-dependent function types, we write plainly $A
\to B$.  Functions are defined as $\lambda\,x.\,t$. We may group multiple
binders, as in $\lambda\,x\,y\,z.\,t$, and optionally add type annotation to binders,
as in $\lambda\,(x : A).\,t$.

We also use Agda-like implicit arguments: a function type $\{a : A\} \to B$
signals that we usually omit the argument in function applications. For example,
if $\id : \{A : \Set\} \to A \to A$, we write $\id\,\true : \Bool$.  We can
still make these arguments explicit, by using bracketed application, as in
$\id\,\{\Bool\}\,\true$. Similarly, we may use bracketed $\lambda$, as in
$\lambda\,\{A : \Set\}\,(x : A).\,x$, to bind implicit arguments.

Sometimes we also write pattern matching abstraction, as in $\lambda\,(x,\,y).\,t$
for a function with a $\Sigma$ domain.

We may use implicit quantification as well: argument binders and types may be
entirely omitted when it is clear where they are quantified. This resembles the
implicit generalization in the Haskell or Idris programming languages. For example,
the $A$ and $B$ types are implicitly quantified below:
\begin{alignat*}{3}
  &\ms{map} : (A \to B) \to \ms{List}\,A \to \ms{List}\,B\\
  &\ms{map} \defn ...
\end{alignat*}

\subsubsection{Identity types}

We use $\blank\!\equiv\!\blank$ and $\blank\!=\!\blank$ to denote identity
types. We always use $\blank\!\equiv\!\blank$ as a ``strict'' equality which
satisfies uniqueness of identity proofs. Reflexivity of identity is always
written as $\refl$. We use $\blank\!=\!\blank$ as intensional identity in
Chapter \ref{chap:simple-inductive-signatures}.  In later chapters,
$\blank\!=\!\blank$ denotes the identity type in the inner layer of a two-level
type theory, and $\blank\!\equiv\!\blank$ denotes the outer identity type.

\subsubsection{Definitions}

We give definitions using $\defn$, for example as in
\begin{alignat*}{3}
  &\ms{id} : \{A : \Set\} \to A \to A\\
  &\ms{id}\,a \defn a
\end{alignat*}
Note that we write the function argument on the left of $\defn$, instead of
writing a $\lambda$ on the right. We may switch between the two styles. The
type signature can be omitted in a definition. We may also use pattern matching,
like in $\ms{foo}\,(x,\,y) \defn ...$.

\chapter{Simple Signatures}
\label{chap:simple-inductive-signatures}

In this chapter, we take a look at a very simple notion of algebraic
signature. The motivation for doing so is to present the basic ideas of this
thesis in the easiest possible setting, with explicit definitions. The later
chapters are greatly generalized and expanded compared to the current one, and
are not feasible (and probably not that useful) to present in full formal
detail. We also include a complete Agda formalization of the contents of this
chapter, in around 250 lines.

The mantra throughout this dissertation is the following: algebraic theories are
specified by typing contexts in certain \emph{theories of signatures}. For each
class of algebraic theories, there is a corresponding theory of signatures,
which is viewed as a proper type theory and comes equipped with a model
theory. \emph{Semantics} of signatures is given by interpreting them in certain
models of the theory of signatures. Semantics should at least provide a notion
of induction principle for each signature; in this chapter we provide a bit more
than that, and we will do substantially more in Chapters \ref{chap:fqiit} and
\ref{chap:iqiit}.

\subsubsection{Metatheory}

We work in an intensional type theory which supports $\Pi$, $\Sigma$, $\top$,
intensional identity $\blank\!=\!\blank$, inductive families, and two universes
$\Set$ and $\Set_1$ closed under the mentioned type formers, with $\Set :
\Set_1$.  Since the contents of this chapter are formalized in Agda, and our
notation is reminiscent of Agda too, we can think of the metatheory as a
subset of Agda.

\section{Theory of Signatures}
\label{sec:simple-signatures}

Generally, more expressive theories of signatures can describe larger classes of
theories. As we are aiming for minimalism right now, the current theory of
signatures is as follows:

\begin{mydefinition}
The \textbf{theory of signatures}, or ToS for short, is a simple type theory
equipped with the following features:
  \begin{itemize}
    \item An empty base type $\iota$.
    \item A \emph{first-order function type} $\iota\!\to\!\blank$; this is a
      function whose domain is fixed to be $\iota$. Moreover, first-order functions only
      have neutral terms: there is application, but no $\lambda$-abstraction.
  \end{itemize}
\end{mydefinition}

We can specify the full syntax using the following Agda-like inductive definitions.
\begin{alignat*}{4}
  & \Ty              &&: \Set           && \Var &&: \Con \to \Ty \to \Set \\
  & \iota            &&: \Ty            && \vz  &&: \Var\,(\Gamma \ext A)\,A \\
  & \iota\!\to\blank &&: \Ty \to \Ty    && \vs  &&: \Var\,\Gamma\,A \to \Var\,(\Gamma \ext B)\,A\\
  & && && &&\\
  & \Con             &&: \Set           && \Tm  &&: \Con \to \Ty \to \Set \\
  & \emptycon        &&: \Con           && \var &&: \Var\,\Gamma\,A \to \Tm\,\Gamma\,A \\
  & \blank\ext\blank &&: \Con \to \Ty \to \Con \hspace{2em} && \app &&: \Tm\,\Gamma\,(\iota\to A) \to \Tm\,\Gamma\,\iota
                                                           \to \Tm\,\Gamma\,A
\end{alignat*}
Here, $\Con$ contexts are lists of types, and $\Var$ specifies well-typed De Bruijn indices, where
$\vz$ represents the zero index, and $\vs$ takes the successor of an index.

\begin{notation} We use capital Greek letters starting from $\Gamma$ to refer to contexts, $A$, $B$, $C$ to
refer to types, and $t$, $u$, $v$ to refer to terms. In examples, we may use a
nameful notation instead of De Bruijn indices. For example, we may write $x :
\Tm\,(\emptycon \ext (x : \iota) \ext (y : \iota))\,\iota$ instead of $\var\,(\vs\,\vz)
: \Tm\,(\emptycon \ext \iota \ext \iota)\,\iota$. Additionally, we may write
$t\,u$ instead of $\app\,t\,u$ for $t$ and $u$ terms.
\end{notation}

\begin{mydefinition} \textbf{Parallel substitutions} map variables to terms.
\begin{alignat*}{3}
&\Sub : \Con \to \Con \to \Set\\
&\Sub\,\Gamma\,\Delta \equiv \{A : \Ty\} \to \Var\,\Delta\,A \to \Tm\,\Gamma\,A
\end{alignat*}
We use $\sigma$ and $\delta$ to refer to substitutions. We also recursively
define the action of substitution on terms:
\begin{alignat*}{3}
  &\rlap{$\blank[\blank] : \Tm\,\Delta\,A \to \Sub\,\Gamma\,\Delta \to \Tm\,\Gamma\,A$}\\
  &(\var\, x)   &&[ \sigma ] \defn \sigma\,x\\
  &(\app\,t\,u) &&[ \sigma ] \defn \app\,(t[\sigma])\,(u[\sigma])
\end{alignat*}
The identity substitution $\id$ is defined simply as $\var$. It is easy to see that
$t[\id] = t$ for all $t$. Substitution composition is as follows.
\begin{alignat*}{3}
  &\blank\!\circ\!\blank : \Sub\,\Delta\,\Xi \to \Sub\,\Gamma\,\Delta \to \Sub\,\Gamma\,\Xi\\
  &(\sigma \circ \delta)\,x \defn (\sigma\,x)[\delta]
\end{alignat*}
\end{mydefinition}

\begin{myexample} We may write signatures for natural numbers and binary trees respectively as follows.
\begin{alignat*}{3}
  & \ms{NatSig}  &&\defn \emptycon \ext (\mi{zero} : \iota) \ext (\mi{suc} : \iota \to \iota)\\
  & \ms{TreeSig} &&\defn \emptycon \ext (\mi{leaf} : \iota) \ext (\mi{node} : \iota \to \iota \to \iota)
\end{alignat*}
\end{myexample}
In short, the current ToS allows signatures which are
\begin{itemize}
\item \emph{Single-sorted}: this means that we have a single type constructor, corresponding to $\iota$.
\item \emph{Closed}: signatures cannot refer to any externally existing type. For example, we cannot write a signature for lists of natural numbers in a direct fashion, since there is no way to refer to the type of natural numbers.
\item \emph{Finitary}: inductive types corresponding to signatures are always
  finitely branching trees.
\end{itemize}

\emph{Remark.} We omit $\lambda$-expressions from the ToS for the sake of
simplicity: this causes terms to be always in normal form (neutral, to be
precise), and thus we can skip talking about conversion rules. Later, starting
from Chapter \ref{chap:fqiit} we include proper $\beta\eta$-rules in theories of
signatures.

\section{Semantics}
\label{sec:simple-semantics}

For each signature, we need to know what it means for a type theory to support
the corresponding inductive type. For this, we need at least a notion of
\emph{algebras}, which can be viewed as a bundle of all type and
value constructors, and what it means for an algebra to support an
\emph{induction principle}.  Additionally, we may want to know what it means to
support a \emph{recursion principle}, which can be viewed as a non-dependent
variant of induction. In the following, we define these notions by induction on
ToS syntax.

\emph{Remark.} We use ``algebra'' and ``model'' synonymously throughout
this thesis.

\subsection{Algebras}

First, we calculate types of algebras. This is simply a standard interpretation
into the $\Set$ universe. We define the following operations by induction; the
$\blank^A$ name is overloaded for $\Con$, $\Ty$ and $\Tm$.
\begin{alignat*}{3}
& \hspace{-4em} \rlap{$\blank^A : \Ty \to \Set \to \Set$} \\
& \hspace{-4em} \iota^A\,&&X \defn X \\
& \hspace{-4em} (\iota\to A)^A\,&&X \defn X \to A^A\,X\\
& \hspace{-4em} && \\
& \hspace{-4em} \rlap{$\blank^A : \Con \to \Set \to \Set$}\\
& \hspace{-4em} \rlap{$\Gamma^A\,X \defn \{A : \Ty\} \to \Var\,\Gamma\,A \to A^A\,X$}\\
& \hspace{-4em} && \\
& \hspace{-4em} \rlap{$\blank^A : \Tm\,\Gamma\,A \to \{X : \Set\} \to \Gamma^A\,X \to A^A\,X$}\\
& \hspace{-4em} (\var\,x)^A\,&&\gamma \defn \gamma\,x\\
& \hspace{-4em} (\app\,t\,u)^A\,&&\gamma \defn t^A\,\gamma\,(u^A\,\gamma)\\
& \hspace{-4em} && \\
& \hspace{-4em} \rlap{$\blank^A : \Sub\,\Gamma\,\Delta \to \{X : \Set\} \to \Gamma^A\,X \to \Delta^A\,X$}\\
& \hspace{-4em} \rlap{$\sigma^A\,\gamma\,x \defn (\sigma\,x)^A\,\gamma$}
\end{alignat*}
Here, types and contexts depend on some $X : \Set$, which serves as the
interpretation of $\iota$. We define $\Gamma^A$ as a product: for each variable
in the context, we get a semantic type. This trick, along with the definition of
$\Sub$, makes formalization a bit more compact. Terms and substitutions are
interpreted as natural maps. Substitutions are interpreted by pointwise interpreting
the contained terms.

\begin{notation}
We may write values of $\Gamma^A$ using notation for $\Sigma$-types. For
example, we may write $(\mi{zero} : X) \times (\mi{suc} : X \to X)$ for the
result of computing $\ms{NatSig}^A\,X$.
\end{notation}

\begin{mydefinition} We define \textbf{algebras} as follows.
\begin{alignat*}{3}
  & \Alg : \Con \to \Set_1 \\
  & \Alg\,\Gamma \defn (X : \Set) \times \Gamma^A\,X
\end{alignat*}
\end{mydefinition}

\begin{myexample} $\Alg\,\ms{NatSig}$ is computed to $(X : \Set)\times(\mi{zero} :
X)\times(\mi{suc} : X \to X)$.
\end{myexample}

\subsection{Morphisms}

Now, we compute notions of morphisms of algebras. In this case, morphisms are
functions between underlying sets which preserve all specified structure. The
interpretation for calculating morphisms is a \emph{logical relation
interpretation} \cite{udayReynolds} over the $\blank^A$ interpretation. The key
part is the interpretation of types:
\begin{alignat*}{3}
  & \hspace{-4em}\rlap{$\blank^M : (A : \Ty)\{X_0\,X_1 : \Set\}(X^M : X_0 \to X_1) \to A^A\,X_0 \to A^A\,X_1 \to \Set$}\\
  & \hspace{-4em}\iota^M\,&&X^M\,\alpha_0\,\,\alpha_1 \defn X^M\,\alpha_0 = \alpha_1 \\
  & \hspace{-4em}(\iota\to A)^M\,&&X^M\,\alpha_0\,\,\alpha_1 \defn
       (x : X_0) \to A^M\,X^M\,(\alpha_0\,x)\,(\alpha_1\,(X^M\,x))
\end{alignat*}
We again assume an interpretation for the base type $\iota$, as $X_0$, $X_1$ and
$X^M : X_0 \to X_1$. $X^M$ is function between underlying sets of algebras, and
$A^M$ computes what it means that $X^M$ preserves an operation with type $A$. At
the base type, preservation is simply equality. At the first-order function
type, preservation is a quantified statement over $X_0$. We define morphisms for
$\Con$ pointwise:
\begin{alignat*}{3}
  &\blank^M : (\Gamma : \Con)\{X_0\,X_1 : \Set\} \to (X_0 \to X_1) \to \Gamma^A\,X_0 \to \Gamma^A\,X_1 \to \Set\\
  &\Gamma^M\,X^M\,\gamma_0\,\gamma_1 \defn
    \{A : \Ty\}(x : \Var\,\Gamma\,A) \to A^M\,X^M\,(\gamma_0\,x)\,(\gamma_1\,x)
\end{alignat*}
For terms and substitutions, we get preservation statements, which are sometimes
called \emph{fundamental lemmas} in discussions of logical relations \cite{udayReynolds}.
\begin{alignat*}{3}
  & \hspace{-10em}\rlap{$\blank^M : (t : \Tm\,\Gamma\,A) \to \Gamma^M\,X^M\,\gamma_0\,\gamma_1 \to A^M\,X^M\,(t^A\,\gamma_0)\,(t^A\,\gamma_1)$}\\
  & \hspace{-10em}(\var\,x)^M    &&\gamma^M \defn \gamma^M\,x \\
  & \hspace{-10em}(\app\,t\,u)^M &&\gamma^M \defn t^M\,\gamma^M\,(u^A\,\gamma_0)
\end{alignat*}
\begin{alignat*}{3}
  & \blank^M : (\sigma : \Sub\,\Gamma\,\Delta) \to \Gamma^M\,X^M\,\gamma_0\,\gamma_1 \to \Delta^M\,X^M\,(\sigma^A\,\gamma_0)\,(\sigma^A\,\gamma_1)\\
  & \sigma^M\, \gamma^M\,x \defn (\sigma\,x)^M\,\gamma^M
\end{alignat*}
The definition of $(\app\,t\,u)^M$ is well-typed by the induction hypothesis
$u^M\,\gamma^M : X^M\,(u^A\,\gamma_0) = u^A\,\gamma_1$.

\begin{mydefinition}
\label{def:simple-morphism}
To get notions of \textbf{algebra morphisms}, we again pack up $\Gamma^M$ with
the interpretation of $\iota$.
\begin{alignat*}{3}
  & \Mor : \{\Gamma : \Con\} \to \Alg\,\Gamma \to \Alg\,\Gamma \to \Set \\
  & \Mor\,\{\Gamma\}\,(X_0,\,\gamma_0)\,(X_1,\,\gamma_1) \defn (X^M : X_0 \to X_1) \times \Gamma^M\,X^M\,\gamma_0\,\gamma_1
\end{alignat*}
\end{mydefinition}
\begin{myexample} We have the following computation:
\begin{alignat*}{3}
  & \hspace{-5em}\rlap{$\Mor\,\{\NatSig\}\,(X_0,\,\mi{zero_0},\,\mi{suc_0})\,(X_1,\,\mi{zero_1},\,\mi{suc_1}) \defn$} \\
           &(X^M : X_0 \to X_1) \\
   \times\,&(X^M\,\mi{zero_0} = \mi{zero_1}) \\
   \times\,&((x : X_0) \to X^M\,(\mi{suc_0}\,x) = \mi{suc_1}\,(X^M\,x))
\end{alignat*}
\end{myexample}

\begin{mydefinition} We state \textbf{initiality} as a predicate on algebras:
\begin{alignat*}{3}
  & \Initial : \{\Gamma : \Con\} \to \Alg\,\Gamma \to \Set\\
  & \Initial\,\{\Gamma\}\,\gamma \defn
    (\gamma' : \Alg\,\Gamma) \to \ms{isContr}\,(\Mor\,\gamma\,\gamma')
\end{alignat*}
Here $\ms{isContr}$ refers to unique existence \cite[Section 3.11]{hottbook}. If we drop
$\ms{isContr}$ from the definition, we get the notion of weak initiality, which
corresponds to the recursion principle for $\Gamma$. Although we call this
predicate $\Initial$, in this chapter we do not yet show that algebras form a
category. We will show this in a more general setting in Chapter \ref{chap:fqiit}.

\end{mydefinition}

\paragraph{Morphisms vs.\ logical relations.}
The $\blank^M$ interpretation can be viewed as a special case of logical
relations over the $\blank^A$ model: every morphism is a \emph{functional}
logical relation, where the chosen relation between the underlying sets happens
to be a function. Consider now a more general relational interpretation for
types:
\begin{alignat*}{3}
  & \hspace{-0.5em}\rlap{$\blank^R : (A : \Ty)\{X_0\,X_1 : \Set\}(X^R : X_0 \to X_1 \to \Set) \to A^A\,X_0 \to A^A\,X_1 \to \Set$}\\
  & \hspace{-0.5em}\iota^R\,&&X^R\,\alpha_0\,\,\alpha_1 \defn X^R\,\alpha_0\,\alpha_1 \\
  & \hspace{-0.5em}(\iota\to A)^R\,&&X^R\,\alpha_0\,\,\alpha_1 \defn
       (x_0 : X_0)(x_1 : X_1) \to X^R\,x_0\,x_1 \to A^R\,X^R\,(\alpha_0\,x_0)\,(\alpha_1\,x_1)
\end{alignat*}
Here, functions are related if they map related inputs to related outputs. If we
know that $X^M\,\alpha_0\,\alpha_1 \equiv (f\,\alpha_0 = \alpha_1)$ for some
function $f$, we get
\[
  (x_0 : X_0)(x_1 : X_1) \to f\,x_0 = x_1 \to A^R\,X^R\,(\alpha_0\,x_0)\,(\alpha_1\,x_1)
\]
Now, we can simply substitute along the input equality proof in the above type,
to get the previous definition for $(\iota \to A)^M$:
\[
  (x_0 : X_0) \to A^R\,X^R\,(\alpha_0\,x_0)\,(\alpha_1\,(f\,x_0))
\]
This substitution along the equation is called ``singleton contraction'' in the
jargon of homotopy type theory \cite{hottbook}. The ability to perform contraction
here is at the heart of the \emph{strict positivity restriction} for inductive
signatures. Strict positivity in our setting corresponds to only having
first-order function types in signatures. If we allowed function domains to be
arbitrary types, in the definition of $(A \to B)^M$ we would only have a
black-box $A^M\,X^M : A^A\,X_0 \to A^A\,X_1 \to \Set$ relation, which is not
known to be given as an equality.

In Chapter \ref{chap:fqiit} we expand on this. As a preliminary summary:
although higher-order functions have relational interpretation, such relations
do not generally compose. What we eventually aim to have is a category of
algebras and algebra morphisms, where morphisms do compose. We need a
\emph{directed} model of the theory of signatures, where every signature becomes
a category of algebras. The way to achieve this is to prohibit higher-order
functions, thereby avoiding the polarity issues that prevent a directed
interpretation for general function types.

\subsection{Displayed Algebras}

At this point we do not yet have specification for induction principles. We use
the term \emph{displayed algebra} to refer to ``dependent'' algebras, where
every displayed algebra component lies over corresponding components in the base
algebra. For the purpose of specifying induction, displayed algebras can be
viewed as bundles of induction motives and methods.

Displayed algebras over some $\gamma : \Alg\,\Gamma$ are equivalent to slices
over $\gamma$ in the category of $\Gamma$-algebras; we will show this in Chapter
\ref{chap:fqiit}. A slice $f : \Gamma^M\,\gamma'\,\gamma$ maps elements of
$\gamma'$'s underlying set to elements in the base algebra. Why do we need
displayed algebras, then? The main reason is that if we are to eventually
implement inductive types in a programming language or proof assistant, we need
to compute induction principles exactly, not merely up to isomorphisms.

For more illustration of using displayed algebras in a type-theoretic
setting, see \cite{displayedcats}. We adapt the term ``displayed algebra'' from
ibid.\ as a generalization of displayed categories (and functors, natural
transformations) to other algebraic structures.

The displayed algebra interpretation is a \emph{logical predicate}
interpretation, defined as follows.
\begin{alignat*}{3}
  & \rlap{$\blank^D : (A : \Ty)\{X\} \to (X \to \Set) \to A^A\,X \to \Set$}\\
  & \iota^D\,       && X^D\,\alpha \defn X^D\,\alpha \\
  & (\iota\to A)^D\,&& X^D\,\alpha \defn (x : X)(x^D : X^D\,x) \to A^D\,X^D\,(\alpha\,x)\\
  & &&\\
  & \rlap{$\blank^D : (\Gamma : \Con)\{X\} \to (X \to \Set) \to \Gamma^A\,X \to \Set$}\\
  & \rlap{$\Gamma^D\,X^D\,\gamma \defn
       \{A : \Ty\}(x : \Var\,\Gamma\,A) \to A^D\,X^D\,(\gamma\,x)$}\\
  & &&\\
  & \rlap{$\blank^D : (t : \Tm\,\Gamma\,A)
      \to \Gamma^D\,X^D\,\gamma \to A^D\,X^D\,(t^A\,\gamma)$}\\
  & (\var\,x)^D\,&&\gamma^D \defn \gamma^D\,x\\
  & (\app\,t\,u)^D\,&&\gamma^D \defn t^D\,\gamma^D\,(u^A\,\gamma)\,(u^D\,\gamma^D)\\
  & &&\\
  & \rlap{$\blank^D : (\sigma : \Sub\,\Gamma\,\Delta)
      \to \Gamma^D\,X^D\,\gamma \to \Delta^D\,X^D\,(\sigma^A\,\gamma)$}\\
  & \rlap{$\sigma^D\,\gamma^D\,x \defn (\sigma\,x)^D\,\gamma^D$}
\end{alignat*}
Analogously to before, everything depends on a predicate interpretation $X^D : X
\to \Set$ for $\iota$. For types, a predicate holds for a function if the
function preserves predicates. The interpretation of terms is again a
fundamental lemma, and we again have pointwise definitions for contexts and
substitutions.
\begin{mydefinition}[\textbf{displayed algebras}]\label{def:simple-section}
\begin{alignat*}{3}
  & \DispAlg : \{\Gamma : \Con\} \to \Alg\,\Gamma \to \Set_1\\
  & \DispAlg\,\{\Gamma\}\,(X,\,\gamma) \defn (X^D : X \to \Set) \times \Gamma^D\,X^D\,\gamma
\end{alignat*}
\end{mydefinition}
\begin{myexample} We have the following computation.
\begin{alignat*}{3}
  & \hspace{-5em}\rlap{$\DispAlg\,\{\ms{NatSig}\}\,(X,\,\mi{zero},\,\mi{suc}) \equiv$}\\
              & (X^D &&: X \to \Set)\\
      \times\,& (\mi{zero^D} &&: X^D\,\mi{zero})\\
      \times\,& (\mi{suc^D} &&: (n : X) \to X^D\,n \to X^D\,(\mi{suc}\,n))
\end{alignat*}
\end{myexample}

\subsection{Sections}

Sections of displayed algebras are ``dependent'' analogues of algebra morphisms,
where the codomain is displayed over the domain.

\begin{alignat*}{3}
  & \hspace{-6em}\rlap{$\blank^S : (A : \Ty)\{X\,X^D\}(X^S : (x : X) \to X^D\,x) \to (\alpha : A^A\,X) \to A^D\,X^D\,\alpha \to \Set$}\\
  & \hspace{-6em}\iota^S\,&&X^S\,\alpha\,\,\alpha^D \defn X^S\,\alpha = \alpha^D \\
  & \hspace{-6em}(\iota\to A)^S\,&&X^S\,\alpha\,\,\alpha^D \defn
  (x : X) \to A^S\,X^S\,(\alpha\,x)\,(\alpha^D\,(X^S\,x))\\
  & \hspace{-6em}&&\\
  &\hspace{-6em}\rlap{$\Con^S : (\Gamma : \Con)\{X\,X^D\}(X^S : (x : X) \to X^D\,x) \to (\gamma : \Gamma^A\,X) \to \Gamma^D\,X^D\,\gamma \to \Set$}\\
  &\hspace{-6em}\rlap{$\Gamma^S\,X^S\,\gamma_0\,\gamma_1 \defn
    \{A : \Ty\}(x : \Var\,\Gamma\,A) \to A^S\,X^S\,(\gamma_0\,x)\,(\gamma_1\,x)$}\\
  & \hspace{-6em} && \\
  & \hspace{-6em}\rlap{$\blank^S : (t : \Tm\,\Gamma\,A) \to \Gamma^S\,X^S\,\gamma\,\gamma^D \to A^S\,X^S\,(t^A\,\gamma)\,(t^D\,\gamma^D)$}\\
  & \hspace{-6em}(\var\,x)^S    &&\gamma^S \defn \gamma^S\,x \\
  & \hspace{-6em}(\app\,t\,u)^S &&\gamma^S \defn t^S\,\gamma^S\,(u^A\,\gamma)\\
  & \hspace{-6em}&& \\
  & \hspace{-6em}\rlap{$\blank^S : (\sigma : \Sub\,\Gamma\,\Delta) \to \Gamma^S\,X^S\,\gamma\,\gamma^D \to \Delta^S\,X^S\,(\sigma^A\,\gamma)\,(\sigma^A\,\gamma^D)$}\\
  & \hspace{-6em} \rlap{$\sigma^S\, \gamma^S\,x = (\sigma\,x)^S\,\gamma^S$}
\end{alignat*}

\begin{mydefinition}[\textbf{Displayed algebra sections} (``sections'' in short)]
\begin{alignat*}{3}
  & \Section : \{\Gamma : \Con\} \to (\gamma : \Alg\,\Gamma) \to \DispAlg\,\gamma \to \Set\\
  & \Section\,(X,\,\gamma)\,(X^D\,\gamma^D) \defn (X^S : (x : X) \to X^D\,x) \times \Gamma^S\,X^S\,\gamma\,\gamma^D
\end{alignat*}
\end{mydefinition}
\begin{myexample} We have the following computation.
\begin{alignat*}{3}
  & \hspace{-5em}\rlap{$\Section\,\{\ms{NatSig}\}\,(X,\,\mi{zero},\,\mi{suc})\,(X^D,\,\mi{zero^D},\,\mi{suc^D}) \equiv$}\\
              & (X^S &&: (x : X) \to X^D\,x)\\
      \times\,& (\mi{zero^S} &&: X^S\,\mi{zero} = \mi{zero^D})\\
      \times\,& (\mi{suc^S} &&: (n : X) \to X^S\,(\mi{suc\,n}) = \mi{suc^D}\,n\,(X^S\,n))
\end{alignat*}
\end{myexample}

\begin{mydefinition}[\textbf{Induction}]
We define a predicate which holds if an algebra supports induction.
\begin{alignat*}{3}
  & \Inductive : \{\Gamma : \Con\} \to \Alg\,\Gamma \to \Set_1\\
  & \Inductive\,\{\Gamma\}\,\gamma \defn
     (\gamma^D : \DispAlg\,\gamma) \to \Section\,\gamma\,\gamma^D
\end{alignat*}
\end{mydefinition}

We can observe that $\Inductive\,\{\ms{NatSig}\}\,(X,\,\ms{zero},\,\ms{suc})$
computes to the usual induction principle for natural numbers, but with
$\beta$-rules given as propositional equalities. The input $\DispAlg$ is a
bundle of the induction motive and the methods, and the output $\Section$
contains the $X^S$ eliminator function together with its $\beta$-rules.

\section{Term Algebras}

In this section we show that if a type theory supports the inductive types comprising
the theory of signatures, it also supports every inductive type which is described
by the signatures.

Note that we specified $\Tm$ and $\Sub$, but did not need either of them when
specifying signatures, or when computing induction principles. That signatures
do not depend on terms is a property specific to simple signatures; this will
not be the case in Chapter \ref{chap:fqiit} when we move to more general
signatures. However, terms and substitutions are already required in the
construction of term algebras.

The idea is that terms in contexts comprise initial algebras. For example,
$\Tm\,\ms{NatSig}\,\iota$ is the set of natural numbers (up to
isomorphism). Informally, this is because the only way to construct terms is by
applying the $\ms{suc}$ variable (given by $\var\,\vz$) finitely many times to
the $\ms{zero}$ variable (given by $\var\,(\vs\,\vz)$).

\begingroup
\allowdisplaybreaks
\begin{mydefinition}[\textbf{Term algebras}]
Fix an $\Omega : \Con$. We abbreviate $\Tm\,\Omega\,\iota$ as $\ms{T}$; this will serve
as the carrier set of the term algebra. We additionally define the following.
\begin{alignat*}{3}
  & \hspace{-5em}\rlap{$\blank^T : (A : \Ty) \to \Tm\,\Omega\,A \to A^A\,\ms{T}$} \\
  & \hspace{-5em}\iota^T\,&&t \defn t \\
  & \hspace{-5em}(\iota\to A)^T\,&&t \defn \lambda\,u.\,A^T\,(\app\,t\,u)\\
  & \hspace{-5em}&& \\
  & \hspace{-5em}\rlap{$\blank^T : (\Gamma : \Con) \to \Sub\,\Omega\,\Gamma \to \Gamma^A\,\ms{T}$}\\
  & \hspace{-5em}\rlap{$\Gamma^T\,\nu\,\{A\}\,x \defn A^T\,(\nu\,x)$}\\
  & \hspace{-5em}&& \\
  & \hspace{-5em}\rlap{$\blank^T : (t : \Tm\,\Gamma\,A)(\nu : \Sub\,\Omega\,\Gamma) \to A^T\,(t[\nu]) = t^A\,(\Gamma^T\,\nu)$}\\
  & \hspace{-5em}(\var\,x)^T\,   &&\nu    \text{   holds by   } \refl\\
  & \hspace{-5em}(\app\,t\,u)^T\,&&\nu \text {   holds by   } t^T\,\nu \text{   and   } u^T\,\nu\\
  & \hspace{-5em}&& \\
  & \hspace{-5em}\rlap{$\blank^T : (\sigma : \Sub\,\Gamma\,\Delta)(\nu : \Sub\,\Omega\,\Gamma)\{A\}(x : \Var\,\Delta\,A)$}\\
  & \hspace{-3em}\rlap{$\to \Delta^T\,(\sigma \circ \nu)\,x = \sigma^A\,(\Gamma^T\,\nu)\,x$}\\
  & \hspace{-5em}\rlap{$ \sigma^T\,\nu\,x \defn (\sigma\,x)^T\,\nu$}
\end{alignat*}
Now we can define the term algebra for $\Omega$ itself:
\begin{alignat*}{3}
  & \TmAlg_{\Omega} : \Alg\,\Omega \\
  & \TmAlg_{\Omega} \defn \Omega^T\,\Omega\,\id
\end{alignat*}
\end{mydefinition}
\endgroup

In the interpretation for contexts, it is important that $\Omega$ is
fixed, and we do induction on all $\Gamma$ contexts such that there is a
$\Sub\,\Omega\,\Gamma$. It would not work to try to compute term algebras by
direct induction on contexts because we need to refer to the same $\ms{T}$ set
in the interpretation of every type in a signature.

The interpretation of types embeds terms as $A$-algebras. At the base type
$\iota$, this embedding is simply the identity function, since $\iota^A\,\ms{T}
\equiv \ms{T} \equiv \Tm\,\Omega\,\iota$. At function types we recursively proceed
under a semantic $\lambda$. The interpretation of contexts is pointwise.

The interpretations of terms and substitutions are coherence properties, which
relate the term algebra construction to term evaluation in the $\blank^A$ model.
For terms, if we pick $\nu \equiv \id$, we get $A^T\,t =
t^A\,\TmAlg_{\Omega}$. The left side embeds $t$ in the term model via
$\blank^T$, while the right hand side evaluates $t$ in the term model.

One way to view the term algebra construction, is that we are working in a
\emph{slice model} over the fixed $\Omega$, and every $\nu :
\Sub\,\Omega\,\Gamma$ can be viewed as an internal $\Gamma$-algebra in this
model. The term algebra construction demonstrates that every such internal
algebra yields an external element of $\Gamma^A$.

%% We will
%% see in Section \ref{sec:fqiit-term-algebras} that we can construct term algebras
%% from \emph{any} model of a ToS, not just the ToS syntax; but while term algebras
%% constructed from ToS syntax are themselves initial algebras, in other cases they
%% may not be initial.

\subsection{Recursor Construction}
\label{sec:simple-weak-initiality}
We show that $\TmAlg_{\Omega}$ supports a recursion principle, i.e.\ it is weakly
initial.

\begin{mydefinition}[\textbf{Recursor construction}]\label{def:simple-recursor} We assume $(X,\,\omega) : \Alg\,\Omega$;
recall that $X : \Set$ and $\omega : \Omega^A\,X$. We define $\ms{R} : \ms{T} \to X$
as $\ms{R}\,t \defn t^A\,\omega$. We additionally define the following.
\begin{alignat*}{3}
& \hspace{-6em}\rlap{$\blank^R : (A : \Ty)(t : \Tm\,\Omega\,A) \to A^M\,\ms{R}\,(A^T\,t)\,(t^A\,\omega)$}\\
& \hspace{-6em}\iota^R\,&&t \defn (\refl : t^A\,\omega = t^A\,\omega)\\
& \hspace{-6em}(\iota\to A)^R\,&&t \defn \lambda\,u.\,A^R\,(\app\,t\,u)
\end{alignat*}
\begin{alignat*}{3}
& \hspace{-10em}\rlap{$\blank^R : (\Gamma : \Con)(\nu : \Sub\,\Omega\,\Gamma) \to \Gamma^M\,\ms{R}\,(\Gamma^T\,\nu)\,(\nu^A\,\omega)$}\\
& \hspace{-10em}\rlap{$\Gamma^R\,\nu\,x \defn A^R\,(\nu\,x)$}
\end{alignat*}
We define the recursor for $\Omega$ as
\begin{alignat*}{3}
  & \Rec_{\Omega} : (\mi{alg} : \Alg\,\Omega) \to \Mor\,\TmAlg_{\Omega}\,\mi{alg}\\
  & \Rec_{\Omega}\,(X,\,\omega) \defn (\ms{R},\,\Omega^R\,\Omega\,\id)
\end{alignat*}
\end{mydefinition}

In short, the way we get recursion is by evaluating terms in arbitrary
$(X,\,\omega)$ algebras using $\blank^A$. The $\blank^R$ operation for types and
contexts confirms that $\ms{R}$ preserves structure appropriately, so that
$\ms{R}$ indeed yields algebra morphisms.

We skip interpreting terms and substitutions by $\blank^R$. It is necessary to
do so with more general signatures, but not in the current chapter.

\subsection{Eliminator Construction}

We take the idea of the previous section a bit further. We have seen that
recursion for term algebras is given by evaluation in the ``standard'' $\Set$
model. Now, we show that induction for term algebras is obtained from the
$\blank^D$ interpretation into the logical predicate model over the $\Set$
model.

\begin{mydefinition}[\textbf{Eliminator construction}]\label{def:simple-eliminator-construction}
We assume $(X^D,\,\omega^D) : \DispAlg\,\TmAlg_\Omega$. Recall that $X^D :
\ms{T} \to \Set$ and $\omega^D : \Omega^D\,X^D\,(\Omega^T\,\Omega\,\id)$. Like
before, we first interpret the underlying set:
\begin{alignat*}{3}
  & \ms{E} : (t : \ms{T}) \to X^D\,t \\
  & \ms{E}\,t \defn t^D\,\omega^D
\end{alignat*}
However, this definition is not immediately well-typed, since $t^D\,\omega^D$
has type $X^D\,(t^A\,(\Omega^T\,\Omega\,\id))$, so we have to show that
$t^A\,(\Omega^T\,\Omega\,\id) = t$. This equation says that nothing happens if
we evaluate a term with type $\iota$ in the term model. We get it from the
$\blank^T$ interpretation of terms: $t^T\,\id : t[\id] =
t^A\,(\Omega^T\,\Omega\,\id)$, and we also know that $t[\id] = t$. We interpret types
and contexts as well:
\begin{alignat*}{3}
  & \hspace{-6em}\rlap{$\blank^E : (A : \Ty)(t : \Tm\,\Omega\,A) \to A^S\,\ms{E}\,(t^A\,(\Omega^T\,\Omega\,\id))\,(t^D\,\omega^D)$}\\
  & \hspace{-6em}\iota^E\,&&t : (t^A\,(\Omega^T\,\Omega\,\id))^D\,\omega^D = t^D\,\omega^D\\
  & \hspace{-6em}(\iota\to A)^E\,&&t \defn \lambda\,u.\, A^E\,(\app\,t\,u)
\end{alignat*}
\begin{alignat*}{3}
  & \blank^E : (\Gamma : \Con)(\nu : \Sub\,\Omega\,\Gamma) \to \Gamma^S\,\ms{E}\,(\nu^A\,(\Omega^T\,\Omega\,\id))\, (\nu^D\,\omega^D)\\
  & \Gamma^E\,\nu\,x \defn A^E\,(\nu\,x)
\end{alignat*}
In $\iota^E$ we use the same equation as in the definition of $\ms{E}$. In
$(\iota\to A)^E$ the definition is well-typed because of the same equation, but
instantiated for the abstracted $u$ term this time. All of this amounts to some
additional path induction and transport fiddling in the (intensional) Agda
formalization. We get induction for $\Omega$ as below.
\begin{alignat*}{3}
  &\Ind_{\Omega} : (\mi{alg} : \DispAlg\,\TmAlg_\Omega) \to \Section\,\TmAlg_\Omega\,\mi{alg}\\
  &\Ind_{\Omega}\,(X^D,\,\omega^D) \defn (E,\, \Omega^E\,\Omega\,\id)
\end{alignat*}
\end{mydefinition}

%% \section{Discussion \& Related Work}

\section{Comparison to Endofunctors as Signatures}

A well-known alternative definition of algebraic signatures is to view certain
cocontinuous endofunctors as such. For example, single-sorted signatures can be
defined to be endofunctors which preserve colimits of some ordinal-indexed
chains. For instance, if we have a $\kappa$-cocontinuous $F : \mbb{C} \to
\mbb{C}$, then algebras are given as $(X : |\mbb{C}|) \times
(\mbb{C}(F\,X,\,X))$, called \emph{F-algebras}, morphisms as commuting squares,
and Adámek's theorem \cite{adamek} establishes the existence of initial
algebras.

An advantage of this approach is that we can describe different classes of
signatures by choosing different $\mbb{C}$ categories:
\begin{itemize}
  \item If $\mbb{C}$ is $\mbf{Set}$, we get simple inductive theories.
  \item If $\mbb{C}$ is $\mbf{Set}^I$ for some set $I$, we get indexed inductive signatures.
  \item If $\mbb{C}$ is $\mbf{Set}/I$, we get inductive-recursive signatures.
\end{itemize}

Another advantage is that signatures are fairly semantic in nature: they make
sense even if we have no syntactic presentation at hand. That said, often we do
need syntactic signatures, for use in proof assistants, or just to have a
convenient notation for a class of cocontinuous functors.

An elegant way of carving out a large class of such functors is to
consider polynomials as signatures. For example, when working in \textbf{Set}, a
signature is an element of $(S : \Set) \times (P : S \to \Set)$, and $(S,\,P)$
is interpreted as a functor as $X \mapsto (s : S) \times (P\,s \to X)$. The
initial algebra is the W-type specified by $S$ shapes and $P$ positions. This
yields infinitary inductive types as well.

However, it is not known how to get \emph{inductive-inductive} signatures by
picking the right $\mbb{C}$ category and a functor. In an inductive-inductive
signature, there may be multiple sorts, which can be indexed over previously
declared sorts. For example, in the signature for categories we have $\Obj :
\Set$ and $\Mor : \Obj \to \Obj \to \Set$, indexed twice over $\Obj$. Some
extensions are required to the idea of $F$-algebras:
\begin{itemize}
\item
  For inductive-inductive definitions with two sorts, Forsberg gives a
  specification with two functors, and a considerably more complex notion of
  algebras, involving dialgebras \cite{forsberg-phd}.\footnote{However, the
  dialgebra specification only covers restricted signatures, where $B : A \to
  \Set$ constructor types may refer to $A : \Set$ constructors, but no other
  dependency is allowed.  There is a more general and yet more complicated
  notion of signature in \cite{forsberg-phd}, which is not anymore represented
  with functors.}
\item
  For an arbitrary number of sorts, Altenkirch et
  al.\ \cite{altenkirch18qiit} use a ``list'' of functors, specified mutually
  with categories of algebras: each functor has as domain the semantic category
  of all previous sorts.
\end{itemize}

The functors-as-signatures approach gets significantly less convenient as we
consider more general specifications. The approach of this thesis is to skip the
middle ground between syntactic signatures and semantic categories of algebras:
we treat syntactic signatures as a key component, and give direct semantic
interpretation for them. Although we lose the semantic nature of signatures,
our approach scales extremely well, all the way up to infinitary
quotient-inductive-inductive types in Chapter \ref{chap:iqiit}, and to some
extent to higher inductive-inductive types as well in Chapter \ref{chap:hiit}.

If we look back at $\blank^A : \Con \to \Set \to \Set$, we may note that
$\Gamma^A$ yields a functor, in fact the same functor (up to isomorphism) that
we would get from an endofunctor presentation. However, this is a coincidence in
the single-sorted case. We can view $(X : |\mbb{C}|) \times (\mbb{C}(F\,X,\,X))$
as specifying the category of algebras as the total category of a displayed
category (by viewing the $\Sigma$-type here as taking total categories; a
$\Sigma$ in $\mbf{Cat}$). In our approach, we aim to get the displayed
categories directly, without talking about functors.

\chapter{Semantics in Two-Level Type Theory}
\label{chap:2ltt}

In this chapter we describe how two-level type theory is used as a metatheoretic
setting in the rest of this thesis. First, we provide motivation and
overview. Second, we describe models of type theories in general, and models of
two-level type theories as extensions. Third, we describe presheaf models of
two-level type theories. Finally, we generalize the semantics and the term
algebra construction from Chapter \ref{chap:simple-inductive-signatures} in
two-level type theory, as a way to illustrate the applications.

\section{Motivation}
\label{sec:2ltt-motivation}
We note two shortcomings of the semantics presented in the
previous chapter.

First, the semantics that we provided was not as general as it could be. We
used the internal $\Set$ universe to specify algebras, but algebras make sense
in many different categories. A crude way to generalize semantics is to simply
say that our formalization, which was carried out in the syntax (i.e.\ initial
model) of some intensional type theory, can be interpreted in any model of the
type theory. But this is wasteful: for simple inductive signatures, it is
enough to assume a category with finite products as semantic setting. We do not
need all the extra baggage that comes with a model of a type theory.

Second, we were not able to reason about definitional equalities, only
propositional ones. We have a formalization of signatures and semantics in
intensional Agda, where the two notions differ\footnote{As opposed to in
extensional type theory, where they are the same.}, but only propositional
equality is subject to internal reasoning. For instance, we would like to show
that term algebras support recursion with strict $\beta$-rules, and for this we
need to reason about strict equality.

\begin{notation}
We use $\emptycon$ for the terminal object in a $\mbb{C}$ category, with
$\epsilon : \mbb{C}(A,\,\emptycon)$ for the unique morphism. For products, we
use $\blank\!\otimes\!\blank$ with $(\blank\!,\!\blank) : \mbb{C}(A,\,B) \to
\mbb{C}(A,\,C) \to \mbb{C}(A,\,B\otimes C)$ and $\p$ and $\q$ for
first and second projections respectively.
\end{notation}

\begin{myexample}
Assuming a category $\mbbC$ with finite products, we specify natural number
algebras and binary tree algebras as follows. Below, $\Alg_{\ms{NatSig}}$ and
$\Alg_{\ms{TreeSig}}$ are both sets in some metatheory, and the $\times$ in the
definitions refer to the metatheoretic $\Sigma$.
\begin{alignat*}{3}
  &\Alg_{\ms{NatSig}} &&\defn (X : |\mbbC|) \times \mbbC(\emptycon,\,X) \times \mbbC(X,\,X)\\
  &\Alg_{\ms{TreeSig}}&& \defn (X : |\mbbC|) \times \mbbC(\emptycon,\,X) \times \mbbC(X \otimes X,\,X)
\end{alignat*}
\end{myexample}
How should we adjust $\Alg$ from the previous chapter to compute algebras in
$\mbbC$, and $\Mor$ to compute their morphisms? While it is possible to do this
in a direct fashion, working directly with objects and morphisms of $\mbbC$ is
rather unwieldy. $\mbbC$ is missing many convenience features of type theories.
\begin{itemize}
\item
  There are no variables or binders. We are forced to work in a point-free style
  or chase diagrams; both become difficult to handle above a certain level of
  complexity.
\item
  There are no functions, universes or inductive types.
\item
  Substitution (with weakening as a special case) has to be handled explicitly
  and manually. Substitutions are certain morphisms, while ``terms'' are also
  morphisms, and we have to use composition to substitute terms. In contrast, if
  we are working internally in a type theory, terms and substitutions are
  distinct, and we only have to explicitly deal with terms, and substitutions
  are automated and implicit.
\end{itemize}

The above overlaps with motivations for working in \emph{internal languages}
\cite{internallogic} of structured categories: they aid calculation and compact
formalization by hiding bureaucratic structural details.

A finite product category $\mbbC$ does not have much of an internal language, it
is too bare-bones. But we can work instead in the internal language of
$\hat\mbbC$, the category of presheaves over $\mbbC$. This allows faithful
reasoning about $\mbbC$, while also including all convenience features of
extensional type theory.

\emph{Two-level type theories} \cite{twolevel}, or 2LTT in short, are type
theories such that they have ``standard'' interpretations in presheaf
categories. A 2LTT has an inner layer, where types and terms arise by embedding
$\mbbC$ in $\hat{\mbbC}$, and an outer layer, where constructions are inherited
from $\hat{\mbbC}$. The exact details of the syntax may vary depending on what
structures $\mbbC$ supports, and which type formers we assume in the outer
layer. Although it is possible to add assumptions to a 2LTT which preclude
standard presheaf semantics \cite[Section 2.4.]{twolevel}, we stick to basic
2LTT in this thesis. By using 2LTT, we are able to use a type-theoretic syntax
which differs only modestly from the style of definitions that we have seen so
far.

From a programming perspective, basic 2LTT provides a convenient syntax for
writing metaprograms. This can be viewed as \emph{two-stage compilation}: if we
have a 2LTT program with an inner type, we can run it, and it returns another
program, which lives purely in the inner theory.

\section{Models of Type Theories}
\label{sec:models-of-tts}

Before explaining 2LTT-specific features, we review models of type theories
in general. Variants of 2LTT will be obtained by adding extra features on the
top of more conventional type theories.

It is also worth to take a more general look at models at this point, because
the notions presented in this subsection (categories with families, type
formers) will be reused several times in this thesis, when specifying theories
of signatures.

\subsection{The Algebraic View}

We take an algebraic view of models and syntaxes of type theories throughout
this thesis. Models of type theories are algebraic structures: they are
categories with certain extra structure. The syntax of a type theory is
understood to be its initial model. In initial models, the underlying category
is the category of typing contexts and parallel substitutions, while the extra
structure corresponds to type and term formers, and equations quotient the
syntax by definitional equality.

Type theories can be described with quotient inductive-inductive (QII)
signatures, and their initial models are quotient inductive-inductive types
(QIITs). Hence, 2LTT is also a QII theory. We will first talk about QIITs in
Chapter \ref{chap:fqiit}. Until then, we shall make do with an informal
understanding of categorical semantics for type theories, without using anything
in particular from the metatheory of QIITs. There is some circularity here, that
we talk about QIITs in this thesis, but we employ QIITs when talking about
them. However, this is only an annoyance in exposition and not a fundamental
issue: Sections \ref{sec:closed-levitation} and \ref{sec:iqii-levitation}
describe how to eliminate circularity by a form of bootstrapping.

The algebraic view lets us dispense with all kinds of ``raw'' syntactic objects.
We only ever talk about well-typed and well-formed objects, moreover, every
construction must respect definitional equalities. For terms in the algebraic
syntax, definitional equality coincides with metatheoretic equality. This
mirrors equality of morphisms in 1-category theory, where we usually reuse
metatheoretic equality in the same way.

In the following we specify notions of models for type theories. We split this
in two parts: categories with families and type formers.

\subsection{Categories With Families}

\begin{mydefinition}
A \textbf{category with families} (cwf) \cite{Dybjer96internaltype} is a way to
specify the basic structural rules for contexts, substitutions, types and
terms. It yields a dependently typed explicit substitution calculus.  A cwf
consists of the following.
\begin{itemize}
\item
  A category with a terminal object. We denote the set of objects as $\Con :
  \Set$ and use capital Greek letters starting from $\Gamma$ to refer to
  objects. The set of morphisms is $\Sub : \Con \to \Con \to \Set$, and we use
  $\sigma$, $\delta$ and so on to refer to morphisms. We write $\id$ for the
  identity morphism and $\blank\circ\blank$ for composition. The terminal
  object is $\emptycon$ with unique morphism $\epsilon :
  \Sub\,\Gamma\,\emptycon$. In initial models (that is, syntaxes) of type
  theories, objects correspond to typing contexts, morphisms to parallel
  substitutions and the terminal object to the empty context; this informs the
  naming scheme.
\item A \emph{family structure}, containing $\Ty : \Con \to \Set$ and $\Tm :
  (\Gamma : \Con) \to \Ty\,\Gamma \to \Set$. We use $A$, $B$, $C$ to refer to
  types and $t$, $u$, $v$ to refer to terms. $\Ty$ is a presheaf over the
  category of contexts and $\Tm$ is a displayed presheaf over $\Ty$. This means
  that types and terms can be substituted:
  \begin{alignat*}{3}
    &\blank[\blank] : \Ty\,\Delta \to \Sub\,\Gamma\,\Delta \to \Ty\,\Gamma\\
    &\blank[\blank] : \Tm\,\Delta\,A \to (\sigma : \Sub\,\Gamma\,\Delta) \to \Tm\,\Gamma\,(A[\sigma])
  \end{alignat*}
  Substitution is functorial: we have $A[\id] \equiv A$ and
  $A[\sigma\circ\delta] \equiv A[\sigma][\delta]$, and likewise for terms.

  A family structure is additionally equipped with \emph{context comprehension}
  which consists of a context extension operation $\blank\ext\blank : (\Gamma :
  \Con) \to \Ty\,\Gamma \to \Con$ together with a natural isomorphism
  $\Sub\,\Gamma\,(\Delta\ext A) \simeq ((\sigma : \Sub\,\Gamma\,\Delta) \times
  \Tm\,\Gamma\,(A[\sigma]))$.
\end{itemize}
\end{mydefinition}

The following notions are derivable from the comprehension structure:
\begin{itemize}
\item
  By going right-to-left along the isomorphism, we recover \emph{substitution
  extension} $\blank,\blank : (\sigma : \Sub\,\Gamma\,\Delta) \to
  \Tm\,\Gamma\,(A[\sigma]) \to \Sub\,\Gamma\,(\Delta\ext A)$. This means that
  starting from $\epsilon$ or the identity substitution $\id$, we can iterate
  $\blank,\blank$ to build substitutions as lists of terms.
\item
  By going left-to-right, and starting from $\id : \Sub\,(\Gamma\ext
  A)\,(\Gamma\ext A)$, we recover the \emph{weakening substitution} $\p :
  \Sub\,(\Gamma\ext A)\,\Gamma$ and the \emph{zero variable} $\q :
  \Tm\,(\Gamma\ext A)\,(A[\p])$.
\item
  By weakening $\q$, we recover a notion of variables as De Bruijn indices. In
  general, the $n$-th De Bruijn index is defined as $\q[\p^{n}]$, where $\p^{n}$
  denotes $n$-fold composition.
\end{itemize}

Comprehension can be characterized either by taking $\blank,\blank$, $\p$ and
$\q$ as primitive, or the natural isomorphism. The two are equivalent, and we
may switch between them, depending on which is more convenient.

There are other ways for presenting the basic categorical structure of models,
which are nonetheless equivalent to cwfs, including natural models
\cite{awodey18natural} and categories with attributes \cite{cartmellthesis}. We
use the cwf presentation for its immediately algebraic character and closeness
to conventional explicit substitution syntax.

\begin{notation}As De Bruijn indices are hard to read, we will mostly use
nameful notation for binders. For example, assuming $\Nat : \{\Gamma : \Con\}
\to \Ty\,\Gamma$ and $\Id : \{\Gamma : \Con\}(A : \Ty\,\Gamma) \to
\Tm\,\Gamma\,A \to \Tm\,\Gamma\,A \to \Ty\,\Gamma$, we may write $\emptycon\,\ext\,
n : \Nat\,\ext\,p : \Id\,\Nat\,n\,n$ for a typing context, instead of using
numbered variables or cwf combinators as in $\emptycon \ext \Nat \ext
\Id\,\Nat\,\q\,\q$.
\end{notation}

\begin{notation}
In the following, we will denote family structures by ($\Ty$,$\Tm$) pairs and overload context
extension $\blank\ext\blank$ for different families.
\end{notation}

\begin{mydefinition} The following derivable operations are commonly used.
\label{def:cwfops}
  \begin{itemize}
    \item \emph{Single substitution} can be derived from parallel substitution
      as follows. Assume $t : \Tm\,(\Gamma\ext A)\,B$, and $u :
      \Tm\,\Gamma\,A$. $t$ is a term which may depend on the last variable in
      the context, which has $A$ type. We can substitute that variable with the
      $u$ term as $t[\id,\,u] : \Tm\,\Gamma\,(\A[\id,\,u])$. Note that term
      substitution causes the type to be substituted as well. $(\id,\,u) :
      \Sub\,\Gamma\,(\Gamma\ext A)$ is well-typed because $u : \Tm\,\Gamma\,A$
      hence also $u : \Tm\,\Gamma\,(A[\id])$.

    \item We can \emph{lift substitutions} over binders as follows. Assuming
      $\sigma : \Sub\,\Gamma\,\Delta$ and $A : \Ty\,\Delta$, we construct a
      lifting of $\sigma$ which maps an additional $A$-variable to itself:
      $(\sigma\circ\p,\,\q) : \Sub\,(\Gamma\ext A[\sigma])\,(\Delta \ext A)$.
      Let us see why this is well-typed. We have $\p : \Sub\,(\Gamma\ext
      A[\sigma])\,\Gamma$ and $\sigma : \Sub\,\Gamma\,\Delta$, so $\sigma \circ
      \p : \Sub\,(\Gamma\ext A[\sigma])\,\Delta$. Also, $\q : \Tm\,(\Gamma\ext
      A[\sigma])\,(A[\sigma][\p])$, hence $\q : \Tm\,(\Gamma\ext
      A[\sigma])\,(A[\sigma \circ \p])$, thus $(\sigma\circ \p,\,\q)$
      typechecks.
  \end{itemize}
\end{mydefinition}

\begin{notation}

As a nameful notation for substitutions, we may write $t[x \mapsto u]$, for
a single substitution, or $t[x \mapsto u_1, y \mapsto u_2]$ and so on.

In nameful notation we leave all weakening implicit, including substitution
lifting. Formally, if we have $t : \Tm\,\Gamma\,A$, we can only mention $t$ in
$\Gamma$. If we need to mention it in $\Gamma \ext B$, we need to use $t[\p]$
instead. In the nameful notation, $t : \Tm\,(\Gamma\ext x : B)\,A$ may be
used.\footnote{Moreover, when working in the internal syntax of a theory, we
just write Agda-like type-theoretic notation, without noting contexts and
substitutions in any way.}
\end{notation}

\subsection{Type formers}
A family structure in a cwf may be closed under certain type formers, such as
functions, $\Sigma$-types, universes or inductive types. We give some examples
here for their specification. First, we look at common negative type formers;
these are the type formers which can be specified using isomorphisms. Then, we
consider positive type formers, and finally universes.

\subsubsection{Negative types}

\begin{mydefinition}
A $(\Ty,\,\Tm)$ family supports \textbf{$\Pi$-types} if it supports the following.
\begingroup
\allowdisplaybreaks
\begin{alignat*}{3}
  &\Pi           &&: (A : \Ty\,\Gamma) \to \Ty\,(\Gamma\ext A) \to \Ty\,\Gamma\\
  &\Pi[]         &&: (\Pi\,A\,B)[\sigma] \equiv \Pi\,(A[\sigma])\,(B[\sigma\circ\p,\,\q])\\
  &\app          &&: \Tm\,\Gamma\,(\Pi\,A\,B) \to \Tm\,(\Gamma \ext A)\,B\\
  &\lam          &&: \Tm\,(\Gamma \ext A)\,B \to \Tm\,\Gamma\,(\Pi\,A\,B)\\
  &\Pi\beta      &&: \app\,(\lam\,t) \equiv t\\
  &\Pi\eta       &&: \lam\,(\app\,t) \equiv t\\
  &\lam[]        &&: (\lam\,t)[\sigma] \equiv \lam\,(t[\sigma\circ\p,\,\q])
\end{alignat*}
\endgroup
Here, $\Pi$ is the type formation rule. $\ms{\Pi[]}$ is the type substitution
rule, expressing that substituting $\Pi$ proceeds structurally on constituent
types.  Note $B[\sigma\circ\p,\,\q]$, where we lift $\sigma$ over the additional
binder.

The rest of the rules specify a natural isomorphism $\Tm\,\Gamma\,(\Pi\,A\,B)
\simeq \Tm\,(\Gamma \ext A)\,B$. We only need a substitution rule (i.e.\ a
naturality rule) for one direction of the isomorphism, since the naturality of
the other map is derivable.

This way of specifying $\Pi$-types is very convenient if we have explicit
substitutions. The usual ``pointful'' specification is equivalent to this. For
example, we have the following derivation of pointful application:
\begin{alignat*}{3}
  &\app' : \Tm\,\Gamma\,(\Pi\,A\,B) \to (u : \Tm\,\Gamma\,A) \to \Tm\,\Gamma\,(B[\id,\,u])\\
  &\app'\,t\,u \defn (\app\,t)[\id,\,u]
\end{alignat*}

\end{mydefinition}

\emph{Remark on naturality.} The above specification for $\Pi$ can be written
more compactly if we assume that everything is natural with respect to
substitution.
\begin{alignat*}{3}
  &\Pi            &&: (A : \Ty\,\Gamma) \to \Ty\,(\Gamma\ext A) \to \Ty\,\Gamma\\
  & (\app,\,\lam) &&: \Tm\,\Gamma\,(\Pi\,A\,B) \simeq \Tm\,(\Gamma \ext A)\,B
\end{alignat*}
This is a reasonable assumption; in the rest of the thesis we only ever define
structures on cwfs which are natural in this way.

\begin{notation} From now on, when specifying type formers in family structures,
we assume that everything is natural, and thus omit substitution equations.
\end{notation}

There are ways to make this idea more precise, and take it a step further by
working in languages where only natural constructions are possible. The term
\emph{higher-order abstract syntax} (HOAS) is sometimes used for this style. It lets us
also omit contexts, so we would only need to write
\begin{alignat*}{3}
  &\Pi            &&: (A : \Ty) \to (\Tm\,A \to \Ty) \to \Ty\\
  & (\app,\,\lam) &&: \Tm\,(\Pi\,A\,B) \simeq ((a : \Tm\,A) \to \Tm\,(B\,a))
\end{alignat*}
Recently several promising works emerged in this area
\cite{uemura,ctt-normalization,bocquet2021induction}. Although this technology
is likely to be the preferred future direction in the metatheory of type
theories, this thesis does not make use of it. The field is rather fresh, with
several different approaches and limited amount of pedagogical exposition, and
the new techniques would also raise the level of abstraction in this thesis,
contributing to making it less accessible. It is also not obvious how exactly
HOAS-style could be employed to aid formalization here, and it would require
significant additional research.  Often, a setup with multiple modalities
(``multimodal'' \cite{gratzer20multimodal}) is required
\cite{bocquet2021induction} because we work with presheaves over different
cwfs. It seems that a synthetic notion of dependent modes would be also required
to formalize constructions in this thesis, since we often work with displayed
presheaves over displayed cwfs. This is however not yet developed in the
literature.

\begin{mydefinition}
\label{def:constant-families}
A family structure supports \textbf{constant families} if we have the following.
\begin{alignat*}{3}
  & \K &&: \Con \to \{\Gamma : \Con \} \to \Ty\,\Gamma \\
  & (\appK,\,\lamK) &&: \Tm\,\Gamma\,(\K\,\Delta) \simeq \Sub\,\Gamma\,\Delta
\end{alignat*}
Constant families express that every context can be viewed as a non-dependent
type in any context. Having constant families is equivalent to the
\emph{democracy} property for a cwf
\cite{clairambault2014biequivalence,forsberg-phd}. Constant families are
convenient when building models because they let us model non-dependent types
as semantic contexts, which are often simpler structures than semantic types.
From a programming perspective, constant families specify closed record types,
where $\K\,\Delta$ has $\Delta$-many fields.

If we have equalities of sets for the specification,
i.e.\ $\Tm\,\Gamma\,(\K\,\Delta) \equiv \Sub\,\Gamma\,\Delta$, we have \textbf{strict
  constant families}.

\end{mydefinition}

\begin{mydefinition}
A family structure supports \textbf{$\Sigma$-types} if we have
\begin{alignat*}{3}
  & \Sigma  &&: (A : \Ty\,\Gamma) \to \Ty\,(\Gamma\ext A) \to \Ty\,\Gamma\\
  & (\proj,\,(\blank,\blank)) &&: \Tm\,\Gamma\,(\Sigma\,A\,B) \simeq ((t : \Tm\,\Gamma\,A) \times \Tm\,\Gamma\,(B[\id,\,t]))
\end{alignat*}
We may write $\proj_1$ and $\proj_2$ for composing the metatheoretic first and
second projections with $\proj$.
\end{mydefinition}

\begin{mydefinition}
A family structure supports the \textbf{unit type} if we have $\top : \Ty\,\Gamma$ such
that $\Tm\,\Gamma\,\top \simeq \top$, where the $\top$ on the right is the
metatheoretic unit type, and we overload $\top$ for the internal unit type.
From this, we get the internal $\tt : \Tm\,\Gamma\,\top$, which is
definitionally unique.
\end{mydefinition}

\begin{mydefinition}
A family structure supports \textbf{extensional identity} types if there is $\Id
: \Tm\,\Gamma\,A \to \Tm\,\Gamma\,A \to \Ty\,\Gamma$ such that
$(\reflect,\,\refl) : \Tm\,\Gamma\,(\Id\,t\,u) \simeq (t \equiv u)$.
\end{mydefinition}

It is also possible to give a positive definition for identity types, in which
case we get intensional identity. Extensional identity corresponds to a
categorical equalizer of terms (a limit), while the Martin-Löf-style intensional
identity is a positive (inductive) type.

This choice between negative and positive specification generally exists for
type formers with a single term construction rule. For example, $\Sigma$ can be
defined as a positive type, with an elimination rule that behaves like pattern
matching. Positive $\Sigma$ is equivalent to negative $\Sigma$, although it only
supports propositional $\eta$-rules. In contrast, positive identity is usually
\emph{not} equivalent to negative identity.

$\refl : t \equiv u \to \Tm\,\Gamma\,(\Id\,t\,u)$ expresses reflexivity of
identity: definitionally equal terms are provably equal. $\reflect$, which goes
the other way around, is called \emph{equality reflection}: provably equal terms
are identified in the metatheory.

Uniqueness of identity proofs (UIP) is often ascribed to the extensional
identity type (see e.g.\ \cite{hofmann95extensional}). UIP means that
$\Tm\,\Gamma\,(\Id\,t\,u)$ has at most a single inhabitant up to $\Id$. However,
UIP is not something which is inherent in the negative specification, instead it
is inherited from the metatheory. If $\Tm$ forms a homotopy set in the
metatheory, then internal equality proofs inherit uniqueness through the
defining isomorphism.

\subsubsection{Positive types}

We do not dwell much on positive types here, as elsewhere in this thesis we talk
a lot about specifying such types anyway. We provide here some background and
a small example.

The motivation is to specify initial internal algebras in a cwf. However,
specifying the uniqueness of recursors using definitional equality is
problematic, if we are to have decidable and efficient conversion checking for a
type theory. Consider the specification of $\Bool$ together with its recursor.
\begingroup
\allowdisplaybreaks
\begin{alignat*}{3}
  & \Bool  &&: \Ty\,\Gamma \\
  & \true  &&: \Tm\,\Gamma\,\Bool \\
  & \false &&: \Tm\,\Gamma\,\Bool \\
  & \ms{BoolRec} &&: (B : \Ty\,\Gamma)\to \Tm\,\Gamma\,B \to \Tm\,\Gamma\,B \to \Tm\,\Gamma\,\Bool \to \Tm\,\Gamma\,B\\
  & \true\beta &&: \ms{BoolRec}\,B\,t\,f\,\true \equiv t\\
  & \false\beta &&: \ms{BoolRec}\,B\,t\,f\,\false \equiv f
\end{alignat*}
\endgroup
$\ms{BoolRec}$ together with the $\beta$-rules specifies an internal
$\Bool$-algebra morphism. A possible way to specify definitional uniqueness is
as follows. Assuming $B : \Ty\,\Gamma$, $t : \Tm\,\Gamma\,B$, $f :
\Tm\,\Gamma\,B$ and $m : \Tm\,(\Gamma\ext b : \Bool)\,B$, such that $m[b \mapsto
  \true] \equiv t$ and $m[b \mapsto \false] \equiv f$, it follows that
$\ms{BoolRec}\,B\,t\,f\,b : \Tm\,(\Gamma\ext b : \Bool)\,B$ is equal to $m$.

Unfortunately, deciding conversion with this rule entails deciding pointwise
equality of arbitrary $\Bool$ functions, which can be done in exponential time
in the number of $\Bool$ arguments. More generally, Scherer presented a decision
algorithm for conversion checking with strong finite sums and products in simple
type theory \cite{scherer17deciding}, which also takes exponential time. If we
move to natural numbers with definitionally unique recursion, conversion
checking becomes undecidable. To illustrate this, consider checking conversion
between any closed term $t : \Nat \to \Nat$ and the identity function. To cover
the $\eta$-rule, we would have to decide $t\,\zero \equiv \zero$ and $(n :
\Tm\,\emptycon\,\Nat) \to t\,(\suc\,n) \equiv \suc\,(t\,n)$. Assuming
canonicity, this is equivalent to deciding $(n : \mathbb{N}) \to
t\,(\suc^{n}\,\zero) \equiv (\suc^{n}\,\zero)$, where $\suc^{n}\,\zero$ denotes
a canonical numeral.

The standard solution is to have dependent elimination principles instead: this
allows inductive reasoning, canonicity and effectively decidable definitional
equality at the same time. For $\Bool$, we would have
\begin{alignat*}{3}
  & \ms{BoolInd} &&: (B : \Ty\,(\Gamma\ext b : \Bool)) \to \Tm\,\Gamma\,(B[b \mapsto \true])\\
  & &&\to \Tm\,\Gamma\,(B[b \mapsto \false]) \to (t : \Tm\,\Gamma\,\Bool) \to \Tm\,\Gamma\,(B[b \mapsto t])
\end{alignat*}
together with $\ms{BoolInd}\,B\,t\,f\,\true \equiv t$ and $\ms{BoolInd}\,B\,t\,f\,\false \equiv f$.

Of course, if we assume extensional identity types, we have undecidable
conversion anyway, and definitionally unique recursion is equivalent to
induction. But decidable conversion is a pivotal part of type theory, which
makes it possible to relegate a deluge of boilerplate to computers, so
decidable conversion should be kept in mind.

\subsubsection{Universes}

Universes are types which classify types. There are several different flavors of
universes.

\begin{mydefinition} A \textbf{Tarski-style} universe consists of
the following data:
\begin{alignat*}{3}
  & \U : \Ty\,\Gamma\hspace{2em}\El : \Tm\,\Gamma\,\U \to \Ty\,\Gamma
\end{alignat*}
\end{mydefinition}
This is a weak classifier, since not all elements of $\Ty\,\Gamma$ are
necessarily represented as terms of the universe. Like families, Tarski
universes can be closed under type formers. For instance, if $\U$ has $\Nat$, we
have the following:
\begin{alignat*}{3}
  &\Nat : \Tm\,\Gamma\,\U
    \hspace{1em}\zero : \Tm\,\Gamma\,(\El\,\Nat)
    \hspace{1em}\suc : \Tm\,\Gamma\,(\El\,\Nat) \to \Tm\,\Gamma\,(\El\,\Nat)
\end{alignat*}
\vspace{-2em}
\begin{alignat*}{3}
  \ms{NatElim} &:\,\,\,(P : \Ty\,(\Gamma\ext n : \El\,\Nat))\\
  &\to \Tm\,\Gamma\,(P[n \mapsto \zero])\\
  &\to \Tm\,(\Gamma\ext n : \El\,\Nat \ext \mi{np} : P[n \mapsto n])\,(P[n \mapsto \suc\,n]) \\
  &\to (n : \Tm\,\Gamma\,(\El\,\Nat)) \to \Tm\,\Gamma\,(P[n \mapsto n])
\end{alignat*}
If all type formers in $\U$ follow this scheme, $\U$ may be called a
\textbf{weakly Tarski} universe. If we assume that every type former in $\U$ is
also duplicated in $(\Ty,\,\Tm)$, moreover $\El$ preserves all type formers, so
that e.g.\ $\El\,\Nat$ is definitionally equal to the natural number type in
$\Ty$, then $\U$ is \textbf{strongly Tarski}.

It is often more convenient to have stronger classifiers as universes, so that
\emph{all} types in a given family structure are represented.

\begin{mydefinition}
Ignoring size issues for now, \textbf{Coquand universes}
\cite{coquand2018canonicity} are specified as follows:
\[
  \U : \Ty\,\Gamma\hspace{1em} (\El,\,\ms{c}) : \Tm\,\Gamma\,\U \simeq \Ty\,\Gamma
\]
$\ms{c}$ maps every type in $\Ty$ to a code in $\U$. Now we can ignore $\El$
when specifying type formers, as $\ms{c}$ can be always used to get a code in
$\U$ for a type.
\end{mydefinition}

Unfortunately, the exact specification above yields an inconsistent
``type-in-type'' system because $\U$ itself has a code in $\U$. The standard
solution is to have multiple family structures $(\Ty_i,\,\Tm_i)$, indexed by
universe levels, and have $\U_i : \Ty_{i + 1}\,\Gamma$ and
$\Tm_{i+1}\,\Gamma\,\U_i \simeq \Ty_i\,\Gamma$. For a general specification of
consistent universe hierarchies, see \cite{kovacs2021generalized}. We omit
universe indices in the following, and implicitly assume ``just enough''
universes for particular purposes.

\begin{mydefinition}
\textbf{Russell universes} are Coquand universes additionally satisfying
$\Tm\,\Gamma\,\U \equiv \Ty\,\Gamma$ as an equality of sets, and also $\ms{El}\,t \equiv
t$. This justifies omitting $\El$ and $\ms{c}$ from informal notation,
implicitly casting between $\Tm\,\Gamma\,\U$ and $\Ty\,\Gamma$.
\end{mydefinition}
Russell-style universes are commonly supported in set-theoretic models. They are
also often inherited from meta-type-theories which themselves have
Russell-universes. Major implementations of type theories (Coq, Lean, Agda,
Idris) are all such.

\section{Two-Level Type Theory}\label{sec:2ltt}

\subsection{Models}

We describe models of 2LTT in the following. This is not the only possible way
to present 2LTT; our approach differs from \cite{twolevel} in some ways. We will summarize
the differences at the end of this section.

\begin{mydefinition}
A model of a \textbf{two-level type theory} is a model of type theory such that
\begin{itemize}
  \item It supports a Tarski-style universe $\Ty_0 : \Ty\,\Gamma$ with decoding $\Tm_0 :
    \Tm\,\Gamma\,\Ty_0 \to \Ty\,\Gamma$.
  \item $\Ty_0$ may be closed under arbitrary type formers, however, it is only possible
    to eliminate from $\Ty_0$ type formers to types in $\Ty_0$.
\end{itemize}
Types in $\Ty_0$ are called \emph{inner} types, while other types are \emph{outer}. Alternatively,
we may talk about \emph{object-level} and \emph{meta-level} types.
\end{mydefinition}

For example, if we have inner functions, we have the following:
\begin{alignat*}{3}
  &\Pi_0 &&: (A : \Tm\,\Gamma\,\Ty_0) \to \Tm\,(\Gamma \ext \Tm_0\,A) \to \Tm\,\Gamma\,\Ty_0\\
  &(\app_0,\,\lam_0) &&: \Tm\,\Gamma\,(\Tm_0\,(\Pi_0\,A\,B)) \simeq \Tm\,(\Gamma \ext \Tm_0\,A)\,(\Tm_0\,B)
\end{alignat*}
If we have inner Booleans, we have the following (with $\beta$-rules omitted):
\begin{alignat*}{3}
  &\Bool_0 &&: \Tm\,\Gamma\,\Ty_0\\
  &\true_0 &&: \Tm\,\Gamma\,(\Tm_0\,\Bool_0)\\
  &\false_0 &&: \Tm\,\Gamma\,(\Tm_0\,\Bool_0)\\
  & \ms{BoolInd_0} &&: (B : \Tm\,(\Gamma\ext b : \Tm_0\,\Bool_0)\,\Ty_0)\\
  & && \to \Tm\,\Gamma\,(\Tm_0\,(B[b \mapsto \true_0]))\\
  & &&\to \Tm\,\Gamma\,(\Tm_0\,(B[b \mapsto \false_0]))\\
  & && \to (t : \Tm\,\Gamma\,(\Tm_0\,\Bool_0)) \to \Tm\,\Gamma\,(\Tm_0\,(B[b \mapsto t]))
\end{alignat*}

Intuitively, we can view outer types and terms as metatheoretical, while $\Ty_0$
represents the set of types in the object theory, and $\Tm_0$ witnesses that any
object type can be mapped to a metatheoretical set of object terms. The
restriction on elimination is crucial. If we have a Boolean term in the object
language, we can use the object-level elimination principle to construct new
object terms. But it makes no sense to eliminate into the metatheory. In fact,
an object-level Boolean term is not necessarily $\true$ or $\false$, it can also
be just a variable or neutral term in some context, or it can be an arbitrary
non-canonical value in a given model.

We review some properties of 2LTT. An important point is the action of $\Tm_0$
on type formers. In general, $\Tm_0$ preserves the negative type formers but not
others.

For example, we have the isomorphism $\Tm_0\,(\Pi_0\,A\,B) \simeq
\Pi_1\,(\Tm_0\,A)\,(\Tm_0\,B)$, where $\Pi_1$ denotes outer functions.  We move
left-to-right by mapping $t$ to $\lam_1\,(\app_1\,t)$, and the other way by
mapping $t$ to $\lam_0\,(\app_0\,t)$. The preservation of $\Sigma$, $\top$, $\K$
and extensional identity is analogous.

In contrast, we can map from outer positive types to inner ones, but not the
other way around. From $b : \Tm\,\Gamma\,\Bool_1$, we can use the outer
$\Bool_1$ recursor to return in $\Tm_0\,\Bool_0$. In the other direction, only
constant functions are definable since the $\Bool_0$ recursor only targets types
in $\Ty_0$.

It may be the case that there are universes in the inner layer. For example,
disregarding size issues (or just accepting an inconsistent inner theory), there
may be an $\U_0$ in $\Ty_0$ such that we have $\Tm\,\Gamma\,(\Tm_0\,\U_0) \equiv
\Tm\,\Gamma\,\Ty_0$. This amounts to having a Russell-style inner universe with
type-in-type. Assume that we have $\U_1$ as well, as a meta-level Russell
universe. Then we can map from $\Tm_0\,\U_0$ to $\U_1$, by taking $A$ to
$\Tm_0\,A$, but we cannot map in the other direction.

\subsection{Internal Syntax and Notation}
\label{sec:2ltt-internal-syntax}

In the rest of this thesis we will often work internally to a 2LTT, i.e.\ we use
2LTT as metatheory. We adapt the metatheoretical notations used so far. We list
used features and conventions below.

\begin{itemize}
  \item
    We keep previous notation for type formers. For instance, $\Pi$-types are
    written as $(x : A) \to B$ or as $A \to B$.
  \item
    We assume a Coquand-style universe in the outer layer, named $\Set$. As
    before, we leave the sizing levels implicit; if we were fully precise, we
    would write $\Set_i$ for a hierarchy of outer universes. Despite having a
    Coquand universe, we shall omit encoding and decoding in the internal
    syntax, and instead work in Russell-style. In practical implementations,
    elaborating Russell-style notation to Coquand-style is straightforward to
    do.
  \item
    If the same type formers are supported both in the inner and outer layers, we
    may distinguish them by $_0$ and $_1$ subscripts, e.g.\ by having $\Bool_0$ and
    $\Bool_1$. We omit some inferable subscripts, e.g.\ for $\Pi$ and
    $\Sigma$-types. In these cases, we usually know from the type parameters which
    type former is meant. For example, $\Tm_0\,\Bool_0 \to \Bool_1$ can only refer
    to outer functions.
  \item
    We have the convention that $\blank\!=\!\blank$ refers to the inner equality
    type, while $\blank\!\equiv\!\blank$ refers to the outer equality type. If
    the inner equality is extensional, the choice between
    $\blank\!=\!\blank$ and $\blank\!\equiv\!\blank$ is immaterial, but in
    Section \ref{sec:2ltt-simple-signatures} and Chapter \ref{chap:hiit} we do
    have intensional inner equality.
  \item
    By having $\Set$, we are able to have $\Ty_0 : \Set$ and $\Tm_0 : \Ty_0 \to
    \Set$. So we do not have to deal with proper meta-level types, and have a
    more uniform notation. Notation and specification for inner type formers
    changes accordingly. For example, for inner $\Pi$-types we may write $(x :
    A) \to B$ if $A : \Ty_0$ and $B$ depends on $x : \Tm_0\,A$. This also
    enables a higher-order specification: if $B : \Tm_0\,A \to \Ty_0$, then $(x
    : A) \to B\,x : \Ty_0$, and the specifying isomorphism for $\Pi$ can be
    written as $\Tm_0\,((x : A) \to B\,x) \simeq ((x : \Tm_0\,A) \to
    \Tm_0\,(B\,x))$. Note that the definition of $\blank\!\simeq\!\blank$
    requires a meta-level identity type.
    \begin{notation}
      An explicit notation for inner function abstraction would look like
      $\lam_0\,t$ for $t : (x : \Tm_0\,A) \to \Tm_0\,(B\,x)$. This results in
      ``double'' abstraction, e.g.\ in
      $\lam_0\,(\lambda\,x.\,\suc_0\,(\suc_0\,x)) : \Tm_0\,(\Nat_0 \to
      \Nat_0)$. Instead of this, we write $\lambda_0\,x.\,t$ as a notation, thus
      we write $\lambda_0\,x.\,\suc_0\,(\suc_0\,x)$ for the above example. We
      may also group multiple $\lambda_0$ binders together the same way as with
      $\lambda$.
    \end{notation}
  \item
    We may omit inferable $\Tm_0$ applications. For instance, $\Bool_1 \to
    \Bool_0$ can be ``elaborated'' to $\Bool_1 \to \Tm_0\,\Bool_0$ without
    ambiguity, since the function codomain must be on the same level as the
    domain, and the only thing we can do to make sense of this is to lift the
    codomain by $\Tm_0$. Sometimes there is some ambiguity: $(\Bool_0 \to
    \Bool_0) \to \Bool_1$ can be elaborated both to $\Tm_0\,(\Bool_0 \to
    \Bool_0) \to \Bool_1$ and to $(\Tm_0\,\Bool_0 \to \Tm_0\,\Bool_0) \to
    \Bool_1$. However, in this case the two output types are definitionally
    isomorphic because of the $\Pi$-preservation by $\Tm_0$. Hence, the
    elaboration choice does not make much difference, so we may still omit
    $\Tm_0$-s in situations like this.
\end{itemize}

\begin{myexample} Working in the internal syntax of 2LTT, the specification of $\Bool_0$
looks like the following (omitting $\beta$ again):
\begin{alignat*}{3}
  &\Bool_0  &&: \Ty_0\\
  &\true_0  &&: \Bool_0\\
  &\false_0 &&: \Bool_0\\
  & \ms{BoolInd_0} &&: (B : \Bool_0 \to \Ty_0) \to B\,\true_0 \to B\,\false_0 \to (t : \Bool_0) \to B\,t
\end{alignat*}
If we elaborate the type of $\ms{BoolInd_0}$, we get the following:
\begin{alignat*}{3}
  & \ms{BoolInd_0} &&: (B : \Tm_0\,\Bool_0 \to \Ty_0) \to \Tm_0\,(B\,\true_0) \to \Tm_0\,(B\,\false_0)\\
  & && \to (t : \Tm_0\,\Bool_0) \to \Tm_0\,(B\,t)
\end{alignat*}
Here, the type is forced to live in the outer level because of the dependency on
$\Ty_0$. Since $\Ty_0$ is an outer type, $\Bool_0 \to \Ty_0$ must be lifted, which
in turn requires all other types to be lifted as well.
\end{myexample}

\subsection{Alternative Presentation for 2LTT}

We digress a bit on a different way to present 2LTT. In the primary 2LTT
reference \cite{twolevel}, inner and outer layers are specified
as follows. We have two different \emph{family structures} on the base
cwf, $(\Ty_0,\,\Tm_0)$ and $(\Ty_1,\,\Tm_1)$, and a morphism between them. A
family morphism is natural transformation mapping types to types and terms to
terms, which is an isomorphism on terms. We might name the component maps
as follows:
\begin{alignat*}{3}
  &\Lift &&: \Ty_0\,\Gamma \to \Ty_1\,\Gamma \\
  &\up   &&: \Tm_0\,\Gamma\,A \to \Tm_1\,\Gamma\,(\Lift\!A)\\
  &\down &&: \Tm_1\,\Gamma\,(\Lift\!A) \to \Tm_0\,\Gamma\,A
\end{alignat*}
An advantage of this presentation is that we may close $(\Ty_0,\,\Tm_0)$ under
type formers without any encoding overhead, for example by having $\Bool_0 :
\Ty_0\,\Gamma$, $\true_0 : \Tm_0\,\Gamma\,\Bool_0$, etc., without the
Tarski-style decoding. On the other hand, we do not automatically get an outer
universe of inner types. We can recover that in two ways:
\begin{itemize}
\item
  We can assume an inner universe $\U_0 : \Ty_0\,\Gamma$, which can be lifted to the
  outer theory as $\Lift\!\U_0$. However, we may not want to make this
  assumption, in order to keep the inner theory as simple as possible.
\item
  We can assume an outer universe which classifies elements of
  $\Ty_0\,\Gamma$. This amounts to reproducing the $\Ty_0$ \emph{type} from our
  2LTT presentation, as an additional assumption. But in this case, we might as
  well skip the two family structures and the $\Lift$ morphism.
\end{itemize}
In this thesis we make ubiquitous use of the outer universe of inner types, so
we choose that to be the primitive notion, instead of having two family
structures.

Do we lose anything by this? For the purposes of this thesis, not
really. However, if we want to implement 2LTT as a system for two-stage
compilation, the $\Lift$ syntax appears to be closer to existing systems.
Staging is about computing all outer redexes but no inner ones, thereby
outputting syntax which is purely in the inner theory. This could be implemented
as a stage-aware variant of normalization-by-evaluation
\cite{abel2013normalization,decidableconv,coqnbe}. We can give an intuitive
staging interpretation for the operators in the $\Lift$ syntax:
\begin{itemize}
\item
  $\Lift\!A$ is the type of $A$-expressions. This corresponds to $a\,\ms{code}$
  in MetaOcaml \cite{kiselyov14metaocaml} and $\ms{TExp}\,a$ in typed Template
  Haskell \cite{typed-th}.
\item
  $\up$ is \emph{quoting}, which creates an expression from any inner term. This is
  $.\lab\blank\rab.$ in MetaOCaml and $[||\blank||]$ in typed Template Haskell.
\item $\down$ is \emph{splicing}, which inserts the result of a meta-level computation into
  an object-level expression. This is $\sim\!(\blank)$ in MetaOCaml and $\$\$(\blank)$ in typed Template Haskell.
\end{itemize}
For example, in the $\Lift$ syntax, we might write a polymorphic identity function
which acts on inner types in two different ways:
\begin{alignat*}{3}
  &\id : (A : \U_0) \to A \to A\hspace{2em} && \id' : (A :\,\Lift\!\U_0) \to\,\Lift\!(\down A) \to\,\Lift\!(\down A)\\
  &\id \defn \lambda_0\,A\,x.\,x && \id' \defn \lambda_1\,A\,x.\,x
\end{alignat*}
The first one lives in the inner family structure. The second one is the same
thing, but lifted to the outer theory. The choice between the two allows us to
control staging-time evaluation. If we write $\id\,\Bool_0\,\true_0$, that is an
inner expression which goes into the staging output as it is. On the other hand,
$\down(\id'\,(\up\,\Bool_0)\,(\up\,\true_0))$ reduces to $\down(\up\,\true_0)$
which in turn reduces to $\true_0$. The same choice can be expressed in our
syntax as well:
\begin{alignat*}{3}
  &\id : \Tm_0\,((A : \U_0) \to A \to A)\hspace{2em}&&\id' : (A : \Tm_0\,\U_0) \to\,\Tm_0\,A \to \Tm_0\,A\\
  &\id \defn \lambda_0\,A\,x.\,x &&\id' \defn \lambda\,A\,x.\,x
\end{alignat*}
It remains to be checked which style is preferable in a staging
implementation. In the $\Lift$ style, the quoting and splicing operations add
noise to core syntax, but they are also mostly inferable during elaboration, and
they pack stage-changing information into $\up$ and $\down$, thereby making it
feasible to omit stage annotations in other places in the core syntax. In the
$\Ty_0$ style, we do not have quote/splice, but we have to keep track of stages
in all type/term formers. It would be interesting to compare the two flavors
in prototype implementations of staged systems.

\section{Presheaf Semantics of 2LTT}

We review the standard semantics of 2LTT which we use in the rest of the
thesis. This justifies the metaprogramming view, that 2LTT allows meta-level
reasoning about an inner theory.

We present it in two steps, by assuming progressively more structure in the inner
theory. First, we only assume a category. This already lets us present a
presheaf semantics for the outer layer. Then, we assume a cwf as the inner
theory, which lets us interpret $\Ty_0$ and $\Tm_0$ and also consider inner type
formers.

\subsection{Presheaf Model of the Outer Layer}

In this subsection we present a presheaf model for the outer layer of 2LTT, that
is, the base category together with the terminal object, the $(\Ty,\,\Tm)$
family and some type formers. This presheaf semantics is well-known in the
literature \cite{Hofmann97syntaxand}. We give a specification which follows
\cite{huber-thesis} most closely.

In the following, we work outside 2LTT (since we are defining a model of 2LTT),
in a suitable metatheory; an extensional type theory with enough $\Set$
universes suffices.

We assume a $\mbbC$ category. We write $i,\,j,\,k : |\mbbC|$ for objects and
$f,\,g,\,h\,: \mbbC(i,\,j)$ for morphisms. We use a different notation than for
cwfs before, in order to disambiguate components in $\mbbC$ from components in
the presheaf model of 2LTT. We use $\hmbbC$ to refer to the model which is being
defined. We use the same component names for $\hmbbC$ as in Section
\ref{sec:models-of-tts}.

\subsubsection{Model of cwf}

\begin{mydefinition}
$\Gamma : \Con$ is a presheaf over $\mbbC$. Its components
are as follows.
\begin{alignat*}{3}
  & |\Gamma|             &&: |\mbbC| \to \Set \\
  & \blank\lab\blank\rab &&: |\Gamma|\,j \to \mbbC(i,\,j) \to |\Gamma|\,i\\
  & \gamma\lab\id\rab &&\equiv \gamma \\
  & \gamma\lab f\circ g\rab &&\equiv \gamma \lab f \rab \lab g \rab
\end{alignat*}
We flip around the order of arguments in the action of $\Gamma$ on
morphisms. This is more convenient because of the contravariance; we can observe
this in the statement of preservation laws already. The action on morphisms is
sometimes called \emph{restriction}.
\end{mydefinition}

\begin{mydefinition}
$\sigma : \Sub\,\Gamma\,\Delta$ is a natural transformation from $\Gamma$ to
$\Delta$. It has action $|\sigma| : |\Gamma|\,i \to |\Delta|\,i$, such that
$|\sigma|(\gamma\lab f \rab) \equiv (|\sigma|\gamma)\lab f \rab$.
\end{mydefinition}

\begin{mydefinition}
\label{def:presheaf-type}
$A : \Ty\,\Gamma$ is a displayed presheaf over $\Gamma$. The
``displayed'' here is used in exactly the same sense as in ``displayed
algebra'' before. As we will see in Chapter \ref{chap:fqiit}, presheaves can be
specified with a signature, in which case a presheaf is an algebra, and a
displayed presheaf is a displayed algebra. The definition here is equivalent
to saying that $A$ is a presheaf over the category of elements of $\Gamma$,
but it is more convenient to use in concrete definitions and calculations. The
components of $A$ are as follows.
\begin{alignat*}{3}
  &|A| &&: |\Gamma|\,i \to \Set\\
  &\blank\lab\blank\rab &&: |A|\,\gamma \to (f : \mbbC(i,\,j)) \to |A|\,(\gamma\lab f \rab)\\
  & \alpha\lab\id\rab &&\equiv \alpha \\
  & \alpha\lab f\circ g\rab &&\equiv \alpha \lab f \rab \lab g \rab
\end{alignat*}
\end{mydefinition}

\begin{mydefinition}
$t : \Tm\,\Gamma\,A$ is a section of the displayed presheaf $A$. This is
again the same notion of section that we have seen before, instantiated for
presheaves.
\begin{alignat*}{3}
  & |t| : (\gamma : |\Gamma|\,i) \to |A|\,\gamma \\
  & |t|(\gamma\lab f \rab) \equiv (|t|\gamma)\lab f \rab
\end{alignat*}
\end{mydefinition}

\begin{mydefinition}
$\Gamma \ext A : \Con$ is the total presheaf of the displayed presheaf $A$. Its action on objects and morphisms is the following.
\begin{alignat*}{3}
  &|\Gamma \ext A| &&\defn (\gamma : |\Gamma|\,i) \times |A\,\gamma|\\
  &(\gamma,\,\alpha)\lab f \rab &&\defn (\gamma\lab f \rab,\, \alpha\lab f \rab)
\end{alignat*}
The $\id$ and $\blank\!\circ\!\blank$ preservation laws follow immediately.
\end{mydefinition}

\begin{mydefinition}
$A[\sigma] : \Ty\,\Gamma$ is defined as follows, assuming
$A : \Ty\,\Delta$ and $\sigma : \Sub\,\Gamma\,\Delta$.
\begin{alignat*}{3}
  & |A[\sigma]|\,\gamma &&\defn |A|\,(|\sigma|\,\gamma) \\
  & \alpha \lab f \rab &&\defn \alpha \lab f \rab
\end{alignat*}
In the second component, we use $\blank\!\lab\!\blank\!\rab$ for $A$ on the right hand
side.  The definition is well-typed since $|A|\,(|\sigma|\,(\gamma\lab f \rab))
\equiv |A|\,((|\sigma|\,\gamma)\lab f \rab)$ by the naturality of
$\sigma$. Functoriality follows from functoriality of $A$.
\end{mydefinition}

It is easy to check that the above definitions can be extended to a cwf.
\begin{itemize}
  \item For the base category, we take the category of presheaves.
  \item The empty context $\emptycon$ is the terminal presheaf, i.e.\ the
        constantly $\top$ functor.
  \item Type substitution is functorial, as it is defined as simple function
    composition of actions on objects.
  \item Term substitution is defined as composition of a section and
    a natural transformation; and also functorial for the same reason.
  \item Context comprehension structure follows from the $\Sigma$-based definition for
    context extension.
\end{itemize}

\subsubsection{Yoneda embedding}

Before continuing with interpreting type formers in $\hmbbC$, we review the
Yoneda embedding, as it is useful in subsequent definitions.

\begin{mydefinition}
The \textbf{Yoneda embedding}, denoted $\ms{y}$, is a functor from $\mbbC$ to
the underlying category of $\hmbbC$, defined as follows.
\begin{alignat*}{3}
  & \yon : |\mbbC| \to \Con \hspace{3em}&& \yon : \mbbC(i,\,j) \to \Sub\,(\yon\,i)\,(\yon\,j)\\
  & \yon\,i \defn \mbbC(\blank,\,i) && |\yon\,f|\,g \defn f \circ g
\end{alignat*}
\end{mydefinition}

\begin{mylemma}[\textbf{Yoneda lemma}] We have $\Sub\,(\yon\,i)\,\Gamma \simeq |\Gamma|\,i$ as an isomorphism of sets, natural in $i$ \cite[Section~III.2]{maclane98categories}.
\end{mylemma}

\noindent\emph{Corollary.} If we choose $\Gamma$ to be $\yon j$, it follows that
$\Sub\,(\yon\,i)\,(\yon\,j) \simeq \mbbC(i,\,j)$, i.e.\ that $\yon$ is
bijective on morphisms; hence it is an embedding.

\begin{notation}
\label{not:yoneda}
For $\gamma : |\Gamma|\,i$, we use $\gamma\lab \blank \rab :
\Sub\,(\yon\,i)\,\Gamma $ to denote transporting right-to-left along the Yoneda
lemma. In the other direction we do not really need a notation, since from
$\sigma : \Sub\,(\yon\,i)\,\Gamma$ we get $|\sigma|\,\id : |\Gamma|\,i$.
\end{notation}

\subsubsection{Type formers}

\begin{mydefinition}
\label{def:k-psh}
\textbf{Constant families} are displayed presheaves which do not depend on their context.
\begin{alignat*}{3}
  & \K &&: \Con \to \{\Gamma : \Con \} \to \Ty\,\Gamma\\
  & |\K\,\Delta|\,\{i\}\,\gamma\,&&\defn |\Delta|\,i \\
  & \delta\lab f \rab &&\defn \delta \lab f \rab
\end{alignat*}
With this definition, we have $\Tm\,\Gamma\,(\K\,\Delta) \equiv \Sub\,\Gamma\,\Delta$
so we have strict constant families.
\end{mydefinition}

\begin{notation}
It is useful to consider any set as a constant presheaf, so
given $A : \Set$ we may write $A : \Con$ for the constant presheaf
as well.
\end{notation}

\begin{mydefinition}
From any $A : \Set$, we get $\K\,A : \Ty\,\Gamma$. This can be used to
model negative or positive \textbf{closed type formers}. For example, natural
numbers are modeled as $\K\,\mbb{N}$, Booleans as $\K\,\Bool$, the unit type as
$\K\,\top$, and so on.
\end{mydefinition}

\begin{mydefinition}
\label{def:presheaf-univ}
\textbf{Coquand universes} can be defined as follows. We write $\Set_{\hmbbC}$
for the outer universe in the model, to distinguish it from the external
$\Set$. Since the $\Set_{\hmbbC}$ is a non-dependent type, it is helpful to
define it as a $\Set_{\hmbbC} : \Con$ such that $\Sub\,\Gamma\,\Set_{\hmbbC}
\simeq \Ty\,\Gamma$.  The usual universe can be derived from this as
$\K\,\Set_{\hmbbC}$. Again, we ignore size issues; the fully formal definition
would involve indexing constructions in $\hmbbC$ by universe levels.

We can take a hint from the Yoneda lemma. We aim to define $|\Set_{\hmbbC}|\,i$,
but by the Yoneda lemma it is isomorphic to $\Sub\,(\yon
i)\,\Set_{\hmbbC}$. However, by specification this should be isomorphic to
$\Ty\,(\yon\,i)$, so we take this as definition:
\begin{alignat*}{3}
  & \Set_{\hmbbC} &&: \Con\\
  &|\Set_{\hmbbC}|\,i &&\defn \Ty\,(\yon\,i)\\
  &A \lab f \rab &&\defn A[\yon f]
\end{alignat*}
\end{mydefinition}
In the $A \lab f \rab$ definition, we substitute $A : \Ty\,(\yon\,i)$ with $\yon
f : \Sub\,(\yon j)\,(\yon i)$ to get an element of $\Ty\,(\yon j)$.  The
required $\Sub\,\Gamma\,\Set_{\hmbbC} \simeq \Ty\,\Gamma$ is straightforward, so
we omit the definition.

We note that Russell universes are not supported in the outer layer, as
$\Sub\,\Gamma\,\Set_{\hmbbC}$ and $\Ty\,\Gamma$ are not strictly the same, in
particular they have a different number of components as iterated
$\Sigma$-types. Nevertheless, as we mentioned in Section
\ref{sec:2ltt-internal-syntax}, we use Russell-style notation in the internal
2LTT syntax, and assume that encoding/decoding is inserted by elaboration.

\begin{mydefinition}
\textbf{$\Sigma$-types} are defined pointwise. The definitions for pairing and
projections follow straightforwardly.
\begin{alignat*}{3}
  & \Sigma  &&: (A : \Ty\,\Gamma) \to \Ty\,(\Gamma\ext A) \to \Ty\,\Gamma\\    & |\Sigma\,A\,B|\,\gamma && \defn (\alpha : |A|\,\gamma) \times |B|\,(\gamma,\,\alpha)\\
  & (\alpha,\,\beta) \lab f \rab && \defn (\alpha \lab f \rab,\, \beta \lab f \rab)
\end{alignat*}
\end{mydefinition}

\begin{mydefinition}
We define \textbf{$\Pi$-types} in the following. This is a bit more complicated,
so first we look at the simpler case of presheaf exponentials. We source this
example from \cite[Section~I.]{maclane2012sheaves}. The reader may refer to
ibid.\ for an overview of constructions in presheaf categories.

The exponential $\Delta^\Gamma : \Con$ is characterized by the isomorphism
$\Sub\,(\Xi \otimes \Gamma)\,\Delta\,\simeq\,\Sub\,\Xi\,(\Delta^\Gamma)$, where
we write $\otimes$ for the pointwise product of two presheaves. We can again use
the Yoneda lemma. We want to define $|\Delta^\Gamma|\,i$, but this is isomorphic
to $\Sub\,(\yon i)\,(\Delta^\Gamma)$, which should be isomorphic to $\Sub\,(\yon
i \otimes \Gamma)\,\Delta$ by the specification of exponentials. Hence:
\begin{alignat*}{3}
  &|\Delta^\Gamma|\,i &&\defn \Sub\,(\yon i \otimes \Gamma)\,\Delta \\
  & \sigma \lab f \rab &&\defn \sigma \circ (\yon f \circ \p,\,\q)
\end{alignat*}
In the definition of presheaf restriction, we use $\p$, $\q$ as projections and
$\blank\!,\!\blank$ as pairing for $\otimes$. In short, $(\yon f \circ \p,\,\q)$
is the same as the morphism lifting from Definition \ref{def:cwfops}: it weakens
$\yon f : \Sub\,(\yon j)\,(\yon i)$ to $\Sub\,(\yon j \otimes \Gamma)\,(\yon i
\otimes \Gamma)$.

The dependently typed case follows the same pattern, except that we use $\Tm$
and $\blank\!\ext\!\blank$ instead of $\Sub$ and
$\blank\!\otimes\!\blank$. Additionally, the action on objects depends on
$\gamma : |\Gamma|\,i$, and we make use of $\gamma \lab \blank \rab\,\,:
\Sub\,(\yon i)\,\Gamma$ (introduced in Notation \ref{not:yoneda}).
\begin{alignat*}{3}
  & \Pi &&: (A : \Ty\,\Gamma) \to \Ty\,(\Gamma \ext A) \to \Ty\,\Gamma\\
  & |\Pi\,A\,B|\,\{i\}\,\gamma &&\defn \Tm\,(\yon i \ext A[\gamma \lab \blank \rab])\,(B[\gamma\lab\blank\rab \circ\,\p,\,\q])\\
  & t\lab f \rab &&\defn t[\yon f \circ \p,\, \q]
\end{alignat*}
Let us unfold the above definition a bit. Assuming $t : |\Pi\,A\,B|\,\{i\}\,\gamma$, we have
\[
|t| : \{j : |\mbbC|\}\to((f,\,\alpha) : (f : \mbbC(j,\,i)) \times |A|\,(\gamma\lab f \rab)) \to
       |B|\,(\gamma\lab f\rab,\,\alpha)
\]
This is a bit clearer if we remove the $\Sigma$-type by currying.
\[
|t| : \{j : |\mbbC|\}(f : \mbbC(j,\,i))(\alpha : |A|\,(\gamma \lab f \rab)) \to
       |B|\,(\gamma\lab f \rab,\,\alpha)
\]

Restriction is functorial since it is defined as $\Tm$ substitution. The definitions
for $\lam$ and $\app$ are left to the reader.
\end{mydefinition}

\begin{mydefinition}
\textbf{Extensional identity} is defined as pointwise equality of sections:
\begin{alignat*}{3}
  & \Id : \Tm\,\Gamma\,A \to \Tm\,\Gamma\,A \to \Ty\,\Gamma\\
  & |\Id\,t\,u|\,\gamma \defn |t|\,\gamma \equiv |u|\,\gamma
\end{alignat*}
For the restriction operation, we have to show that $|t|\,\gamma \equiv
|u|\,\gamma$ implies $|t|\,(\gamma\lab f \rab) \equiv |u|\,(\gamma \lab f
\rab)$. This follows from congruence by $\blank \lab f \rab$ and naturality of
$t$ and $u$.  The defining $(\reflect,\,\refl) : \Tm\,\Gamma\,(\Id\,t\,u) \simeq
(t \equiv u)$ isomorphism is evident from UIP and function extensionality for
the metatheoretic $\blank\!\equiv\!\blank$ relation.
\end{mydefinition}

\subsection{Modeling the Inner Layer}

We assume now that $\mbbC$ is a cwf. We write types as $a,\,b,\,c :
\Ty_\mbbC\,i$ and terms as $t,\,u,\,v : \Tm_\mbbC\,i\,a$. We reuse $\emptycon$
for the terminal object and $\blank\ext\blank$ for context extension, and
likewise reuse notation for substitutions.

\begin{mydefinition}[\textbf{$\Ty_{0}$, $\Tm_{0}$}]
\label{def:psh-ty0-tm0}
First, note that $\Ty_{\mbbC}$ is a presheaf over $\mbbC$, and $\Tm_{\mbbC}$ is
a displayed presheaf over $\Ty_{\mbbC}$; this follows from the requirement that
they form a family structure over $\mbbC$. Hence, in the presheaf model
$\Ty_{\mbbC}$ is an element of $\Con$ and $\Tm_{\mbbC}$ is an element of
$\Ty\,\Ty_{\mbbC}$. Also recall from Definition \ref{def:k-psh} that
$\Tm\,\Gamma\,(\K\,\Delta) \equiv \Sub\,\Gamma\,\Delta$. With this is mind, we
give the following definitions:
\begin{alignat*}{3}
  & \Ty_0 : \Ty\,\Gamma                   && \Tm_0 : \Tm\,\Gamma\,\Ty_0 \to \Ty\,\Gamma\\
  & \Ty_0 \defn \K\,\Ty_{\mbbC}\hspace{2em} && \Tm_0\,A \defn \Tm_{\mbbC}[A]
\end{alignat*}
$\Tm_{\mbbC}[A]$ is well-typed since $A : \Tm\,\Gamma\,(\K\,\Ty_{\mbbC})$, thus
$A : \Sub\,\Gamma\,\Ty_{\mbbC}$. In other words, $A$ is a natural transformation
from $\Gamma$ to the presheaf of inner types.
\end{mydefinition}

\subsubsection{Inner type formers}

Can type formers in $(\Ty_{\mbbC},\,\Tm_{\mbbC})$ be transferred to
$(\Ty_0,\,\Tm_0)$ in the presheaf model of 2LTT? For example, if $\mbbC$
supports $\Bool$, we would like to model $\Bool_0$ in $\Ty_0$ as well. The
following explanation is adapted from Capriotti \cite[Section
  2.3]{capriotti2017models}.

Generally, a type former in $\mbbC$ transfers to $\hmbbC$ if it can be specified
in the internal language of $\hmbbC$. This is also a core idea of HOAS: when
working in $\hmbbC$, everything is natural, and we can omit boilerplate related
to contexts and substitutions. Recall from Section \ref{sec:2ltt-internal-syntax}
the higher-order specification of inner functions:
\begin{alignat*}{3}
  &\Pi_0             &&: (A : \Ty_0) \to (\Tm_0\,A \to \Ty_0) \to \Ty_0\\
  &(\app_0,\,\lam_0) &&: \Tm_0\,(\Pi_0\,A\,B) \simeq ((a :\,\Tm_0\,A) \to \,\Tm_0\,(B\,a))
\end{alignat*}
We can say that this \emph{defines} what it means for $\mbbC$ to support
$\Pi$. More precisely:
\begin{itemize}
 \item A type former in an object theory is specified with a closed
       type $A$ in a 2LTT.
 \item A model of the object theory $\mbbC$ supports a type former $A$ if there is a
       global section of the $\hmbbC$ interpretation of $A$.
 \item Thus, if $\mbbC$ supports a type former, it is immediate that
       the specifying type is inhabited in $\hmbbC$.
\end{itemize}
In this thesis we only mention type formers in type theories which
can be specified in such manner.

\subsection{Functions With Inner Domains}
\label{sec:functions-with-inner-domains}

There is a useful semantic simplification in the standard presheaf model, in
cases where we have functions of the form $\Pi\,(\Tm_0\,A)\,B$. This greatly
reduces encoding overhead when interpreting inductive signatures in 2LTT; we
look at examples in Section \ref{sec:2ltt-simple-signatures}. First we look at
the simply-typed case with presheaf exponentials.
\begin{mylemma}
$\yon$ preserves finite products up to isomorphism, i.e.\ $\yon \emptycon \simeq
  \emptycon$ and $\yon (i \otimes j) \simeq (\yon i \otimes \yon j)$.
\end{mylemma}
\begin{proof}
$\yon \emptycon$ is $\mbbC(\blank,\,\emptycon)$ by definition, which is
pointwise isomorphic to $\top$, hence isomorphic to $\emptycon \equiv
\K\,\top$. $\yon (i \otimes j)$ is $\mbbC(\blank,\,i \otimes j)$, which is
isomorphic to $\yon i \otimes \yon j$ by the specification of products.
\end{proof}
\begin{mylemma} We have the following isomorphism.
\begin{alignat*}{3}
  & |\Gamma^{\yon i}|\,j \equiv \hspace{3em}&&\\
  & \Sub\,(\yon j \otimes \yon i)\,\Gamma \simeq \hspace{3em} &&\text{by product preservation}\\
  & \Sub\,(\yon (j \otimes i))\,\Gamma \simeq \hspace{3em} &&\text {by Yoneda lemma}\\
  & |\Gamma|\,(j \otimes i)&&
\end{alignat*}
\end{mylemma}

It is possible to rephrase the above derivation for $\Pi$-types. For that, we
would need to define the action of $\yon$ on types and terms, consider the
preservation of $\blank\ext\blank$ by $\yon$, and also specify a ``dependent''
Yoneda lemma for $\Tm$. For the sake of brevity, we omit this, and present the
result directly:
\begin{alignat*}{3}
  & |\Pi\,(\Tm_0\,A)\,B|\,\{i\}\,\gamma \simeq |B|\,\{i \ext |A|\,\gamma\}\,(\gamma \lab \p \rab,\,\q)
\end{alignat*}
In short, depending on an inner domain is the same as depending on an extended
context in $\mbbC$.  We expand a bit on the typing of the right hand side. We
have $\gamma : |\Gamma|\,i$, moreover
\begingroup
\allowdisplaybreaks
\begin{alignat*}{3}
  & |B| &&: \{j : \mbbC\} \to |\Gamma\,\,\ext \Tm_0\,A|\,j \to \Set\\
  & |B| &&: \{j : \mbbC\} \to ((\gamma' : |\Gamma|\,j)\times \Tm_{\mbbC}\,j\,(|A|\,\gamma')) \to \Set\\
  & |B|\,\{i \ext |A|\,\gamma\} &&: ((\gamma' : |\Gamma|\,(i \ext |A|\,\gamma))\times \Tm_{\mbbC}\,(i \ext |A|\,\gamma)\,(|A|\,\gamma')) \to \Set\\
  & \gamma \lab \p \rab &&: |\Gamma|\,(i \ext |A|\,\gamma)\\
  & \q &&: \Tm_{\mbbC}\,(i \ext |A|\,\gamma)\,((|A|\,\gamma)[\p])\\
  & \q &&: \Tm_{\mbbC}\,(i \ext |A|\,\gamma)\,(|A|\,(\gamma \lab \p \rab))\\
\end{alignat*}
\endgroup

\section{Simple Signatures in 2LTT}
\label{sec:2ltt-simple-signatures}

We revisit simple inductive signatures in this section, working internally to
2LTT. We review the concepts introduced in Chapter
\ref{chap:simple-inductive-signatures} in the same order.

\begin{notation}
In this section we shall be fairly explicit about writing $\Tm_0$-s and
transporting along definitional isomorphisms. The simple setting makes it
feasible to be explicit; in later chapters we are more terse, as signatures and
semantics get more complicated.
\end{notation}

\subsection{Theory of Signatures}
Signatures are defined exactly in the same way as before: we have $\Con : \Set$,
$\Ty : \Set$, $\Sub : \Con \to \Con \to \Set$, $\Var : \Con \to \Ty \to \Set$ and
$\Tm : \Con \to \Ty \to \Set$. However, now by $\Set$ we mean the outer universe
of 2LTT. Thus signatures are inductively defined in the outer layer.

\subsection{Algebras}
\label{sec:2ltt-simple-algebras}

Again we compute algebras by induction on signatures, but now we use inner
types for carriers of algebras. We interpret types as follows:
\begin{alignat*}{3}
& \hspace{-4em} \rlap{$\blank^A : \Ty \to \Ty_0 \to \Set$} \\
& \hspace{-4em} \iota^A\,&&X \defn \Tm_0\,X \\
& \hspace{-4em} (\iota\to A)^A\,&&X \defn \Tm_0\,X \to A^A\,X
\end{alignat*}
Elsewhere, we change the type of the $X$ parameters accordingly:
\begin{alignat*}{3}
& \blank^A &&: \Con \to \Ty_0 \to \Set\\
& \blank^A &&: \Var\,\Gamma\,A \to \{X : \Ty_0\} \to \Gamma^A\,X \to A^A\,X\\
& \blank^A &&: \Tm\,\Gamma\,A \to \{X : \Ty_0\} \to \Gamma^A\,X \to A^A\,X\\
& \blank^A &&: \Sub\,\Gamma\,\Delta \to \{X : \Ty_0\} \to \Gamma^A\,X \to \Delta^A\,X
\end{alignat*}
We also define $\ms{Alg}\,\Gamma$ as $(X : \Ty_0) \times \Gamma^A\,X$.

\begin{myexample}
Inside 2LTT we have the following:\footnote{Up to isomorphism, since we previously defined $\Gamma^A$ as a function type instead of an iterated product type.}
\[ \Alg\,\ms{NatSig} \equiv (X : \Ty_0)\times(\mi{zero} : \Tm_0\,X)\times(\mi{suc} : \Tm_0\,X \to \Tm_0\,X) \]
Then, we may assume any cwf $\mbbC$, and interpret the above closed type in the
presheaf model $\hmbbC$, and evaluate the result at $\emptycon$ and
the unique element of the terminal presheaf $\K\,\top$:
\[
  |\Alg\,\ms{NatSig}|\,\{\emptycon\}\,\tt : \Set
\]
We compute the definitions now. We use the simplified semantics for $\mi{suc} :
\Tm_0\,X \to \Tm_0\,X$, since the function domain is an inner type.
\[
|\Alg\,\ms{NatSig}|\,\{\emptycon\}\,\tt \equiv
(X : \Ty_{\mbbC}\,\emptycon) \times (\mi{zero} : \Tm_{\mbbC}\,\emptycon\,X) \times (\mi{suc} : \Tm_{\mbbC}\,(\emptycon \ext X)\,X)
\]
Using the same computation, we get the following for binary trees:
\[
|\Alg\,\ms{TreeSig}|\,\{\emptycon\}\,\tt \equiv
(X : \Ty_{\mbbC}\,\emptycon) \times (\mi{leaf} : \Tm_{\mbbC}\,\emptycon\,X) \times (\mi{node} : \Tm_{\mbbC}\,(\emptycon \ext X \ext X)\,X)
\]
\end{myexample}

We can also get internal algebras in any $\mbbC$ category with finite products
because we can build cwfs from all such $\mbbC$.

\begin{mydefinition} Assuming $\mbbC$ with finite products, we build a cwf by setting
$\Con \defn |\mbbC|$, $\Ty\,\Gamma \defn |\mbbC|$, $\Sub\,\Gamma\,\Delta \defn \mbbC(\Gamma,\,\Delta)$, $\Tm\,\Gamma\,A \defn \mbbC(\Gamma,\,A)$, $\Gamma \ext A \defn \Gamma \otimes A$ and $\emptycon \defn \emptycon_{\mbbC}$. In short, we build a non-dependent (simply-typed) cwf.
\end{mydefinition}

Now we can effectively interpret signatures in finite product categories. For
example:
\[
|\Alg\,\ms{NatSig}|\,\{\emptycon\}\,\tt \equiv
(X : |\mbbC|) \times (\mi{zero} : \mbbC(\emptycon,\,X)) \times (\mi{suc} : \mbbC(\emptycon \otimes X,\,X))
\]
This is almost the same as what we would write by hand for the specification of
natural number objects; the only difference is the extra $\emptycon
\otimes\blank$ in $\mi{suc}$.

\subsection{Morphisms}

We get an additional degree of freedom in the computation of morphisms:
preservation equations can be inner or outer. The former option is \emph{weak}
or \emph{propositional} preservation, while the latter is \emph{strict}
preservation. In the presheaf model of 2LTT, outer equality is definitional
equality of inner terms, while inner equality is propositional equality in the
inner theory. Of course, if the inner theory has extensional identity type,
weak and strict equations in 2LTT are equivalent for inner types. We compute
weak preservation for types as follows.
\begin{alignat*}{3}
  & \hspace{-5em}\rlap{$\blank^M : (A : \Ty)\{X_0\,X_1 : \Ty_0\}(X^M : \Tm_0\,X_0 \to \Tm_0\,X_1) \to A^A\,X_0 \to A^A\,X_1 \to \Set$}\\
  & \hspace{-5em}\iota^M\,&&X^M\,\alpha_0\,\,\alpha_1 \defn \Tm_0\,(X^M\,\alpha_0 = \alpha_1) \\
  & \hspace{-5em}(\iota\to A)^M\,&&X^M\,\alpha_0\,\,\alpha_1 \defn
       (x : \Tm_0\,X_0) \to A^M\,X^M\,(\alpha_0\,x)\,(\alpha_1\,(X^M\,x))
\end{alignat*}
For strict preservation, we simply change $\Tm_0\,(X^M\,\alpha_0 = \alpha_1)$ to
$X^M\,\alpha_0 \equiv \alpha_1$. The definition of morphisms is the same as
before:
\begin{alignat*}{3}
  &\blank^M : (\Gamma : \Con_1)\{X_0\,X_1 : \Ty_0\} \to (\Tm_0\,X_0 \to \Tm_0\,X_1) \to \Gamma^A\,X_0 \to \Gamma^A\,X_1 \to \Set\\
  &\Gamma^M\,X^M\,\gamma_0\,\gamma_1 \defn
  \{A\}(x : \Var_1\,\Gamma\,A) \to A^M\,X^M\,(\gamma_0\,x)\,(\gamma_1\,x)\\
  & \\
  &\Mor : \{\Gamma : \Con_1\} \to \Alg\,\Gamma \to \Alg\,\Gamma \to \Set \\
  &\Mor\,\{\Gamma\}\,(X_0,\,\gamma_0)\,(X_1,\,\gamma_1) \defn (X^M : \Tm_0\,X_0 \to \Tm_0\,X_1) \times \Gamma^M\,X^M\,\gamma_0\,\gamma_1
\end{alignat*}
We omit here the $\blank^M$ definitions for terms and substitutions.

\subsection{Displayed Algebras}

We present $\blank^D$ only for types below.
\begin{alignat*}{3}
  & \rlap{$\blank^D : (A : \Ty)\{X\} \to (\Tm_0\,X \to \Ty_0) \to A^A\,X \to \Set$}\\
  & \iota^D\,       && X^D\,\alpha \defn \Tm_0\,(X^D\,\alpha) \\
  & (\iota\to A)^D\,&& X^D\,\alpha \defn (x : \Tm_0\,X)(x^D : \Tm_0\,(X^D\,x)) \to A^D\,X^D\,(\alpha\,x)
\end{alignat*}
Note that in the presheaf model, inhabitants of $\Tm_0\,X \to \Ty_0$ are inner types
depending on contexts extended with the interpretation of $X$.
\begin{myexample} Assume a closed $(X,\,\mi{zero},\,\mi{suc})$ $\Nat$-algebra in 2LTT. We have the following computation:
\begin{alignat*}{3}
  & \hspace{-5em}\rlap{$\DispAlg\,\{\ms{NatSig}\}\,(X,\,\mi{zero},\,\mi{suc}) \equiv$}\\
              & (X^D &&: \Tm_0\,X \to \Ty_0)\\
      \times\,& (\mi{zero^D} &&: \Tm_0\,(X^D\,\mi{zero}))\\
      \times\,& (\mi{suc^D} &&: (n : \Tm_0\,X) \to \Tm_0\,(X^D\,n) \to \Tm_0\,(X^D\,(\mi{suc}\,n)))
\end{alignat*}
Let us look at the presheaf interpretation. We simplify functions with inner
domains everywhere. Also note that for $\ms{suc} : \Tm_0\,X \to \Tm_0\,X$, we
get $|\ms{suc}|\,\tt : \Tm_{\mbbC}\,(\emptycon\ext n : |X|\,\tt)\,(|X|\,\tt)$ in
the semantics, so a $\ms{suc}\,t$ application is translated as a substitution
$(|\ms{suc}|\,\tt)[n \mapsto |t|\,\tt]$.
\begin{alignat*}{3}
  & \rlap{$|\DispAlg\,\{\ms{NatSig}\}\,(X,\,\mi{zero},\,\mi{suc})|\,\{\emptycon\}\,\tt \equiv$}\\
              & (X^D &&: \Ty_{\mbbC}\,(\emptycon\ext n : |X|\,\tt))\\
      \times\,& (\mi{zero^D} &&: \Tm_{\mbbC}\,\emptycon\,(X^D[n \mapsto |\mi{zero}|\,\tt]))\\
      \times\,& (\mi{suc^D} &&: \Tm_{\mbbC}\,(\emptycon\ext n : |X|\,\tt \ext n^D : X^D[n \mapsto |\mi{zero}|\,\tt])\,(X^D[n \mapsto (|\mi{suc}|\,\tt)[n \mapsto n]))
\end{alignat*}
\end{myexample}
To explain $(|\ms{suc}|\,\tt)[n \mapsto n])$: we have $\ms{suc}\,n$ in 2LTT,
where $n$ is an inner variable, and in the presheaf model inner variables become
actual variables in the inner theory. Hence, we map the $n$ which $\ms{suc}$
depends on to the concrete $n$ in the context.

We can also interpret displayed algebras in finite product categories:
\begin{alignat*}{3}
  & \hspace{-2em}\rlap{$|\DispAlg\,\{\ms{NatSig}\}\,(X,\,\mi{zero},\,\mi{suc})|\,\{\emptycon\}\,\tt \equiv$}\\
              & (X^D &&: |\mbbC|)\\
      \times\,& (\mi{zero^D} &&: \mbbC(\emptycon,\,X^D))\\
      \times\,& (\mi{suc^D} &&: \mbbC(\emptycon \otimes |X|\,\tt \otimes X^D,\, X^D))
\end{alignat*}

While displayed algebras in cwfs can be used as bundles of induction motives and
methods, in finite product categories they are argument bundles to
\emph{primitive recursion}; this is sometimes also called a
\emph{paramorphism} \cite{bananas}. In an internal syntax, the type of primitive
recursion for natural numbers could be written more compactly as:
\begin{alignat*}{3}
  & \ms{primrec} : (X : \Set) \to X \to (\Nat \to X \to X) \to \Nat \to X
\end{alignat*}
This is not the same thing as the usual recursion principle (corresponding to
weak initiality) because of the extra dependency on $\Nat$ in the method for
successors.

\subsection{Sections}
Sections are analogous to morphisms. We again have a choice between weak and
strict preservation; below we have weak preservation.
\begin{alignat*}{3}
  & \rlap{$\blank^S : (A : \Ty)\{X\,X^D\}(X^S : (x : \Tm_0\,X)\to \Tm_0\,(X^D\,x))$}\\
  & \hspace{2em}\rlap{$\to (\alpha : A^A\,X) \to A^D\,X^D\,\alpha \to \Set$}\\
  & \iota^S\,&&X^S\,\alpha\,\,\alpha^D \defn \Tm_0\,(X^S\,\alpha = \alpha^D) \\
  & (\iota\to A)^S\,&&X^S\,\alpha\,\,\alpha^D \defn
  (x : \Tm_0\,X) \to A^S\,X^S\,(\alpha\,x)\,(\alpha^D\,(X^S\,x))
\end{alignat*}

\subsection{Term Algebras}
\label{sec:simple-2ltt-term-algebras}

For term algebras, we need to assume a bit more in the inner theory. For
starters, it has to support the theory of signatures. In order to avoid name
clashes down the line, we use $\SigTy_0$ to refer to signature types, and
$\SigTm_0$ for terms. That is, we have
\begin{alignat*}{3}
  & \SigTy_0  &&: \Ty_0\\
  & \Con_0   &&: \Ty_0\\
  & \Var_0   &&: \Tm_0\,\Con_0 \to \Tm_0\,\SigTy_0 \to \Ty_0\\
  & \SigTm_0 &&: \Tm_0\,\Con_0 \to \Tm_0\,\SigTy_0 \to \Ty_0\\
  & \Sub_0   &&: \Tm_0\,\Con_0 \to \Tm_0\,\Con_0 \to \Ty_0
\end{alignat*}
together with all constructors and induction principles. We also assume inner
$\Pi$-types because we previously defined $\Sub$ using functions.

\emph{Remark.} If we only want to construct term algebras, it is not necessary
to assume inner induction principles. In this section, our goal is to redo the
constructions of Chapter \ref{chap:simple-inductive-signatures} without making
essential changes, so we just assume everything that was available there.

We still have ToS in the outer layer. To make the naming
scheme consistent, we shall write outer ToS types as $\SigTy_1$,
$\SigTm_1$, $\Con_1$, $\Var_1$ and $\Sub_1$. We have conversion functions from
the outer ToS to the inner ToS:
\begin{mydefinition}
\label{def:simple-lowering}
We have the following \textbf{lowering} functions which
preserve all structure.
\begin{alignat*}{3}
  & \down\,: \SigTy_1 &&\to \Tm_0\,\SigTy_0\\
  & \down\,: \Con_1 &&\to \Tm_0\,\Con_0\\
  & \down\,: \Var_1\,\Gamma\,A &&\to \Tm_0\,(\Var_0\,(\down\!\Gamma)\,(\down\!A))\\
  & \down\,: \SigTm_1\,\Gamma\,A &&\to \Tm_0\,(\SigTm_0\,(\down\!\Gamma)\,(\down\!A))\\
  & \down\,: \Sub_1\,\Gamma\,\Delta &&\to \Tm_0\,(\Sub_0\,(\down\!\Gamma)\,(\down\!\delta))
\end{alignat*}
These are called ``lifting'' or ``serialization'' in the context of
multi-stage programming; see e.g.\ the $\ms{Lift}$ typeclass in Haskell
\cite{pickering-multistage}. There, like here, the point is to build
object-language terms from meta-level (``compile-time'') values.

Lowering is straightforward to define for types, contexts, variables and terms,
but there is a bit of a complication for $\Sub$. Unfolding the definitions, we
need to map from $\{A\} \to \Var_1\,\Delta\,A \to \SigTm_1\,\Gamma\,A$ to
$\Tm_0\,(\{A\} \to \Var_0\,(\down\Delta)\,A \to \SigTm_0\,(\down\Gamma)\,A)$. It
might appear problematic that we have types and variables in \emph{negative}
position because we cannot map inner types/variables to outer ones.
Fortunately, $\Sub_1\,\Gamma\,\Delta$ is isomorphic to a finite product type,
and we can lower a finite product component-wise.

Concretely, we define lowering by induction on $\Delta$, while making use of
a case splitting operation for $\Var_0$. We use an informal $\ms{case}$
operation below, which can be defined using inner induction. Note that since
$\Var_0\,\emptycon\,A$ is an empty type, case splitting on it behaves like
elimination for the empty type.
\begin{alignat*}{3}
  &\rlap{$\hspace{0.3em}\down_{\Delta} : \Sub_1\,\Gamma\,\Delta \to \Tm_0\,(\Sub_0\,(\down\!\Gamma)\,(\down\!\Delta))$}\\
  &\down_{\emptycon}\,&&\sigma \defn
      \lambda\,\{A\}\,(x : \Var_0\,\emptycon\,A).\,\ms{case}\,x\,\ms{of}\,()\\
  &\down_{\Delta\ext B}\,&&\sigma \defn
      \lambda\,\{A\}\,(x : \Var_1\,(\down\!\Delta\,\ext \down\!B)\,A).\,\ms{case}\,x\,\ms{of}\\
  & &&\vz\hspace{0.65em}\to\,\,\down\!(\sigma\,\vz)\\
  & &&\vs\,x \to\,\,\down_{\Delta}\!(\sigma \circ \vs)\,x
\end{alignat*}
In general, for finite $A$ type, functions of the form $A \to \Tm_0\,B$ can be
represented as inner types up to isomorphism; they can be viewed as finite
products of terms.

\emph{Remark.} For infinite $A$ this does not work anymore in our system. In
\cite{twolevel}, the assumption that this still works with $A \equiv \Nat_1$ is
an important axiom (``cofibrancy of $\Nat_1$'') which makes it possible to embed
higher categorical structures in 2LTT. From the metaprogramming perspective,
cofibrancy of $\Nat_1$ implies that the inner theory is \emph{infinitary}, since
we can form inner terms from infinite collections of inner terms. We do not
assume this axiom in 2LTT, although we will consider infinitary (object) type
theories in Chapters \ref{chap:fqiit} and \ref{chap:iqiit}.
\end{mydefinition}

\noindent
We proceed to the definition of term algebras. We fix an $\Omega : \Con_1$, and
define $\ms{T} : \Ty_0$ as $\SigTm_0\,(\down\!\Omega)\,\iota$.
\begingroup
\allowdisplaybreaks
\begin{alignat*}{3}
  & \hspace{-5em}\rlap{$\blank^T : (A : \SigTy_1) \to \Tm_0\,(\SigTm_0\,(\down\!\Omega)\,(\down\!A)) \to A^A\,\ms{T}$} \\
  & \hspace{-5em}\iota^T\,&&t \defn t \\
  & \hspace{-5em}(\iota\to A)^T\,&&t \defn \lambda\,u.\,A^T\,(\app\,t\,u)\\
  & \hspace{-5em}&&\\
  & \hspace{-5em}\rlap{$\blank^T : (\Gamma : \Con_1) \to \Sub_1\,\Omega\,\Gamma \to \Gamma^A\,\ms{T}$}\\
  & \hspace{-5em}\rlap{$\Gamma^T\,\nu\,\{A\}\,x \defn A^T\,(\down\!(\nu\,x))$}\\
  & \hspace{-5em} && \\
  & \hspace{-5em}\rlap{$\TmAlg_{\Omega} : \Alg\,\Omega$}\\
  & \hspace{-5em}\rlap{$\TmAlg_{\Omega} \defn \Omega^T\,\Omega\,\id$}
\end{alignat*}
\endgroup
We omit the $\blank^T$ interpretation for terms and substitutions for now, as
they require a bit more setup, and they are not needed just for term algebras.

\subsection{Recursor Construction}

Recall from Section \ref{sec:simple-weak-initiality} that recursion is
implemented using the $\blank^A$ interpretation of terms. Since terms are now in
the inner theory, we need to define an inner version of the same interpretation.
We need to compute types by inner induction, so we additionally assume a
Russell-style inner $\U_0$ universe. The Russell style means that we may freely
coerce between $\Tm_0\,\U_0$ and $\Ty_0$. The following are defined the same way
as $\blank^A$ before.
\begin{alignat*}{3}
  &\blank^A : \Tm_0\,(\SigTy_0 \to \U_0 \to \U_0)\\
  &\blank^A : \Tm_0\,(\Con_0 \to \U_0 \to \U_0)\\
  &\blank^A : \Tm_0\,(\SigTm_0\,\Gamma\,A \to \{X : \U_0\} \to \Gamma^A\,X \to A^A\,X)\\
  &\blank^A : \Tm_0\,(\Sub_0\,\Gamma\,\Delta \to \{X : \U_0\} \to \Gamma^A\,X \to \Delta^A\,X)
\end{alignat*}
Since lowering preserves all structure, and $\blank^A$ is defined in the same
way in both the inner and outer theories, lowering is compatible with
$\blank^A$ in the following way.
\begin{mylemma}\label{lem:down-compat-alg} Assume $A : \SigTy_1$, $\Gamma : \Con_1$, $X : \Ty_0$, $\gamma : \Gamma^A\,X$ and $t : \SigTm_1\,\Gamma\,A$. We have the following:
  \begin{itemize}
  \item $(A^A_{\to},\,A^A_{\leftarrow}) : \Tm_0\,((\down\!A)^A\,X) \simeq A^A\,X$
  \item $(\Gamma^A_{\to},\,\Gamma^A_{\leftarrow}) : \Tm_0\,((\down\!\Gamma)^A\,X) \simeq \Gamma^A\,X$
  \item $t^A\,\gamma \equiv A^A_{\to}\,((\down\!t)^A\,(\Gamma^A_{\leftarrow}\,\gamma))$
  \end{itemize}
\end{mylemma}
\begin{proof}
By induction on $\Gamma$, $A$ and $t$.
\end{proof}
We construct recursors now, yielding strict algebra morphisms.  We
assume $(X,\,\omega) : \Alg\,\Omega$. Recall that $\omega : \Omega^A\,X$, thus
$\Omega^A_{\leftarrow}\,\omega : \Tm_0\,((\down\!\Omega)^A\,X)$. We define $\ms{R} :
\Tm_0\,\ms{T} \to \Tm_0\,X$ as $\ms{R}\,t \defn t^A\,(\Omega^A_{\leftarrow}\,\omega)$.
\begin{alignat*}{3}
& \hspace{-10em}\rlap{$\blank^R : (A : \SigTy_1)(t : \Tm_0\,(\SigTm_0\,(\down\!\Omega)\,(\down\!A))) \to A^M\,\ms{R}\,(A^T\,t)\,(\A^A_{\rightarrow}\,(t^A\,(\Omega^A_{\leftarrow}\,\omega)))$}\\
& \hspace{-10em}\iota^R\,&&t : t^A\,(\Omega^A_{\leftarrow}\,\omega) \equiv \iota^A_{\to}\,(t^A\,(\Omega^A_{\leftarrow}\,\omega))\\
& \hspace{-10em}(\iota\to A)^R\,&&t \defn \lambda\,u.\,A^R\,(\app\,t\,u)
\end{alignat*}
\begin{alignat*}{3}
& \hspace{-9em}\rlap{$\blank^R : (\Gamma : \Con_1)(\nu : \Sub_1\,\Omega\,\Gamma) \to \Gamma^M\,\ms{R}\,(\Gamma^T\,\nu)\,(\nu^A\,\omega)$}\\
& \hspace{-9em}\rlap{$\Gamma^R\,\nu\,\{A\}\,x \defn A^R\,(\down\!(\nu\,x))$}
\end{alignat*}
In the proof obligation for $t^A\,(\Omega^A_{\leftarrow}\,\omega) \equiv
\iota^A_{\to}\,(t^A\,(\Omega^A_{\leftarrow}\,\omega))$, $\iota^A_{\to}$ computes
to the identity function; note that $\iota^A_{\to} : \Tm_0\,X \to \Tm_0\,X$. Hence
the equality becomes reflexive.

In $\Gamma^R\,\nu\,\{A\}\,x \defn A^R\,(\down\!(\nu\,x))$, we have that
\[
  A^R\,(\down\!(\nu\,x)) : A^M\,\ms{R}\,(A^T\,(\down\!(\nu\,x)))\,(A^A_{\to}\,(\down\!(\nu\,x)^A\,(\Omega^A_{\leftarrow}\,\omega)))
\]
Hence by Lemma \ref{lem:down-compat-alg}, we have
\[
  A^R\,(\down\!(\nu\,x)) : A^M\,\ms{R}\,(A^T\,(\down\!(\nu\,x)))\,((\nu\,x)^A\,\omega)
\]
Hence, by the definition of $\blank^A$ for substitutions:
\[
  A^R\,(\down\!(\nu\,x)) : A^M\,\ms{R}\,(A^T\,(\down\!(\nu\,x)))\,(\nu^A\,\omega\,x)
\]
Which is exactly what is required when we unfold the expected return type:
\begin{alignat*}{3}
  & \blank^R : (\Gamma : \Con_1)(\nu : \Sub_1\,\Omega\,\Gamma) \to \Gamma^M\,\ms{R}\,(\Gamma^T\,\nu)\,(\nu^A\,\omega)\\
  & \blank^R : (\Gamma : \Con_1)(\nu : \Sub_1\,\Omega\,\Gamma) \to \{A\}(x : \Var_1\,\Gamma\,A) \to
    A^M\,\ms{R}\,(A^T\,(\down\!(\nu\,x)))\,(\nu^A\,\omega\,x)
\end{alignat*}
The recursor is defined the same way as in Definition \ref{def:simple-recursor}:
\begin{alignat*}{3}
  & \Rec_{\Omega} : (\mi{alg} : \Alg\,\Omega) \to \Mor\,\TmAlg_{\Omega}\,\mi{alg}\\
  & \Rec_{\Omega}\,(X,\,\omega) \defn (\ms{R},\,\Omega^R\,\Omega\,\id)
\end{alignat*}

\subsection{Eliminator Construction}

For induction, we need to additionally define $\blank^D$ in the inner layer.
\begin{alignat*}{3}
  &\blank^D : \Tm_0\,((A : \SigTy_0)\{X\} \to (\Tm_0\,X \to \U_0) \to A^A\,X \to \U_0)\\
  &\blank^D : \Tm_0\,((\Gamma : \Con_0)\{X\} \to (\Tm_0\,X \to \U_0) \to \Gamma^A\,X \to \U_0)\\
  &\blank^D : \Tm_0\,((t : \SigTm_0\,\Gamma\,A) \to \Gamma^D\,X^D\,\gamma \to A^D\,X^D\,(t^A\,\gamma))\\
  &\blank^D : \Tm_0\,((\sigma : \Sub_0\,\Gamma\,\Delta) \to \Gamma^D\,X^D\,\gamma \to \Delta^D\,X^D\,(\sigma^A\,\gamma))
\end{alignat*}

\begin{mylemma}
We have again compatibility of lowering with $\blank^D$. Assuming
$(X,\,\gamma) : \Alg\,\Gamma$, $(X^D,\,\gamma^D) : \DispAlg\,(X,\,\gamma)$,
$t : \SigTm_1\,\Gamma\,A$, and $\alpha : A^A\,X$, we have
\begin{itemize}
  \item $(\A^D_{\to},\,\A^D_{\leftarrow}) :
    \Tm_0\,((\down\!\A)^D\,X^D\,(\A^A_{\leftarrow}\,\alpha)) \simeq \A^D\,X^D\,\alpha$
  \item $(\Gamma^D_{\to},\,\Gamma^D_{\leftarrow}) :
    \Tm_0\,((\down\!\Gamma)^D\,X^D\,(\Gamma^A_{\leftarrow}\,\gamma)) \simeq \Gamma^D\,X^D\,\gamma$
  \item $t^D\,\gamma^D \equiv A^D_{\to}\,((\down\!t)^D\,(\Gamma^D_{\leftarrow}\,\gamma^D))$
\end{itemize}
The equation for $t^D\,\gamma^D$ is well-typed because of the term equation in Lemma
\ref{lem:down-compat-alg}.
\end{mylemma}
\begin{proof} Again by induction on $\Gamma$, $A$ and $t$.
\end{proof}

We also need to extend $\blank^T$ with action on terms. Note that we return
an inner equality, since we can only compute such equality by induction on the
inner term input:
\[
\blank^T : (t : \SigTm_0\,(\down\!\Gamma)\,(\down\!A))(\nu : \Sub_1\,\Omega\,\Gamma) \to
  \Tm_0\,(A^A_{\leftarrow}\,(A^T\,(t[\down\!\nu])) = t^A\,(\Gamma^A_{\leftarrow}\,\nu))
\]
We assume $(X^D,\,\omega^D) : \DispAlg\,\TmAlg_\Omega$, and define elimination
as follows:
\begin{alignat*}{3}
  & \ms{E} : (t : \Tm_0\,\ms{T}) \to \Tm_0\,(X^D\,t) \\
  & \ms{E}\,t \defn t^D\,(\Omega^D_{\leftarrow}\,\omega^D)
\end{alignat*}
This definition is well-typed only up to $t^T\,\id : \Tm_0\,(t =
t^A\,(\Omega^A_{\leftarrow}\,(\Omega^T\,\Omega\,\id)))$. Since $t^T\,\id$ is an
inner equality, in a fully formal intensional presentation we would have to
write an explicit transport in the definition.

We shall skip the remainder of the eliminator construction; it goes the same way
as in Definition \ref{def:simple-eliminator-construction}. Intuitively, this is
possible since the inner theory has all necessary features to reproduce the
eliminator construction, and lowering preserves all structure.

Since $t^T$ yields inner equations, this implies that the displayed algebra
sections returned by the eliminator are \emph{weak sections}, i.e.\ they contain
$\beta$-rules expressed in inner equalities.

\subsection{Strict Elimination}

If we want to use term algebras in generic programming, having only weak
$\beta$-rules in elimination is inconvenient. We make a brief digression here,
to define an alternative eliminator which computes strictly. The idea is to
specialize the notion of displayed algebras to the term algebra, and likewise
give a specialized definition for the eliminator function. We fix $\Omega :
\Con_1$ and $X^D : \Tm_0\,(\SigTm_0\,(\down\!\Omega)\,\iota) \to \Ty_0$.
\begin{alignat*}{3}
  & \rlap{$\blank^D : (A : \SigTy_1) \to \Tm_0\,(\SigTm_0\,(\down\!\Omega)\,(\down\!A)) \to \Ty_0$}\\
  & \iota^D\,       && X^D\,\alpha \defn X^D\,\alpha \\
  & (\iota\to A)^D\,&& X^D\,\alpha \defn (u : \SigTm_0\,\Omega\,\iota) \to X^D\,u \to A^D\,X^D\,(\alpha\,u)
\end{alignat*}
\begin{alignat*}{3}
  &\Omega^D : \Set\\
  &\Omega^D \defn \{A : \SigTy_1\}(x : \Tm_0\,(\Var_0\,(\down\!\Omega)\,(\down\!A))) \to \Tm_0\,(A^D\,(\var\,x))
\end{alignat*}
\begin{alignat*}{3}
  & \hspace{-10em}\rlap{$\ms{Elim} : \{A : \SigTy_1\} \to \Omega^D \to  (t : \Tm_0\,(\SigTm_0\,(\down\!\Omega)\,(\down\!A))) \to \Tm_0\,(A^D\,t)$}\\
  & \hspace{-10em}\ms{Elim}\,\omega^D\,(\var\,x) &&\defn \omega^D\,x\\
  & \hspace{-10em}\ms{Elim}\,\omega^D\,(\app\,t\,u) &&\defn \ms{Elim}\,\omega^D\,t\,u\,(\ms{Elim}\,\omega^D\,u)
\end{alignat*}
Now, $\ms{Elim}\,\{\iota\}$ has type $\Omega^D \to \Tm_0\,((t :
\SigTm_0\,(\down\!\Omega)\,\iota) \to X^D\,t)$. Since $\Omega^D$ is a finite
product of inner types, it's isomorphic to an inner type, so we can extract a
purely inner eliminator:
\[ \ms{Elim} : (X^D : \SigTm_0\,(\down\!\Omega)\,\iota \to \U_0) \to \Omega^D \to (t : \SigTm_0\,(\down\!\Omega)\,\iota) \to X^D\,t.\]
Here, $\Omega^D$ specifies induction methods, and the eliminator is defined by
inner induction on terms. Compare this to the previously constructed weak
eliminator, where we had to transport the result over $t^T\,\id$. The extra
transport precluded strict $\beta$-rules in that case, since transports do not
definitionally compute on inductive constructors in the inner theory.

However, the weak eliminator construction is overall more regular and scales
better to more complicated theories of signatures, as we will see in Sections
\ref{sec:fqiit-term-algebras} and \ref{sec:inf-term-algebras}. Also, in these
Sections we will assume equality reflection everywhere so weak and strict
$\beta$-rules will coincide. Another advantage of the weak eliminator
construction is that it builds on definitions that are already available from
the semantics of signatures. In contrast, strict eliminators should be connected
back to the semantics in a separate step: we should show that strict elimination
yields a displayed algebra section, and that the two definitions of displayed
algebras are equivalent. We do not detail these here.

\chapter[Finitary QII Signatures]{Finitary Quotient Inductive-Inductive Signatures}
\label{chap:fqiit}

In this chapter we bump the expressive power of signatures by a large margin,
and also substantially extend the semantics. However, we keep the basic approach
the same; indeed its advantages become apparent with the more sophisticated
signatures.

We use two different setups for semantics in this chapter.
\begin{itemize}
  \item In Sections \ref{sec:fqiit-tos}-\ref{sec:fqii-left-adjoint} we work in 2LTT, thereby
        getting a generalized semantics for signatures. Here we keep details about universe
        levels to the minimum.
  \item In Section \ref{sec:fqiit-term-algebras}, we work in an
        extensional type theory with cumulative universes. This is more suited for the term
        algebra construction, where (as we will see) 2LTT does not bring any advantage, but
        we do need to be more precise about universes.
\end{itemize}

\section{Theory of Signatures}
\label{sec:fqiit-tos}

Signatures are once again given by contexts of a type theory, but now it is a
dependent type theory, given as a cwf with certain type formers, in the style
of Section \ref{sec:models-of-tts}.

\subsubsection{Metatheory and terminology}
We work in 2LTT with $\Ty_0$ and $\Tm_0$, and make the following assumptions:
\begin{itemize}
  \item $\Ty_0$ is closed under $\top$, $\Sigma$ and extensional identity $\blank\!=\!\blank$.
        The inner identity reflects the outer one.
  \item The outer identity $\blank\!\equiv\!\blank$ is also extensional; it
    reflects strict equality in some unspecified metatheory outside 2LTT. This
    justifies omitting transports along $\blank\!\equiv\!\blank$ in our
    notation.
\end{itemize}

In the following we specify models of the theory of \emph{finitary quotient
inductive-inductive signatures}. The names involved are a bit of a mouthful, so
we abbreviate ``finitary quotient inductive-inductive'' as FQII, and like
before, we abbreviate ``theory of signatures'' as ToS. In this chapter, by
signature we mean an FQII signature unless otherwise specified.

Additionally, we abbreviate ``quotient inductive-inductive types'' as QIIT, and
we may qualify it to FQIIT if it is finitary. A \emph{type} in this sense is
simply the initial algebra for a given FQII signature. We shall use this naming
in the rest of the thesis; an \emph{inductive type} is an initial algebra for a
signature. Also, we use \emph{syntax} as a synonym for initial algebra.

\begin{mydefinition}
\label{def:fqiit-tos}
A \textbf{model of the theory of signatures} consists of the following.
  \begin{itemize}
    \item A \textbf{cwf} with underlying sets $\Con$, $\Sub$, $\Ty$ and $\Tm$, all returning in
      the outer $\Set$ universe of 2LTT.
    \item A \textbf{Tarski-style universe} $\U$ with decoding $\El$.
    \item An \textbf{extensional identity type} $\Id : \Tm\,\Gamma\,A \to
      \Tm\,\Gamma\,A \to \Ty\,\Gamma$, specified by $(\reflect,\,\refl) :
      \Tm\,\Gamma\,(\Id\,t\,u) \simeq (t \equiv u)$.
    \item An \textbf{internal product type} $\Pi : (a : \Tm\,\Gamma\,\U) \to
      \Ty\,(\Gamma\ext\El\,a) \to \Ty\,\Gamma$, specified by
      $(\app,\,\lam) : \Tm\,\Gamma\,(\Pi\,a\,B) \simeq \Tm\,(\Gamma \ext \El\,a)\,B$.
    \item An \textbf{external product type} $\Pie : (\mi{Ix} : \Ty_0) \to (\mi{Ix} \to \Ty\,\Gamma) \to \Ty\,\Gamma$, specified by
      $(\appe,\,\lame) : \Tm\,\Gamma\,(\Pie\,\mi{Ix}\,B) \simeq ((i : \mi{Ix}) \to \Tm\,\Gamma\,(B\,i))$.
  \end{itemize}
\end{mydefinition}
At this point we only have a notion of model for ToS, but as we will see in
Chapter \ref{chap:iqiit}, ToS is also an algebraic theory, more specifically an
infinitary QII one. It is infinitary because $\Pie$ and $\lame$ allow branching
which is indexed over elements of arbitrary $\mi{Ix} : \Ty_0$ types.

Because of the algebraic character of ToS, there is a category of ToS models
where morphisms strictly preserve all structure, and the initial model
corresponds to the syntax. We will make this precise in Chapter
\ref{chap:iqiit}. We also assume that the ToS syntax exists.

\begin{mydefinition} An FQII \textbf{signature} is an element of $\Con$ in the syntax of ToS.
\end{mydefinition}
We review several example signatures in the following, using progressively
more ToS type formers.  We also introduce progressively more compact notation
for signatures. As a rule of thumb, we shall use compact notation for larger and
more complex signatures, but we shall be more explicit when we specify models of
ToS later in this chapter.
\begin{myexample}
  Simple inductive signatures can be evidently expressed using $\U$ and
  $\Pi$. By adding a single $\U$ to the signature, we introduce the inductive
  sort, while $\Pi$ adds an inductive parameter to an entry.
  \begin{alignat*}{3}
    & \ms{NatSig} &&\defn \emptycon \ext (N : \U) \ext (\mi{zero} : \El\,N)
                        \ext (\mi{suc} : \Pi (n : N) (\El\,N))\\
    & \ms{TreeSig} &&\defn \emptycon \ext (T : \U) \ext (\mi{leaf} : \El\,T)
                         \ext (\mi{node} : \Pi (t_1 : T) (\Pi (t_2 : T) (\El\,T)))
  \end{alignat*}
  Observe that the domains in $\Pi$ are terms with type $\U$, while the codomains are proper types.
\end{myexample}

\begin{notation} We write non-dependent product types in ToS as follows.
  \begin{itemize}
  \item $a \funi B$ for $\Pi\,(\_ : a)\,B$.
  \item $\mi{Ix} \fune B$ for $\Pie\,\mi{Ix}\,(\lambda\,\_.\,B)$.
  \end{itemize}
\end{notation}
\noindent
Using this notation, we may write $\mi{suc} : N \funi \El\,N$ and $\mi{node} : T
\funi T \funi \El\,T$.

\begin{notation}
The ``categorical'' application $\app$ with explicit substitutions is a bit
inconvenient. Instead, we simply write whitespace for $\Pii$ and $\Pie$
application:
\begin{alignat*}{3}
  & t\,u \defn (\appi\,t)[\id,\,u]\\
  & t\,u \defn (\appe\,t)\,u
\end{alignat*}
\end{notation}

\begin{myexample}
We may have any number of sorts by adding more $\U$ to the signatures. Moreover,
sorts can be indexed over previous sorts. Hence, using only $\U$, $\El$ and
$\Pi$, we can express any closed inductive-inductive type
\cite{forsberg-phd}. The following fragment of the the signature for categories
is such:
\begin{alignat*}{3}
  & \emptycon \ext (\mi{Obj} : \U) \ext (\mi{Hom} : \mi{Obj} \funi \mi{Obj} \funi \U)
      \ext (\mi{id} : \Pi(i : \mi{Obj})\,(\El\,(\mi{Hom}\,i\,i)))
\end{alignat*}
These inductive-inductive signatures are more flexible than those in prior
literature \cite{forsberg-phd}, since we allow type constructors (sorts) and
point constructors to be arbitrarily mixed, as opposed to mandating that sorts
are declared first. For example:
\begin{alignat*}{3}
  & \emptycon \ext (A : \U) \ext (a : \El\,A) \ext (B : A \funi \U) \ext (C : B\,a \funi \U)
\end{alignat*}
Here $C$ is indexed over $B\,a$, where $a$ is a point constructor of $a$, so a
sort specification mentions a point constructor.
\end{myexample}

\begin{myexample} $\Id$ lets us add equations to signatures. With this, we can write down the full
signature for categories:
\begin{alignat*}{3}
  & \emptycon &&\ext (\mi{Obj}   &&: \U)\\
  &           &&\ext (\mi{Hom}   &&: \mi{Obj} \funi \mi{Obj} \funi \U)\\
  &           &&\ext (\mi{id}    &&: \Pi(i : \mi{Obj})\,(\El\,(\mi{Hom}\,i\,i)))\\
  &           &&\ext (\mi{comp}  &&: \Pi\,(i\,j\,k : \mi{Obj})\,(\mi{Hom}\,j\,k \funi \mi{Hom}\,i\,j \funi \El\,(\mi{Hom}\,i\,k)))\\
  &           &&\ext (\mi{idr}   &&: \Pi\,(i\,j : \mi{Obj})(f : \mi{Hom}\,i\,j)\,(\Id\,(\mi{comp}\,i\,i\,j\,f\,(\mi{id}\,i))\,f))\\
  &           &&\ext (\mi{idl}   &&: \Pi\,(i\,j : \mi{Obj})(f : \mi{Hom}\,i\,j)\,(\Id\,(\mi{comp}\,i\,j\,j\,(\mi{id}\,j)\,f)\,f))\\
  &           &&\ext (\mi{assoc} &&: \Pi\,(i\,j\,k\,l : \mi{Obj})(f : \mi{Hom}\,j\,l)(g : \mi{Hom}\,j\,k)(h : \mi{Hom}\,i\,j)\\
  &           && &&\hspace{1.7em} (\Id
                 \,(\mi{comp}\,i\,j\,l\,(\mi{comp}\,j\,k\,l\,f\,g)\,h)
                 \,(\mi{comp}\,i\,k\,l\,f\,(\mi{comp}\,i\,j\,k\,g\,h))
\end{alignat*}
Now, this is already rather hard to read, even together with a compressed
notation for multiple $\Pi$ binders.
\begin{notation}
For more complex signatures, we may entirely switch to an internal notation,
where we mostly reuse the conventions in the metatheories, including implicit
arguments and implicit quantification. We use $(x : a) \to B$ for internal
products, $(x : A) \toe B$ for external products, but we still write $\Id$ for
the identity type and make $\U$ and $\El$ explicit. In this notation, a
signature is just a listing of binders. The category signature becomes the
following:
\begin{alignat*}{3}
  & \ms{Obj} &&: \U\\
  & \ms{Hom} &&: \ms{Obj} \to \ms{Obj} \to \U\\
  & \ms{id}  &&: \El\,(\ms{Hom}\,i\,i)\\
  & \ms{\blank\!\circ\!\blank} &&: \ms{Hom}\,j\,k \to \ms{Hom}\,i\,j \to \El\,(\ms{Hom}\,i\,k)\\
  & \ms{idr} &&: \Id\,(f \circ \ms{id})\,f\\
  & \ms{idl} &&: \Id\,(\ms{id} \circ f)\,f\\
  & \ms{assoc} &&: \Id\,(f \circ (g \circ h))\,((f \circ g) \circ h)
\end{alignat*}
\end{notation}
\end{myexample}

\begin{myexample}
The external product type makes it possible to reference inner types (in 2LTT)
in signatures. Here ``external'' is meant relative to a given signature, and
refers to types and inhabitants which are not introduced inside a signature.
For example, we give a signature for lists by assuming $A : \Ty_0$ for the
(external) type of list elements:
\begin{alignat*}{3}
  &\ms{List} &&: \U\\
  &\ms{nil}  &&: \El\,\ms{List}\\
  &\ms{cons} &&: A \toe \ms{List} \to \El\,\ms{List}
\end{alignat*}
Hence, ``parameters'' are always assumptions made in the metatheory. We can
also \emph{index} sorts by external values. Let us specify length-indexed vectors
now; we keep the $A : \Ty_0$ assumption, but also assume that $\Ty_0$ has
natural numbers, with $\Nat_0 : \Ty_0$, $\zero_0$ and $\suc_0$.
\begin{alignat*}{3}
  &\ms{Vec}  &&: \Nat_0 \toe \U \\
  &\ms{nil}  &&: \El\,(\ms{Vec}\,\zero_0)\\
  &\ms{cons} &&: (n : \Nat_0) \toe A \toe \ms{Vec}\,n \to \El\,(\ms{Vec}\,(\suc_0\,n))
\end{alignat*}
\end{myexample}

\begin{myexample}
We can also introduce \emph{sort equations} using $\Id$: this means equating
terms of $\U$, i.e.\ inductively specified sets. This is useful for specifying
certain strict type formers. For example, a signature for cwfs can be extended with
a specification for strict constant families.
\begin{alignat*}{3}
  & \ms{Con}     &&: \U\\
  & \ms{Sub}     &&: \ms{Con} \to \ms{Con} \to \U \\
  & \ms{Ty}      &&: \ms{Con} \to \U\\
  & \ms{Tm}      &&: (\Gamma : \ms{Con}) \to \ms{Ty}\,\Gamma \to \U\\
  & ...          &&\\
  & \ms{K}       &&: \ms{Con} \to \{\Gamma : \ms{Con}\} \to \El\,(\ms{Ty}\,\Gamma)\\
  & \ms{K_{spec}} &&: \Id\,(\ms{Tm}\,\Gamma\,(\ms{K}\,\Delta))\,(\ms{Sub}\,\Gamma\,\Delta)
\end{alignat*}
The equation for Russell-style universes is likewise a sort equation:
\begin{alignat*}{3}
  &\ms{Univ}    &&: \El\,(\ms{Ty}\,\Gamma)\\
  &\ms{Russell} &&: \Id\,(\ms{Tm}\,\Gamma\,\ms{Univ})\,(\ms{Ty}\,\Gamma)
\end{alignat*}
\end{myexample}

\begin{myexample}
\label{ex:presheaf-sig}
As we mentioned in Definition \ref{def:presheaf-type}, there is a signature for
presheaves, so let us look at that now. Assume a category $\mbbC$ in the inner
theory; this means that objects and morphisms of $\mbbC$ are in $\Ty_0$.
\begin{alignat*}{3}
  & \ms{Obj}         &&: |\mbbC| \toe \U\\
  & \ms{Hom}         &&: \mbbC(i,\,j) \toe \ms{Obj}\,j \to \El\,(\ms{Obj}\,i)\\
  & \ms{Hom_{\ms{id}}} &&: \Id\,(\ms{Hom}\,\id\,x)\,x\\
  & \ms{Hom_{\circ}}  &&: \Id\,(\ms{Hom}\,(f \circ g)\,x)\,(\ms{Hom}\,f\,(\ms{Hom}\,g\,x))
\end{alignat*}
We depart from the sugary naming scheme in Definition \ref{def:presheaf-type},
and name the action on objects $\ms{Obj}$ and the action on morphisms
$\ms{Hom}$. When we give semantics to this signature in Section \ref{sec:fqiit-semantics}, we
will get as algebras functors from $\mbbC^{\ms{op}}$ to the category of inner
types. That category has elements of $\Ty_0$ as objects and $\Tm_0\,A \to
\Tm_0\,B$ functions as morphisms.
\end{myexample}

\subsubsection{Strict positivity}

Only strictly positive signatures are expressible. Similarly to the case with
simple signatures, there is no way to abstract over internal products, since
internal products are indexed over $\U$-small types, and $\U$ has no type
formers at all. With $\Pie$, we can abstract over functions, but only those
which are external to a signature and do not depend on internally specified
constructions.

\subsubsection{Limitation: nested induction}

Nested induction means that external type functions may be applied to
expressions internal to the theory of signatures. This is not possible in any of
the signatures in this thesis. A common example is rose trees, assuming
external $\ms{List} : \Set \to \Set$:
\begin{alignat*}{3}
  &\ms{Tree} &&: \Set\\
  &\ms{node} &&: \ms{List}\,\ms{Tree} \to \ms{Tree}
\end{alignat*}
The $\ms{List}\,\ms{Tree}$ expression is not representable in a signature; the
$\ms{List}$ function is external, while $\ms{Tree}$ would be an internal sort.
This style of inductive definition requires reasoning about the polarity of all
external type functions: only the strictly positive $\Set \to \Set$ functions
should be allowed. With general type functions we would also need to track
polarity of multiple parameters, or even higher-order polarity.

Many use cases of nested induction can be removed by ``including'' the external
type constructor into the signature. In the case of rose trees, this means
defining lists and trees mutually:
\begin{alignat*}{3}
  &\ms{List} &&: \U\\
  &\ms{Tree} &&: \U\\
  &\ms{nil}  &&: \El\,\ms{List}\\
  &\ms{cons} &&: \ms{Tree} \to \ms{List} \to \El\,\ms{List}\\
  &\ms{node} &&: \ms{List} \to \El\,\ms{Tree}
\end{alignat*}
Of course, nested induction would be still desirable because of the code reuse
that it enables.

\section{Semantics}
\label{sec:fqiit-semantics}

\subsection{Overview}

For simple signatures, we only gave semantics in enough detail so that notions
of recursion and induction could be recovered. We aim to do more now. For each
signature, we would like to have
\begin{enumerate}
  \item A category of algebras, with homomorphisms as morphisms.
  \item A notion of induction, which requires a notion of dependent algebras.
  \item A proof that for algebras, initiality is equivalent to supporting induction.
\end{enumerate}

We do this by creating a model of ToS where contexts (signatures) are categories
with certain extra structure and substitutions are structure-preserving
functors. Then, ToS signatures can be interpreted in this model, using the
initiality of ToS syntax (i.e.\ the recursor).

Our semantics has a type-theoretic flavor, which is inspired by the cubical set
model of Martin-Löf type theory by Bezem et al. \cite{cubical}. The idea is
to avoid strictness issues by starting from basic ingredients which are already
strict enough. Hence, instead of modeling ToS types as certain morphisms and
substitution by pullback, we model types as displayed categories with extra
structure, which naturally support strict reindexing/substitution.

We make a similar choice in the interpretation of signatures themselves: we use
structured cwfs of algebras, where types correspond to displayed algebras. This
choice is in contrast to having finitely complete categories of algebras.
Preliminarily, the reason is that ``native'' displayed algebras and sections
allow us to compute induction principles strictly as one would write in a type
theory. In fact, in this chapter we recover exactly the same semantics for
simple signatures that we already specified. In contrast, in finitely complete
categories there is no primitive notion of displayed objects, and we can only
specify induction principles up to equivalences.

This issue is perhaps not relevant from a purely categorical perspective, but we
are concerned with eventually implementing QIITs in proof assistants. If we do
not compute induction principles here in an exact way, we do not get them from
anywhere else.

\subsection{Separate vs.\ Bundled Models}

Previously, we defined $\blank^A$, $\blank^M$, $\blank^D$ and $\blank^S$
interpretations of signatures separately, by doing induction anew for each
one. Formally, this amounts to giving a plain model of ToS in order to define
$\blank^A$, but then giving three \emph{displayed} models of ToS to specify the
other interpretations because they sometimes need to refer to the recursors or
eliminators of other interpretations.

For example, $\blank^A : \Con \to \Set$ while $\blank^D : (\Gamma : \Con) \to
\Gamma^A \to \Set$, so displayed algebras already refer to $\blank^A$, which is
part of the recursor for the corresponding model.

However, this piecewise style can be avoided: we can give a single non-displayed
model which packs everything in a $\Sigma$-type, yielding just one
interpretation function for signatures. Let us call that function $\blank^M$ now:
\begin{alignat*}{5}
  & \blank^M : \Con &&\to \,&&(A &&: \Set)\\
  & &&\hspace{0.3em}\times &&(M  &&: A \to A \to \Set)\\
  & &&\hspace{0.3em}\times &&(D  &&: A \to \Set)\\
  & &&\hspace{0.3em}\times &&(S  &&: (a : A) \to D\,a \to \Set)
\end{alignat*}
Note that it is often not possible to merge multiple recursors/eliminators by
packing models together. For example, addition on natural numbers is defined by
recursion, and so is multiplication; but since multiplication calls addition in
an iterated fashion, it is not possible to define both operations by a single
algebra. Nevertheless, merging does work in our case. We will, in fact, get a
formal vocabulary for merging models (and manipulating them in other ways) from
the semantics of ToS itself.

In simple cases, and in Agda, the piecewise style is convenient, since we do not
have to deal with $\Sigma$-s. However, for larger models, important organizing
principles may become more apparent if we bundle things together.

In the following, we shall define a model $\bM : \ToS$ such that its $\Con$
component is a bundle containing all $A$, $M$, $D$, $S$ components, plus a
number of additional components. We present the components of $\bM$ in the same
order as in Definition \ref{def:fqiit-tos}. There is significant overlap in
names and notations, so we use \textbf{bold} font to disambiguate components of
$\bM$ from components of other structures. For example, we use $\bs{\sigma :
  \Sub\,\Gamma\,\Delta}$ to denote a substitution in $\bM$, while there could be
$\Sub$-named components in other structures under consideration.

\subsection{Finite Limit Cwfs}

We define $\bCon : \Set$ as the type of finite limit cwfs (flcwfs). Recall that this
specifies the objects of the underlying cwf of $\bM$. In the following we
specify flcwfs and describe some internal constructions.

\begin{mydefinition}\label{def:flcwf}
We define $\flcwf : \Set$ as an iterated $\Sigma$-type with the following components:
\begin{enumerate}
  \item A cwf with $\Con$, $\Sub$, $\Ty$, $\Tm$ all returning in $\Set$. \emph{Remark:}
        this implies that $\flcwf : \Set$ is in a larger universe than all of these
        internal components. We continue to elide universe sizing details.
  \item $\Sigma$-types.
  \item Extensional identity type $\Id$ with $\refl$ and $\reflect$.
  \item Strict constant families $\K$.
\end{enumerate}
\end{mydefinition}
\begin{mydefinition}
We abbreviate the additional structure on cwfs consisting of $\Sigma$, $\Id$ and
$\K$ as \textbf{fl-structure}.
\end{mydefinition}

We recover previous concepts as follows. Assuming $\Gamma$ signature, we get an
flcwf by interpreting $\Gamma$ in $\bM$. In that flcwf we have
\begin{itemize}
  \item $\Con$ as the type of algebras.
  \item $\Sub$ as the type of algebra morphisms.
  \item $\Ty$ as the type of displayed algebras.
  \item $\Tm$ as the type of displayed algebra sections.
\end{itemize}
From this, notions of initiality and induction are apparent as well. Initiality
is the usual categorical notion. Also note that the unit type can be derived as
$\K\,\emptycon$.

\begin{mydefinition}
\label{def:induction-predicate}
Assuming $\bGamma : \bCon$ in a cwf, we define the \textbf{induction predicate} on
objects:
\begin{alignat*}{3}
  & \Inductive : \Con_{\bGamma} \to \Set\\
  & \Inductive\,\Gamma \defn (A : \Ty_{\bGamma}\,\Gamma)\ra \Tm_{\bGamma}\,\Gamma\,A
\end{alignat*}
In our semantics this will say that an algebra is inductive if every displayed
algebra over it has a section. We also know that induction and initiality are
equivalent.
\end{mydefinition}

\begin{theorem}\label{thm:initiality-induction}
We assume a cwf $\bGamma$ with $\Id$ and weak $\K$.
An object $\Gamma : \Con_{\bGamma}$ supports induction if and only if it is
initial. Moreover, induction and initiality are both mere properties.
\end{theorem}

\begin{proof}
First, we show that induction implies initiality. We assume $\Gamma : \Con$,
$\ms{ind} : \ms{Inductive}\,\Gamma$ and $\Delta : \Con$. We aim to show that
there is a unique inhabitant of $\Sub\,\Gamma\,\Delta$. We have
$\ms{ind}\,(\K\,\Delta) : \Tm\,\Gamma\,(\K\,\Delta)$, hence
$\appK\,(\ms{ind}\,(\K\,\Delta)) : \Sub\,\Gamma\,\Delta$. We only need to show that this
is unique.  Assume $\delta : \Sub\,\Gamma\,\Delta$. Now,
$\ms{ind}\,(\Id\,(\lamK\,\delta)\,(\ms{ind}\,(\K\,\Delta))) :
\Tm\,\Gamma\,(\Id\,(\lamK\,\delta)\,(\ms{ind}\,(\K\,\Delta)))$, and it follows by
equality reflection that $\lamK\,\delta \equiv \ms{ind}\,(\K\,\Delta)$, thus
$\delta \equiv \appK\,(\ms{ind}\,(\K\,\Delta))$.

Second, the other direction. We assume that $\Gamma$ is initial, and also $A :
\Ty\,\Gamma$, and aim to inhabit $\Tm\,\Gamma\,A$. By initiality we get a unique
$\sigma : \Sub\,\Gamma\,(\Gamma \ext A)$. Now, $\q[\sigma] : \Tm\,\Gamma\,(A[\p \circ \sigma])$,
but since $\p \circ \sigma : \Sub\,\Gamma\,\Gamma$, it must be equal to $\id$ by the initiality
of $\Gamma$. Hence, $\q[\sigma] : \Tm\,\Gamma\,A$.

Lastly: it is well-known that initiality is a mere property, so let us show the
same for induction.  We assume $\ms{ind},\,\ms{ind'} : \ms{Inductive}\,\Gamma$
and $A : \Ty\,\Gamma$. We have
$\reflect\,(\ms{ind}\,(\Id\,(\ms{ind}\,A)\,(\ms{ind'}\,A))) : \ms{ind}\,A \equiv
\ms{ind'}\,A$. Since $A$ is arbitrary, by function extensionality we also have $\ms{ind} \equiv \ms{ind'}$.
\end{proof}

\begin{theorem}
\label{thm:flcwf-term-uniqueness}
$\Tm\,\Gamma\,A$ in an cwf with $\Id$ and weak $\K$ is propositional
when $\Gamma$ is initial.
\end{theorem}
\begin{proof} Assuming $t,\,u : \Tm\,\Gamma\,A$, we have
$\reflect\,(\ms{ind}\,(\Id\,t\,u)) : t \equiv u$.
\end{proof}

Note that the above proofs do not rely on $\Sigma$-types, so why do we include
them in the semantics? One reason is the prior result by Clairmabault and Dybjer
\cite{clairambault2014biequivalence}, that a slightly different formulation of
flcwfs is biequivalent to finitely complete categories. More concretely, in
ibid.\ there is a 2-category of cwfs with $\Sigma$, $\Id$ and ``democracy'', the
last of which is equivalent to the weak formulation of constant families. Then,
it is shown that this 2-category is biequivalent to the 2-category of finitely
complete categories. Thus, including $\Sigma$ is a good deal, as this allows us
to connect our semantics back to finitely complete categories, which are more
common in categorical settings.

We recover finite limits in an flcwf as follows. The product of $\Gamma$ and
$\Delta$ is given by $\Gamma \ext \K\,\Delta$, and we get projection and pairing
from context comprehension. The equalizer of $\sigma,\,\delta :
\Sub\,\Gamma\,\Delta$ is given by $\Gamma \ext \Id\,\sigma\,\delta$, which is
well-typed because morphisms can be viewed as terms, e.g.\ $\sigma :
\Tm\,\Gamma\,(\K\,\Delta)$. The unique morphism out of the equalizer is $\p :
\Sub\,(\Gamma \ext \Id\,\sigma\,\delta)\,\Gamma$.

Our Definition \ref{def:flcwf} for flcwfs is not exactly the same as in
\cite{clairambault2014biequivalence} because our constant families are
strict. However, this only strengthens our semantics in this section, since weak
constant families can be trivially recovered from strict ones. We present some
results from the existing literature in the following.

\begin{mydefinition}[Type categories, c.f.\ {\cite[Section 2.2]{clairambault2014biequivalence}}]
\label{def:type_categories}
We work in a cwf. For each $\Gamma : \Con$, there is a category whose objects
are types $A : \Ty\,\Gamma$, and morphisms from $A$ to $B$ are terms $t :
\Tm\,(\Gamma\,\ext\,A)\,(B[\p])$. Identity morphisms are given by $\q :
\Tm\,(\Gamma\,\ext\,A)\,(A[\p])$, and composition $t \circ u$ by $t[\p, u]$. The
assignment of type categories to contexts extends to a split indexed
category. For each $\sigma : \Sub\,\Gamma\,\Delta$, there is a functor from
$\Ty\,\Delta$ to $\Ty\,\Gamma$, which sends $A$ to $A[\sigma]$ and $t :
\Tm\,(\Gamma\,\ext\,A)\,(B[\p])$ to $t[\sigma\circ \p, \q]$.
\end{mydefinition}

\begin{notation}
  ~\\
  \begin{itemize}
  \item  \vspace{-1.7em}
         In any cwf, we use $\sigma : \Gamma \simeq \Delta$ to indicate
         that $\sigma : \Sub\,\Gamma\,\Delta$ is an isomorphism with inverse $\sigma^{-1}$.
    \item A \emph{type isomorphism}, written as $t : A \simeq B$ is an isomorphism in a
         type category, with inverse as $t^{-1}$.
  \end{itemize}
\end{notation}

\begin{theorem}[Equivalence of type and slice categories, c.f.\ {\cite[Section 2.2]{clairambault2014biequivalence}}]\label{thm:ty-slice}
Assume that we work in an cwf $\bGamma$ with $\Sigma$, $\Id$ and weak
$\K$. For each $\Gamma : \Con$, the type category $\Ty\,\Gamma$ is equivalent to
the slice category $\bGamma/\Gamma$.
\end{theorem}

\emph{Remark.} In the flcwf of sets where types are $A \to \Set$ families, the
above theorem yields the equivalence of $A \to \Set$ and $(B : \Set) \times (B
\to A)$. This is sometimes called the ``family-fibration'' equivalence. It is
also a notable motivating example for univalence in type theory: it is not an
isomorphism of sets, but only an equivalence up to isomorphism of sets. So this
is an example for an equivalence which quite naturally arises even if we only
care about sets, but one which is not covered by set-level univalence, and
additionally requires univalence for groupoids, if we want to prove it as a
propositional equality.

\subsection{The Cwf of Finite Limit Cwfs}

The next task is to define the cwf part of $\bM$. We already know that objects
are flcwfs.

\subsubsection{Category}

A \textbf{morphism} $\bsigma : \bSub\,\bGamma\,\bDelta$ is an algebra
homomorphism, viewing flcwfs as algebraic structures. Hence, $\bsigma$ includes
a functor between underlying categories, but it also maps types to types and
terms to terms, and strictly preserves all structure.

\begin{notation}
We may implicitly project out the underlying maps from $\bsigma$. Hence, we
have the following four maps:
\begin{alignat*}{3}
  & \bsigma &&: \Con_{\bGamma} \to \Con_{\bDelta} \\
  & \bsigma &&: \Sub_{\bGamma}\,\Gamma\,\Delta \to \Sub_{\bDelta}\,(\bsigma\,\Gamma)\,(\bsigma\,\Delta)\\
  & \bsigma &&: \Ty_{\bGamma}\,\Gamma \to \Ty_{\bDelta}\,(\bsigma\,\Gamma)\\
  & \bsigma &&: \Tm_{\bGamma}\,\Gamma\,A \to \Tm_{\bDelta}\,(\bsigma\,\Gamma)\,(\bsigma\,A)
\end{alignat*}
We list some of the preservation equations as examples of usage:
\begin{alignat*}{3}
  \bsigma\,\emptycon &\equiv \emptycon \\
  \bsigma\,(\Gamma \ext A) &\equiv \bsigma\,\Gamma \ext \bsigma\,A\\
  \bsigma\,(A[\sigma]) &\equiv (\bsigma\,A)[\bsigma\,\sigma]\\
  \bsigma\,(t[\sigma]) &\equiv (\bsigma\,t)[\bsigma\,\sigma]\\
  \bsigma\,(\Sigma\,A\,B) &\equiv \Sigma\,(\bsigma\,A)\,(\bsigma\,B)\\
  \bsigma\,(\proj_1\,t) &\equiv \proj_1\,(\bsigma\,t)
\end{alignat*}
Above, we could have also included subscripts indicating the $\bGamma$ or
$\bDelta$ flcwf, as in $\bsigma\,\emptycon_{\bGamma} \equiv
\emptycon_{\bDelta}$; but these are quite easily inferable, so we omit them.
\end{notation}

\textbf{Identity morphisms} and \textbf{composition} are defined in the evident
way using identity functions and function composition in underlying maps, and
they satisfy the category laws.

The terminal object $\bemptycon:\bCon$ is given by having $\Con_{\bemptycon} \defn
\top$, $\Sub_{\bemptycon}\,\Gamma\,\Delta \defn \top$, $\Ty_{\bemptycon}\,\Gamma \defn \top$ and
$\Tm_{\bemptycon}\,\Gamma\,A \defn \top$, and all structure and equations are defined trivially.

\subsubsection{Family structure}
\label{sec:fqiit-family}

A type $\bA : \bTy\,\bGamma$ is a displayed flcwf over $\bGamma$.  As we have
seen before, displayed algebras can be computed as logical predicate
interpretations of algebraic signatures. Every $\bA$ component lies over the
corresponding $\bGamma$ component. Also note that a displayed flcwf includes a
displayed category, for which some results have been worked out in
\cite{displayedcats}.

\begin{notation} In situations where we need to refer to both ``base'' and
displayed things, we give \ul{underlined} names to contexts, substitutions,
types and terms in a base flcwf. For example, we may have $\ulGamma :
\Con_{\bGamma}$ living in $\bs{\Gamma : \Con}$, and $\Gamma :
\Con_{\bA}\,\ulGamma$ living in a displayed flcwf over $\bGamma$. We only use
underlining on 2LTT variable names, and overload flcwf component names for
displayed counterparts. For example, a $\Con$ component is named the same in
a base flcwf and a displayed one.
\end{notation}

Concretely, a displayed flcwf $\bA$ over $\bGamma$ has the following underlying
sets, which we call displayed contexts, substitutions, types and terms
respectively.
\begin{alignat*}{3}
  & \Con_{\bA} && : \Con_{\bGamma}\ra \Set\\
  & \Sub_{\bA} && : \Con_{\bA}\,\ulGamma \ra \Con_{\bA}\,\ulDelta \ra \Sub_{\bGamma}\,\ulGamma\,\ulDelta \ra \Set \\
  & \Ty_{\bA}  && : \Con_{\bA}\,\ulGamma \ra \Ty_{\bGamma}\,\ulGamma \ra \Set\\
  & \Tm_{\bA}  && : (\Gamma : \Con_{\bA}\,\ulGamma)\ra \Ty_{\bA}\,\Gamma\,\ulA \ra \Tm_{\bGamma}\,\ulGamma\,\ulA \ra \Set
\end{alignat*}
We list several components of $\bA$ below; note how every $\bA$ operation lies
over the corresponding $\bGamma$ operation. In our notation with implicit
arguments, equations in $\bA$ can be written the same way as in $\bGamma$, but
of course there is extra indexing involved, and the displayed equations are
well-typed because of their counterparts in the base.
\begingroup
\allowdisplaybreaks
\begin{alignat*}{3}
  & \id_{\bA} &&: \Sub_{\bA}\,\Gamma\,\Gamma\,\ul{\id_{\bGamma}}\\
  & \blank\circ_{\bA}\blank &&: \Sub_{\bA}\,\Delta\,\Xi\,\ulsigma \to \Sub_{\bA}\,\Gamma\,\Delta\,\uldelta
    \to \Sub_{\bA}\,\Gamma\,\Xi\,(\ulsigma \circ_{\bGamma} \uldelta)\\
  & \ms{idl}_{\bA} &&:  \id_{\bA} \circ_{\bA} \sigma \equiv \sigma \\
  & \ms{idr}_{\bA} &&:  \sigma \circ_{\bA} \id_{\bA} \equiv \sigma \\
  & \emptycon_{\bA} && : \Con_{\bA}\,\emptycon_{\bGamma}\\
  & \blank\ext_{\bA}\blank && : (\Gamma : \Con_{\bA}\,\ulGamma)\ra \Ty_{\bA}\,\Gamma\,\ulA \ra
                     \Con_{\bA}\,\Gamma\,(\ulGamma \ext_{\bGamma} \ulA)\\
  & \blank[\blank]_{\bA} && : \Ty_{\bA}\,\Delta\,\ulA \ra \Sub_{\bA}\,\Gamma\,\Delta\,\ulsigma
                     \ra \Ty_{\bA}\,\Gamma\, (\ulA[\ulsigma]_{\bGamma})\\
  & \blank[\blank]_{\bA} && : \Tm_{\bA}\,\Delta\,A\,\ult \ra (\sigma : \Sub_{\bA}\,\Gamma\,\Delta\,\ulsigma)
                        \ra \Tm_{\bA}\,\Gamma\, (A[\sigma]_{\bA})\,(\ult[\ulsigma]_{\bGamma})\\
  &\Id_{\bA} &&: \Tm_{\bA}\,\Gamma\,A\,\ult \to \Tm_{\bA}\,\Gamma\,A\,\ulu \to \Ty_{\bA}\,\Gamma\,(\Id_{\bGamma}\,\ult\,\ulu)\\
  &\K_{\bA} &&: \Con_{\bA}\,\ulDelta \to \{\Gamma : \Con_{\bA}\,\ulGamma\} \to \Ty_{\bA}\,\Gamma\,(\K_{\bGamma}\,\ulDelta)\\
  &\Sigma_{\bA} &&: (A : \Ty_{\bA}\,\Gamma\,\ulA) \to \Ty_{\bA}\,(\Gamma \ext_{\bA} A)\,\ulB \to
                   \Ty_{\bA}\,\Gamma\,(\Sigma_{\bGamma}\,\ulA\,\ulB)
\end{alignat*}
\endgroup
In the following we will often omit $_{\bGamma}$ and $_{\bA}$ subscripts on
components; for example, in the type $\Con_{\bA}\,\emptycon$, the $\emptycon$ is
clearly a base component in $\bGamma$.

A substituted type $\bs{A[\sigma] : \Ty\,\Gamma}$ is defined as follows, for
$\bs{A : \Ty\,\Delta}$ and $\bs{\sigma : \Sub\,\Gamma\,\Delta}$. We simply compose
underlying functions in $\bsigma$ with the underlying predicates in $\bA$:
\begin{alignat*}{3}
  & \Con_{\bs{A[\sigma]}}\,\ulGamma && \defn \Con_{\bA}\,(\bsigma\,\ulGamma)\\
  & \Sub_{\bs{A[\sigma]}}\,\Gamma\,\Delta\,\ulsigma && \defn
    \Sub_{\bA}\,\Gamma\,\Delta\,(\bsigma\,\ulsigma)\\
  & \Ty_{\bs{A[\sigma]}}\,\Gamma\,\ulA && \defn
      \Ty_{\bA}\,\Gamma\,(\bsigma\,\ulA)\\
  & \Tm_{\bs{A[\sigma]}}\,\Gamma\,A\,\ult && \defn
      \Tm_{\bA}\,\Gamma\,A\,(\bsigma\,\ult)
\end{alignat*}
It should be clear that $\bs{A[\sigma]}$ thus defined still supports all
displayed flcwf structure. For example, the displayed contexts in
$\bs{A[\sigma]}$ are elements of $\Con_{\bA}\,(\bsigma\,\ulGamma)$, but since
$\bsigma$ preserves all $\bGamma$-structure, we can also recover all displayed
structure. For example, if $\ulGamma$ is $\ul{\emptycon}$, we have
$\bsigma\,\ul{\emptycon} \equiv \ul{\emptycon}$, and we can reuse
$\emptycon_{\bA} : \Con_{\bA}\,\ul{\emptycon}$ to define the displayed empty
context in $\bs{A[\sigma]}$, and we can proceed analogously for all other
structure in $\bs{A[\sigma]}$.

Additionally, type substitution is functorial, i.e.\ $\bs{A[\id]} \equiv \bA$
and $\bs{A[\sigma \circ \delta]} \equiv \bs{A[\sigma][\delta]}$. This holds
because the underlying set families are defined by function composition.

\emph{Remark.} Types could be equivalently defined as objects in
$\bs{\ms{flcwf}}/\bGamma$, and type substitution could be given as pullback, but
in that case we would run into the well-known strictness issue, that type
substitution is functorial only up to isomorphism. This is not a critical issue,
as there are standard solutions for recovering strict substitutions from weak
ones \cite{kapulkin12simplicial,local-univ,clairambault2014biequivalence}. But
if we ever need to look inside the definitions in the model, using displayed
algebras yields less encoding overhead than strictifying pullbacks.

A term $\bs{t : \Tm\,\Gamma\,A}$ is a displayed flcwf section, which again
strictly preserves all structure. We use the same notation for the action of
$\bt$ that we use for $\bSub$. We have the following underlying maps:
\begin{alignat*}{3}
  & \bt &&: (\ulGamma : \Con_{\bGamma}) \to \Con_{\bA}\,\ulGamma \\
  & \bt &&: (\ulsigma : \Sub_{\bGamma}\,\Gamma\,\Delta) \to
             \Sub_{\bA}\,(\bt\,\Gamma)\,(\bt\,\Delta)\,\ulsigma\\
  & \bt &&: (\ulA : \Ty_{\bGamma}\,\Gamma) \to \Ty_{\bA}\,(\bt\,\Gamma)\,\ulA\\
  & \bt &&: (\ult : \Tm_{\bGamma}\,\Gamma\,A) \to \Tm_{\bA}\,(\bt\,\Gamma)\,(\bt\,A)\,\ult
\end{alignat*}

A substituted term $\bs{t[\sigma]}$ for $\bs{t : \Tm\,\Delta\,A}$ and
$\bs{\sigma : \Sub\,\Gamma\,\Delta}$ is again given by component-wise function
composition.

An extended context $\bs{\Gamma \ext A}$ is the \emph{total flcwf} of
$\bA$. This is defined by combining corresponding underlying sets with
$\Sigma$-types:
\begin{alignat*}{3}
  & \Con_{\bs{\Gamma \ext A}} &&\defn (\ulGamma : \Con_{\bGamma}) \times \Con_{\bA}\,\ulGamma\\
  & \Sub_{\bs{\Gamma \ext A}}\,(\ulGamma,\,\Gamma)\,(\ulDelta,\,\Delta) &&\defn (\ulsigma : \Sub_{\bGamma}\,\ulGamma\,\ulDelta) \times \Sub_{\bA}\,\Gamma\,\Delta\,\ulsigma\\
  & \Ty_{\bs{\Gamma \ext A}}\,(\ulGamma,\,\Gamma) &&\defn (\ulA : \Ty_{\bGamma}\,\ulGamma) \times \Ty_{\bA}\,\Gamma\,\ulA\\
  & \Tm_{\bs{\Gamma \ext A}}\,(\ulGamma,\,\Gamma)\,(\ulA,\,A) &&\defn (\ult : \Tm_{\bGamma}\,\ulGamma\,\ulA) \times \Tm_{\bA}\,\Gamma\,A\,\ult
\end{alignat*}
All structure is defined pointwise, using $\bGamma$-structure for first
projections and $\bA$-structure for second projections. $\bs{\Gamma \ext A}$ may
be viewed as a dependent generalization of products of flcwfs.

\textbf{Comprehension structure} follows from the above definition: $\bs{\p}$ is
component-wise first projection, $\bs{\q}$ is second projection and substitution
extension $\bs{\blank,\blank}$ is pairing.

With this, we have a cwf of flcws. \emph{Remark:} the theory of flcwfs is itself
algebraic and has a finitary QII signature. Hence, if we succeed building
semantics for finitary QII signatures, we get ``for free'' an flcwf of
flcwfs. Of course, we cannot rely on this when we are in the process of defining
the $\bM$ model in the first place. Checking that the $\bM$ model indeed works,
is the somewhat tedious task that we have to perform \emph{once}, in order to
get semantics for any other finitary QII theory.

\subsection{Type Formers}

\subsubsection{Strict constant families}

This was not included in the ToS specification, but it is quite useful, so we
shall define it. $\bs{\K\,\Delta : \Ty\,\Gamma}$ is defined by ignoring
$\bGamma$ inhabitants in all underlying sets:
\begin{alignat*}{3}
  & \Con_{\bs{\K\,\Delta}}\,\ulGamma &&\defn \Con_{\bDelta}\\
  & \Sub_{\bs{\K\,\Delta}}\,\Gamma\,\Delta\,\ulsigma &&\defn \Sub_{\bDelta}\,\Gamma\,\Delta\\
  &  \Ty_{\bs{\K\,\Delta}}\,\Gamma\,\ulA &&\defn \Ty_{\bDelta}\,\Gamma\\
  &  \Tm_{\bs{\K\,\Delta}}\,\Gamma\,A\,\ult &&\defn \Tm_{\bDelta}\,\Gamma\,A
\end{alignat*}
All structure is inherited from $\bDelta$. There is also a type substitution
rule, expressing that for $\bs{\sigma : \Sub\,\Gamma\,\Xi}$, we have
$\bs{(\K\,\{\Xi\}\,\Delta)[\sigma]} \equiv \bs{\K\,\{\Gamma\}\,\Delta}$. This
follows immediately from the above definition and the definition of type
substitution, since the base inhabitants are ignored the same way on both sides
of the equation. We also need to show $\bs{\Tm\,\Gamma\,(\K\,\Delta)} \equiv
\bs{\Sub\,\Gamma\,\Delta}$. This again follows directly from the $\bK$
definition. From $\bK$, we get
\begin{itemize}
\item The unit type, defined as $\bK\,\bemptycon : \bTy\,\bGamma$.
\item Categorical products of $\bGamma$ and $\bDelta$, defined as $\bs{\Gamma \ext \K\,\Delta}$.
\item The ability to define closed type formers as elements of $\bCon$.
\end{itemize}

\subsubsection{Universe}

Similarly to what we did in Definition \ref{def:presheaf-univ}, we define $\bU$
as a context, and use $\bK$ later to get the universe as a type. $\bU :
\bCon$ is defined to be the flcwf where objects are inner types, and morphisms
are outer functions between them:
\begin{alignat*}{3}
  &\Con_{\bU} &&\defn \Ty_0 \\
  &\Sub_{\bU}\,\Gamma\,\Delta &&\defn \Tm_0\,\Gamma \to \Tm_0\,\Delta \\
  &\Ty_{\bU}\,\Gamma    &&\defn \Tm_0\,\Gamma \to \Ty_0\\
  &\Tm_{\bU}\,\Gamma\,A &&\defn (\gamma : \Tm_0\,\Gamma) \to \Tm_0\,(A\,\gamma)
\end{alignat*}
Substitution for types and terms is defined by function composition. The empty
context is defined as the inner unit type $\top_0$, and context extension
$\Gamma \ext_{\bU} A$ is defined as $(\gamma : \Gamma) \times A\,\gamma$ using
inner $\Sigma$. We can also define $\Sigma_{\bU}$ and $\Id_{\bU}$ using inner
$\Sigma$ and identity.

For constant families, we do not need any additional assumption in the inner
theory, since it can be defined as $\K_{\bU}\,\{\Gamma\}\,\Delta \defn \Delta$,
and $\Sub_{\bU}\,\Gamma\,\Delta \equiv \Tm_{\bU}\,\Gamma\,(\K_{\bU}\,\Delta)$
follows immediately.
\\\\
\indent For $\bs{a : \Sub\,\Gamma\,\U}$, we have to define $\bs{\El\,a :
  \Ty\,\Gamma}$. This is given as the \emph{displayed flcwf of elements} of
$\ba$.

Background: from any functor $F : \mbbC \to \bs{\Set}$ we can construct the
category of elements $\int\!F$, where objects are in $(i : |\mbbC|) \times F\,i$
and morphisms between $(i,\,x)$ and $(j,\,y)$ are in $(f : \mbbC(i,\,j)) \times
(F\,f\,x \equiv y)$. If we take the second projections of components in $\int\!F$,
we get the displayed category of elements, which lies over $\mbbC$. We may also call
this a \emph{discrete displayed category}, in analogy to discrete categories whose
morphisms are trivial.

We extend this to flcwfs in the definition of $\bs{\El\,a}$. With this
definition, $\bs{\Gamma \ext \El\,a}$ will yield the flcwf of elements of $\ba$.
\begin{alignat*}{3}
  &\Con_{\bs{\El\,a}}\,\ulGamma &&\defn \Tm_0\,(\ba\,\ulGamma) \\
  &\Sub_{\bs{\El\,a}}\,\Gamma\,\Delta\,\ulsigma &&\defn \ba\,\ulsigma\,\Gamma \equiv \Delta \\
  &\Ty_{\bs{\El\,a}}\,\Gamma\,\ulA &&\defn \Tm_0\,(\ba\,\ulA\,\Gamma)\\
  &\Tm_{\bs{\El\,a}}\,\Gamma\,A\,\ult &&\defn \ba\,\ult\,\Gamma \equiv A
\end{alignat*}
Let us check that we have all other structure as well.
\begin{itemize}
  \item For contexts and types, the task is to exhibit elements of $\ba$ lying over
        specific base contexts and types.
  \item For terms and substitutions, the task is to exhibit equations which
    specify the action of $\ba$.
  \item Equations between terms and substitutions are trivial because of UIP (we need
    to show equations between equality proofs).
\end{itemize}
We summarize below the additional structure on top of the displayed category
part of $\bs{\El\,a}$.
\begin{itemize}
  \item
  For $\emptycon_{\bs{\El\,a}} : \Con_{\bs{\El\,a}}\,\ulemptycon$, the type can
  be simplified along the definition of $\Con_{\bs{\El\,a}}$ and
  structure-preservation by $\ba$ to $\Tm_0\,\top_0$. Hence,
  $\emptycon_{\bs{\El\,a}} \defn \tt_0$ is the unique definition. For $\epsilon
  : \Sub_{\bs{\El\,a}}\,\Gamma\,\emptycon_{\bs{\El\,a}}\,\ul{\epsilon}$, we have to show
  $\ba\,\ul{\epsilon}\,\Gamma \equiv \tt_0$, which holds by the uniqueness of $\tt_0$.

  \item
  For $\Gamma \ext_{\bs{\El\,a}} A : \Con_{\bs{\El\,a}}\,(\ulGamma \ext \ulA)$,
  the target type unfolds to $\Tm_0\,(\ba\,(\ulGamma \ext \ulA))$, which in
  turn simplifies to $\Tm_0\,((\gamma : \ba\,\ulGamma) \times \ba\,\ulA\,\gamma)$.
  Since $\Gamma : \Tm_0\,(\ba\,\ulGamma)$ and $A : \Tm_0\,(\ba\,\ulA\,\Gamma)$,
  we define $\Gamma \ext_{\bs{\El\,a}} A$ as $(\Gamma,\,A)$.

  \item
  For comprehension, we have to show the following, after simplifying types:
  \begin{alignat*}{3}
    & \p &&: \ba\,\ul{\p}\,(\Gamma,\,A) &&\equiv \Gamma \\
    & \q &&: \ba\,\ul{\q}\,(\Gamma,\,A) &&\equiv A \\
    & (\sigma,\,t) &&: \ba\,(\ulsigma,\,\ult)\,\Gamma &&\equiv (\Delta,\,A)
  \end{alignat*}
  For $\p$ and $\q$, equations follow from preservation by $\ba$. For pairing,
  the goal further simplifies to $(\ba\,\ulsigma\,\Gamma,\,\ba\,\ult\,\Gamma)
  \equiv (\Delta,\,A)$. Then, the first and second components are equal by
  the $\sigma$ and $t$ hypotheses.

  \item
  Assuming $A : \Ty_{\bs{\El\,a}}\,\Delta\,\ulA$ and $\sigma :
  \Sub_{\bs{\El\,a}}\,\Gamma\,\Delta\,\ulsigma$, we aim to define
  $A[\sigma]_{\bs{\El\,a}} :
  \Ty_{\bs{\El\,a}}\,\Gamma\,(\ulA[\ulsigma])$. Simplifying types, $A :
  \Tm_0\,(\ba\,\ulA\,\Delta)$, $\sigma : \ba\,\ulsigma\,\Gamma \equiv \Delta$
  and the target type is $\Tm_0\,(\ba\,(\ulA[\ulsigma])\,\Gamma)$, which is the
  same as $\Tm_0\,(\ba\,\ulA\,(\ba\,\ulsigma\,\Gamma))$, by the preservation of
  $\blank[\blank]$ by $\ba$. Hence, by the $\sigma$ assumption, the target
  type is $\Tm_0\,(\ba\,\ulA\,\Delta)$, so we give the following definition:
  \[
    A[\sigma]_{\bs{\El\,a}} \defn A
  \]
  This is clearly functorial; moreover, substitution rules for the other type
  formers hold trivially.
  \item Term substitution is given by transitivity of equality.

  \item For $\Id_{\bs{\El\,a}}\,t\,u : \Ty_{\bs{\El\,a}}\,\Gamma\,(\Id\,\ult\,\ulu)$,
    the goal type is $\Tm_0\,(\ba\,(\Id\,\ult\,\ulu)\,\Gamma)$, hence
    $\Tm_0\,(\ba\,\ult\,\Gamma = \ba\,\ulu\,\Gamma)$. This holds by
    $t : \ba\,\ult\,\Gamma \equiv A$ and $u : \ba\,\ult\,\Gamma \equiv A$.
    Reflexivity and equality reflection are trivial by UIP.

  \item
    For $A : \Ty_{\bs{\El\,a}}\,\Gamma\,\ulA$ and $B :
    \Ty_{\bs{\El\,a}}\,(\Gamma \ext A)\,\ulB$, we aim to define
    $\Sigma_{\bs{\El\,a}}\,A\,B :
    \Ty_{\bs{\El\,a}}\,\Gamma\,(\Sigma\,\ulA\,\ulB)$, hence
    \begin{alignat*}{3}
      &\Sigma_{\bs{\El\,a}}\,A\,B : \Tm_0\,(\ba\,(\Sigma\,\ulA\,\ulB)\,\Gamma)\\
      &\Sigma_{\bs{\El\,a}}\,A\,B : \Tm_0\,((A : \ba\,\ulA\,\Gamma) \times \ba\,\ulB\,(\Gamma,\,A))\\
      &\Sigma_{\bs{\El\,a}}\,A\,B \defn (A,\,B)
    \end{alignat*}
    Projections and pairing proceed analogously to what we did for comprehension.
  \item
  For $\K_{\bs{\El\,a}}\,\Delta : \Ty_{\bs{\El\,a}}\,\Gamma\,(\K\,\ulDelta)$,
  the target type simplifies to $\Tm_0\,(\ba\,\ulDelta)$, hence we have
  $\K_{\bs{\El\,a}}\,\Delta \defn \Delta$. For the specifying sort equation of $\K$,
  we have to show
  \[
  \Sub_{\bs{\El\,a}}\,\Gamma\,\Delta\,\ulsigma \equiv
  \Tm_{\bs{\El\,a}}\,\Gamma\,(\K_{\bs{\El\,a}}\,\Delta)\,\ulsigma
  \]
  where $\ulsigma : \Sub\,\ulGamma\,\ulDelta$ but at the same time $\ulsigma :
  \Tm\,\ulGamma\,(\K\,\ulDelta)$ because of the $\K$ sort equation in the base.
  Fortunately, both sides simplify to $\ba\,\ulsigma\,\Gamma \equiv \Delta$.
\end{itemize}

\noindent We still have to check $\bs{(\El\,a)[\sigma]} \equiv \bs{\El\,(a \circ
  \sigma)}$, the naturality rule for $\bEl$. We only have to check equality of
underlying sets, $\Con$ and $\Ty$ formers, since terms and substitutions are
equal by UIP. For underlying sets, both sides compute to the following:
\begin{alignat*}{3}
  &\Con\,\ulGamma &&\defn \Tm_0\,(\ba\,(\bsigma\,\ulGamma)) \\
  &\Sub\,\Gamma\,\Delta\,\ulsigma &&\defn \ba\,(\bsigma\,\ulsigma)\,\Gamma \equiv \Delta \\
  &\Ty\,\Gamma\,\ulA &&\defn \Tm_0\,(\ba\,(\bsigma\,\ulA)\,\Gamma)\\
  &\Tm\,\Gamma\,A\,\ult &&\defn \ba\,(\bsigma\,\ult)\,\Gamma \equiv A
\end{alignat*}
Since $\bsigma$ also strictly preserves all structure, and we simply replace $\ba$ action
by the composite $\bs{a \circ \sigma}$ action, it is straightforward to check that $\Con$
and $\Ty$ formers are also the same on both sides.

At this point, we have $\bs{\U : \Con}$ and $\bs{\El : \Sub\,\Gamma\,\U}$. Let us
rename them to $\bU'$ and $\bEl'$ respectively, and define the usual ``open''
versions:
\begin{alignat*}{3}
  &\bs{\U} : \bs{\Ty\,\Gamma} &&\bs{\El} : \bs{\Tm\,\Gamma\,\U} \to \bs{\Ty\,\Gamma}\\
  &\bs{\U} \defn \bs{\K\,\U'}\hspace{2em}&&\bs{\El}\,\ba \defn \bEl'\,\ba
\end{alignat*}

\subsubsection{Identity}

Assuming $\bs{t,\,u : \Tm\,\Gamma\,A}$, extensional identity $\bs{\Id\,t\,u}$ is
defined as component-wise equality:
\begin{alignat*}{3}
  & \Con_{\bs{\Id\,t\,u}}\,\ulGamma &&\defn \bt\,\ulGamma \equiv \bu\,\ulGamma\\
  & \Sub_{\bs{\Id\,t\,u}}\,\Gamma\,\Delta\,\ulsigma &&\defn \bt\,\ulsigma \equiv \bu\,\ulsigma\\
  & \Ty_{\bs{\Id\,t\,u}}\,\Gamma\,\ulA &&\defn \bt\,\ulA \equiv \bu\,\ulA\\
  & \Tm_{\bs{\Id\,t\,u}}\,\Gamma\,A\,\ult &&\defn \bt\,\ult \equiv \bu\,\ult
\end{alignat*}
All other structure follows from structure-preservation of $\bt$ and $\bu$. For
the simplest example, $\emptycon_{\bs{\Id\,t\,u}} : \bt\,\ulemptycon \equiv
\bu\,\ulemptycon$ holds because $\bt$ and $\bu$ both preserve $\ulemptycon$. The
rule $\bs{(\Id\,t\,u)[\sigma]} \equiv \bs{\Id\,(t[\sigma])\,(u[\sigma])}$ is
straightforward to check: we only have to look at the underlying sets, where
e.g. both sides have $\Con\,\ulGamma \equiv (\bt\,(\bsigma\,\ulGamma) \equiv
\bu\,(\bsigma\,\ulGamma))$. It is also evident that
$\bs{\Tm\,\Gamma\,(\Id\,t\,u)}$ is equivalent to $\bt \equiv \bu$, that is, we
have reflexivity and equality reflection.

\subsubsection{Internal product type}

For $\bs{a : \Tm\,\Gamma\,\U}$ and $\bs{B : \Ty\,(\Gamma \ext \El\,a)}$, we aim
to define $\bs{\Pi\,a\,B : \Ty\,\Gamma}$. This is a dependent product of displayed
flcwfs, indexed over a \emph{discrete} domain. Discreteness is critical: since
morphisms in $\bs{\El\,a}$ are proof-irrelevant and invertible (because they are
equations), we avoid the variance issues that preclude general $\Pi$-types in
the cwf of categories \cite[Secion~A1.5]{johnstone2002sketches}.

The direct definition would be to define underlying sets as products, indexed
over corresponding components in $\bs{\El\,a}$:
\begin{alignat*}{3}
  & \Con_{\bs{\Pi\,a\,B}}\,\ulGamma &&\defn (\gamma : \ba\,\ulGamma) \to \Con_{\bB}\,(\ulGamma,\,\gamma)\\
  & \Sub_{\bs{\Pi\,a\,B}}\,\Gamma\,\Delta\,\ulsigma &&\defn
    \{\gamma : \ba\,\ulGamma\}\{\delta : \ba\,\ulDelta\}(\sigma : \Sub_{\bs{\El\,a}}\,\gamma\,\delta\,\ulsigma) \to \Sub_{\bB}\,(\Gamma\,\gamma)\,(\Delta\,\delta)\,(\ulsigma,\,\sigma)\\
  & \Ty_{\bs{\Pi\,a\,B}}\,\Gamma\,\ulA &&\defn
    \{\gamma : \ba\,\ulGamma\}(\alpha : \ba\,\ulA\,\gamma) \to \Ty_{\bB}\,(\Gamma\,\gamma)\,(\ulA,\,\alpha)\\
  & \Tm_{\bs{\Pi\,a\,B}}\,\Gamma\,A\,\ult &&\defn
    \{\gamma : \ba\,\ulGamma\}\{\alpha : \ba\,\ulA\,\gamma\}(t : \Tm_{\bs{\El\,a}}\,\gamma\,\delta\,\ult)
      \to
    \Tm_{\bB}\,(\Gamma\,\gamma)\,(\A\,\alpha)\,(\ult,\,t)
\end{alignat*}
But just like in Definitions \ref{def:simple-morphism} and
\ref{def:simple-section}, we can contract the $\Sub$ and $\Tm$ definitions,
since $\Sub_{\bs{\El\,a}}\,\gamma\,\delta\,\ulsigma \equiv
(\ba\,\ulsigma\,\gamma \equiv \delta)$ and
$\Tm_{\bs{\El\,a}}\,\gamma\,\alpha\,\ult \equiv (\ba\,\ult\,\gamma \equiv
\alpha)$.
\begin{alignat*}{3}
  & \Con_{\bs{\Pi\,a\,B}}\,\ulGamma &&\defn (\gamma : \ba\,\ulGamma) \to \Con_{\bB}\,(\ulGamma,\,\gamma)\\
  & \Sub_{\bs{\Pi\,a\,B}}\,\Gamma\,\Delta\,\ulsigma &&\defn
    (\gamma : \ba\,\ulGamma)\to \Sub_{\bB}\,(\Gamma\,\gamma)\,(\Delta\,(\ba\,\ulsigma\,\gamma))\,(\ulsigma,\,\refl)\\
  & \Ty_{\bs{\Pi\,a\,B}}\,\Gamma\,\ulA &&\defn
  \{\gamma : \ba\,\ulGamma\}(\alpha : \ba\,\ulA\,\gamma)
  \to \Ty_{\bB}\,(\Gamma\,\gamma)\,(\ulA,\,\alpha)\\
  & \Tm_{\bs{\Pi\,a\,B}}\,\Gamma\,A\,\ult &&\defn
    (\gamma : \ba\,\ulGamma) \to \Tm_{\bB}\,(\Gamma\,\gamma)\,(\A\,(\ba\,\ult\,\gamma))\,(\ult,\,\refl)
\end{alignat*}
With the contracted definition, $\Sub$ and $\Tm$ are only indexed over displayed
objects and types, but not over displayed morphisms or terms anymore. So it is
apparent that we cannot have issues with indexing variance. All structure in
$\bs{\Pi\,a\,B}$ is pointwise inherited from $\bB$. We list some examples below
for definitions.
\begingroup
\allowdisplaybreaks
\begin{alignat*}{3}
  &\emptycon_{\bs{\Pi\,a\,B}}\,\gamma &&\defn \emptycon_{\bB}\\
  &(\Gamma \ext_{\bs{\Pi\,a\,B}} A)\,(\gamma,\,\alpha) &&\defn (\Gamma\,\gamma \ext_{\bB} A\,\alpha)\\
  &\id_{\bs{\Pi\,a\,B}}\,\gamma &&\defn \id_{\bB}\\
  &(\sigma \circ_{\bs{\Pi\,a\,B}} \delta)\,\gamma &&\defn \sigma\,\gamma \circ_{\bB} \delta\,\gamma\\
  &A[\sigma]_{\bs{\Pi\,a\,B}}\,\{\gamma\}\,\alpha &&\defn (A\,\alpha)[\sigma\,\gamma]_{\bB}\\
  &\K_{\bs{\Pi\,a\,B}}\,\Delta\,\alpha &&\defn \K_{\bB}\,(\Delta\,\alpha)
\end{alignat*}
\endgroup
For the specifying isomorphism $\bs{(\app,\,\lam) : \Tm\,\Gamma\,(\Pi\,a\,B)
  \simeq \Tm\,(\Gamma \ext \El\,a)\,B}$, note that the difference in
presentation is exactly component-wise currying and uncurrying. For instance, in
$\bs{t : \Tm\,\Gamma\,(\Pi\,a\,B)}$, the underlying action on contexts has the following
type:
\[
  (\ulGamma : \Con_{\bGamma})(\gamma : \ba\,\ulGamma) \to \Con_{\bB}\,(\ulGamma,\,\gamma)
\]
While in $\bs{t : \Tm\,(\Gamma \ext \El\,a)\,B}$, we have
\[
  ((\ulGamma,\,\gamma) : (\ulGamma : \Con_{\bGamma}) \times \ba\,\ulGamma) \to \Con_{\bB}\,(\ulGamma,\,\gamma)
\]
So $\bs{\app}$ and $\bs{\lam}$ are defined as component-wise uncurrying and
currying respectively.  Naturality of $\bs{\Pi}$ and $\bs{\app}$ again follows
from the fact that flcwf morphisms strictly preserve all structure, and
substitution is component-wise function composition.

\subsubsection{External product type}

For $\Ix : \Ty_0$ and $\bB : \Tm_0\,\Ix \to \bs{\Ty\,\Gamma}$, we define
$\bs{\Pie}\,\Ix\,\bB : \bs{\Ty\,\Gamma}$ as the $\Ix$-indexed product of
a family of displayed flcwfs.
\begin{alignat*}{3}
  & \Con_{\bs{\Pie}\,\Ix\,\bB}\,\ulGamma &&\defn (i : \Tm_0\,\Ix) \to \Con_{\bB\,i}\,\ulGamma\\
  & \Sub_{\bs{\Pie}\,\Ix\,\bB}\,\Gamma\,\Delta\,\ulsigma &&\defn (i : \Tm_0\,\Ix) \to \Sub_{\bB\,i}\,(\Gamma\,i)\,(\Delta\,i)\,\ulsigma\\
  &  \Ty_{\bs{\Pie}\,\Ix\,\bB}\,\Gamma\,\ulA &&\defn (i : \Tm_0\,\Ix) \to \Ty_{\bB\,i}\,(\Gamma\,i)\,\ulA\\
  &  \Tm_{\bs{\Pie}\,\Ix\,\bB}\,\Gamma\,A\,\ult &&\defn (i : \Tm_0\,\Ix) \to \Tm_{\bB\,i}\,(\Gamma\,i)\,(A\,i)\,\ult
\end{alignat*}
All structure is defined in the evident pointwise way. $\bs{\appe}$ and
$\bs{\lame}$ are defined by component-wise flipping of function arguments. This
concludes the definition of the $\bM$ model.

\begin{myexample}
We look at the computation of a semantic flcwf, in the simple case of the flcwf
of $\Nat$-algebras. Recall that the signature is
\[
  \ms{NatSig} \defn \emptycon \ext (N : \U) \ext (\zero : \El\,N) \ext (\suc : N \funi \El\,N)
\]
We evaluate $\ms{NatSig}$ in $\bM$ entry-wise. We start from $\bemptycon$, the
terminal flcwf where algebras are elements of $\top$. Then, moving left to
right, we take the total flcwf of each type in the signature. From $\U$, we get
the product of $\top$ and the flcwf of sets, which is equivalent to simply the
flcwf of sets. Second, we extend this with the semantic $\El\,N$, which is the
\emph{displayed flcwf of points of sets}, to get the flcwf of pointed
sets. Finally, by extension with $N \funi \El\,N$, we get the flcwf of
$\Nat$-algebras.

Let us also look at some components of the resulting flcwf. Algebras, displayed
algebras, morphisms and sections have been already discussed before, so we look
at other components. We omit the leading $\top$ components everywhere in the following.

$\emptycon$ is the terminal $\Nat$-algebra, i.e.\ $\emptycon \equiv
(\top,\,\tt,\,\lambda\,\_.\,\tt)$. Context extension $\blank\ext\blank : (\Gamma
: \Con) \to \Ty\,\Gamma \to \Con$ constructs the total algebra of a displayed
algebra.
\begin{alignat*}{3}
&(N,\,z,\,s) \ext (N^D,\,z^D,\,s^D) \equiv\\
&\hspace{1em}(((n : N) \times N^D\,n),\,(z,\,z^D),\,(\lambda\,(n,\,n^D).\,(s\,n,\,s^D\,n\,n^D)))
\end{alignat*}
$\p$ and $\q$ respectively project first and second components from a total
algebra. For $t,\,u : \Tm\,(N,\,z,\,s)\,(N^D,\,z^D,\,s^D)$, $\Id\,t\,u$ is the
displayed $\Nat$-algebra which expresses equality of $\Nat$-algebra sections.
Let us review the definition of sections:
\begin{alignat*}{3}
&\rlap{$\Tm\,(N,\,z,\,s)\,(N^D,\,z^D,\,s^D) \equiv$}\\
  &\hspace{1em}        \,&&(N^S &&: (n : N) \to N^D\,n)\\
  &\hspace{1em} \times \,&&(z^S &&: N^S\,z \equiv z^D)\\
  &\hspace{1em} \times \,&&(s^S &&: (n : N) \to N^S\,(s\,n) \equiv s^D\,n\,(N^S\,n))
\end{alignat*}
We have that
\begin{alignat*}{3}
  & \Id\,(N^S_0,\,z^S_0,\,s^S_0)\,(N^S_1,\,z^S_1,\,s^S_1) \equiv\\
  & \hspace{1em}((\lambda\,n.\,N^S_0\,n \equiv N^S_1\,n),\,(\_ : N^S_0\,z \equiv N^S_1\,z),\,
  (\lambda\,n.\,(\_ : N^S_0\,(s\,n) \equiv N^S_1\,(s\,n))))
\end{alignat*}
The underscores denote omitted equality proofs; they follow from the $z^S$ and
$s^S$ components. It should be apparent that $\Tm\,\Gamma\,(\Id\,t\,u)$ is isomorphic to $t
\equiv u$; this follows from function extensionality and decomposition of
equalities of pairs. Thus, equality reflection holds in the flcwf of
$\Nat$-algebras. Note that we do not need to use equality reflection for
$\blank\equiv\blank$ to show this; it is simply a reshuffling of components along $\funext$.

$\K : \Con \to \{\Delta : \Con\} \to \Ty\,\Delta$ yields a non-dependent displayed algebra:
\[
\K\,(N,\,z,\,s)\,\{N',\,z',\,s'\} \equiv (\lambda\,\_.\,N,\,z,\,\lambda\,n\,\_.\,s\,n)
\]
With this definition, we indeed have that $\Tm\,\Gamma\,(\K\,\Delta) \equiv
\Sub\,\Gamma\,\Delta$.

$\Sigma : (A : \Ty\,\Gamma) \to \Ty\,(\Gamma \ext A) \to \Ty\,\Gamma$ is the
evident parameterized variant of $\blank\ext\blank$:
\begin{alignat*}{3}
  &\Sigma\,(N^D,\,z^D,\,s^D)\,(N^{D'},\,z^{D'},\,s^{D'}) \defn\\
  & \hspace{1em}((\lambda\,n.\,(n^D : N^D\,n) \times N^{D'}\,(n,\,n^D)),\,\\
  & \hspace{1em}(z^D,\,z^{D'}),\,\\
  & \hspace{1em}(\lambda\,n\,(n^D,\,n^{D'}).\,(s^D\,n\,n^D,\,s^{D'}\,(n,\,s^D\,n\,n^D)\,n^{D'})))
\end{alignat*}
\end{myexample}

\subsection{Recovering AMDS Interpretations}

We have defined the $\bM$ model in a ``bundled'' fashion, but sometimes we will
also need to refer to pieces of it. In Figure \ref{fig:fqiit-model} we have a
summary of the model. On the left, the rows are labeled with components of ToS,
while on the top we have components of flcwf. The individual rows can be further
unfolded, as each of them contains multiple components. Likewise the $\Sigma$,
$\Id$ and $\K$ columns can be unfolded. We get the whole model by filling every
cell of the unfolded table with a definition. Of course, many of these cells are
equations between equations, hence trivial by UIP.

This setup is very regular and convenient because we can extract a displayed
ToS model from any column, which may depend on columns to the left. The whole
model is the total model of all columns. For example, the $\Con$ column does not
depend on anything, so it is a plain model. The $\Ty$ column is displayed over
$\Con$. The $\Tm$ column depends on $\Con$ and $\Ty$, but it does not depend on $\Sub$.

See also Appendix \ref{app:fqii-amds} for a tabular specification of the AMDS
interpretations.

From each displayed model, we get an eliminator, i.e.\ a family of
interpretation functions. We note $\blank^A$, $\blank^M$, $\blank^D$ and
$\blank^S$ in the table, but in principle we could refer to the eliminators of
other columns as well. The interpretation functions can be in multiple ways.

One extreme choice is to separately take the eliminator for each column, and
refer to previous eliminators in each displayed model; e.g.\ referring to the
eliminator functions $\blank^A$ in the definition of the $\Ty$ column.  Another
extreme choice is to take the recursor for the entire model, then project out
components from the result.E.g.\ we get $\blank^A$ by projecting out the first
components of the interpretations of ToS objects.

Generally, we can bundle columns as we like. However, all variations coincide
because of the initiality of ToS syntax.

\begin{figure}
\begin{center}
\begin{tabular}{ |c|c|c|c|c|c|c|c|c|  }
 \hline
   & \multicolumn{5}{|c|}{cwf} & \multicolumn{3}{c|}{fl}\\
 \cline{2-9}
   & $\Con$ & $\Sub$ & $\Ty$ & $\Tm$ & $...$ & $\Sigma$ & $\Id$ & $\K$ \\
 \hline
   cwf    & \multirow{5}{2em}{$\blank^A$}&\multirow{5}{2em}{$\blank^M$}&\multirow{5}{2em}{$\blank^D$}&\multirow{5}{2em}{$\blank^S$}& & & & \\
   $\U$   &                              &              &                  &    & & & & \\
   $\Id$  &                              &              &                  &    & & & & \\
   $\Pi$  &                              &              &                  &    & & & & \\
   $\Pie$ &                              &              &                  &    & & & & \\
 \hline
\end{tabular}
\end{center}
\caption{The flcwf model of the theory of signatures}
\label{fig:fqiit-model}
\end{figure}

\subsection{Left Adjoints of Substitutions}
\label{sec:fqii-left-adjoint}

In this section we show that if all signatures have initial algebras, then the
semantic interpretation of each $\nu : \Sub\,\Omega\,\Delta$ has a left adjoint
functor. We have the following setup.
\begin{itemize}
  \item
  We write $\llb\blank\rrb$ for the interpretation into the flcwf model $\bM$.
  \item
  We close types in ToS under $\top$ and $\Sigma$, that is, we have $\top :
  \Ty\,\Gamma$ and $\Sigma : (A : \Ty\,\Gamma) \to \Ty\,(\Gamma \ext A) \to
  \Ty\,\Gamma$. The flcwf semantics can be immediately extended with these type
  formers: since flcwfs are given by an FQII signature, they form an flcwf
  themselves and support $\top$ (as $\K\,\emptycon$) and $\Sigma$. In the
  following we will need to talk about signatures depending on signatures, and
  $\top$ and $\Sigma$ are more convenient for this purpose than telescopes.
\end{itemize}

Given $\nu : \Sub\,\Omega\,\Delta$ in the ToS syntax, we get $\sem{\nu} :
\sem{\Omega} \to \sem{\Delta}$ as a functor between $\sem{\Omega}$ and
$\sem{\Delta}$ categories of algebras. We seek to construct some $\ms{L} :
\sem{\Delta} \to \sem{\Omega}$ such that $\ms{L} \dashv \sem{\nu}$.

The basic idea is the following: the existence of left adjoints is equivalently
characterized by having an initial object in the comma category
$\delta/\sem{\nu}$ for each $\delta : \Delta^A$
\cite[Section~IV]{maclane98categories}. Thus, it is enough to find some
signature $\Psi$ such that $\sem{\Psi}$ is equivalent to
$\delta/\sem{\nu}$, and by assumption we get an initial object. The objects
of $\delta/\sem{\nu}$ consist of the following:
\[
  (\omega : \Omega^A) \times (\eta : \Delta^M\,\delta\,(\nu^A\,\omega))
\]
Of the two components, $\omega : \Omega^A$ can be clearly represented as the
$\Omega$ signature. The $\eta$ component is a bit
more complicated. We need to represent a $\Delta$-morphism, but whose domain is
an external algebra, and whose codomain is an algebra internal to the ToS
syntax. In other words, we need a notion of ``heterogeneous'' morphism, where
the domain lives in the usual flcwf semantics, but the codomain lives in the
syntactic slice model $\ToS/\Omega$.

\begin{mydefinition}[\textbf{Heterogeneous morphisms}]
Fixing $\Omega : \Con$, we define $\blank^{HM}$ by induction on the ToS.
\begin{alignat*}{3}
  &\blank^{HM} &&: (\Gamma : \Con) \to \Gamma^A \to \Sub\,\Omega\,\Gamma \to \Ty\,\Omega\\
  &\blank^{HM} &&: (\sigma : \Sub\,\Gamma\,\Delta) \to \Tm\,\Omega\,(\Gamma^{HM}\,\gamma_0\,\gamma_1) \to \Tm\,\Omega\,(\Delta^{HM}\,(\sigma^A\,\gamma_0)\,(\sigma \circ\,\gamma_1))\\
  &\blank^{HM} &&: (A : \Ty\,\Gamma) \to A^A\,\gamma_0 \to \Tm\,\Omega\,(A[\gamma_1])
  \to \Tm\,\Omega\,(\Gamma^{HM}\,\gamma_0\,\gamma_1) \to \Ty\,\Omega\\
  &\blank^{HM} &&: (t : \Tm\,\Gamma\,A)(\gamma^{HM} : \Tm\,\Omega\,(\Gamma^{HM}\,\gamma_0\,\gamma_1))
   \to \Tm\,\Omega\,(A^{HM}\,(t^A\,\gamma_0)\,(t[\gamma_1])\,\gamma^{HM})
\end{alignat*}
The interpretation on contexts sums up the difference between the
``homogeneous'' $\blank^{M}$ and the current one. In the homogeneous
interpretation, we have $\Gamma^{M} : \Gamma^A \to \Gamma^A \to \Set$, in the
heterogeneous one the codomain of the relation is syntactic, and the return type
as well. We use $\top$ and $\Sigma$ in ToS to interpret contexts:
\begin{alignat*}{3}
  &\emptycon^{HM}\,\gamma_0\,\gamma_1 &&\defn \top\\
  &(\Gamma\ext A)^{HM}\,(\gamma_0,\,\alpha_0)\,(\gamma_1\,\alpha_1) &&\defn
    \Sigma\,(\gamma^{HM} : \Gamma^{HM}\,\gamma_0\,\gamma_1)\,(A^{HM}\,\alpha_0\,\alpha_1\,\gamma^{HM})
\end{alignat*}
We use a nameful notation for $\Sigma$-binding on the right hand side. In the
cwf interpretation we similarly reuse ToS type formers in a mechanical way,
following the definitions of the homogeneous $\blank^{HM}$.

\noindent $\U$ is interpreted using external function types:
\begin{alignat*}{3}
  & \U^{HM} : (a_0 : \Ty_0)(a_1 : \Tm\,\Omega\,\U) \to \Tm\,\Omega\,(\Gamma^{HM}\,\gamma_0\,\gamma_1)
  \to \Ty\,\Omega\\
  & \U^{HM}\,a_0\,a_1\,\gamma^{HM} \defn a_0 \toe \El\,a_1
\end{alignat*}
Note that this does not work if $a_0$ is syntactic and $a_1$ is external, as we
have no function type in ToS with external codomain; so $\blank^{HM}$ would not
work with an external second parameter. $\El^{HM}$ uses the $\Id$ type in ToS:
\begin{alignat*}{3}
  & (\El\,a)^{HM} : a^A\,\gamma_0 \to \Tm\,\Omega\,(\El\,(a[\gamma_1])) \to \Ty\,\Omega\\
  & (\El\,a)^{HM}\,\alpha_0\,\alpha_1\,\gamma^{HM} \defn \Id\,(a^{HM}\,\gamma^{HM}\,\alpha_0)\,\alpha_1
\end{alignat*}
In $\Pi$ we give the usual pointwise definition, using the external product type:
\[
  (\Pi\,a\,B)^{HM}\,t_0\,t_1\,\gamma^{HM} \defn (\alpha : a^A\,\gamma_0) \toe B^{HM}\,(t_0\,\alpha)\,(t_1\,(a^{HM}\,\gamma^{HM}\,\alpha))\,(\gamma^{HM},\,\refl)
\]
In $\Id$, we reuse the $\Id$ in ToS:
\begin{alignat*}{3}
  &(\Id\,t\,u)^{HM} : t^A\,\gamma_0 \equiv u^A\,\gamma_0 \to \Tm\,\Omega\,(\Id\,(t[\gamma_1])\,(u[\gamma_1]))
   \to \Tm\,\Omega\,(\Gamma^{HM}\,\gamma_0\,\gamma_1) \to \Ty\,\Omega\\
  &(\Id\,t\,u)^{HM}\,p_0\,p_1\,\gamma^{HM} \defn  \Id\,(t^{HM}\,\gamma^{HM})\,(u^{HM}\,\gamma^{HM})
\end{alignat*}
External products are again external products.
\[
  (\Pie\,\ms{Ix}\,B)^{HM}\,t_0\,t_1\,\gamma^{HM} \defn (i : \ms{Ix}) \toe (B\,i)^{HM}\,(t_0\,i)\,(t_1\,i)\,\gamma^{HM}
\]
The newly added $\top$ and $\Sigma$ type formers are evident:
\begin{alignat*}{3}
  &\top^{HM}\,\tt\,\tt\,\gamma^M \defn \top\\
  &(\Sigma\,A\,B)^{HM}\,(\alpha_0,\,\beta_0)\,(\alpha_1,\,\beta_1) \defn\\
  &\hspace{2em}\Sigma\,(\alpha^{HM} : A^{HM}\,\alpha_0\,\alpha_1\,\gamma^M)\,
          (B^{HM}\,\beta_0\,\beta_1\,(\gamma^{HM},\,\alpha^{HM}))
\end{alignat*}
\end{mydefinition}

\begin{mydefinition}[\textbf{Representing signature}]
\label{def:fqii-rep-signature}
\noindent For $\nu : \Sub\,\Omega\,\Delta$ and $\delta : \Delta^A$, we define the
signature which represents $\delta/\sem{\nu}$:
\[
  \ms{Sig}_{\delta/\sem{\nu}} \defn \Omega \ext \Delta^{HM}\,\delta\,\nu
\]
Now, we have that
\begin{alignat*}{3}
  &(\ms{Sig}_{\delta/\sem{\nu}})^A &&\equiv (\omega : \Omega^A) \times ((\Delta^{HM}\,\delta\,\nu)^A\,\omega)\\
  &(\ms{Sig}_{\delta/\sem{\nu}})^M\,(\omega_0,\,\eta_0)\,(\omega_1,\,\eta_1) &&\equiv (\omega^M : \Omega^M\,\omega_0\,\omega_1) \times ((\Delta^{HM}\,\delta\,\nu)^M\,\eta_0\,\eta_1\,\omega^M)
\end{alignat*}
It remains to show that $\sem{\ms{Sig}_{\delta/\sem{\nu}}}$ is indeed equivalent
to $\delta/\sem{\nu}$. It suffices to show that sets of objects and morphisms
are isomorphic. We need the following:
\begin{alignat*}{3}
  & A^{\simeq} &&: (\Delta^{HM}\,\delta\,\nu)^A\,\omega &&\simeq \Delta^M\,\delta\,(\nu^A\,\omega)\\
  & M^{\simeq} &&: (\Delta^{HM}\,\delta\,\nu)^M\,\eta_0\,\eta_1\,\omega^M &&\simeq (\nu^M\,\omega^M \circ A^{\simeq}\,\eta_0 \equiv A^{\simeq}\,\eta_1)
\end{alignat*}
These can be shown by induction on ToS again; we omit describing this here.
\end{mydefinition}

\begin{theorem}
If every FQII signature has an initial algebra, then for every $\nu :
\Sub\,\Omega\,\Delta$, there exists a left adjoint of $\sem{\nu} : \sem{\Omega} \to \sem{\Delta}$.
\end{theorem}
\begin{proof}
For each $\delta : \Delta^A$, the comma category $\delta/\sem{\nu}$ can be
specified with $\ms{Sig}_{\delta/\sem{\nu}}$ by Definition
\ref{def:fqii-rep-signature}, hence it has an initial object. The left adjoint
$\ms{L} : \sem{\Delta} \to \sem{\Omega}$ sends each $\delta : \Delta^A$ to the
$\omega : \Omega^A$ component of the initial algebra of
$\ms{Sig}_{\delta/\sem{\nu}}$.
\end{proof}

\section{Discussion of Semantics}

\subsection{Flcwfs For Free}

We give a quick summary for using the semantics of FQII signatures. As input we
pick a) a signature $\Gamma$ b) a cwf $\mbbC$ with $\Sigma$, $\top$ and
extensional $\Id$. Then, we interpret the signature in $\bM$, thereby getting an
flcwf in 2LTT. Then, we interpret that in presheaves over $\mbbC$, and we get
the flcwf whose objects are internal $\Gamma$-algebras in $\mbbC$.

One use case is in building models of certain type theories. Usually, this
starts with constructing the base cwf. But if the objects can be specified using
an FQII signature, we get an flcwf for free. In some cases, we get exactly
what is needed. For example, the flcwf of presheaves can be used as it is in the
presheaf models of type theories.

In other cases, the flcwf that we get has to be extended in some
ways. This often happens if the objects in the model have an internal notion of
``equivalence'' which has to be respected by types.
\begin{itemize}
  \item In the setoid model, objects are setoids and types are displayed
    setoids with additional fibrancy structure \cite{setoidtt}.
  \item The groupoid model \cite{hofmann96groupoidmodel} is analogous; again types are displayed groupoids
    with fibrancy structure.
  \item Likewise, in the cubical set model \cite{cubical}, types are displayed
    presheaves together with fibrancy structure (Kan composition).
\end{itemize}

In all these cases, the semantic objects have FQII signatures. We can interpret
their flwcfs in $\bs{\Set}$ and add fibrancy conditions. The cubical set model
is presented exactly in this way in \cite{cubical}, using displayed
algebras. The groupoid model in \cite{hofmann96groupoidmodel} instead presents
types as $\bs{\Gamma} \to \bs{\ms{Gpd}}$ functors, i.e.\ uses an indexed style
instead of the displayed style.

In the indexed-style groupoid model, we get strictly functorial type
substitution, just like in the displayed style. However, the displayed style
appears to be a more general way to get strict substitution, as it works for
every FQII theory. Again, although finitely complete categories can be always
strictified to cwfs, if we ever need to perform calculations with the
internal definitions of a model, the displayed style is much more direct.

\subsection{Variations of the Semantics}
\label{sec:fqii-variations}

In Section \ref{sec:fqiit-tos}, we required that the inner theory has $\Sigma$,
$\top$ and extensional $\Id$, and then used the assumed type formers in the
definition of $\bU$. Hence, when we interpret the semantic flcwf of a signature
in the presheaf model, we again need to assume these type formers in the base
cwf $\mbbC$.

However, we can drop $\Id$ from the requirements on the inner theory, and
likewise drop the identity type from flcwfs, and the model still works. In this
case we have a somewhat more general semantics. In particular, like in Section
\ref{sec:2ltt-simple-algebras}, we can interpret signatures in finite product
categories because $\top$ and $\Sigma$ can be derived from finite products in
the constructed ``simply typed'' cwf. On the other hand, we get less out of the
semantics. For instance, we cannot show equivalence of initiality and induction
without $\Id$.

If we want to trim down the assumptions on the inner theory to the minimum, we
can make do with simply an inner cwf with no type formers at all. This
implies that for each signature we can build a category of algebras, plus extra
structure which does not require $\Sigma$ or $\top$ in the $\bU$ definition. So
we may have displayed algebras, sections, and also functorial substitution for
these, but we do not have terminal algebras and total algebras.

We could also add more type formers to the semantics. For instance, we may add
an external product $\Pie$ (specified the same way as in signatures). Extending
flcwfs with $\Pie$ requires $\Pi$-types in the inner theory of 2LTT, hence in
$\mbbC$ as well. The reason is that indexed products of algebras require
functions in the underlying sorts. More concretely, in the definition of $\bU$
we have to interpret
\[
  \Pie_{\bU} : (\Ix : \Ty_0) \to (\Tm_0\,\Ix \to \Ty_{\bU}\,\Gamma) \to \Ty_{\bU}\,\Gamma
\]
hence
\[
  \Pie_{\bU} : (\Ix : \Ty_0) \to (\Tm_0\,\Ix \to \Tm_0\,\Gamma \to \Ty_0) \to \Tm_0\,\Gamma \to \Ty_0
\]
This works if we can return an inner $\Pi$ type in the definition:
\[
  \Pie_{\bU}\,\Ix\,B\,\gamma \defn (i : \Ix) \to B\,i\,\gamma
  \]
In this case, the flcwf semantics can be completed. We omit checking the details
here. If we have both extensional $\Id$ and $\Pie$, that yields small limits of algebras. If
we want to have ``simply typed'' semantics for this configuration, it is enough
to assume a cartesian closed base category $\mbbC$.

\subsection{Substitutions}

Interpreting signatures is not the only potentially useful thing that we get out
of the semantics. Each $\sigma : \Sub\,\Gamma\,\Delta$ can be viewed as a free
interpretation of the $\Delta$ theory in $\Gamma$, and we get a strict flcwf
morphism from the semantics.

\subsubsection{Ornaments}

One use case of $\Sub$ is to specify \emph{ornaments} \cite{ornaments}, i.e.\ ways
to decorate structures with additional information, or dually, to erase parts of
some structure. Ornaments differ from the usual forgetful maps because they
forget structure in \emph{negative} position, i.e.\ in assumptions of
construction rules.

\begin{myexample}
We assume $A : \Ty_0$. We define the substitution which forgets elements of
$A$-lists.
\begin{alignat*}{3}
  & \sigma : \Sub\,&&
  (\emptycon \ext (\mi{Nat} : \U) \ext (\mi{zero} : \El\,\mi{Nat}) \ext (\mi{suc} : \mi{Nat}))\\
  & &&(\emptycon \ext (\mi{List} : \U) \ext (\mi{nil} : \El\,\mi{List}) \ext (\mi{cons} : A \toe \mi{List} \to \mi{List}))
\end{alignat*}
The map goes from numbers to lists because of the ``contravariant''
forgetfulness. We define $\sigma$ by listing its component definitions.
\begin{alignat*}{3}
  &\ms{List} &&\defn \ms{Nat}\\
  &\ms{nil} &&\defn \ms{zero}\\
  &\ms{cons} &&\defn \lambda^{ext}\,\_.\,\lambda\,n.\,\ms{suc}\,n
\end{alignat*}
\end{myexample}

\begin{myexample}
  We assume $\Nat_0 : \Ty_0$ with $\zero_0$ and $\suc_0$, and define $\sigma : \Sub\,\ms{NatSig}\,\ms{FinSig}$, where $\ms{FinSig}$ is as follows:
\begin{alignat*}{3}
  &\ms{Fin}  &&: \Nat_0 \toe \U\\
  &\ms{zero} &&: (n : \Nat_0) \toe \El\,(\ms{Fin}\,(\suc_0\,n))\\
  &\ms{suc}  &&: (n : \Nat_0) \toe \ms{Fin}\,n \to \El\,(\ms{Fin}\,(\suc_0\,n))
\end{alignat*}
$\sigma$ is defined as
\begin{alignat*}{3}
  &\ms{Fin}  &&\defn \lambda^{ext}\,\_.\,\ms{Nat}\\
  &\ms{zero} &&\defn \lambda^{ext}\,\_.\,\ms{zero}\\
  &\ms{suc}  &&\defn \lambda^{ext}\,\_.\,\lambda\,n.\,\ms{suc}\,n
\end{alignat*}
\end{myexample}
For a specific programming use case, if we have any recursive function defined
on an ``erased'' type, we can convert that to a recursive function which acts on
an ``ornamented'' type. For example, if we have some $\Nat$-algebra $\Gamma$,
the recursor yields a morphism from the initial algebra to $\Gamma$. We can map
$\Gamma$ to a list-algebra or a $\ms{Fin}$-algebra, and then we can also use
recursors for lists or $\ms{Fin}$. Equivalently, we can map the unique morphism
to $\Gamma$ directly to a morphism between ornamented algebras.

Note though that a number of features and concepts from prior work on ornaments
are not yet reproduced. For example, we do not yet have an analogue of
\emph{algebraic ornaments}, which would allow us produce an ornamented signature
as an \emph{output} of a generic operation, instead of assuming it to begin
with. Exploring ornaments with QII signatures could be part of future work.

\subsubsection{Model constructions}

In a broader context, ToS provides a synthetic language for specifying
\emph{model constructions}.

\begin{myexample}
For a simple example, we might want to show that constant families are
equivalent to democracy in cwfs. Democracy means that for each $\Gamma : \Con$
there is a $\overline{\Gamma} : \Ty\,\emptycon$ such that $\Gamma \simeq
(\emptycon\ext\overline{\Gamma})$ \cite[Section~3.1]{flccc-undecidability}.

We can define a $\sigma : \Sub\,\ms{cwf^K}\,\ms{cwf^{dem}}$ which interprets
democracy using constant families. It is the identity morphism on the cwf parts
and interprets democracy as $\overline{\Gamma} \defn \ms{K}\,\Gamma$. The
isomorphism $\Gamma \simeq (\emptycon \ext \K\,\Gamma)$ follows from the
specification of $\K$. We can also define a morphism $\sigma^{-1} :
\Sub\,\ms{cwf^K}\,\ms{cwf^{dem}}$, which interprets $\K\,\Delta$ as
$\overline{\Delta}[\epsilon]$. It is easy to check that $\sigma^{-1}$ is indeed
the inverse of $\sigma$. Thus we get an isomorphism of flcwfs of models
from the ToS semantics.

This construction is very simple, and would not be difficult to check without
the ToS semantics. But it is generally not obvious that a certain mapping from
models to models extends to an flcwf morphism, so it may be helpful
to work inside ToS.
\end{myexample}

\begin{myexample}
There is a simple way to show that if a type theory does not support $\eta$ for
$\Pi$, then function extensionality is not provable in the theory
\cite{next700}.\footnote{It is also possible to show unprovability of function
extensionality \emph{assuming} $\eta$ for functions, but in significantly more
complicated ways. To the author's best knowledge, the set-based polynomial model
is the easiest solution \cite{von2015polynomials}.} Assume some type theory with $\Sigma$,
$\Pi$, $\Id$ and $\Bool$, and abbreviate its signature as $\ms{TT}$. We define a
$\sigma : \Sub\,\ms{TT}\,\ms{TT}$ which has identity action everywhere except on
$\Pi$. There, we have
\begin{alignat*}{3}
  &\Pi    &&\defn \lambda\,A\,B.\,\Pi\,A\,B \times \Bool \\
  &\app   &&\defn \lambda\,t.\,\app\,(\proj_1\,t)\\
  &\lam   &&\defn \lambda\,t.\,(\lam\,t,\,\true)
\end{alignat*}
In short, we tag functions with a $\Bool$ value. This equips $\Pi$ with
``intensional'' information, contradicting extensionality. If we have two
functions which are pointwise equal, that only specifies that the function parts
are equal, but does not say anything about the $\Bool$ tags. Hence, if we take
any model of $\ms{TT}$, we get a new model by the semantic action of $\sigma$,
where function extensionality is false. Note though that the $\eta$ rule also
fails in the new model, so we had to drop $\eta$ from the $\ms{TT}$ signature
as well.

In \cite{next700}, this construction is presented for the special case where the
starting model is initial. While it is easy to generalize to arbitrary starting
models, it is less obvious to extend the construction to a functor of categories
of models - which we do get for free here.
\end{myexample}

\begin{myexample}
\label{ex:gluing}
The gluing construction by Kaposi, Huber and Sattler \cite{gluing} takes as
input two models of some type theory together with a weak cwf-morphism between
them, and produces as output a displayed model over the first model. Depending
on the choice of the inputs, the gluing construction can yield parametricity
translations and canonicity proofs as well.

Let us use $\ms{TT} : \Ty\,\emptycon$ for the signature of the type theory,
given as an iterated large $\Sigma$-type. Then, the notion of weak cwf-morphism
is also expressible in ToS as $\ms{morph} : \Ty\,(\emptycon\,\ext(M_0 :
\ms{TT})\ext(M_1 : \ms{TT}))$, and the notion of displayed model as well, as
$\ms{TT}^D : \Ty\,(\emptycon\ext(M : \ms{TT})).$\footnote{We will be also able to automatically derive $\ms{TT}^D$ from $\ms{TT}$, in Section \ref{sec:signature-semantics}.} Thus, we can give a ``type''
for the gluing construction, as follows:
\begin{alignat*}{3}
  &\Tm\,&&(\emptycon\,\ext (M_0 : \ms{TT}) \ext (M_1 : \ms{TT}) \ext (f : \ms{morph}[M_0 \mapsto M_0,\,M_1 \mapsto M_1]))\\
  & && (\ms{TT}^D[M \mapsto M_0])
\end{alignat*}
Moreover, the gluing construction itself can be given as an inhabitant of the
above type. This construction works in the ToS because it only reuses structure
from $M_1$ to define the displayed model over $M_0$.
\end{myexample}

\textbf{Limitations.}
In the finitary ToS syntax, when defining substitutions we can only ever use
assumed type constructors. If we assume $\Sigma$ and $\top$ type formers in
the domain signature of a construction, we might be able to work around the lack
of $\Sigma$ and $\top$ in $\U$ in the ToS itself. This does not always work
though; for example, take the substitution with type
$\Sub\,\ms{MonoidSig}\,\ms{CatSig}$ which maps a monoid to a single-object
category. Assuming $\ms{M} : \U$ is the carrier set in $\ms{MonoidSig}$, we
would need to have the following:
\begin{alignat*}{3}
  &\ms{Obj} &&\defn \top \\
  &\ms{Hom} &&\defn \lambda\,\_\,\_.\, \ms{M}
\end{alignat*}
But we have $\Obj : \U$ in $\ms{CatSig}$, so we would need to have $\top : \U$.
In Chapter \ref{chap:iqiit}, we present a more expressive ToS which does include
$\top : \U$.

\subsection{Using Signatures in Implementations}
\label{sec:implementation}

We may ask whether the current ToS is suitable for implementations of type
theories. The answer is not wholly straightforward.

Note that we must choose a concrete surface syntax in an implementation, and
there are many design choices. The surface syntax would be almost certainly
nameful, and may or may not leave $\El$-s implicit, since they are not difficult
to insert by bidirectional elaboration. Besides the elaboration of surface
syntax, we should have at least the computation of induction principles.

Equality reflection in the ToS is a complication. If we have ``silent''
transports along equality reflection, that makes elaboration of surface
signatures undecidable. We might make transports explicit, which restores
decidable checking, but that requires the ToS to be deeply embedded in some
ambient theory.\footnote{Equality reflection is simply an equality constructor in
the embedded syntax, and has no bearing on decidability of type checking in the
metalanguage.}

Alternatively, we may just drop equality reflection from the ToS, and use
transport and UIP as primitives. This recovers decidable surface syntax, but now
we have to cover transport and UIP in the semantics, to be able to compute
induction principles. This is not too difficult; in Chapter \ref{chap:hiit} we
do the same for path induction $\J$ in the ToS. In that case, we even have a
Haskell implementation of signature elaboration and computation of induction
principles \cite{hiit-sig-program}.

Hence, handling signatures and computing induction principles is not difficult.
Instead, the real gap between our ToS and practical implementations is that we
need to have computationally adequate treatment of quotients. In plain
Martin-Löf type theories, computation gets stuck on quotients. We need to use
more recent systems, such as a cubical type theories \cite{cubicalagda,xtt}, or
some flavor of observational \cite{altenkirch2007observational} or setoid
\cite{setoidtt} type theory. In each of these systems, the signatures and their
semantics would need to be adapted, and we would need to work out additional
details. For example, we would need to produce extra computation rules which
explain the behavior of coercion or transport on QIIT constructors.

\section{Term Algebras}
\label{sec:fqiit-term-algebras}

In this section we proceed with the construction of term algebras for FQII
signatures, together with their recursors and eliminators. We make two
significant modifications to the setup.
\\\\
\indent First, \textbf{we drop the outer theory}, and work exclusively inside an
extensional type theory. The reason is the following. The main purpose of 2LTT
is to generalize the semantics of signatures. In the previous section, we
presented semantics for signatures, where algebras are internal to arbitrary
cwfs with $\Sigma$, $\top$ and extensional $\Id$. This is quite general; in
particular we can interpret signatures in any finitely complete category. We
also described dropping assumptions in Section \ref{sec:fqii-variations},
thereby getting semantics in yet more general settings.

In contrast, we make a lot more assumptions in the inner theory when we develop
initial term algebras; we essentially have to replicate the outer features
verbatim.  Thus, we gain nothing by using 2LTT, compared to working in a model
of an extensional TT.

What about the term model construction for simple signatures in Section
\ref{sec:simple-2ltt-term-algebras}, why did we use 2LTT there? In that case,
the inner theory was intensional, i.e.\ lacked equality reflection. So
there remained an interesting distinction between the inner and outer layer,
which allowed us to prove definitional $\beta$-rules for recursors. In
contrast, here we assume inner equality reflection, so we have no distinction
between propositional and definitional inner equality.
\\\\
\indent Second, \textbf{we make universe levels explicit} in the semantics and
constructions. So far, we have been consistently ignoring universe levels. Now,
size questions are less obvious, and quite relevant to a) ensuring the
consistency of assumed induction principles b) laying groundwork for
bootstrapped semantics and self-describing signatures in Section \ref{sec:closed-levitation}.

Universe levels are a fairly bureaucratic detail in type theories. In the
following we try to be as informal as possible, while still representing the
essential sizing aspects. In the following, we describe the new universe setup,
and adapt the previously described signatures and semantics to it.

\subsection{Universes \& Metatheory}
\label{sec:cumulative-ett}

We have $\mbb{N}$-indexed Russell-style $\Set_i$ universes, which are \emph{cumulative},
meaning that any type in $\Set_i$ is also an element of $\Set_{i+1}$. We use a
surface syntax which is similar to Coq, where cumulativity is implicit. This
contrasts the formal (``algebraic'') specification of cumulativity
\cite{sterling2019algebraic,kovacs2021generalized}, which involves rather heavy
explicit annotation.

Also following Coq, we have implicit \emph{cumulative subtyping}
\cite{timany18cumulative}. In our case, this means that cumulativity distributes
through basic type formers. We have a $\blank\!\leq\!\blank$ subtyping relation
on types, specified in Figure \ref{fig:cumulativity}.
\begin{figure}
\begin{mathpar}
  \inferrule*{i \leq j}
             {\Gamma \vdash \Set_i \leq \Set_j}

  \inferrule*{\Gamma,\,x : A \vdash B \leq B'}
             {\Gamma \vdash (x : A)\ra B \leq (x : A)\ra B'}

  \inferrule*{\Gamma\vdash A \leq A' \\ \Gamma,\,x : A \vdash B \leq B'}
             {\Gamma \vdash (x : A)\times B \leq (x : A') \times B'}

  \inferrule*{\\}
             {\Gamma \vdash A \leq A}

  \inferrule*{\Gamma \vdash A \leq B \\ \Gamma\vdash B \leq C}
             {\Gamma \vdash A \leq C}

  \inferrule*{\Gamma\vdash A \leq A' \\ \Gamma\vdash t : A}
             {\Gamma \vdash t : A'}
\end{mathpar}
\caption{Rules for cumulative subtyping}
\label{fig:cumulativity}
\end{figure}
This is subtyping for \emph{surface syntax}; it is expected that
surface syntax can be elaborated to \emph{coercions} in a formal syntax with
algebraic cumulativity.

Note that we have an invariant rule for function domain types. This is to
match Coq and \cite{timany18cumulative}, and also because we will not need a
contravariant rule in any case.

We assume that $\Pi$ and $\Sigma$ types return in least upper bounds of levels. For instance,
assuming $A : \Set_i$ and $B : A \to \Set_j$, we have $(x : A) \to B : \Set_{i \lub j}$.

\subsection{Signatures \& Semantics}
\label{sec:ett-signatures}

First, we parameterize the notion of ToS-model with levels.
\begin{mydefinition}
\label{def:ftos-models}
For levels $i$ and $j$, $\ToS_{i,j} : \Set_{i+1\lub j+1}$ is the type of ToS
models, defined as before, but where $\Con$, $\Sub$, $\Ty$ and $\Tm$ all return
in $\Set_i$, and $\Pie$ abstracts over $\Set_j$.
\end{mydefinition}

We have that $\ToS_{i,j} \leq \ToS_{i+1,j}$. This follows from the rules
in Figure \ref{fig:cumulativity}. All underlying sets return in $\Set_i$, which
can be bumped to $\Set_{i + 1}$.  Th $j$ level does not change, which is as
expected, since $\Set_j$ appears in a negative position in the type of $\Pie$,
and has to be invariant.

\textbf{Assumption.}
We assume that for all $j$, there exists $\syn_j : \ToS_{j+1,j}$ which supports
induction. Note the level bump in the first index; this is to avoid
inconsistency from type-in-type:
\begin{alignat*}{3}
  &\Ty &&: \Con \to \Set_{j+1}\\ &\Pie &&: (A : \Set_j) \to (A \to \Ty\,\Gamma)
  \to \Ty\,\Gamma
\end{alignat*}
With $\Ty$ returning in $\Set_j$, $\Pie$ would ``contain'' a $\Set_j$, but at
the same time return in a type in $\Set_j$, and by induction we would be able to
derive a Russell-like paradox. Likewise, all other underlying sets
must be bumped to $\Set_{j+1}$ because of their mutual nature: contexts, terms
and substitutions all ``contain'' types through some of their constructors.

\begin{mydefinition}[\textbf{Signatures}]
We define $\Sig_j : \Set_{j+1}$ as the type of signatures where $\Pie$ may
abstract over $\Set_j$, so we have $\Sig_j \defn \Con_{\syn_j}$.
\end{mydefinition}

\begin{mydefinition}[\textbf{Flwcf model}]
For levels $i$ and $j$, we have $\bM_{i,j} : \ToS_{(i+1\lub j)+1,j}$ as the
model where contexts are flcwfs, and objects in the flcwf are algebras.  The
model is defined in essentially the same way as in Section
\ref{sec:fqiit-semantics}. The algebras have underlying sets in $\Set_i$ and
(semantic) external products are indexed over types in $\Set_j$. Hence, every
algebra in $\bM_{i,j}$ is in $\Set_{i+1\lub j}$.
\end{mydefinition}

\begin{myexample}
We may define $\ms{NatSig}$ as an element of $\Sig_0$. Then, by interpreting the
signature in $\bM_{i,0}$, we get $\ms{NatSig}^A \equiv (N : \Set_i) \times (N
\to N) \times N$, hence $\ms{NatSig}^A : \Set_{i + 1 \lub 0}$.
\end{myexample}

\begin{notation}
For a signature $\Gamma : \Sig_j$ and level $i$, we may write $\Gamma^A_i$ for
the type of $\Gamma$-algebras with underlying sets in $\Set_i$, which is
computed by interpreting $\Gamma$ in $\bM_{i,j}$. We may use similar notation
for $\blank^M$, $\blank^D$ and $\blank^S$.
\end{notation}

\noindent\textbf{Cumulativity of algebras.} In the following, we shall assume
that for $\Gamma : \Sig_j$ and $i \leq i'$, we have $\Gamma^A_i \leq
\Gamma^A_{i'}$. For any concrete signature $\Gamma$, this is clearly the case,
but $\blank\!\leq\!\blank$ is not subject to propositional reasoning, so we
cannot prove this by internal induction on signatures. We can prove by induction
on signatures that there exists a \emph{lifting}, a $\ms{Lift}\,\Gamma^A_i :
\Set_{i+1 \lub j}$ which is isomorphic to $\Gamma^A_{i}$. Instead, we take
liberties, and work as if we had actual cumulative subtyping. This seems
acceptable, since by using implicit cumulativity, we are already taking the same
liberty everywhere, by omitting formal lifts and isomorphisms.

\subsection{Term Algebra Construction}
\label{sec:fqii-term-algebra-construction}
We fix $\Omega : \Sig_j$ for some $j$ level. We define $\blank^T$ by induction
on $\syn_j$. In the following we write $\blank^A$ for $\blank^A_{j+1}$,
i.e.\ the algebra interpretation where underlying sets are in
$\Set_{j+1}$. Formally, we need a displayed model over $\syn_j$, but we instead
present the resulting eliminator, which is perhaps easier to read. The
underlying functions have the following types.
\begin{alignat*}{3}
  &\blank^T &&: (\Gamma : \Con)&&(\nu : \Sub\,\Omega\,\Gamma) \to \Gamma^A\\
  &\blank^T &&: (\sigma : \Sub\,\Gamma\,\Delta)&&(\nu : \Sub\,\Omega\,\Gamma) \to \Delta^T\,(\sigma \circ \nu) \equiv \sigma^A\,(\Gamma^T\,\nu)\\
  &\blank^T &&: (A : \Ty\,\Gamma)&&(\nu : \Sub\,\Omega\,\Gamma) \to \Tm\,\Omega\,(A[\nu])
  \to A^A\,(\Gamma^T\,\nu)\\
  &\blank^T &&: (t : \Tm\,\Gamma\,A)&&(\nu : \Sub\,\Omega\,\Gamma) \to A^T\,\nu\,(t[\nu]) \equiv t^A\,(\Gamma^T\,\nu)
\end{alignat*}
We review the idea of term algebras. In any model of ToS, we might think of a
$\Sub\,\emptycon\,\Gamma$ as a $\Gamma$-algebra internal to the model. In the
$\blank^T$ interpretation we can assume $\Omega \equiv \emptycon$; this means
that from any internal $\Gamma$-algebra we can extract an ``external''
$\Gamma$-algebra. This is possible because every sort $a : \Tm\,\Gamma\,\U$ in
ToS induces an external type of terms as $\Tm\,\Gamma\,(\El\,a)$.

We can view the generalization from $\emptycon$ to arbitrary $\Omega$ as
switching from working in the syntactic model $\syn_j$, to working in the
\emph{slice model} $\syn_j/\Omega$, where contexts are given as $\Omega$ extended with
zero or more entries. And in $\syn_j/\Omega$, we have an $\Omega$-algebra quite
trivially, by taking the identity morphism $\id :
\Sub\,\Omega\,\Omega$.\footnote{Writing $\blank^{\syn_j/\Omega}$ for the interpretation
of syntax in the slice model,
$\Sub_{\syn_j/\Omega}\,\emptycon\,(\Omega^{\syn_j/\Omega})$ is isomorphic to, but not
strictly the same as $\Sub_{\syn_j}\,\Omega\,\Omega$.} Hence, term algebras
arise by first taking the trivial internal algebra $\id$ in $\syn_j/\Omega$,
then converting it to an external algebra as $\Omega^T\,\id : \Omega^A$.

\emph{Remark.} We could have presented $\blank^T$ and slice models separately.
We instead chose to merge them into the current $\blank^T$, since we do not use
slice models elsewhere, and we can skip their definition this way. Slice models
would require the specification of \emph{telescopes}, used to extend the base
context, but this entails a fair amount of bureaucratic detail.

We explain the $\blank^T$ specification in the following. Term and substitution
equations are given by UIP. We omit cases for substitutions and terms.

For contexts, we simply recurse on the entries. We use a pattern matching
notation for $\Sub\,\Omega\,(\Gamma\ext A)$, since any $\nu$ with this type is
uniquely determined by its first and second projections $\p\circ\nu$ and
$\q[\nu]$.
\begin{alignat*}{3}
  &\emptycon^T\,\nu           &&\defn \tt\\
  &(\Gamma \ext A)^T(\nu,\,t) &&\defn (\Gamma^T\,\nu,\,A^T\,\nu\,t)
\end{alignat*}
Type substitution with $\sigma : \Sub\,\Gamma\,\Delta$ is as follows. This is well-typed by
$\sigma^T\,\nu : \Delta^T\,(\sigma \circ \nu) \equiv \sigma^A\,(\Gamma^T\,\nu)$.
\[ (A[\sigma])^T\,\nu\,t \defn A^T\,(\sigma\circ\nu)\,t \]

\subsubsection{Universe}
For the universe, note that $\U^A_{j+1}\,\gamma \equiv \Set_{j+1}$.  As we
mentioned before, this is the key part when we map from internal sorts to
external sets. The levels line up, since in $\syn_j$ we have $\Tm$ returning in $\Set_{j+1}$.
\begin{alignat*}{3}
  &\U^T : (\nu : \Sub\,\Omega\,\Gamma) \to \Tm\,\Omega\,\U \to
          \Set_{j+1}\\
  &\U^T\,\nu\,a \defn \Tm\,\Omega\,(\El\,a)
\end{alignat*}
For $\El$, we have to define
\[
  (\El\,a)^T : (\nu : \Sub\,\Omega\,\Gamma)
          \to \Tm\,\Omega\,(\El\,(a[\nu])) \to a^A\,(\Gamma^T\,\nu)
\]
but since $a^T\,\nu : \Tm\,\Omega\,(\El\,(a[\nu]))
      \equiv a^A\,(\Gamma^T\,\nu)$, we have
\begin{alignat*}{3}
  &(\El\,a)^T : (\nu : \Sub\,\Omega\,\Gamma)
          \to \Tm\,\Omega\,(\El\,(a[\nu])) \to \Tm\,\Omega\,(\El\,(a[\nu]))\\
  &(\El\,a)^T\,\nu\,t \defn t
\end{alignat*}
The $a^T\,\nu$ equation is worth noting. If we have $\nu \equiv \id$, the
equation is $a^T\,\id : \Tm\,\Omega\,(\El\,a) \equiv
a^A\,(\Omega^T\,\id)$, that is, if we evaluate a signature sort in the term
model $\Omega^T\,\id$, we get a type of inner terms.

\subsubsection{Identity}
We have to show that provably equal terms are evaluated to the same value in the
term model.
\begin{alignat*}{3}
  &(\Id\,t\,u)^T : (\nu : \Sub\,\Omega\,\Gamma)
    \to \Tm\,\Omega\,(\Id\,(t[\nu])\,(u[\nu])) \to t^A\,(\Gamma^T\,\nu) \equiv u^A\,(\Gamma^T\,\nu)
\end{alignat*}
We know by equality reflection that $t[\nu] \equiv u[\nu]$, and
we also get
\begin{alignat*}{3}
  &t^T\,\nu &&: A^T\,\nu\,(t[\nu]) &&\equiv t^A\,(\Gamma^T\,\nu)\\
  &u^T\,\nu &&: A^T\,\nu\,(u[\nu]) &&\equiv u^A\,(\Gamma^T\,\nu)
\end{alignat*}
from which the target equality follows. Equality reflection for inner $\Id$ is
crucial here. It is the reason why $\blank^T$ works for \emph{quotient
signatures}; equality reflection is in fact the ``quotient'' rule which
identifies provably equal terms. For a simple example, terms with type
\[
  \Tm\,
  (\emptycon \ext (I : \U) \ext (\mi{left} : \El\,I) \ext (\mi{right} : \El\,I) \ext (\mi{seg} : \Id\,l\,r))\,
  (\El\,I)
\]
are quotiented by $\mi{seg}$, which is a provable equation in the context.

\subsubsection{Internal product type}
Here we have to convert an inner term with $\Pi$ type to an outer function.
\begin{alignat*}{3}
  &(\Pi\,a\,B)^T : (\nu : \Sub\,\Omega\,\Gamma)
                 &&\to \Tm\,\Omega\,(\Pi\,(a[\nu])\,(B[\nu\circ\p,\,\q]))\\
  &              &&\to (\alpha : a^A\,(\Gamma^T\,\nu)) \to B^A\,(\Gamma^T\,\nu,\,\alpha)\\
  &\rlap{$(\Pi\,a\,B)^T\,\nu\,t \defn \lambda\,\alpha.\,B^T\,(\nu,\,\alpha)\,(t\,\alpha)$}
\end{alignat*}
This is well-typed by $a^T\,\nu :
\Tm\,\Omega\,(\El\,(a[\nu])) \equiv
a^A\,(\Gamma^T\,\nu)$, which allows us to consider $\alpha$ to be an inner term
in $\lambda\,\alpha.\,B^T\,(\nu,\,\alpha)\,(t\,\alpha)$.

\subsubsection{External product type}
In this case we just recurse through the specifying isomorphism:
\begin{alignat*}{3}
  &(\Pie\,A\,B)^T : (\nu : \Sub\,\Omega\,\Gamma)
                 &&\to \Tm\,\Omega\,(\Pie\,A\,(\lambda\,\alpha.\,(B\,\alpha)[\nu]))\\
  &              &&\to (\alpha : A) \to (B\,\alpha)^A\,(\Gamma^T\,\nu)\\
  &\rlap{$(\Pie\,A\,B)^T\,\nu\,t \defn \lambda\,\alpha.\,(B\,\alpha)^T\,(\nu,\,\alpha)$}
\end{alignat*}
This concludes the definition of $\blank^T$.

\begin{mydefinition}
For an $\Omega : \Con_{\syn_j}$ signature, the corresponding \textbf{term
algebra} is given as $\Omega^T\,\id : \Omega^A_{j+1}$.
\end{mydefinition}

\emph{Remark.} If we start with a signature in $\syn_j$, then the underlying
sets in the term algebra are all in $\Set_{j+1}$. Hence, the term algebra for
$\NatSig : \Sig_0$ has an underlying set in $\Set_1$. This is perhaps
inconvenient, since normally we would have natural numbers in $\Set_0$. However,
we argue that this is no issue because we are free to specify $\Set_0$ as we
like. In particular, we can say that $\Set_0$ is an \emph{empty universe},
closed under no type formers at all (or explicitly isomorphic to $\bot$) in
which case $\Sig_0$ stands for \emph{closed} signatures (since $\Pie$ cannot be
constructed), and it is expected that any closed inductive type would be placed
in $\Set_1$. Alternatively, we could name the bottom-most universe
$\Set_{empty}$ or $\Set_{-1}$, and start counting non-empty universes from
$\Set_0$.

\subsection{Recursor Construction}

We continue with the construction of recursors. This is not necessary, strictly
speaking, since recursion is derivable from elimination, so it would suffice to
only construct eliminators. We still present recursors, for the sake of matching
the presentation in Chapter \ref{chap:simple-inductive-signatures}.

The goal is to construct a morphism from a term algebra to any other $\omega :
\Omega^A$ algebra. However, we have to handle universe levels as well. We want
to be able to eliminate from the term algebra, which was constructed at the
lowest possible level, to any other universe. So far we have not introduced a
``heterogeneous'' notion of morphism, between algebras at different levels. We
get this from cumulativity.
\begin{itemize}
  \item We assume $\Omega : \Sig_j$, for which we already have the term algebra $\Omega^T\,\id : \Omega^A_{j+1}$.
  \item We assume some $k \geq j + 1$, and an $\omega : \Omega^A_{k}$, the target of recursion.
  \item We implicitly lift $\Omega^T\,\id$ from level $j + 1$ to level $k$ by cumulativity, and construct
        a ``homogeneous'' morphism from the lifted term algebra to $\omega$.
\end{itemize}
This allows us to eliminate from $\Omega^T\,\id$ to any level. If we want to
eliminate to $k \geq j + 1$, we can lift the term algebra, and use a constructed
recursor. On the other hand, if we want to eliminate to $k < j + 1$, we can
instead lift the target $\omega : \Omega^A_{k}$ algebra to $j + 1$, and again
use a constructed recursor.

In general, for any $\omega : \Omega^A_{i}$ and $\omega' : \Omega^A_{j}$, the
notion of heterogeneous morphism between them arises by lifting both algebras to
$i \lub j$, and taking homogeneous morphisms between these.

\begin{myexample}
The $\ms{NatSig} : \Sig_0$ signature gives rise to $\ms{NatSig}^T\,\id :
\ms{NatSig}^A_1$. This consists of $\ms{Nat} : \Set_1$ together with $\zero$ and
$\suc$. Assuming a recursion principle as described above, and $\Bool : \Set_0$,
we may define an $\ms{isZero} : \Nat \to \Bool$ function by ``downwards''
elimination.  We have that $(\Bool,\,\true,\,\lambda\,\_.\,\false) :
\ms{NatSig}^A_0$, so by cumulativity we also have
$(\Bool,\,\true,\,\lambda\,\_.\,\false) : \ms{NatSig}^A_1$, hence by recursion
we get the desired morphism from $\ms{NatSig}^T\,\id$ to this model, which
contains the $\Nat \to \Bool$ function. We can also eliminate ``upwards'' by
lifting $\ms{NatSig}^T\,\id$ to any $\ms{NatSig}^A_i$ for $i > 1$.
\end{myexample}

We define $\blank^R$ by induction on $\syn_j$. From this, we will obtain the
recursor as $\Omega^R\,\id$.
\begin{alignat*}{3}
  &\blank^R &&: (\Gamma : \Con)&&(\nu : \Sub\,\Omega\,\Gamma) \to \Gamma^M\,(\nu^A\,(\Omega^T\,\id))\,(\nu^A\,\omega)\\
  &\blank^R &&: (\sigma : \Sub\,\Gamma\,\Delta)&&(\nu : \Sub\,\Omega\,\Gamma) \to \Delta^R\,(\sigma \circ \nu) \equiv \sigma^M\,(\Gamma^R\,\nu)\\
  &\blank^R &&: (A : \Ty\,\Gamma)&&(\nu : \Sub\,\Omega\,\Gamma)(t : \Tm\,\Omega\,(A[\nu]))
     \to A^M\,(t^A\,(\Omega^T\,\id))\,(t^A\,\omega)\,(\Gamma^R\,\nu)\\
  &\blank^R &&: (t : \Tm\,\Gamma\,A)&&(\nu : \Sub\,\Omega\,\Gamma) \to A^R\,\nu\,(t[\nu]) \equiv t^M\,(\Gamma^R\,\nu)
\end{alignat*}
Let us refresh some details about the involved operations. The reader may also
refer to Appendix \ref{app:fqii-amds} for definitions of the AMDS
interpretations.
\begin{itemize}
\item
  For $\nu : \Sub\,\Omega\,\Gamma$, we get $\nu^A : \Omega^A \to \Gamma^A$. In
  the semantics, $\nu$ is a functor, and $\nu^A$ is its action on
  objects. Analogously, for a term $t : \Tm\,\Omega\,A$, we have $t^A : (\gamma
  : \Omega^A) \to A^A\,\gamma$, also an action on objects.
\item
  $\Gamma^M$ is the set of $\Gamma$-morphisms. $A : \Ty\,\Gamma$ is a displayed
  flcwf in the semantics. $A^M$ yields sets of displayed morphisms,
  corresponding to the semantic $\Sub$ component. So we have
  \[ A^M : A^A\,\gamma_0 \to A^A\,\gamma_1 \to \Gamma^M\,\gamma_0\,\gamma_1 \to \Set_{k} \]
\item $t^M$ and $\sigma^M$ yield actions on morphisms. For $t : \Tm\,\Gamma\,A$ and $\sigma : \Sub\,\Gamma\,\Delta$, we have
  \begin{alignat*}{3}
    &t^M      &&: (\gamma^M : \Gamma^M\,\gamma_0\,\gamma_1) \to A^M\,(t^A\,\gamma_0)\,(t^A\,\gamma_1)\,\gamma^M\\
    &\sigma^M &&: (\gamma^M : \Gamma^M\,\gamma_0\,\gamma_1) \to \Delta^M\,(\sigma^A\,\gamma_0)\,(\sigma^A\,\gamma_1)
  \end{alignat*}
\end{itemize}

Again, we follow the ``sliced'' pattern that we have seen in the term model
construction. Another way to view this, is that getting term algebras or
recursors by direct induction on signatures is futile, since in the
construction we have to refer to the \emph{whole} $\Omega$ signature, but when
we recurse inside $\Omega$ we necessarily get \emph{smaller} signatures.

Hence, the sliced induction can be viewed as induction on arbitrary $\Gamma$
signatures which are smaller than $\Omega$, in the sense that there is a
$\Sub\,\Omega\,\Gamma$. Of course, $\Sub\,\Omega\,\Gamma$ includes ``being
smaller'', but it is more general.

We look at the interpretation of type formers. Again, term and substitution
equations are given by UIP, and we omit term and substitution formers.  For
contexts, we again simply recurse:
\begin{alignat*}{3}
  &\emptycon^R\,\nu           &&\defn \tt\\
  &(\Gamma \ext A)^R(\nu,\,t) &&\defn (\Gamma^R\,\nu,\,A^R\,\nu\,t)
\end{alignat*}
Type substitution with $\sigma : \Sub\,\Gamma\,\Delta$ also follows the same
pattern. The following is well-typed by $\sigma^R\,\nu :
\Delta^R\,(\sigma\circ\nu) \equiv \sigma^M\,(\Gamma^R\,\nu)$.
\[ (A[\sigma])^R\,\nu\,t \defn A^R\,(\sigma\circ\nu)\,t \]

\subsubsection{Universe}
We need to define
\begin{alignat*}{3}
  &\U^R : (\nu : \Sub\,\Omega\,\Gamma)(a : \Tm\,\Omega\,\U) \to \U^M\,(a^A\,(\Omega^T\,\id))\,(a^A\,\omega)\,(\Gamma^R\,\nu)
\end{alignat*}
Morphisms in the semantics of $\U$ are simply functions. Moreover, we have
$a^T\,\id : \Tm\,\Omega\,(\El\,a) \equiv a^A\,(\Omega^T\,\id)$.
\begin{alignat*}{3}
  &\U^R : (\nu : \Sub\,\Omega\,\Gamma)(a : \Tm\,\Omega\,\U) \to \Tm\,\Omega\,(\El\,a) \to a^A\,\omega\\
  &\U^R\,\nu\,a\,t \defn t^A\,\omega
\end{alignat*}
Thus, we evaluate $t$ in the $\omega$ algebra, the same way as we did in Chapter
\ref{chap:simple-inductive-signatures}.
\\\\
\noindent For $\El$, we need to show
\[
  (\El\,a)^R : (\nu : \Sub\,\Omega\,\Gamma)(t : \Tm\,\Gamma\,(\El\,(a[\nu]))) \to a^M\,(\Gamma^R\,\nu)\,(t^A\,(\Omega^T\,\id)) \equiv t^A\,\omega
\]
We have $a^R\,\nu : U^R\,\nu\,(a[\nu]) \equiv a^M\,(\Gamma^R\,\nu)$. Hence,
$U^R\,\nu\,(a[\nu])\,t \equiv a^M\,(\Gamma^R\,\nu)\,t$, and by computing $\U^R$
we have $t^A\,\omega \equiv a^M\,(\Gamma^R\,\nu)\,t$. The target equation then
follows by $t^T\,\id : t^A\,(\Omega^T\,\id) \equiv t$.

\subsubsection{Identity}

We need to show:
\[
(\Id\,t\,u)^R : (\nu : \Sub\,\Omega\,\Gamma)(e : \Tm\,\Gamma\,(\Id\,(t[\nu])\,(u[\nu])))
  \to t^M\,(\Gamma^R\,\nu) \equiv u^M\,(\Gamma^R\,\nu)
  \]
This follows from equality reflection on $e$, together with
\begin{alignat*}{3}
  & t^R\,\nu &&: A^R\,\nu\,(t[\nu]) \equiv t^M\,(\Gamma^R\,\nu)\\
  & u^R\,\nu &&: A^R\,\nu\,(u[\nu]) \equiv u^M\,(\Gamma^R\,\nu)
\end{alignat*}

\subsubsection{Internal product type}
We get the following target type after unfolding $(\Pi\,a\,B)^M$:
\begin{alignat*}{3}
 &(\Pi\,a\,B)^R : (\nu : \Sub\,\Omega\,\Gamma)(t : \Tm\,\Omega\,(\Pi\,(a[\nu])\,(B[\nu\circ\p,\,\q])))\\
 & \hspace{2em}\to (\alpha : a^A\,(\nu^A\,(\Omega^T\,\id))) \to B^M\,(t^A\,(\Omega^T\,\id)\,\alpha)\,(t^A\,\omega\,(a^M\,(\Gamma^R\,\nu)\,\alpha))\,(\Gamma^R\,\nu,\,\refl)
\end{alignat*}
We have
\begin{alignat*}{3}
  & \nu^T\,\id &&: \Gamma^T\,\nu \equiv \nu^A\,(\Omega^T\,\id) \\
  & a^T\,\nu   &&: a^A\,(\Gamma^T\,\id) \equiv \Tm\,\Omega\,(\El\,(a[\nu]))
\end{alignat*}
Hence, $a^A\,(\nu^A\,(\Omega^T\,\id)) \equiv \Tm\,\Omega\,(\El\,(a[\nu]))$. We
also have $a^R\,\nu : (\lambda\,\alpha.\,\alpha^A\,\omega) \equiv a^M\,(\Gamma^R\,\nu)$, therefore
$\alpha^A\,\omega \equiv a^M\,(\Gamma^R\,\nu)$. With this in mind, the goal type can be rewritten as
\begin{alignat*}{3}
 &(\Pi\,a\,B)^R : (\nu : \Sub\,\Omega\,\Gamma)(t : \Tm\,\Omega\,(\Pi\,(a[\nu])\,(B[\nu\circ\p,\,\q])))\\
 & \hspace{2em}\to (\alpha : \Tm\,\Omega\,(\El\,(a[\nu]))) \to B^M\,(t^A\,(\Omega^T\,\id)\,\alpha)\,(t^A\,\omega\,(\alpha^A\,\omega))\,(\Gamma^R\,\nu,\,\refl)
\end{alignat*}
We have the following typing now:
\[
  B^R\,(\nu,\,\alpha)\,(t\,\alpha) : B^M\,((t\,\alpha)^A\,(\Omega^T\,\id))\,((t\,\alpha)^A\,\omega)\,(\Gamma^T\,\nu,\,\refl)
\]
By the action of $\blank^A$ on internal application, we have
\[
   B^R\,(\nu,\,\alpha)\,(t\,\alpha) :
      B^M\,(t^A\,(\Omega^T\,\id)\,(\alpha^A\,(\Omega^T\,\id)))\,(t^A\,\omega\,(\alpha^A\,\omega))\,(\Gamma^T\,\nu,\,\refl)
\]
But since $\alpha^T\,\id : \alpha^A\,(\Omega^T\,\id) \equiv \alpha$, this is
exactly the target type. Therefore the definition is:
\[
  (\Pi\,a\,B)^R\,\nu\,t \defn \lambda\,\alpha.\,B^R\,(\nu,\,\alpha)\,(t\,\alpha)
\]

\subsubsection{External product type}
We again simply recurse through the indexing:
\begin{alignat*}{3}
  & (\Pie\,A\,B)^R : (\nu : \Sub\,\Omega\,\Gamma)(t : \Tm\,\Omega\,(\Pie\,A\,(\lambda\,\alpha.\,(B\,\alpha)[\nu])))\\
  & \hspace{2em}\to (\alpha : A) \to (B\,\alpha)^M\,(t^A\,(\Omega^T\,\id)\,\alpha)\,(t^A\,\omega\,\alpha)\,(\Gamma^R\,\nu)\\
  &(\Pie\,A\,B)^R\,\nu\,t \defn \lambda\,\alpha.\,(B\,\alpha)^R\,\nu\,(t\,\alpha)
\end{alignat*}
This concludes the definition of $\blank^R$.

\begin{mydefinition}[\textbf{Recursors}]
Assuming $\Omega : \Sig_j$, a $k$ level such that $k \geq j + 1$ and $\omega :
\Omega^A_{k}$, we have $\Omega^R\,\id : \Omega^M\,(\Omega^T\,\id)\,\omega$ as
the recursor for the term algebra.
\end{mydefinition}

\subsection{Eliminator Construction}
\label{sec:fqii-eliminator-construction}

We assume $\Omega : \Sig_j$ and $\omega^D : \Omega^D_{k}\,(\Omega^T\,\id)$,
where $k \geq j + 1$. Again we implicitly lift the term algebra from level $j+1$
to $k$. Here, $\omega^D$ is a displayed algebra over the term algebra. We seek
to construct an inhabitant of $\Omega^S\,(\Omega^T\,\id)\,\omega^D$. We define
$\blank^E$ by induction.

Constructing eliminators is on the whole quite similar to the recursor
construction. The switch from morphisms to sections is mechanical. We shall only
look at $\U$, $\El$ and $\Pi$ here.
\begin{alignat*}{3}
  &\blank^E &&: (\Gamma : \Con)&&(\nu : \Sub\,\Omega\,\Gamma) \to \Gamma^S\,(\nu^A\,(\Omega^T\,\id))\,(\nu^D\,\omega^D)\\
  &\blank^E &&: (\sigma : \Sub\,\Gamma\,\Delta)&&(\nu : \Sub\,\Omega\,\Gamma) \to \Delta^E\,(\sigma \circ \nu) \equiv \sigma^S\,(\Gamma^E\,\nu)\\
  &\blank^E &&: (A : \Ty\,\Gamma)&&(\nu : \Sub\,\Omega\,\Gamma)(t : \Tm\,\Omega\,(A[\nu]))
     \to A^S\,(t^A\,(\Omega^T\,\id))\,(t^D\,\omega^D)\,(\Gamma^E\,\nu)\\
  &\blank^E &&: (t : \Tm\,\Gamma\,A)&&(\nu : \Sub\,\Omega\,\Gamma) \to A^E\,\nu\,(t[\nu]) \equiv t^S\,(\Gamma^E\,\nu)
\end{alignat*}
For the \textbf{universe}, we have the following.
\begin{alignat*}{3}
  &\U^E : (\nu : \Sub\,\Omega\,\Gamma)(a : \Tm\,\Omega\,\U) \to (\alpha : a^A\,(\Omega^T\,\id)) \to a^D\,\omega^D\,\alpha
\end{alignat*}
By $a^T\,\id : a^A\,(\Omega^T\,\id) \equiv \Tm\,\Omega\,(\El\,a)$, we can give the following definition:
\begin{alignat*}{3}
  &\U^E : (\nu : \Sub\,\Omega\,\Gamma)(a : \Tm\,\Omega\,\U) \to (\alpha : \Tm\,\Omega\,(\El\,a)) \to a^D\,\omega^D\,\alpha\\
  &\U^E\,\nu\,a\,\alpha \defn \alpha^D\,\omega^D
\end{alignat*}
In other words, we evaluate $\alpha$ in the $\omega^D$ displayed algebra. Let us check that this is well-typed:
\begin{alignat*}{3}
  & \alpha^D &&: \{\omega : \Omega^A\}(\omega^D : \Omega^D\,\omega) \to a^D\,\omega^D\,(\alpha^A\,\omega)\\
  & \alpha^D\,\omega^D &&: a^D\,\omega^D\,(\alpha^A\,(\Omega^T\,\id) \\
  & \alpha^T\,\id      &&: \alpha^A\,(\Omega^T\,\id) \equiv \alpha
\end{alignat*}
Thus $\alpha^D\,\omega^D : a^D\,\omega^D\,\alpha$. Recall that $\alpha^D$ can
be viewed as the logical predicate interpretation of $\alpha$, which expresses
that $\alpha^A$ preserves $\blank^D$ predicates.
\\\\
\noindent For \textbf{$\El$}, we need to show
\[
  (\El\,a)^S : (\nu : \Sub\,\Omega\,\Gamma)(t : \Tm\,\Gamma\,(\El\,(a[\nu]))) \to a^S\,(\Gamma^E\,\nu)\,(t^A\,(\Omega^T\,\id)) \equiv t^D\,\omega^D
\]
This follows from $t^T\,\id : t^A\,(\Omega^T\,\id) \equiv t$ and $a^E\,\nu : (\lambda\,t.\,t^D\,\omega^D) \equiv a^S\,(\Gamma^E\,\nu)$.
\\\\
\noindent The \textbf{internal product} interpretation is defined similarly as before:
\begin{alignat*}{3}
 &(\Pi\,a\,B)^E : (\nu : \Sub\,\Omega\,\Gamma)(t : \Tm\,\Omega\,(\Pi\,(a[\nu])\,(B[\nu\circ\p,\,\q])))\\
 & \hspace{2em}\to (\alpha : \Tm\,\Omega\,(\El\,(a[\nu]))) \to B^S\,(t^A\,(\Omega^T\,\id)\,\alpha)\,(t^D\,\omega^D\,(\alpha^D\,\omega^D))\,(\Gamma^E\,\nu,\,\refl)\\
 &\rlap{$(\Pi\,a\,B)^E\,\nu\,t \defn \lambda\,\alpha.\,B^E\,(\nu,\,\alpha)\,(t\,\alpha)$}
\end{alignat*}
We make use of $\nu^T\,\id$, $u^T\,\id$, $a^E\,\nu$ and $a^T\,\nu$ to type-check the definition.

Interpretations for contexts and other type formers are also essentially the same as with recursors.

\begin{mydefinition}[\textbf{Eliminators}]
\label{def:fqiit-eliminator}
Assuming $\Omega : \Sig_j$, a $k$ level such that $k \geq j + 1$ and $\omega^D :
\Omega^D_{k}\,(\Omega^T\,\id)$, we have $\Omega^E\,\id : \Omega^S\,(\Omega^T\,\id)\,\omega^D$ as
the eliminator.
\end{mydefinition}

\begin{theorem}
  $\Omega^T\,\id$ is initial when lifted to any $k \geq j + 1$ level.
\end{theorem}
\begin{proof}
  $\Omega^T\,\id : \Omega^A_{k}$ supports elimination by Definition \ref{def:fqiit-eliminator},
  and elimination is equivalent to initiality by Theorem \ref{thm:initiality-induction}.
\end{proof}

\section{Levitation and Bootstrapping for Closed Signatures}
\label{sec:closed-levitation}

When we previously introduced the ToS, we only specified the notion of model,
and simply assumed that there is an evident notion of model morphism and also a
notion of induction. For the theory of \emph{closed} signatures, we can do
better because ToS is itself a closed FQII theory. This is
\emph{levitation} \cite{chapman2010gentle}, i.e.\ the situation where a ToS contains a
signature for itself. Levitation is useful for bootstrapping: it shall be
sufficient to specify only the notion of model for ToS, and notions of
ToS-morphisms, initiality and induction can be computed from that. This
bootstrapping process eliminates the need for either
\begin{itemize}
  \item Assuming that the syntax of ToS already exists as a QIIT. Here, the assumed
        syntax is necessarily ad-hoc, since we are still in the process
        of building metatheory for QII theories.
  \item Bootstrapping the ToS syntax as ``raw'' syntax, using simple inductive
        types, typing/conversion relations and quotients. This is very tedious and
        should be avoided if possible. See Section \ref{sec:finitary-reductions} for
        a discussion of this approach, although used for slightly different
        purposes.
\end{itemize}

In this section we describe levitation for closed signatures. The theory of
closed signatures does not have $\Pie$, but is otherwise the same as before.  As
we have seen, the inclusion of $\Pie$ yields a ToS which is itself infinitary,
which breaks levitation. Moving to a theory of infinitary signatures will
restore levitation; we revisit this is Section \ref{sec:iqii-levitation}.

\subsection{Models \& Signatures}

Since we do not have $\Pie$, we only need a single universe level for indexing models.

\begin{mydefinition} For some $i$ level, we have $\ToS_i : \Set_{i+1}$ as the type of
models of ToS, where all underlying sets return in $\Set_i$.
\end{mydefinition}

\begin{mydefinition}[\textbf{Flcwf model}]
For $i$, we have $\bM_{i} : \ToS_{i + 2}$ as the model where contexts are flcwfs
of algebras, and algebras have underlying sets in $\Set_i$. To see how $i + 2$
checks out: if algebras contain $\Set_i$-s, the category of algebras has a
$\Set_{i + 1}$ for a set of objects, and $\bM_{i}$ itself includes a category of
these categories.
\end{mydefinition}
So far, this can be defined while only using the notion of model for ToS. What
about signatures though? Previously we had that signatures are contexts in ToS
\emph{syntax}, and to talk about syntax, we need to know at least the notion of
ToS model morphism.

Actually, if we only want to be able to write down signatures and interpret them
in the semantics, we do not need a ToS syntax. A functional encoding suffices.

\begin{mydefinition} A \textbf{bootstrap signature} is a function which for every ToS model
yields a context in that model. The type of bootstrap signatures is:
\[
  \ms{BootSig} \defn (i : \ms{Level}) \to (M : \ToS_i) \to \Con_{M}
\]
Note that this is a universe-polymorphic type. This is not an issue; universe
polymorphism is a sensible feature in type theories, or alternatively we may
assume that quantification over levels takes place in some outer theory.

We do not get induction on bootstrap signatures, nor do we automatically
get any naturality or parametricity property.

\end{mydefinition}

\begin{myexample}
For $\ms{NatSig}$, we define the expected signature, but we specify it in an arbitrary
$M$ model instead of the syntax.
\begin{alignat*}{3}
  & \NatSig : \ms{BootSig}\\
  & \NatSig \defn \lambda (i :\ms{Level})(M : \ToS_i).\\
  & \hspace{5.2em}(\emptycon_M\,\ext_M\, (N : \U_M) \,\ext_M\,(\mi{zero} : \El_M\,N)
      \ext_M (\mi{suc} : N\arri_M\El_M\,N))
\end{alignat*}
\end{myexample}
We might as well use the same notations for signatures as in Section
\ref{sec:fqiit-tos}, as every signature from before can be unambiguously
rewritten as a bootstrap signature.

With this, we can interpret each signature in an arbitrary ToS model, by
applying a signature to a model. $\ms{BootSig}_j$ can be viewed as a precursor
to a Böhm-Berarducci encoding \cite{boehm-berarducci} for the theory of
signatures, but we only need contexts encoded in this way, and not other ToS
components. In functional programming, this style of encoding is sometimes
called ``finally tagless'' \cite{carette2007finally}.

If we only want to build the 2LTT-based semantics of signatures, we are done
with bootstrapping right now. In the 2LTT semantics, we never needed induction
on ToS, we only needed to be able to write down signatures and interpret them in
models - which we can do. Going forward, we only need to assume an inner
$(\Ty_0,\,\Tm_0)$ layer with appropriate type formers, and define the flcwf model
the same way as before.

On the other hand, if we want to consider term models, we do need a notion of
induction on ToS.

\begin{mydefinition}[\textbf{Signature for ToS}]
We define $\ToSSig : \ms{BootSig}$ as the bootstrap signature for the theory of
signatures. We present an excerpt from $\ToSSig$ below using internal notation;
it should be clear that every component can be reproduced. We use $\ms{SigU}$
and $\ms{SigEl}$ to disambiguate components inside the signature from ToS
components.
\begingroup
\allowdisplaybreaks
\begin{alignat*}{3}
  & \Con       &&: \U\\
  & \Sub       &&: \Con \to \Con \to \U\\
  & \Ty        &&: \Con \to \U\\
  & \Tm        &&: (\Gamma : \Con) \to \Ty\,\Gamma \to \U\\
  & ...        &&\\
  & \ms{SigU}  &&: \{\Gamma : \Con\} \to \El\,(\Ty\,\Gamma)\\
  & \ms{SigEl} &&: \{\Gamma : \Con\} \to \Tm\,\Gamma\,\ms{SigU} \to \El\,(\Ty\,\Gamma)\\
  & \Pi        &&: \{\Gamma : \Con\}(a : \Tm\,\Gamma) \to \Ty\,(\Gamma\ext\,\ms{SigEl}\,a)
                   \to \El\,(\Ty\,\Gamma)\\
  & ...        &&
\end{alignat*}
\end{mydefinition}
\endgroup
For each $i$, the interpretation of $\ToSSig$ in $\bM_{i}$ yields an flcwf
$\bGamma$ such that $\Con_{\bGamma} \equiv \ToS_{i}$, that is, objects are
models of ToS at level $i$. This yields a model theory for ToS, which includes
the notion of induction at level $i$.

We also know by the definition of $\ToS_{i}$ that we have cumulativity,
i.e.\ $\ToS_{i} \leq \ToS_{i+1}$.\footnote{For concrete bootstrap signatures we
may conclude cumulativity of algebras, but we cannot conclude this universally
for all bootstrap signatures, since we cannot do induction on them, and we do not
even assume that they are parametric in levels.} Hence, we can make the
following definition:

\begin{mydefinition} $M : \ToS_0$ supports elimination into any universe if
it supports elimination when lifted by cumulativity to any $\ToS_i$.
\end{mydefinition}

This notion of (large) elimination is sufficient for the term algebra and
eliminator constructions in Section \ref{sec:fqiit-term-algebras}. Thus, we were
able to derive all required concepts just from the notion of model of ToS.

\section{Reductions to Basic Type Formers}
\label{sec:finitary-reductions}

From the construction of term algebras and eliminators, we get a reduction of
all QIITs to a single infinitary QIIT, namely the syntax of ToS. We spell this out:

\begin{theorem} If an extensional type theory supports syntax for $\ToS_{j+1,j}$, it supports
initial algebras for each signature in $\Sig_j$.
\end{theorem}

Ideally, we would like to reduce QIITs to some collection of basic type
formers. The ToS syntax is far from being a basic type former, it is rather
large and complicated. Therefore, the remaining job is to construct the ToS
syntax from simpler types.

We do not attempt here to construct the entire ToS syntax as specified.
Lumsdaine and Shulman \cite[Section 9]{lumsdaineShulman} showed that infinitary
QIITs are not constructible from inductive types and simple quotienting with
relations. Recently, Fiore, Pitts and Steenkamp showed that a class of
infinitary quotient inductive types, called QWI-types, can be reduced to
inductive types, quotients and the axiom of weakly initial sets of covers (WISC)
\cite{DBLP:journals/corr/abs-2101-02994}. The setting additionally assumes
extensional equality and propositional extensionality for an impredicative
universe of propositions. The infinitary ToS syntax is not immediately a
QWI-type because it is inductive-inductive.  Nevertheless, it is a reasonable
conjecture that infinitary QIITs are also constructible from the WISC
principle. We leave this to future work.

In this section we show constructions of certain fragments of the full ToS
syntax. We first give a general description of QIIT constructions, then describe
two specific constructions, for a) finitary inductive-inductive signatures b)
closed QII signatures.

\subsection{Finitary QIIT Constructions}
\label{sec:fqiit-constructions}

The general recipe of constructing finitary QIITs from basic type formers is the
following. This is more or less adapted from Streicher \cite{streicher93habil}
and Brunerie et al.\ \cite{brunerie}.
\begin{enumerate}
  \item
    We define the \emph{raw} syntax, using at most inductive families, but
    no induction-induction. These definitions include all value constructors of
    the goal QIIT, but there is no indexing involved, constructors only store
    the raw inductive data. For example, the raw syntax of closed ToS would
    include the following:
    \begin{alignat*}{5}
      & \Con &&: \Set \hspace{3em}&& \emptycon             &&: \Con\\
      & \Sub &&: \Set && \blank\!\ext\!\blank  &&: \Con \to \Ty \to \Ty \\
      & \Ty  &&: \Set && \id                   &&: \Con \to \Sub \\
      & \Tm  &&: \Set && \blank\!\circ\!\blank &&: \Con \to \Con \to \Con \to \Sub \to \Sub \to \Sub\\
      & ... && &&
    \end{alignat*}
    This can be given by a simple mutual inductive definition, which can be
    represented as an indexed inductive family. Indexed families can be
    reduced to indexed W-types \cite{mutualinductive}, which can be
    reduced in turn to W-types and the identity type.
  \item
    We define typing and conversion relations on the raw syntax. For dependent
    type theories, the two are usually mutual: typing includes the rule which
    coerces terms along type conversion, and conversion is usually defined only
    on well-typed terms.  However, it is still possible to define everything
    using only indexed inductive families.
  \item
    The underlying sets are given as follows: we take raw syntactic objects
    which are \emph{merely} well-formed (i.e.\ proofs of well-formedness are
    propositionally truncated, or defined in a universe of irrelevant
    propositions to begin with), and quotient them by conversion.
  \item
    We show that these underlying sets support all constructors of the target
    QIIT: value constructors are defined using raw constructors, while
    equality constructors follow from conversion rules and quotienting.
  \item
    We construct a unique morphism from the above term model to an arbitrary
    model of the QII theory. This usually requires several steps. One approach
    is to first define by induction on raw syntax a family of partial
    functions into the assumed model, then separately show that these functions
    are total on well-typed input. The separation is necessary because the
    induction principle for the raw syntax is too weak: it cannot express
    the inductive-inductive indexing dependencies which would be required to
    construct the morphism in one go. For instance, if we have the QIIT syntax
    for ToS, and we have some displayed model $A$ over the syntax, the
    eliminator contains the following:
    \begin{alignat*}{3}
      &\Con^S &&: (\Gamma : \Con) \to \Con_A\,\Gamma \\
      &\Sub^S &&: (\Gamma\,\Delta : \Con)(\sigma : \Sub\,\Gamma\,\Delta) \to \Sub_{A}\,(\Con^S\,\Gamma)\,(\Con^S\,\Delta)\,\sigma
    \end{alignat*}
    But with the raw syntax, we can only eliminate using a displayed model of
    the raw syntax, and the eliminator contains the following:
    \begin{alignat*}{3}
      &\Con^S &&: (\Gamma : \Con) \to \Con_A\,\Gamma \\
      &\Sub^S &&: (\sigma : \Sub) \to \Sub_A\,\sigma
    \end{alignat*}
    Lastly, we show that the constructed morphism is unique. This is done by
    induction on raw syntax, and is generally possible in just one elimination.
\end{enumerate}

Note that the above recipe permits a large number of design variations. Some
examples:
\begin{itemize}
\item We may omit fields from raw syntax which are fully determined by type indices.
      This may make subsequent work easier or harder depending on particulars.
\item We may start from a \emph{well-scoped} raw syntax, if there is a notion of
      scoping in the goal QIIT. In general, we may start from some kind of partially
      raw syntax, which is well-typed to some extent. This extent is bounded by what
      is expressible only using indexed inductive families but not
      induction-induction.
\item We may move along a spectrum of ``paranoia'' in the specification of
      well-typing \cite[Section~9.2]{winterhalter-thesis}. A paranoid typing rule
      assumes the well-formedness of everything involved, for example assumes the
      well-formedness of a context $\Gamma$ before it assumes well-formedness of a
      type in $\Gamma$. In contrast, an ``economic'' specification tries to make the
      minimum necessary assumptions, relying on admissibility properties. It is
      possible that well-formedness of $\Gamma$ is derivable from the
      well-formedness of a type in $\Gamma$, so the assumption can be dropped.

      However, if we omit too much, then some other admissibility properties may
      break! Design decisions along the paranoia spectrum are often all tangled
      up like this; hence the name ``paranoid'', which probably stems from the
      anxiety of breaking things by making too many shortcuts.
    \item Instead of using partial maps from raw syntax to the the assumed model
      in step 5, we may define \emph{relations} between \emph{well-formed} raw
      syntax and the given model, and later show that these relations are
      functional.  This seems to be a technically easier approach. The reason is
      that we do not have decidable definedness of the partial maps, which makes
      them more complicated.  A decidably defined partial function has type $A
      \to \ms{Maybe}\,B$. For any $a : A$ we can look at whether the function is
      defined on it. A more general partial function has type $A \to ((P :
      \ms{Prop}) \times (P \to B))$.  If we forget about the $\ms{Prop}$-ness of
      $P$ for the time being, we can equivalently have a relation $A \to B \to
      \Set$ instead.  This is a more ``indexed'' definition compared to the
      ``fibered'' presentation with $P : \ms{Prop}$, and indexed presentations
      in type theory usually enjoy more definitional computation rules - this
      is also the reason why displayed algebras are better-behaved computationally
      than fibered algebras.
\end{itemize}

It should be apparent that constructing QIITs is tedious, and especially so for
large QIITs like type theories. Hence, it is best if we do it just once, for a
theory of signatures from which every other QIIT can be constructed.

\subsubsection{Connection to the initiality conjecture}

The initiality conjecture was made by Voevodsky \cite{voevodsky-initiality}, and
it is essentially the conjecture that the above construction (``initiality
construction'') can be carried out in sufficient formal detail for ``usual''
type theories.

There has been much debate about the merits of initiality constructions. See
\cite{initiality-project} for a hub of such discussions.  On one hand, some
people believed that the initiality construction is essential for reconciling
the usage of raw syntax and categorical notions of models. On the other hand,
some people dismissed the initiality construction as a pointless exercise,
considering the categorical syntax to be the actual syntax, and raw syntax as
merely notation for that. The author of this thesis is of a somewhat different
opinion than either of the above.

First, as a moral justification for the usage of raw syntax, the initiality
construction is indeed mostly pointless. That is because \emph{elaboration}
comprises the true justification for that. Elaboration is the effective
algorithm which converts raw syntax to ``core syntax'', i.e.\ typed categorical
syntax. Given a piece of raw syntax, even if we have done the initiality
construction, we have no effective way of learning which core syntactic object it
corresponds to!  The elaboration literature is mainly about practical
justifications for using certain raw syntaxes, and it comes with established
ways to talk about strength and correctness of elaboration algorithms.

Second, there is a different motivation for the initiality construction:
\emph{foundational minimalism}, the reduction of a complicated QIIT to basic
type formers. Elaboration merely assumes that a categorical core syntax already
exists, as the target of elaboration, but it is orthogonal to the construction of
the core syntax. If we have elaboration, we may still want to show a reduction
of the core syntax, but now we are free to perform this construction in whatever
way is the easiest. We do not have to construct the QIIT out of a raw syntax
which is intentionally close to the raw syntax that we use in practice! In the
author's opinion, a great deal of confusion arises from the conflation of the
two different motivations for the initiality construction.

As to which way of construction is easiest: the author does not know of any
truly easy way, but this thesis shows that we only have to do it once, for a
theory of signatures, and then we can construct all other QIITs from that in a
generic way. In particular, almost all type theories in the wild are finitary closed QII
theories (with the notable exceptions of our ToS-es), so if we can construct
closed signatures, we can construct initial models of almost all type theories.

What about generic ways to formalize elaboration algorithms? This seems to be a
lot more difficult. To the author's knowledge there has not been notable
research in this area. Decidability of conversion is already very hard to
analyze in a generic way, and the simplest possible bidirectional elaboration
algorithms rely on decidable conversion. To formalize practically realistic
elaboration (i.e.\ elaboration which includes unification) is yet more
difficult.

\subsection{Reduction of Finitary Inductive-Inductive Types}

This section is based on the author's joint work with Kaposi and Lafont
\cite{ind-ind-reduction}. The core idea is the following: a certain fragment of
ToS can be constructed in a far simpler way than what we described in Section
\ref{sec:fqiit-constructions}, with fewer assumptions in the ambient theory.  We
call this fragment the theory of \emph{finitary inductive-inductive} signatures.
This theory has the following type formers (on the top of the base cwf):
\begin{itemize}
  \item The $\U$ universe with $\El$.
  \item Inductive function type $\Pi$, but without $\lami$, and thus without $\beta\eta$-rules.
  \item External function type $\Pie$, but again without $\lame$.
\end{itemize}
This ToS is tuned so that
\begin{enumerate}
  \item No quotients are required in its construction.
  \item The generic term model construction still goes through for every signature in
        the ToS.
\end{enumerate}
We explain in the following. First, the equational theory of ToS only specifies
substitution, but it contains no computation rules for type formers. Thus, ToS
is a theory of \emph{neutral terms} and substitutions. This allows us to define
a raw syntax which only includes normal forms, and to define substitution as
recursive functions acting on normal forms. This trivializes the conversion
relation: conversion is simply propositional equality of raw terms. Thus, there
is no need to quotient by conversion. Note that our raw syntax is infinitary
because we have to represent the branching in $\Pie$. This is fine though: we
only run into the issue of the missing choice principle (presumably, the
WISC principle) if we try to mix quotients and infinite branching. Without
quotients, infinite branching is not an issue.

Second, we do not include an identity type in ToS. This blocks the other way for
quotients to enter the picture. With identity types, the generic term model
construction relies on equality reflection in ToS. But when we construct ToS
syntax, the only way to show equality reflection is to quotient raw syntax by
internally provable equalities.

Third, it remains to check that the generic term model construction works with
the pared-down ToS. We only need to check that the omission of $\lam$ and
$\lame$ does not mess things up. Looking at Sections
\ref{sec:fqii-term-algebra-construction} and
\ref{sec:fqii-eliminator-construction}, we see that it does not: the
interpretations of $\Pi$ and $\Pie$ only require applications in ToS, not
abstractions.

\emph{Remark.} Although we have not yet talked about infinitary signatures, we
can give a short summary why the current construction fails to work for their
ToS.  The generic term algebra construction in Section
\ref{sec:iqii-term-algebra-construction} for infinitary signatures relies on
there being both $\lam$ and $\app$ for ``infinitary'' function types, with
$\beta\eta$-rules. This makes the equational theory of ToS non-trivial, so
quotients are necessary in the construction of the syntax. However, this
requires mixing quotients and infinite branching, which we cannot yet handle.

We summarize the construction of the ToS syntax below. We refer the reader to
\cite{ind-ind-reduction} for details.

\begin{enumerate}
\item We define raw syntax by mutual induction. Substitutions are in normal form:
      they are just lists of raw terms. Variables are also normalized as de Bruijn indices.
      We define the action of substitution by recursion on raw syntax. In \cite{ind-ind-reduction},
      raw syntax is not well-scoped, and substitution is partial, but it would be also possible
      to start from well-scoped raw syntax.
\item We inductively define well-formedness relations for contexts,
      substitutions, types and terms, and show by induction on raw syntax that
      well-formedness is propositional, i.e.\ proof-irrelevant. Alternatively, we
      could have defined well-formedness by recursion on raw syntax.
\item We construct a term model of ToS from well-formed raw syntax. All equations in the model
      are provable from the properties of recursive substitution on raw terms.
\item We pick a ToS model, and inductively define a family of relations between
      the term model and the given model, which define the function graphs of the
      model morphism that we aim to define. Then we show in order:
      \begin{enumerate}
        \item Right-uniqueness of the relation, by induction on well-formedness derivations.
        \item Stability of the relation under substitution.
        \item Left-totality of the relation, by induction on well-formedness derivations.
      \end{enumerate}
      We then define the actual model morphism using the functionality of the relation.
\item For the uniqueness of the constructed morphism, we exploit
      right-uniqueness of the relation: it is enough to show that any other model
      morphism maps syntactic input to related semantic output.
\end{enumerate}

This construction is formalized in Agda; see \cite{ind-ind-reduction}. It uses
indexed inductive families, UIP, function extensionality, and equality
reflection in the form of Agda rewrite rules, although the latter could be in
principle omitted from the formalization. Thus, it follows that any model
of ETT with inductive families supports finitary inductive-inductive types.

\subsection{Reduction of Closed QIITs}

For closed QIITs, there is unfortunately no direct formalization which
constructs the ToS. There is one though which is \emph{close enough}, by Menno
de Boer and Guillaume Brunerie \cite{initiality-agda}; see also De Boer's thesis
\cite{deboer-initiality}. This constructs a type theory with the following
features:
\begin{itemize}
  \item A contextual category for base (instead of a cwf).
  \item Countable predicative universes.
  \item $\mathbb{N}$, $\Sigma$, $\Pi$, $\top$, $\bot$, $\blank\!+\!\blank$ and
        intensional $\Id$.
\end{itemize}
The construction follows the 1-5 steps that we described previously in Section \ref{sec:fqiit-constructions}. It makes the following assumptions:
\begin{itemize}
\item A universe of strict propositions $\ms{Prop}$. Every type in this universe
      enjoys definitional proof-irrelevance. This $\ms{Prop}$ is used to define
      partial functions and well-formedness relations.
\item Function extensionality.
\item Propositional extensionality for $\ms{Prop}$.
\item Quotients by relations valued in $\ms{Prop}$.
\item Indexed inductive families returning in $\Set$ or in $\ms{Prop}$.
\end{itemize}
UIP is not assumed, instead the irrelevant equality in $\ms{Prop}$ is used
everywhere. Although it is only possible to eliminate from such equalities to
$\ms{Prop}$, this issue is sidestepped by using an essentially algebraic
specification of models, which is fibered using $\ms{Prop}$ equations.

It is very plausible that this construction can be adapted to our theory of
closed QII signatures. De Boer and Brunerie construct a complicated open
finitary QIIT, while ours is a fairly similar closed QIIT, with fewer and more
restricted type formers. The openness comes from the use of contextual
categories, which involve indexing by external natural numbers. Contextuality
does not make much difference in the construction though, since raw syntax is
always contextual by the inductive nature of raw contexts.

Hence, it is safe to say that any model of an extensional type theory which
supports the assumptions of De Boer and Brunerie, also supports all closed
QIITs.

\section{Related Work}
\label{sec:fqii-related-work}

This chapter is based on the following publications, all coauthored by the
current thesis' author.
\begin{enumerate}
  \item ``Constructing Quotient Inductive-Inductive Types'' \cite{kaposi2019constructing}.
  \item ``Large and Infinitary Quotient Inductive-Inductive Types'' \cite{iqiit}.
  \item ``For Finitary Induction-Induction, Induction is Enough'' \cite{ind-ind-reduction}.
\end{enumerate}
We summarize the differences and enhancements in this chapter, in comparison to the above
(1)-(3) sources.

The theory of signatures is similar to that in (1), except (1) does not include
eliminators for $\Pi$ and $\Pie$, and it has $\Id : \Tm\,\Gamma\,(\El\,a) \to
\Tm\,\Gamma\,(\El\,a) \to \Ty\,\Gamma$, i.e.\ it cannot equate terms with arbitrary
types.

The usage of 2LTT is novel compared to (1)-(3). In (1), the semantics had a cwf
with $\Id$ and $\K$ for each signature; this was extended with $\Sigma$ in (2)
to get the notion of flcwf that we also use in this chapter.

The construction of left adjoints of substitutions is novel.

The current term algebra construction is the same as in (1), but universe levels
were not treated rigorously in (1); instead we adapt the more precise universe
treatment from (2). Notions of bootstrapping and levitation are also
``backported'' from (2) to closed finitary signatures.

(3) is summarized in the current chapter without any notable change.

\subsubsection{ToS-style presentations}

Carette and O'Connor \cite{presentation-combinators} presented algebraic
signatures as contexts in type theories. Altenkirch and Kaposi \cite{ttintt}
observed that induction methods and motives can be computed as logical predicate
translations on typing contexts.

\subsubsection{Generalized algebraic theories}

FQII signatures and Cartmell's generalized algebraic theories \cite{gat} are
close in expressive power, but they do not appear to be equivalent.

GATs may contain an infinite number of rules, while FQII signatures are finitely
long. On the other hand, FQII signatures have $\Pie$ and GATs do not. It appears
that infinite signatures are stronger than $\Pie$: it is possible to recover
$\Pie$ by adding a rule for every value of the external index, but it is not
possible to recover infinite signatures with $\Pie$. The reason is that in $\Pie
: (\mi{Ix} : \Ty_0) \to (\mi{Ix} \to \Ty\,\Gamma) \to \Ty\,\Gamma$, the $\Gamma$
context is fixed, so it is not possible to represent a family of signature entries
where each entry may refer to the previous entry within the same family. For example,
the following (pseudo)-GAT has no corresponding FQII signature:
\begin{alignat*}{3}
  &A_0 : \U\\
  &A_1 : A_0 \to \U\\
  &A_2 : (a_0 : A_0) \to A_1\,a_1 \to \U\\
  &A_3 : (a_0 : A_0)(a_1 : A_1\,a_0) \to A_2\,a_0\,a_1 \to \U\\
  &...
\end{alignat*}
Could we somehow include these? The most convenient way would be to define
signatures coinductively. However, that would cause a mismatch, that described
theories are inductive, while the ToS itself is coinductive, which rules out
levitation and bootstrapping. It is potential future work to investigate such
coinductive signatures.

This leads us to the main difference in formalization between GATs and FQII
signatures: the theory of GATs itself is not presented as a GAT, instead it has
a low-level presentation with raw syntax and well-formedness relations. As a
result, the immediate metatheory of GATs is roughly as tedious as we can expect
from raw syntaxes.

This is a motivation for formally getting away from GATs, by showing their
equivalence to contextual categories. Contextual categories are algebraic and
more convenient to handle than GATs. In \cite{cartmellthesis} one leg of this
equivalence is the construction of a classifying contextual category for each
GAT, which is essentially a term model construction from quotiented raw
syntax. A downside of this setup is that classifying contextual categories
cannot be easily written out by hand like GATs. Thus, GATs necessarily remain
the practical way for specifying the classifying categories.

In contrast, the theory of FQII signatures is itself algebraic, possesses a nice
model theory (as an infinitary QII theory), and it is only mildly more complex
than the theory of contextual categories. Since the immediate theory of
signatures is already quite nice, we do not feel as much pressure to look for
nicer presentations.

Nevertheless, compact alternative presentations would be still interesting to
research.
\begin{itemize}
  \item We could look for for an analogue of the
    GAT-contextual-category correspondence for our signatures. This would
    send each signature to its classifying category.
  \item We could also look for an analogue of the Gabriel-Ulmer duality
    \cite{gabriel2006lokal}. This would send each signature to its category of models in
    $\bs{\Set}$. In the other direction, we would need a way to restrict categories
    to those which are categories of algebras.
\end{itemize}

\subsubsection{Essentially algebraic theories}

Essentially algebraic theories (EATs) \cite{freyd1972aspects} are categories
with finite limits. This is a more semantic notion of an algebraic signature,
much like how contextual categories are a more semantic presentation of
``syntactic'' GATs. For EATs $\Gamma$ and $\Delta$, the $\Gamma$-algebras
internal to $\Delta$ are simply the finite limit preserving functors from
$\Gamma$ to $\Delta$, while algebra morphisms are natural transformations.

We have more syntactic notions of essentially algebraic signatures as well. For
example, the signatures of Adámek and Rosicky
\cite[Section~3.D]{adamek1994locally} or the Partial Horn theories of Palmgren
and Vickers \cite{partialhorn} are such. These signatures are also specified
using raw syntax, but they are significantly easier to formalize than GATs, as
the syntax of signatures admits fewer dependencies. However, the lack of
dependency also causes a significant encoding overhead on comparison to GATs or
FQII signatures. For a classic example, the theory of transitive directed graphs
is given with an FQII signature as
\begin{alignat*}{3}
  & \ms{V} &&: \U\\
  & \ms{E} &&: \ms{V} \to \ms{V} \to \U\\
  & \blank\circ\blank &&: (i\,j\,k : \ms{V}) \to \ms{E}\,i\,j \to \ms{E}\,j\,k \to \El\,(\ms{E}\,i\,k)
\end{alignat*}
The same in a pseudo-EA notation could be:
\begingroup
\allowdisplaybreaks
\begin{alignat*}{3}
  & \ms{V}            &&: \Set\\
  & \ms{E}            &&: \Set\\
  & \ms{src}          &&: \ms{E} \to \ms{V}\\
  & \ms{tgt}          &&: \ms{E} \to \ms{V}\\
  & \blank\circ\blank &&: (f\,g : \ms{E}) \to \ms{tgt}\,f = \ms{src}\,g \to (h : \ms{E}) \times (\ms{src}\,h = \ms{src}\,f) \times (\ms{tgt}\,h = \ms{tgt}\,g)
\end{alignat*}
\endgroup
In short, the FQII notation is ``indexed'', while the EA is ``fibered''. Also
recall Theorem \ref{thm:ty-slice}. While this example is not wildly different in
the two cases, if we move to more complex theories, such as type theories, the
encoding overhead of EA signatures is much greater. In informal mathematics,
this is still not an issue, but in mechanized mathematics, it is. Type
dependencies are a formal complication, but in proof assistants they enable more
compact definitions. They also often force indices to be particular values, which
enables inference and unification to fill in more details in surface syntaxes.

Sketches (see e.g.\ \cite[Section~4]{barr1985toposes}) are another way to
specify EATs. They lie somewhere between the syntactic/logical styles of
specification, and just taking EATs to be finitely complete categories. They
support an elegant metatheory, but they involve an encoding overhead which is
likely unworkable in mechanized settings.

All in all, there is a rich literature on EATs, sketches and related topics, and
it would be interesting to try to connect our signatures to any of these, or try
to reproduce the numerous related results in categorical universal algebra. This
remains future work for now.

\subsubsection{Prior work on (quotient) inductive types}

The current work grew out of a line of research in the field of type
theory. This involved working out more and more expressive classes of inductive
types.

Martin-Löf's W-types \cite{martinlof84sambin} are an early example for a scheme
for inductive types. In fact, it is better viewed as a single parameterized
inductive type, which allows construction of a remarkable range of inductive
types \cite{whynotw}, although with some encoding overheads.

Inductive families \cite{inductivefamilies} allow indexing the inductive sort
with external types. This directly supports only single-sorted signatures, but
some form of mutual induction can be easily modeled through the
indexing. Inductive families have become a core feature in all major
implementations of type theories, such as Coq, Idris, Lean or Agda.

Inductive-recursive types \cite{dybjer99finite} allow mutual definition of an
inductive sort and a function which acts on the sort. These types are absent
from this thesis, they are not representable with any of our theories of
signatures. Induction-recursion is notable for tremendously boosting the
proof-theoretic strength of a type theory; a primary motivation for it was to
explore the limits of predicative constructive mathematics. It is useful for
modeling a wide variety of universe features internally to a type theory
\cite{kovacs2021generalized}.

Induction-induction was described in \cite{altenkirch11ii} and in
\cite{forsberg-phd}. This notion allowed two inductive sorts, where the second
one may be indexed over the first. As we mentioned previously, this notion is
more restricted than what was covered in this chapter.

\cite{altenkirch18qiit} investigated QIITs. The notion of signature
here is more of a semantic nature than ours. Signatures are defined
simultaneously with their categories of algebras. A signature is a inductive
list of functors: at each signature entry, we extend the category of algebras
with a functor whose domain is the current category of algebras. This can be
viewed as a generalization of F-algebras as a form of specification. However,
there is no strict positivity restriction in signatures, hence no attempt at
constructing initial algebras either.

We will look at work related to infinitary QITs in Section
\ref{sec:iqii-related-work} and at work related to higher inductive types in
Section \ref{sec:hii-related-work}.

\chapter[Infinitary QII Signatures]{Infinitary Quotient Inductive-Inductive Signatures}
\label{chap:iqiit}

In this chapter we present another theory of signatures, for \emph{infinitary
quotient inductive-inductive} signatures. As we will see, the reason for
considering the finitary and infinitary cases separately is that they support
different semantics.

First, we specify signatures and define semantics in 2LTT. Then, like in the
previous chapter, we switch to a extensional TT setting and look at term
algebras and related constructions.

\section{Theory of Signatures}

\textbf{Metatheory.} We work in 2LTT. We assume the following type formers in
the inner theory: $\top$, $\Sigma$, extensional identity $\blank\!=\!\blank$ and
$\Pi$. Note that $\Pi$ is an extra assumption compared to what we had in the
finitary case.

\begin{mydefinition}
\label{def:iqiit-tos}
A \textbf{model of the theory of signatures} consists of the following.
  \begin{itemize}
    \item A \textbf{cwf} with underlying sets $\Con$, $\Sub$, $\Ty$ and $\Tm$, all returning in
      the outer $\Set$ universe of 2LTT.
    \item A \textbf{Tarski-style universe} $\U$ with decoding $\El$. $\U$ is closed under the following type formers:
      \begin{itemize}
        \item The \textbf{unit type} $\top$.
        \item \textbf{$\Sigma$-types} $\Sigma : (a : \Tm\,\Gamma\,\U) \to \Tm\,(\Gamma\,\ext\,\El\,a)\,\U \to \Tm\,\Gamma\,\U$, with specifying isomorphism
          \[(\proj,\,\blank\!,\!\blank) : \Tm\,\Gamma\,(\El\,(\Sigma\,a\,b))\simeq (t : \Tm\,\Gamma\,(\El\,a)) \times \Tm\,\Gamma\,(\El\,(b[\id,\,t]))\]
        \item \textbf{Extensional identity} $\Id : \Tm\,\Gamma\,(\El\,a) \to \Tm\,\Gamma\,(\El\,a) \to \Tm\,\Gamma\,\U$,
          specified by $(\reflect,\,\refl) : \Tm\,\Gamma\,(\El\,(\Id\,t\,u)) \simeq (t \equiv u)$.
        \item \textbf{Small external product type} $\Piinf : (\mi{Ix} : \Ty_0) \to (\mi{Ix} \to \Tm\,\Gamma\,\U) \to \Tm\,\Gamma\,\U$, specified by $(\appinf,\,\laminf) : \Tm\,\Gamma\,(\Piinf\,\mi{Ix}\,b) \simeq ((i : \mi{Ix}) \to \Tm\,\Gamma\,(\El\,(b\,i)))$.
      \end{itemize}
    \item \textbf{Internal product type} $\Pi : (a : \Tm\,\Gamma\,\U) \to
      \Ty\,(\Gamma\ext\El\,a) \to \Ty\,\Gamma$, specified by
      $(\app,\,\lam) : \Tm\,\Gamma\,(\Pi\,a\,B) \simeq \Tm\,(\Gamma \ext \El\,a)\,B$.
    \item \textbf{External product type} $\Pie : (\mi{Ix} : \Ty_0) \to (\mi{Ix} \to \Ty\,\Gamma) \to \Ty\,\Gamma$, specified by
      $(\appe,\,\lame) : \Tm\,\Gamma\,(\Pie\,\mi{Ix}\,B) \simeq ((i : \mi{Ix}) \to \Tm\,\Gamma\,(B\,i))$.
  \end{itemize}
\end{mydefinition}
Once again we assume that an initial model for ToS exists, and a signature is a
context in the initial model.

\begin{notation}
  We employ the same notations for signatures as in Section \ref{sec:fqiit-tos}. In addition to that,
  we have the usual internal notation for $\top$ and $\Sigma$, and we write $(x : A) \toinf B$ for $\Piinf$
  and $\lambdainf$ for $\laminf$.
\end{notation}

Let us do a comparison to the finitary case. First, the new signatures do not
support sort equations, since there is no identity type for arbitrary terms,
only for terms with types in $\U$. Second, the universe is not empty anymore, it
supports $\top$, $\Sigma$ and the small external product type $\Piinf$, which
can be viewed as an analogue of $\Pie$ inside $\U$. We look at example
signatures.

\begin{myexample}
Infinitary constructors can be given with $\Piinf$. A classic example is
W-types. Assuming $S : \Ty_0$ and $P : S \to \Ty_0$, we have the following
signature for $P$-branching well-founded trees:
\begin{alignat*}{3}
  &W &&: \U\\
  &\ms{sup} &&: (s : S) \toe (P\,s \toinf W) \to \El\,W
\end{alignat*}
Note that since $P\,s \toinf W$ is in $\U$, it can appear on the left side of
$\to$. If $P\,s$ is an infinite type, $\ms{sup}$ branches with an infinite
number of inductive subtrees. Of course, finitary branching can be also
expressed with $\Piinf$, but that use case was already possible with finitary
signatures, by iterating $\to$ finite times.
\end{myexample}

\begin{myexample} Equations can appear as assumptions now. The simplest
example is set-truncation for some $A : \Ty_0$:
\begin{alignat*}{3}
  &|A|_0      &&: \U\\
  &\ms{embed} &&: A \toe \El\,|A|_0 \\
  &\ms{trunc} &&: (x\,y : |A|_0)(p\,q : \Id\,x\,y) \to \El\,(\Id\,p\,q)
\end{alignat*}
However, this ends up being redundant in our semantics, since we assume UIP, and
every semantic underlying type will be a set. Does this mean that recursive
equations are useless? We do not think so. In the specification of cubical type
theories, there are \emph{boundary conditions} which can be given as $\Id$
assumptions \cite{cchm,angiuli2016computational,angiuli2018cartesian}. Also, it
seems that these conditions cannot be easily contracted away. For an example of
contraction, the signature
\[\emptycon \ext (A : \U) \ext (c_1 : \El\,A) \ext (c_2 : (x : A) \to \Id\,x\,c_1 \to \El\,A)\]
can be rewritten to the equivalent
\[\emptycon \ext (A : \U) \ext (c_1 : \El\,A) \ext (c_2 : \El\,A)\]
signature. However, we cannot mechanically eliminate the $\Id$ from the following signature.
\begin{alignat*}{3}
  &\ms{A  } &&: \U\\
  &\ms{B  } &&: \ms{A} \to \U\\
  &\ms{b_1} &&: \ms{A} \to \El\,\ms{B}\\
  &\ms{b_2} &&: \ms{A} \to \El\,\ms{B}\\
  &...&&\\
  &\ms{a}   &&: (x\,y : \ms{A}) \to \Id\,(\ms{b_1}\,x)\,(\ms{b_2}\,y) \to \El\,\ms{A}
\end{alignat*}
Whether we can reformulate $a$ without the $\Id$ condition depends on what kind
of equational theory we specify for $B$ in the omitted parts of the signature.

However, recursive equations can be always encoded by internalizing extensional
equality in signatures. For example:
\begin{alignat*}{3}
  &\ms{A}   &&: \U \\
  &\ms{EqA} &&: \ms{A} \to \ms{A} \to \U\\
  &\refl    &&: \El\,(\ms{EqA}\,a\,a)\\
  &\reflect &&: \ms{EqA}\,a_0\,a_1 \to \El\,(\Id\,a_0\,a_1)\\
  &\UIP     &&: (p\,q : \ms{EqA}\,a_0\,a_1) \to \El\,(\Id\,p\,q)
\end{alignat*}
Still, we keep recursive equalities around, since they are more ergonomic than
the above encoding, and they pose no extra difficulty in the semantics. The
current formulation of the $\Id$ type will be more useful in Chapter \ref{chap:hiit},
where higher equalities can be proof-relevant.

\end{myexample}

\begin{myexample}
All theories of signatures that we discussed so far, have (infinitary)
signatures.

For finitary signatures, the ToS is itself infinitary because of $\Pie$. We
assume an universe $\U_0$ in $\Ty_0$. In the signature, we have
\begin{alignat*}{3}
  &\ms{Con} &&: \U \\
  &\ms{Ty} &&: \ms{Con} \to \U\\
  &\Pie &&: \{\Gamma : \ms{Con}\} \to (A : \U_0) \toe (A \toinf \ms{Ty}\,\Gamma) \to \El\,(\ms{Ty}\,\Gamma)
\end{alignat*}
In the signature for infinitary ToS, we have
\begin{alignat*}{3}
  &\ms{Univ} &&: \{\Gamma : \ms{Con}\} \to \ms{Ty}\,\Gamma\\
  &\Piinf    &&: \{\Gamma : \ms{Con}\} \to (A : \U_0) \toe (A \toinf \ms{Tm}\,\Gamma\,\ms{Univ}) \to \ms{Tm}\,\Gamma\,\ms{Univ}
\end{alignat*}
\emph{Remark.}
When we will take the semantics of the above signature, we will not exactly get
back the theory of signatures that we are using right now. We have ToS in 2LTT
now, but the semantics is in the inner theory. What we can do though, is to
assume that the inner theory is also a 2LTT. Then we might assume that the inner
theory of \emph{that} is again a 2LTT, and so on. This is a possible (and quite
natural) generalization of 2LTT to n-level type theory. In this setting, one
round of self-description requires a bumping of levels in the sense of n-level TT. In
this thesis we do not explore this, instead we use a more conventional universe
hierarchy in an extensional TT, to investigate self-description.
\end{myexample}

\begin{myexample}
We have seen in Example \ref{ex:presheaf-sig} that $\Ty_0$-valued presheaves have finitary signatures.
With infinitary signatures, we can also cover monads on $\Ty_0$. We assume a
universe $\U_0 : \Ty_0$.
\begingroup
\allowdisplaybreaks
\begin{alignat*}{3}
  &\ms{M}         &&: \U_0 \to \U\\
  &\ms{map}       &&: (A \to B) \toe \ms{M}\,A \to \El\,(\ms{M}\,B)\\
  &\ms{map_{id}}   &&: \El\,(\Id\,(\ms{map}\,\id\,m)\,m)\\
  &\ms{map_{\circ}} &&: \El\,(\Id\,(\ms{map}\,(f \circ g)\,m)\,(\ms{map}\,f\,(\ms{map}\,g\,m)))\\
  &\ms{return}    &&: A \toe \El\,(\ms{M}\,A)\\
  &\ms{bind}      &&: M\,A \to (A \toinf \ms{M}\,B) \to \El\,(\ms{M}\,B)\\
  &\ms{return_r}  &&: \El\,(\Id\,(\ms{bind}\,m\,\ms{return})\,m)\\
  &\ms{return_l}  &&: \El\,(\Id\,(\ms{bind}\,(\ms{return}\,a)\,f)\,(f\,a))\\
  &\ms{assoc}     &&: \El\,(\Id\,(\ms{bind}\,(\ms{bind}\,m\,f)\,g)\,(\ms{bind}\,m\,(\lambda\,a.\,\ms{bind}\,(f\,a)\,g)))
\end{alignat*}
\endgroup

We rely on $\toinf$ to specify binding. The $\mi{join}$-based specification
would not work, since $\ms{M}\,(\ms{M}\,A)$ is not valid in signatures. The
above signature can be helpful for deriving some of the metatheory of Dijkstra
monads \cite[Section~5]{dijkstramonad}.

In the 2LTT-based semantics, we will get $\ms{M} : \U_0 \to \Ty_0$, which is not
quite an endofunctor. In the ETT-based semantics in Section
\ref{sec:inf-term-algebras} we will be able to pick universe levels more
precisely, so we can specify algebras where $\ms{M} : \Set_i \to \Set_i$.
However, we will not get free monads from the term algebra construction, because
the universe levels do not match up as needed. Recall from Section
\ref{sec:fqii-term-algebra-construction} that the level of sets of terms is $j +
1$ when $j$ is the level of external indices in a signatures. Hence, if the
parameter types to $\ms{M}$ are in $\Set_j$, then external indices are in
$\Set_{j+1}$, so we get $\ms{M} : \Set_j \to \Set_{j+2}$ in the term algebra for
monads.

\end{myexample}

\begin{myexample}
It is worth to note that every set-truncated higher inductive type from the
Homotopy Type Theory book \cite{hottbook} is covered. This includes
\begin{itemize}
\item The cumulative hierarchy of sets \cite[Section~10.5]{hottbook}.
\item Cauchy real numbers \cite[Section~11.3]{hottbook}.
\item Surreal numbers \cite[Section~11.6]{hottbook}.
\end{itemize}
\end{myexample}

\section{Semantics}

\subsection{Overview}

As we mentioned, we need a different semantics for infinitary signatures.
First, we look at why the previous semantics fails. We try to model signatures
again as flcwfs, and morphisms as strict flcwf-morphisms. The simplest point of
failure is the interpretation of the unit type $\bs{\top : \Tm\,\Gamma\,\U}$.

In the semantics, this is the same as defining $\bs{\top :
  \Sub\,\Gamma\,\Ty_0}$, where $\bs{\Ty_0}$ is the flcwf of inner types. The
only sensible definition here is the functor which is constantly $\top_0$. But
this does not strictly preserve context comprehension or the finite limit type formers.
If we have
\begin{alignat*}{3}
  &\bs{\top} : \Con_{\bGamma} \to \Ty_0\\
  &\bs{\top}\,\Gamma \defn \top_0
\end{alignat*}
then we have $\bs{\top}\,(\Gamma \ext_{\bGamma} A) \equiv \top_0$, but
$\bs{\top}\,\Gamma \ext_{\bs{\Ty_0}} \bs{\top}\,A \equiv \top_0 \times \top_0$.
Thus, $\top_0 \not\equiv \top_0 \times \top_0$, but of course $\top_0 \simeq
\top_0 \times \top_0$.

Let us look at $\bs{\Piinf} : (A : \bs{\Ty_0}) \to (A \to \bs{\Tm\,\Gamma\,\U}) \to
\bs{\Tm\,\Gamma\,\U}$ as well, since that is a more interesting new feature than the
unit type. The only viable definition is to take the $A$-indexed product of
$\bs{\Sub\,\Gamma\,\Ty_0}$ morphisms, so we map objects of $\bGamma$ to function types:
\begin{alignat*}{3}
  &\Con_{\bs{\Piinf\,A\,b}}\,\Gamma \defn (\alpha : A) \to \Con_{\bb\,\alpha}\,\Gamma
\end{alignat*}
But now we have
\begin{alignat*}{3}
  &(\bs{\Piinf\,A\,b})\,\emptycon_{\bGamma}
     \equiv (\alpha : A) \to \Con_{\bb\,\alpha}\,\emptycon_{\bGamma}
     \equiv A \to \top_0
\end{alignat*}
Also, $\bemptycon_{\bU} \equiv \top_0$. Hence,
$(\bs{\Piinf\,A\,b})\,\emptycon_{\bGamma} \not\equiv \top_0$, although
$(\bs{\Piinf\,A\,b})\,\emptycon_{\bGamma} \simeq \top_0$.

Intuitively, if $\bU$ has no type formers, the terms in $\bU$ are neutral,
i.e.\ variables applied to zero or more neutral terms. But variables in the
semantics simply project out components from iterated $\Sigma$-types. For
example, the action of $\bs{\q} : \bs{\Tm\,(\Gamma \ext A)\,(A[\p])}$ on
objects, types, morphisms and terms is given by taking second
projections. Since all structure in $\bs{\Gamma\,\ext\,A}$ is given by pairing
things, $\bs{\q}$ strictly preserves all structure, and the same goes
for all variables.

Substitutions and terms in the finitary ToS are only allowed to freely reshuffle
structure. We can forget, duplicate, or permute signature entries, or build
neutral expressions from assumptions. In contrast, the infinitary ToS allows us
to take small limits of assumptions, using $\bs{\top}$, $\bs{\Sigma}$, $\bs{\Id}$ and $\bs{\Piinf}$
to build new inhabitants of $\bU$. We summarize the process of getting the new
semantics:
\begin{enumerate}
\item Strict structure-preservation for type formers in $\bU$ generally fails, but they still
      preserve structure up to isomorphism.
\item Hence, we switch from strict flcwf-morphisms to weak ones, which preserve $\emptycon$, comprehension
      and fl-structure weakly.
\item However, in the finitary case we often relied on transporting along
      preservation equations. We need to recover transports along isomorphism.
\item Hence, we extend semantic types from displayed flcwfs to \emph{isofibrations}, which
      support the required transports.
\item However, this rules out sort equations because they are not stable under
      isomorphisms. For example, for sets $A$, $B$, $C$ such that $A \simeq B$ and
      $A \simeq C$, it is not necessarily the case that $B \equiv C$.
\end{enumerate}

\subsubsection{Univalent semantics}

The isofibrant semantics will turn out to be significantly more technical than
the strict semantics. Instead of working with isofibrations in an extensional
setting, could we work with univalent structures in homotopy type theory? In
other words, work with univalent categories of algebras, and univalent displayed
categories over them \cite{displayedcats}. A major benefit of the univalent
setting is that we would get a \emph{structure identity principle}
\cite{aczel2011voevodsky} out of the semantics, which says that for algebras,
isomorphism is the same as equality.

However, it appears that univalent \emph{cwfs} are overall yet more technical to
handle than isofibrations. In an univalent cwf, objects and types are generally
h-groupoids, so we would have groupoids of algebras instead of sets of
algebras. This implies that \emph{type equalities} are between groupoids, so
they need to be coherent, if we want them to be well-behaved. Hence, $\Ty$ is
not an 1-presheaf over contexts, but rather a $(2,1)$-presheaf.

Alternatively, we could simplify the task by only constructing univalent
categories of algebras, and skipping the family structure (and
fl-structure). This would be the minimum amount of effort that would yield the
structure identity principle.

Both of these would be interesting to check in future work. As a third
alternative, instead of stopping at set-truncated algebras in HoTT, we might as
well consider types at arbitrary h-levels, and construct
$(\omega,\,1)$-categories of algebras. This comprises a semantics of higher
inductive-inductive signatures. We do not present a full higher-categorical
semantics in this thesis; we only present a fragment of it in Chapter
\ref{chap:hiit}.

\subsection{Model of the Theory of Signatures}

In the following we present a model of ToS. We call it $\bM$, and like before,
we use \textbf{bold} font to refer to components of $\bM$.

\subsubsection{Contexts}

$\bGamma : \bCon$ is again an flcwf, but with a minor change: $\K$ is not strict
anymore, so we have $(\appK,\,\lamK) : \Tm\,\Gamma\,(\K\,\Delta) \simeq
\Sub\,\Gamma\,\Delta$. As we will see shortly, $\bs{A[\sigma]}$ does not
support strict displayed $\K$ anymore, hence the change.

\subsubsection{Substitutions}
\label{sec:iqiit-substitutions}

$\bs{\sigma : \Sub\,\Gamma\,\Delta}$ is a \emph{weak flcwf-morphism}, that is, a
functor between underlying categories, which also maps types to types and terms
to terms, and satisfies the following mere properties:
  \begin{enumerate}
    \item $\bsigma\,(A[\sigma]) \equiv (\bsigma\,A)\,[\bsigma\,\sigma]$
    \item $\bsigma\,(t[\sigma]) \equiv (\bsigma\,t)\,[\bsigma\,\sigma]$
    \item The unique map $\epsilon : \Sub\,(\bsigma\,\emptycon)\,\emptycon$ has a retraction.
    \item Each $(\bsigma\,\p,\,\bsigma\,\q) : \Sub\,(\bsigma\,(\Gamma\,\ext\,A))\,(\bsigma\,\Gamma\,\ext\,\bsigma\,A)$ has an inverse.
  \end{enumerate}

In short, $\bsigma$ preserves substitution strictly and preserves empty context
and context extension up to isomorphism. We notate the evident isomorphisms as
$\bsigma_{\emptycon} : \bsigma\,\emptycon \simeq \emptycon$ and $\bsigma_{\ext}
: \bsigma\,(\Gamma\,\ext\,A)\,\simeq\,\bsigma\,\Gamma\,\ext\,\bsigma\,A$. Our
notion of weak morphism is the same as in \cite{dependentrightadjoints}, when
restricted to cwfs.

\begin{theorem}\label{thm:flpres}
Every $\bs{\sigma : \Sub\,\Gamma\,\Delta}$ preserves fl-structure up to
type isomorphism. That is, we have
\begin{alignat*}{3}
  & \bsigma_{\Sigma} : \bsigma\,(\Sigma\,A\,B) \simeq \Sigma\,(\bsigma\,A)\,((\bsigma\,B)[\bsigma_{\ext}^{-1}]) \\
  & \bsigma_{\Id} : \bsigma\,(\Id\,t\,u) \simeq \Id\,(\bsigma\,t)\,(\bsigma\,u) \\
  & \bsigma_{\K} : \bsigma\,(\K\,\Delta) \simeq \K\,(\bsigma\,\Delta)
\end{alignat*}
These are all natural in the following sense: for $\sigma :
\Sub_{\bGamma}\,\Gamma\,\Delta$, if we have $\bsigma_{\Sigma}$ as a type
isomorphism in $\bsigma\,\Delta$, if we reindex it by $\sigma$, we get
$\bsigma_{\Sigma}$ as a type isomorphism in $\bsigma\,\Gamma$. The same holds
for $\bsigma_{\Id}$ and $\bsigma_{\K}$.

Moreover, $\bsigma$ preserves all term and substitution formers in the
fl-structure. For example, $\bsigma\,(\proj1\,t) \equiv \proj1\,
(\bsigma_{\Sigma}[\id, \bsigma\,t])$.
\end{theorem}
\begin{proof}
For $\bsigma_{\Sigma}$, we construct the following context isomorphism:
\begin{alignat*}{3}
& (\bsigma\,\Gamma\,\ext\,\bsigma\,(\Sigma\,A\,B)) \simeq
  (\bsigma\,\Gamma\,\ext\,\bsigma\,A\,\ext\,(\bsigma\,B)[\bsigma_{\ext}^{-1}]) \\
& \simeq (\bsigma\,\Gamma\,\ext\,\Sigma\,(\bsigma\,A)\,((\bsigma\,B)[\bsigma_{\ext}^{-1}]))
\end{alignat*}
This isomorphism is the identity on $\bsigma\,\Gamma$, hence we can extract the
desired $\bsigma_{\Sigma} : \bsigma\,(\Sigma\,A\,B) \simeq
\Sigma\,(\bsigma\,A)\,((\bsigma\,B)[\bsigma_{\ext}^{-1}])$ from it.

For $\bsigma_{\Id}$, both component morphisms can be constructed by $\refl$ and
equality reflection, and the morphisms are inverses by UIP. We omit here the
verification of naturality and that $\bsigma$ preserves term and substitution
formers in the fl-structure.

For $\bsigma_{\K}$, note the following:
\begin{alignat*}{3}
  & (\emptycon\,\ext\,\bsigma\,(\K\,\Delta)) \simeq
    (\bsigma\,\emptycon\,\ext\,\bsigma\,(\K\,\Delta)) \simeq
    \bsigma\,(\emptycon\,\ext\,\K\,\Delta)\\
  & \simeq \bsigma\,\Delta \simeq (\emptycon\,\ext\,K\,(\bsigma\,\Delta))
\end{alignat*}
This yields a type isomorphism $\bsigma\,(\K\,\Delta) \simeq
\K\,(\bsigma\,\Delta)$ in the empty context, and we can use the functorial action of
$\epsilon : \Sub\,\Gamma\,\emptycon$ to weaken it to any $\Gamma$ context.
\end{proof}

\subsubsection{Identity and composition}
\label{sec:idcomp}

$\bid : \bSub\,\bGamma\,\bGamma$ is defined in the obvious way, with identities for
underlying functions and for preservation morphisms.

For $\bs{\sigma \circ \delta}$, the underlying functions are given by
function composition, and the preservation morphisms are given as follows:
\begin{alignat*}{3}
  & (\bs{\sigma \circ \delta})_{\emptycon}^{-1} \defn
    \bsigma\,\bdelta_{\emptycon}^{-1} \circ \bdelta_{\emptycon}^{-1} \\
  & (\bs{\sigma \circ \delta})_{\ext}^{-1} \defn
    \bsigma\,\bdelta_{\ext}^{-1} \circ \bdelta_{\ext}^{-1}
\end{alignat*}

It is easy to verify the left and right identity laws and associativity for
$\bs{\blank\circ\blank}$.

\begin{mylemma}\label{lem:idcomppres}
The derived preservation isomorphisms for the fl-structure can be decomposed
analogously; all derived isomorphisms in $\bid$ are identities, and we have
\begin{alignat*}{3}
  & (\bs{\sigma \circ \delta})_{\Sigma} \equiv
  \bsigma\,\bdelta_{\Sigma} \circ \bdelta_{\Sigma}\\
  & (\bs{\sigma \circ \delta})_{\Id} \equiv
  \bsigma\,\bdelta_{\Id} \circ \bdelta_{\Id}\\
  & (\bs{\sigma \circ \delta})_{\K} \equiv
  \bsigma\,\bdelta_{\K} \circ \bdelta_{\K}
\end{alignat*}
On the right sides, $\blank\circ\blank$ refers to composition of type morphisms.
\end{mylemma}
\begin{proof}
In the case of $\Id$, the equations hold immediately by UIP. For $\Sigma$ and
$\K$, we prove by flcwf computation and straightforward unfolding of
definitions.
\end{proof}

\subsubsection{Empty context}
The empty context $\bemptycon : \bCon$ is the same as before, i.e.\ the terminal
flcwf. Since the unique $\bs{\epsilon} : \bSub\,\bGamma\,\bemptycon$ morphism
strictly preserves all structure, it also a weak morphism.

\subsubsection{Types}
We define $\bs{\Ty\,\Gamma} : \Set$ as the type of split flcwf-isofibrations
over $\bGamma$. This consists of a displayed flcwf together with \emph{split iso-cleaving}
structure. For the displayed flcwf part, we reuse previous notation from Section
\ref{sec:fqiit-family}. For the iso-cleaving, we make some auxiliary definitions first.

\begin{mydefinition}[Displayed type categories]
For each $\Gamma : \Con_{\bA}\,\ulGamma$, there is a displayed category over the
type category $\Ty_{\bGamma}\,\ulGamma$, whose objects over $\ulA :
\Ty_{\bGamma}\,\ulGamma$ are elements of $\Ty_{\bA}\,\Gamma\,\ulA$, and
displayed morphisms over $\ult : \Tm_{\bGamma}\,(\ulGamma \ext
\ulA)\,(\ulB[\p])$ are elements of $\Tm_{\bA}\,(\Gamma \ext
A)\,(B[\p])\,\ult$. The identity morphism is given by $\q_{\bA}$, and the
composition of $t$ and $u$ is $t[\p_{\bA},u]$. Analogously to Definition
\ref{def:type_categories}, this extends to a displayed split indexed category.
\end{mydefinition}

\begin{mydefinition}[Displayed isomorphisms]
\label{def:displayed-iso}
A \emph{displayed context isomorphism} over $\ulsigma : \ulGamma \simeq
\ulDelta$, notated $\sigma : \Gamma \simeq_{\ulsigma} \Delta$, is an invertible
displayed morphism $\sigma : \Sub_{\bA}\,\Gamma\,\Delta\,\ulsigma$, with inverse
$\sigma^{-1} : \Sub_{\bA}\,\Delta\,\Gamma\,\ulsigma^{-1}$. A \emph{displayed
  type isomorphism} over $\ult : \ulA \simeq \ulB$, notated $t : A \simeq_{\ult}
B$, is an isomorphism in a displayed type category.
\end{mydefinition}

\begin{mydefinition}
A \emph{vertical morphism} lies over an identity morphism. We use this
definition for context morphisms (substitutions) and type morphisms as well.
\end{mydefinition}

\begin{mydefinition}[Split iso-cleaving for contexts] This lifts a base context isomorphism to a displayed one. It consists of
\begingroup
\allowdisplaybreaks
\begin{alignat*}{3}
  & \coe &&: \ulGamma \simeq \ulDelta \ra \Con_{\bA}\,\ulGamma \ra \Con_{\bA}\,\ulDelta\\
  & \coh &&: (\ulsigma : \ulGamma \simeq \ulDelta)(\Gamma : \Con_{\bA}\,\ulGamma)
           \ra \Gamma \simeq_{\ulsigma} \coe\,\ulsigma\,\Gamma\\
  & \coe^{\id} && : \coe\,\id\,\Gamma \equiv \Gamma\\
  & \coe^{\circ} && : \coe\,(\ulsigma\circ\uldelta)\,\Gamma \equiv \coe\,\ulsigma\,(\coe\,\uldelta\,\Gamma)\\
  & \coh^{\id} && : \coh\,\id\,\Gamma \equiv \id\\
  & \coh^{\circ} && : \coh\,(\ulsigma\circ\uldelta)\,\Gamma \equiv \coh\,\ulsigma\,(\coe\,\uldelta\,\Gamma)
          \circ \coh\,\uldelta\,\Gamma
\end{alignat*}
\endgroup
Here, $\coe$ and $\coh$ abbreviate ``coercion'' and ``coherence'' respectively.
\end{mydefinition}

\begin{mydefinition}[Split iso-cleaving for types] This consists of
\begingroup
\allowdisplaybreaks
\begin{alignat*}{3}
  & \coe &&: \ulA \simeq \ulB \ra \Ty_{\bA}\,\Gamma\,\ulA \ra \Ty_{\bA}\,\Gamma\,\ulB\\
  & \coh &&: (\ult : \ulA \simeq \ulB)(A : \Ty_{\bA}\,\Gamma\,\ulA)
           \ra A \simeq_{\ult} \coe\,\ult\,A\\
  & \coe^{\id} && : \coe\,\id\,A \equiv A\\
  & \coe^{\circ} && : \coe\,(\ult\circ\ulu)\,A \equiv \coe\,\ult\,(\coe\,\ulu\,A)\\
  & \coh^{\id} &&: \coh\,\id\,A \equiv \id\\
  & \coh^{\circ} &&: \coh\,(\ult\circ\ulu)\,A \equiv \coh\,\ult\,(\coe\,\ulu\,A)
          \circ \coh\,\ulu\,A
\end{alignat*}
\endgroup
Additionally, for $\sigma : \Sub_{\bA}\,\Gamma\,\Delta\,\ulsigma$, we have
\begin{alignat*}{3}
  & \coe[] &&: \coe\,(\ult[\ulsigma\circ \p,\q])\,(A[\sigma]) \equiv (\coe\,\ult\,A)[\sigma]\\
  & \coh[] &&: \coh\,(\ult[\ulsigma\circ \p,\q])\,(A[\sigma]) \equiv (\coh\,\ult\,A)[\sigma]
\end{alignat*}

\end{mydefinition}

\begin{mydefinition} A \emph{split flcwf isofibration} is a displayed flCwF equipped with split iso-cleaving for contexts and types.
\end{mydefinition}

\emph{Remark.} It is not possible to model types as fibrations or opfibrations
because we have no restriction on the variance of ToS types. For example, the
type which extends a pointed set signature to a natural number signature, is
neither a fibration nor an opfibration.

\subsubsection{Type substitution}
We aim to define $\bs{\blank[\blank] : \Ty\,\Delta \ra \Sub\,\Gamma\,\Delta \ra
  \Ty\,\Gamma}$, such that $\bs{A[\id]} \equiv \bA$ and $\bs{A[\sigma\circ\delta]} \equiv
\bs{A[\sigma][\delta]}$. As before, the underlying sets are given by simple
composition:
\begin{alignat*}{3}
  & \Con_{\bs{A[\sigma]}}\,\ulGamma && \defn \Con_{\bA}\,(\bsigma\,\ulGamma)\\
  & \Sub_{\bs{A[\sigma]}}\,\Gamma\,\Delta\,\ulsigma && \defn
    \Sub_{\bA}\,\Gamma\,\Delta\,(\bsigma\,\ulsigma)\\
  & \Ty_{\bs{A[\sigma]}}\,\Gamma\,\ulA && \defn
      \Ty_{\bA}\,\Gamma\,(\bsigma\,\ulA)\\
  & \Tm_{\bs{A[\sigma]}}\,\Gamma\,A\,\ult && \defn
      \Tm_{\bA}\,\Gamma\,A\,(\bsigma\,\ult)
\end{alignat*}
The difference from the finitary case is that instead of preservation equations,
we have isomorphisms, coercions and coherence. However, we can recover
essentially the same reasoning as before because all the previous transports
still work. Context and type formers are given by coercing $\bA$ structures
along preservation isomorphisms by $\bsigma$. For example:
\begin{alignat*}{3}
  &\emptycon_{\bs{A[\sigma]}} && \defn
    \coe\,\bsigma_{\emptycon}^{-1}\,\emptycon_{\bA}\\
  &\Gamma\ext_{\bs{A[\sigma]}}A && \defn
    \coe\,\bsigma_{\ext}^{-1}\,(\Gamma\ext_{\bA} A)\\
  &\Id_{\bs{A[\sigma]}}\,t\,u && \defn
    \coe\,\bsigma_{\Id}^{-1}\,(\Id_{\bA}\,t\,u)\\
  &\K_{\bs{A[\sigma]}}\,\Delta && \defn
    \coe\,\bsigma_{\K}^{-1}\,(\K_{\bA}\,\Delta)
\end{alignat*}
Term and substitution formers are given by composing $\coh$-lifted
isomorphisms with term and substitution formers from $\bA$. For example:
\begin{alignat*}{3}
  & \epsilon_{\bs{A[\sigma]}} && \defn
    \coh\,\bsigma_{\emptycon}^{-1}\,\emptycon_{\bA} \circ \epsilon_{\bA}\\
  & \p_{\bs{A[\sigma]}} && \defn
    \p_{\bA} \circ (\coh\,\bsigma_{\ext}^{-1}\,(\Gamma\ext A))^{-1}\\
  & (\sigma,_{\bs{A[\sigma]}}\,t) && \defn \coh\,\bsigma_{\ext}^{-1}\,(\Delta\ext A) \circ (\sigma,_{\bA}\,t)
\end{alignat*}
As we mentioned, only weak $\K$ is supported in $\bs{A[\sigma]}$. For strict $\K$
we would have to show:
\[
\Sub_{\bA}\,\Gamma\,\Delta\,(\bsigma\,\ulsigma)
\equiv \Tm_{\bA}\,\Gamma\,(\coe\,\bsigma_{\K}^{-1}\,(\K_{\bA}\,\Delta))\,(\bsigma\,\ulsigma)
\]
By strict $\K$ in $\bA$, it would be enough to show
\[
\Tm_{\bA}\,\Gamma\,(\K_{\bA}\,\Delta)\,(\bsigma\,\ulsigma)
\equiv \Tm_{\bA}\,\Gamma\,(\coe\,\bsigma_{\K}^{-1}\,(\K_{\bA}\,\Delta))\,(\bsigma\,\ulsigma)
\]
But there is no reason why these sets should be equal, so we instead produce an isomorphism.

Equations for term and type substitution follow from naturality of preservation
isomorphisms in $\bsigma$, $\coe[]$, $\coh[]$ and substitution equations in
$\bA$.

  %% & \appK_{\bs{A[\sigma]}}\,t && \defn
  %%   \appK_{\bA}\,((\coh\,\bsigma_{\K}\,(\K\,\Delta))^{-1}\circ t)

Iso-cleaving is given by iso-cleaving in $\bA$ and the action of $\bsigma$ on
isomorphisms, so that we have $\coe_{\bs{A[\sigma]}}\,\ulsigma\,\Gamma
\defn \coe_{\bA}\,(\bsigma\,\ulsigma)\,\Gamma$ and $\coh_{\bs{A[\sigma]}}\,\ulsigma\,\Gamma
\defn \coh_{\bA}\,(\bsigma\,\ulsigma)\,\Gamma$.

Functoriality of type substitution, i.e.\ $\bs{A[\id]} \equiv \bA$ and
$\bs{A[\sigma\circ\delta]} \equiv \bs{A[\sigma][\delta]}$, follows
from Lemma \ref{lem:idcomppres} and split cleaving given by $\coe^{\id}$,
$\coe^{\circ}$, $\coh^{\id}$ and $\coh^{\circ}$ laws in $\bA$.

\subsubsection{Terms}

$\bs{\Tm\,\Gamma\,A} : \Set$ is defined as the type of
\emph{weak flCwF sections} of $\bA$. The underlying functions of $\bt :
\bTm\,\bGamma\,\bA$ are as follows:
\begingroup
\allowdisplaybreaks
\begin{alignat*}{3}
  & \bt : (\ulGamma : \Con_{\bGamma}) \ra \Con_{\bA}\,\ulGamma\\
  & \bt : (\ulsigma : \Sub_{\bGamma}\,\ulGamma\,\ulDelta)
         \ra \Sub_{\bA}\,(\bt\,\ulGamma)\,(\bt\,\ulDelta)\,\ulsigma\\
  & \bt : (\ulA : \Ty_{\bGamma}) \ra \Ty_{\bA}\,(\bt\,\ulGamma)\,\ulA\\
  & \bt : (\ult : \Tm_{\bGamma}\,\ulGamma\,\ulA) \ra
          \Tm_{\bA}\,(\bt\,\ulGamma)\,(\bt\,\ulA)\,\ult
\end{alignat*}
\endgroup
Such that
\begin{enumerate}
  \item $\bt\,(\ulA[\ulsigma]) \equiv (\bt\,\ulA)\,[\bt\,\ulsigma]$
  \item $\bt\,(\ult[\ulsigma]) \equiv (\bt\,\ult)\,[\bt\,\ulsigma]$
  \item The unique $\epsilon_{\bA} : \Sub\,(\bt\,\emptycon)\,\emptycon\,\id$ has a vertical retraction.
  \item Each $(\bt\,\p,\,\bt\,\q) : \Sub\,(\bt\,(\ulGamma\,\ext\,\ulA))\,(\bt\,\ulGamma\,\ext\,\bt\,\ulA)\,\id$ has a vertical inverse.
\end{enumerate}

Similarly to what we had in $\bSub$, we denote the evident preservation
isomorphisms as $\bt_{\emptycon} : \bt\,\emptycon \simeq_{\id} \emptycon$ and
$\bt_{\ext} : \bt\,(\ulGamma\ext \ulA) \simeq_{\id} \bt\,\ulGamma \ext
\bt\,\ulA$. In short, weak sections are dependently typed analogues of weak
morphisms, with dependent underlying functions and displayed preservation
isomorphisms. We also have the derived fl-preservation isomorphisms.

\begin{theorem} A weak section $\bs{t : \Tm\,\Gamma\,A}$ preserves fl-structure up to vertical type isomorphisms, that is, the following are derivable:
\begin{alignat*}{3}
  & \bt_{\Sigma} : \bt\,(\Sigma\,\ulA\,\ulB) \simeq_{\id} \Sigma\,(\bt\,\ulA)\,((\bt\,\ulB)[\bt_{\ext}^{-1}]) \\
  & \bt_{\Id} : \bt\,(\Id\,\ult\,\ulu) \simeq_{\id} \Id\,(\bt\,\ult)\,(\bt\,\ulu)  \\
  & \bt_{\K} : \bt\,(\K\,\ulDelta) \simeq_{\id} \K\,(\bt\,\ulDelta)
\end{alignat*}
Also, the above isomorphisms are natural in the sense of Theorem
\ref{thm:flpres}, and $\bt$ preserves term and substitution formers in the
fl-structure.
\end{theorem}
\begin{proof}
The construction of isomorphisms is the same as in Theorem
\ref{thm:flpres}. Indeed, every construction there has a displayed counterpart
which we can use here.
\end{proof}

We note though that the move from Theorem \ref{thm:flpres} to here is not simply
a logical predicate translation because we are only lifting the codomain of a
weak morphism to a displayed version, and we leave the domain non-displayed. We
leave to future work the investigation of such asymmetrical logical predicate
translations.

\subsubsection{Term substitution}

$\bs{\blank[\blank] : \Tm\,\Delta\,A \ra (\sigma : \Sub\,\Gamma\,\Delta)
  \ra \Tm\,\Gamma\,(A[\sigma])}$ is given similarly to
$\bs{\blank\!\circ\!\blank}$ in Section \ref{sec:idcomp}. Underlying functions
are given by function composition, and preservation morphisms are also similar:
\begin{alignat*}{3}
  & (\bs{t[\sigma]})_{\emptycon}^{-1} \defn
    \bt\,\bsigma_{\emptycon}^{-1} \circ \bt_{\emptycon}^{-1} \\
  & (\bs{t[\sigma]})_{\ext}^{-1} \defn
    \bt\,\bsigma_{\ext}^{-1} \circ \bt_{\ext}^{-1}
\end{alignat*}
We also have the same decomposition of derived isomorphisms as in Lemma
\ref{lem:idcomppres}. We do not have to show functoriality of term substitution
here, since that is derivable in any cwf, see e.g. \cite{kaposi2019constructing}.

\subsubsection{Comprehension}

$\bs{\Gamma\,\ext A : \Con}$ is defined as the total flcwf of $\bA$, in exactly
the same way as in the finitary case, since the additional iso-cleaving
structure plays no role in the result. $\bs{\p : \Sub\,(\Gamma\ext A)\,\Gamma}$
and $\bs{\q : \Tm\,(\Gamma\ext A)\,(A[\p])}$ are likewise unchanged; they are
strict morphisms, so also automatically weak morphisms. Substitution extension
$\bs{(\sigma,\,t)}$ is given by pointwise combining $\bsigma$ and $\bt$,
e.g.\ $\Con_{\bs{(\sigma,t)}}\,\ulGamma \defn (\bsigma\,\ulGamma,\,\bt\,\ulGamma)$.

\subsubsection{Strict constant families}
We have the same definition for $\bs{\K\,\Delta : \Ty\,\Gamma}$ as in the
finitary case, although we need to define iso-cleaving in addition. Fortunately,
coercions and coherences are all trivial because $\bK\,\bDelta$ does not actually
depend on $\bGamma$.
\begin{alignat*}{3}
  &\coe_{\bs{\K\,\Delta}}\,\ulsigma\,\Gamma &&\defn \Gamma\\
  &\coe_{\bs{\K\,\Delta}}\,\ult\,A          &&\defn A
\end{alignat*}

\subsubsection{Universe}

$\bU : \bTy\,\bGamma$ is exactly the same as before. We define it as the type
which is constantly the flcwf of inner types, so it inherits the trivial
iso-cleaving from $\bK$.

$\bs{\El\,a : \Ty\,\Gamma}$ is again the displayed flcwf of
the elements of $\bs{a : \Tm\,\Gamma\,\U}$. The underlying sets are unchanged:
\begingroup
\allowdisplaybreaks
\begin{alignat*}{3}
  & \Con_{\bEl\,\ba}\,\ulGamma && \defn \Tm_0\,(\ba\,\ulGamma)\\
  & \Sub_{\bEl\,\ba}\,\Gamma\,\Delta\,\ulsigma && \defn \ba\,\ulsigma\,\Gamma \equiv \Delta\\
  & \Ty_{\bEl\,\ba}\,\Gamma\,\ulA && \defn \Tm_0\,(\ba\,\ulA\,\Gamma)\\
  & \Tm_{\bEl\,\ba}\,\Gamma\,A\,\ult && \defn \ba\,\ult\,\Gamma \equiv A
\end{alignat*}
\endgroup
We need to adjust definitions to show that $\bs{\El\,a}$ supports all required
structure. Previously, all context and type formers were inherited from $\bU$,
since $\ba$ strictly preserved them. Now, $\ba$ preserves structure up to
(definitional) isomorphism of inner types. Hence, the adjustments are quite
mechanical; they are like wrapping all definitions in ``unary record constructors''
given by preservation isomorphisms. For example:
\begin{alignat*}{3}
  & \emptycon_{\bEl\,\ba} && \defn \ba_{\emptycon}^{-1}\,\tt\\
  & (\Gamma\ext_{\bEl\,\ba} A) && \defn \ba_{\ext}^{-1}\,(\Gamma,\,A)
\end{alignat*}
We likewise use preservation isomorphisms to define $\K$, $\Id$ and $\Sigma$.
Context coercion is $\coe\,\ulsigma\,\Gamma \defn \ba\,\ulsigma\,\Gamma$. Type
coercion, for $A : \ba\,\ulA\,\Gamma$ is given as $\coe\,\ult\,A \defn
\ba\,\ult\,(\ba_{\ext}^{-1}\,(\Gamma,\,A))$.

\subsubsection{Unit type}

$\bs{\top : \Tm\,\Gamma\,\U}$ is the constantly $\top_0$ morphism, i.e.\ it maps
objects to $\top_0$ and types to $\lambda\,\_.\,\top_0$, and maps morphisms and
terms to the identity function. It clearly preserves $\emptycon$ and $\blank\!\ext\!\blank$
up to isomorphism.

\subsubsection{Sigma type}

For $\bs{a : \Tm\,\Gamma\,\U}$ and $\bs{b : \Tm\,(\Gamma\ext \El\,a)\,\U}$, we
define $\bs{\Sigma\,a\,b : \Tm\,\Gamma\,\U}$ as the component-wise $\Sigma$ of $\ba$
and $\bb$. For the action on $\ulGamma : \Con_{\bGamma}$, we have:
\begin{alignat*}{3}
  &(\bs{\Sigma\,a\,b})\,\ulGamma && : \Ty_0 \\
  &(\bs{\Sigma\,a\,b})\,\ulGamma &&\defn (\alpha : \ba\,\ulGamma) \times \bb\,(\ulGamma,\,\alpha)
\end{alignat*}
For the action on $\ulsigma : \Sub\,\ulGamma\,\ulDelta$, we have:
\begin{alignat*}{3}
  &(\bs{\Sigma\,a\,b})\,\ulsigma && : (\alpha : \ba\,\ulGamma) \times \bb\,(\ulGamma,\,\alpha)
    \to (\alpha : \ba\,\ulDelta) \times \bb\,(\ulDelta,\,\alpha)\\
  &(\bs{\Sigma\,a\,b})\,\ulsigma &&\defn \lambda\,(\alpha,\,\beta).\,(\ba\,\ulsigma\,\alpha,\,\bb\,(\ulsigma,\,\refl)\,\beta)
\end{alignat*}
Above, the second field should have type
$\bb\,(\ulDelta,\,\ba\,\ulsigma\,\alpha)$, while $\beta :
\bb\,(\ulGamma,\,\alpha)$. Therefore we need a morphism in $\bs{\Gamma \ext
  \El\,a}$ from $(\ulGamma,\,\alpha)$ to $(\ulDelta,\,\ba\,\ulsigma\,\alpha)$,
which is defined as $(\ulsigma,\,\refl)$, where $\refl : \ba\,\ulsigma\,\alpha
\equiv \ba\,\ulsigma\,\alpha$.
The action on $\ulA : \Ty\,\ulGamma$ is
\begin{alignat*}{3}
  &(\bs{\Sigma\,a\,b})\,\ulA &&: (\alpha : \ba\,\ulGamma) \times \bb\,(\ulGamma,\,\alpha)
    \to \Ty_0\\
  &(\bs{\Sigma\,a\,b})\,\ulA &&\defn \lambda\,(\alpha,\,\beta).\,(\alpha' : \ba\,\ulA\,\alpha) \times \bb\,(\ulA,\,\alpha')\,\beta
\end{alignat*}
Here we are somewhat running out of notation: we use $\alpha'$ to refer to a
\emph{type} over $\alpha : \ba\,\ulGamma$ in the displayed cwf of elements
$\bs{\El\,a}$. The action on terms is analogous:
\begin{alignat*}{3}
  &(\bs{\Sigma\,a\,b})\,\ult && : ((\alpha,\,\beta) : (\alpha : \ba\,\ulGamma) \times \bb\,(\ulGamma,\,\alpha))
    \to (\alpha' : \ba\,\ulA\,\alpha) \times \bb\,(\ulA,\,\alpha')\,\beta\\
  &(\bs{\Sigma\,a\,b})\,\ult &&\defn \lambda\,(\alpha,\,\beta).\,(\ba\,\ult\,\alpha,\,\bb\,(\ult,\,\refl)\,\beta)
\end{alignat*}
For the preservation of $\emptycon$, we need to show
$(\bs{\Sigma\,a\,b})\,\ulemptycon \simeq \top_0$. Unfolding the definition, we
get $((\alpha : \ba\,\ulemptycon) \times \bb\,(\ulemptycon,\,\alpha)) \simeq
\top_0$. This holds since $\ba\,\ulemptycon \simeq \top_0$, so
$\ba\,\ulemptycon$ is contractible, thus $(\ulemptycon,\,\alpha) \equiv
\ulemptycon_{\bs{\Gamma \ext \El\,a}}$, and we also know $\bb\,\ulemptycon
\simeq \top_0$. For the preservation $\blank\!\ext\!\blank$, we need
\[
(\bs{\Sigma\,a\,b})\,(\ulGamma \ext \ulA) \simeq (\gamma : (\bs{\Sigma\,a\,b})\,\ulGamma) \times
  (\bs{\Sigma\,a\,b})\,\ulA\,\gamma
\]
Unfolding definitions and reassociating $\Sigma$ on the right side:
\begin{alignat*}{3}
   &(\alpha : \ba\,(\ulGamma \ext \ulA)) \times \bb\,((\ulGamma \ext \ulA),\,\alpha) \\
   &\simeq \\
   &(\alpha : \ba\,\ulGamma)
           \times (\beta : \bb\,(\ulGamma,\,\alpha))
           \times (\alpha' : \ba\,\ulA\,\alpha)
           \times \bb\,(\ulA,\,\alpha')\,\beta
\end{alignat*}
Since $\ba_{\ext} : \ba\,(\ulGamma \ext \ulA)) \simeq (\alpha : \ba\,\ulGamma)
\times (\beta : \bb\,(\ulGamma,\,\alpha))$, we can rewrite the left side using
pattern matching notation as
\[
  (\ba_{\ext}^{-1}\,(\gamma,\,\alpha) : \ba\,(\ulGamma \ext \ulA))
     \times \bb\,((\ulGamma \ext \ulA),\,(\gamma,\,\alpha))
\]
Now, since $((\ulGamma \ext \ulA),\,(\gamma,\,\alpha)) \equiv (\ulGamma,\,\gamma)
\ext_{\bs{\Gamma\ext\El\,a}} (\ulA,\,\alpha)$, we know that $\bb\,((\ulGamma
\ext \ulA),\,(\gamma,\,\alpha))$ is also isomorphic to the evident $\Sigma$
type, and the preservation isomorphism follows.

Projections and pairing for $\bs{\Sigma\,a\,b}$ are defined in the obvious way by
component-wise projection and pairing.

\subsubsection{Identity}

For $\bt$ and $\bu$ in $\bTm\,\bGamma\,(\bEl\,\ba)$, we define $\bId\,\bt\,\bu
\boldsymbol{:} \bTm\,\bGamma\,\bU$ as expressing pointwise equality of weak
sections. We rely on the assumption that $\Ty_0$ has identity type.
\begin{alignat*}{3}
& (\bId\,\bt\,\bu)\,\ulGamma &&\defn (\bt\,\ulGamma = \bu\,\ulGamma)\\
& (\bId\,\bt\,\bu)\,\ulA     &&\defn \lambda\,e.\, (\bt\,\ulA = \bu\,\ulA)
\end{alignat*}
Above, $\bt\,\ulA = \bu\,\ulA$ is well-typed because of $e :
\bt\,\ulGamma = \bu\,\ulGamma$. For substitutions, we have to complete a square
of equalities:
\begin{alignat*}{3}
  (\bId\,\bt\,\bu)\,(\ulsigma : \Sub\,\ulGamma\,\ulDelta) : (\bt\,\ulGamma = \bu\,\ulGamma) \ra
       (\bt\,\ulDelta = \bu\,\ulDelta)
\end{alignat*}
This can be given by $\bt\,\ulsigma : \ba\,\ulsigma\,(\bt\,\ulGamma) =
\bt\,\ulDelta$ and $\bu\,\ulsigma : \ba\,\ulsigma\,(\bu\,\ulGamma) =
\bu\,\ulDelta$. The action on terms is analogous.

The $\emptycon$-preservation $(\bt\,\ulemptycon = \bu\,\ulemptycon) \simeq
\top_0$ follows from $\ba\,\ulemptycon \simeq \top_0$. For $\ext$-preservation,
we need to show
\[
 (\bt\,(\ulGamma \ext \ulA) = \bu\,(\ulGamma \ext \ulA)) \simeq
 ((e : \bt\,\ulGamma = \bu\,\ulGamma) \times (\bt\,\ulA = \bu\,\ulA))
\]
This follows from $\ext$-preservation by $\ba$. Equality reflection and
$\bs{\refl :} \bId\,\bt\,\bt$ are also evident.

\subsubsection{Small external product type}

For $\mi{Ix} : \Ty_0$ and $\bb : \mi{Ix} \ra \bTm\,\bGamma\,\bU$, we aim to define
$\bPiinf\,\mi{Ix}\,\bb \bs{:} \bTm\,\bGamma\,\bU$. The underlying functions
are:
\begin{alignat*}{3}
  & (\bPiinf\,\mi{Ix}\,\bb)\,\ulGamma    &&:= (i : \mi{Ix})\ra \bb\,i\,\ulGamma\\
  & (\bPiinf\,\mi{Ix}\,\bb)\,\ulsigma    &&:= \lambda\,f\,i.\, \bb\,i\,\ulsigma\,(f\,i)\\
  & (\bPiinf\,\mi{Ix}\,\bb)\,\ulA\       &&:= \lambda\,\Gamma.\,(i : \mi{Ix})\ra \bb\,i\,\ulA\,(\Gamma\, i)\\
  & (\bPiinf\,\mi{Ix}\,\bb)\,\ult        &&:= \lambda\,f\,i.\, \bb\,i\,\ult\,(f\,i)
\end{alignat*}
We rely on $\Pi$ in the inner theory. The preservation isomorphisms are
pointwise inherited from $\bb$. One direction of the isomorphisms is defined as
follows. Note that $\emptycon_{\bU} \equiv \top$ and $\ext_{\bU}$ is $\Sigma$.
\begin{alignat*}{3}
  &(\bPiinf\,\mi{Ix}\,\bb)_{\emptycon}^{-1} && : \top\ra (\bPiinf\,\mi{Ix}\,\bb)\,\emptycon\\
  &(\bPiinf\,\mi{Ix}\,\bb)_{\emptycon}^{-1} && \defn \lambda\,\_\,i.\,(\bb\,i)_{\emptycon}^{-1}\,\tt
\end{alignat*}
\begin{alignat*}{3}
  &(\bPiinf\,\mi{Ix}\,\bb)_{\ext}^{-1} && : (\Gamma : (\bPiinf\,\mi{Ix}\,\bb)\,\ulGamma)\times((\bPiinf\,\mi{Ix}\,\bb)\,\ulA\,\Gamma)\\
  & && \hspace{0.5em}\ra (\bPiinf\,\mi{Ix}\,\bb)\,(\ulGamma \ext \ulA)\\
  & (\bPiinf\,\mi{Ix}\,\bb)_{\ext}^{-1} && \defn \lambda\,(\Gamma,A)\,i.\,(\bb\,i)_{\ext}^{-1}(\Gamma\,i,\,A\,i)
\end{alignat*}

\subsubsection{Internal product type}

For $\boldsymbol{a : \Tm\,\Gamma\,\U}$ and $\boldsymbol{B :
  \Ty\,(\Gamma\ext\El\,a)}$, we aim to define $\boldsymbol{\Pi\,a\,B}
\boldsymbol{:} \bTy\,\bGamma$. The underlying sets are unchanged.
\begin{alignat*}{3}
  & \Con_{\bs{\Pi\,a\,B}}\,\ulGamma &&\defn (\gamma : \ba\,\ulGamma) \to \Con_{\bB}\,(\ulGamma,\,\gamma)\\
  & \Sub_{\bs{\Pi\,a\,B}}\,\Gamma\,\Delta\,\ulsigma &&\defn
    (\gamma : \ba\,\ulGamma)\to \Sub_{\bB}\,(\Gamma\,\gamma)\,(\Delta\,(\ba\,\ulsigma\,\gamma))\,(\ulsigma,\,\refl)\\
  & \Ty_{\bs{\Pi\,a\,B}}\,\Gamma\,\ulA &&\defn
  \{\gamma : \ba\,\ulGamma\}(\alpha : \ba\,\ulA\,\gamma)
  \to \Ty_{\bB}\,(\Gamma\,\gamma)\,(\ulA,\,\alpha)\\
  & \Tm_{\bs{\Pi\,a\,B}}\,\Gamma\,A\,\ult &&\defn
    (\gamma : \ba\,\ulGamma) \to \Tm_{\bB}\,(\Gamma\,\gamma)\,(\A\,(\ba\,\ult\,\gamma))\,(\ult,\,\refl)
\end{alignat*}
Likewise, all structure is defined pointwise using $\bB$ structure. Similarly to
the $\bEl$ case, we have to sometimes fall through the defining isomorphisms for
$\ba$ structure. For comparison, in the finitary case we had the following definition:
\[
  (\Gamma \ext_{\bs{\Pi\,a\,B}} A)\,(\gamma,\,\alpha) \defn (\Gamma\,\gamma \ext_{\bB} A\,\alpha)
\]
Here, $(\gamma,\,\alpha) : \ba\,(\ulGamma \ext \ulA)$, so also
$(\gamma,\,\alpha) : (\gamma : \ba\,\ulGamma) \times \ba\,\ulA\,\gamma$, so the
$\Sigma$ pattern-matching notation is justified in the definition. In the
current infinitary case, we have $(\ba_{\ext},\,\ba_{\ext}^{-1}) :
\ba\,(\ulGamma \ext \ulA) \simeq ((\gamma : \ba\,\ulGamma) \times
\ba\,\ulA\,\gamma)$ instead. But we can use the intuition that set isomorphisms
are like unary record types, so we can still give a pattern-matching definition:
\[
  (\Gamma \ext_{\bs{\Pi\,a\,B}} A)\,(\ba_{\ext}^{-1}(\gamma,\,\alpha)) \defn
   (\Gamma\,\gamma \ext_{\bB} A\,\alpha)
\]
For the definitions of other type and term formers, we likewise insert the
isomorphisms appropriately. It remains to define iso-cleaving
$\bs{\Pi}$. Coercion is given by mapping indices backwards in $\bEl\,\ba$ and
coercing outputs forwards in $\bB$.
\begin{alignat*}{3}
  & \coe\,\ulsigma\,\Gamma &&\defn
    \lambda\,\gamma.\,\coe_{\bB}\,(\ulsigma,\refl)\,(\Gamma\,(\ba\,(\ulsigma^{-1})\,\gamma))\\
  & \coe\,\ult\,A &&\defn
    \lambda\,\gamma\,a.\,\coe_{\bB}\,(\ult,\refl)\,(A\,(\ba\,(\ult^{-1})\,(\ba_{\ext}^{-1}(\gamma,a))))
\end{alignat*}
Likewise, $\coh$-s are given by backwards-forwards $\coh$-s. As before,
$\bs{\appi}$ and $\bs{\lami}$ are defined as currying and uncurrying the
underlying functions.

\subsubsection{External product type}

For $\mi{Ix} : \Set_j$ and $\bB : \mi{Ix} \ra \bTy\,\bGamma$, we define
$\bPie\,\mi{Ix}\,\bB \bs{:} \bTy\,\bGamma$ as the $\mi{Ix}$-indexed direct
product of $\bB$. Since the indexing is given by a metatheoretic function, every
component is given in the evident pointwise way, including iso-cleaving. This
concludes the definition of the $\bM$ model.

\section{Left Adjoints of Substitutions}
\label{sec:iqii-left-adjoints}

In the following we adapt Section \ref{sec:fqii-left-adjoint} to infinitary
signatures.

\begin{itemize}
  \item
  We again write $\llb\blank\rrb$ for the interpretation into the flcwf model $\bM$.
  \item
  We also add $\top : \Ty\,\Gamma$ and $\Sigma : (A : \Ty\,\Gamma) \to
  \Ty\,(\Gamma \ext A) \to \Ty\,\Gamma$ to the ToS. Again we do not elaborate
  much on their semantics; $\top$ is given as ${\bs{\ms{K}}\,\bemptycon}$ in
  the model and $\Sigma$ is given by component-wise $\Sigma$.
\end{itemize}
We again fix $\Omega : \Con$ and define heterogeneous morphisms. The types of
eliminators stay exactly the same:
\begin{alignat*}{3}
  &\blank^{HM} &&: (\Gamma : \Con) \to \Gamma^A \to \Sub\,\Omega\,\Gamma \to \Ty\,\Omega\\
  &\blank^{HM} &&: (\sigma : \Sub\,\Gamma\,\Delta) \to \Tm\,\Omega\,(\Gamma^{HM}\,\gamma_0\,\gamma_1) \to \Tm\,\Omega\,(\Delta^{HM}\,(\sigma^A\,\gamma_0)\,(\sigma \circ\,\gamma_1))\\
  &\blank^{HM} &&: (A : \Ty\,\Gamma) \to A^A\,\gamma_0 \to \Tm\,\Omega\,(A[\gamma_1])
  \to \Tm\,\Omega\,(\Gamma^{HM}\,\gamma_0\,\gamma_1) \to \Ty\,\Omega\\
  &\blank^{HM} &&: (t : \Tm\,\Gamma\,A)(\gamma^{HM} : \Tm\,\Omega\,(\Gamma^{HM}\,\gamma_0\,\gamma_1))
   \to \Tm\,\Omega\,(A^{HM}\,(t^A\,\gamma_0)\,(t[\gamma_1])\,\gamma^{HM})
\end{alignat*}
We only need to show that the new type formers in $\U$, namely $\top$, $\Sigma$,
$\Id$ and $\Piinf$, can be also covered in the definition of $\blank^{HM}$. The
new type formers turn out to work exactly as mechanically as the previous
ones. We have the following:
\begin{alignat*}{3}
  & \top^{HM}\,\gamma^{HM} : \Tm\,\Omega\,(\top \toe \El\,\top)\\
  & \top^{HM}\,\gamma^{HM} \defn \lambdae\,\_.\,\tt
\end{alignat*}
\begin{alignat*}{3}
  & (\Sigma\,a\,b)^{HM}\,\gamma^{HM} : \\
  & \hspace{2em} \Tm\,\Omega\,(((\alpha : a^A\,\gamma_0) \times (b^A\,(\gamma_0,\,\alpha))) \toe \El\,(\Sigma\,(\alpha : a[\gamma_1])\,(b[\gamma_1,\,\alpha])))\\
  & (\Sigma\,a\,b)^{HM}\,\gamma^{HM} \defn \lambdae\,(\alpha,\,\beta).\,(a^{HM}\,\gamma^{HM}\,\alpha,\,b^{HM}\,(\gamma^{HM},\,\refl)\,\beta)\\
  & \\
  & (\Piinf\,\ms{Ix}\,b)^{HM}\,\gamma^{HM} : \Tm\,\Omega\,(((i : \ms{Ix}) \to (b\,i)^A\,\gamma_0) \toe \El\,((i : \ms{Ix}) \toinf (b\,i)[\gamma_1]))\\
  & (\Piinf\,\ms{Ix}\,b)^{HM}\,\gamma^{HM} \defn \lambdae\,f.\,\lambdainf\,i.\,(f\,i)^{HM}\,\gamma^{HM}
\end{alignat*}
For $\Id$, we again have to complete a square.
\[
(\Id\,t\,u)\,\gamma^{HM} : \Tm\,\Omega\,((t^A\,\gamma_0 = u^A\,\gamma_0) \toe \El\,(\Id\,(t[\gamma_1])\,(u[\gamma_1])))
\]
This follows from $t^{HM}\,\gamma^{HM}$ and $u^{HM}\,\gamma^{HM}$, the same way
as in the flcwf semantics before.

\begin{theorem}
If every infinitary QII signature has an initial algebra, then for every $\nu :
\Sub\,\Omega\,\Delta$, there exists a left adjoint of $\sem{\nu} : \sem{\Omega}
\to \sem{\Delta}$.
\end{theorem}
\begin{proof}
For each $\delta : \Delta^A$, the comma category $\delta/\sem{\nu}$ can be
specified (up to isomorphism) by the signature $\Omega \ext
\Delta^{HM}\,\delta\,\nu$, thus it has an initial object. Hence $\sem{\nu}$ has
a left adjoint.
\end{proof}

\section{Signature-Based Semantics of Signatures}
\label{sec:signature-semantics}

We have seen that the $\blank^{HM}$ interpretation yields a notion of algebra
morphism that is specified inside ToS. What else can we represent in ToS? For
example, can we internalize $\blank^D$, $\blank^M$ and $\blank^S$? In this
section we show that the full flcwf semantics can be expressed internally to
the ToS syntax.

This means that for each $\Gamma : \Con$, we get $\Gamma^A : \Ty\,\emptycon$ as
the notion of algebras, $\Gamma^M : \Ty\,(\emptycon\ext(\gamma_0 :
\Gamma^A)\ext(\gamma_1 : \Gamma^A))$ as the notion of morphisms, $\id :
\Tm\,(\emptycon\ext(\gamma : \Gamma^A))\,(\Gamma^M[\gamma_0 \mapsto
  \gamma,\,\gamma_1 \mapsto \gamma])$ for the identity morphisms, and likewise
we get the whole flcwf of algebras in such an internal manner.

As we will shortly see, capturing the full flcwf semantics is possible with the
infinitary ToS, but not with the finitary ToS because it lacks the necessary
type formers in $\U$.

It would be needlessly tedious and repetitive to redo the flcwf semantics while
explicitly working with ToS components. Instead, we repurpose 2LTT for this use
case. Recall that 2LTT allows to get semantics internally to any cwf with $\Pi$,
$\Sigma$, $\top$ and $\Id$. In the current section we aim to get semantics
internally to the ToS syntax. In short, this means that we work in a 2LTT where
the inner theory is the theory of signatures. The picture is a bit more
nuanced though.

First, since ToS lives inside 2LTT, and we want to get presheaves over ToS in
the presheaf model, the metatheoretic setting of the presheaf model must be also
a 2LTT. This might get a bit confusing, so let us expand:
\begin{itemize}
 \item The syntax of 2LTT internalizes the ToS syntax as an assumed type former.
 \item The presheaf model of 2LTT lives inside yet another 2LTT, let us call it 2LTT*,
       which embeds \emph{both} the 2LTT syntax and the ToS syntax separately.
 \item In the presheaf model, the base cwf is the cwf of the ToS syntax in 2LTT*.
 \item The $\Ty_0$ type former in 2LTT is interpreted in the presheaf model using the $\Ty_0$ type former
       in 2LTT*.
 \item We add $\Tys : \Set$ and $\Tms : \Tys \to \Set$ to 2LTT. $\Tys$ is
   interpreted as the presheaf of ToS types, and $\Tms$ is interpreted using the
   displayed presheaf of ToS terms, following Definition \ref{def:psh-ty0-tm0}.
 \item
   We close $\Tys$ under type formers which represent all type formers in ToS.
   Like in the previous section, we assume that ToS types are closed under
   $\top$ and $\Sigma$, so we have $\top$, $\Sigma$, inductive $\Pi$, $\Pie$ and
   $\U$ in $\Tys$.  The $\U$ in $\Tys$ has $\El : \Tms\,\U \to \Tys$, and it is
   closed under $\top$, $\Sigma$, $\Piinf$ and $\Id$. In the presheaf model, all
   structure in $\Tys$ is interpreted using ToS type formers in the evident way.
\end{itemize}

\begin{notation}
We shall omit $\Tms$ in the following, similarly to how we previously omitted
$\Tm_0$. We keep omitting $\Tm_0$ in the new setup as well. However, we
will still mark $\El : \U \to \Tys$ explicitly.
\end{notation}

For reference, we list type formers in $\Tys$ below.
\begingroup
\allowdisplaybreaks
\begin{alignat*}{3}
  & \U               &&: \Tys\\
  & \El              &&: \U \to \Tys\\
  & \top             &&: \U \\
  & \Sigma           &&: (a : \U) \to (\El\,a \to \U) \to \U \\
  & \Id              &&: \El\,a \to \El\,a \to \U \\
  & \Piinf           &&: (\ms{Ix} : \Ty_0) \to (\ms{Ix} \to \U) \to \U\\
  & \Pi              &&: (a : \U) \to (\El\,a \to \Tys) \to \Tys\\
  & \Pie             &&: (\ms{Ix} : \Ty_0) \to (\ms{Ix} \to \Tys) \to \Tys\\
  & \Sigma           &&: (A : \Tys) \to (A \to \Tys) \to \Tys\\
  & \top             &&: \Tys
\end{alignat*}
\endgroup

\begin{notation}
We will use the $\blank\!\toe\!\blank$ and $\blank\!\toinf\!\blank$ notations in
the following for $\Piinf$ and $\Pie$, but additionally we use
$\blank\!\toind\!\blank$ for internal products, to disambiguate them from outer
functions in 2LTT.
\end{notation}

We revisit now the flcwf semantics in the new setting. The goal is to produce
output by the signature-based semantics, such that if we use the original
$\blank^A$ interpretation on that, we get results that are equivalent to what we
get from the direct semantics. For the simplest example, for $\Gamma : \Con$, we
get $\Gamma^A_{\ms{sig}} : \Ty\,\emptycon$ from the signature-based semantics,
then we get $(\Gamma^A_{\ms{sig}})^A\,\tt : \Set$, which should be equivalent to
$\Gamma^A : \Set$.

In this section, we only describe the signature-based semantics, and we do not
formally check the round-trip property. The round-tripping seems very plausible
though, since as we will see, the signature-based semantics is exactly the same
as the direct semantics, modulo the change of universes and type formers.

We look at key parts of the model. In each case, we generally only check that we
have sufficient type formers. We again write components of the model in
\textbf{bold} font.

\subsubsection{Base cwf}
Contexts in the model are still flcwfs, but now $\Con$, $\Sub$, $\Ty$ and $\Tm$ in
flcwfs all return in $\Tys$. Hence, assuming $\bGamma : \bCon$, we have
\begin{alignat*}{3}
  &\Con_{\bGamma} &&: \Tys \\
  &\Sub_{\bGamma} &&: \Con_{\bGamma} \to \Con_{\bGamma} \to \Tys \\
  &\Ty_{\bGamma}  &&: \Con_{\bGamma} \to \Tys \\
  &\Tm_{\bGamma}  &&: (\Gamma : \Con_{\bGamma}) \to \Ty_{\bGamma}\,\Gamma \to \Tys
\end{alignat*}
We specify all equations using outer equality (since the $\Id$ types in $\Tys$
are extensional, this makes no difference). Similarly, components of $\bA :
\bTy\,\bGamma$ return in $\Tys$. Substitutions and terms in the model are
unchanged, they are weak morphisms and sections respectively. For $\bemptycon :
\bCon$, we use $\top : \Tys$ to define the components. Likewise, we use
the $\Sigma$ type in $\Tys$ to define $\blank\bs{\ext}\blank$.

If we write $\Tms$ explicitly, we have e.g.\ $\Sub_{\bGamma} :
\Tms\,\Con_{\bGamma} \to \Tms\,\Con_{\bGamma} \to \Tys$. Thus, we may use the
simplified interpretation of functions with inner domains, from Section
\ref{sec:functions-with-inner-domains}, and if we interpret the type of
$\Sub_{\bGamma}$ at the empty context in the presheaf model, we get
$\Ty\,(\emptycon\,\ext|\Con_{\bGamma}|\,\tt\,\ext |\Con_{\bGamma}|\,\tt)$.

\subsubsection{Universe}

$\bU : \bTy\,\bGamma$ is defined as $\bU : \bCon$, and we take the constant
displayed flcwf of the definition. Now, we have $\bU : \bCon$ as the flcwf of
types in $\U : \Tys$.
\begingroup
\allowdisplaybreaks
\begin{alignat*}{3}
  &\Con_{\bU} && \defn \U \\
  &\Sub_{\bU}\,\Gamma\,\Delta && \defn \Gamma \toind \El\,\Delta\\
  &\Ty_{\bU}\,\Gamma && \defn \Gamma \toind \U\\
  &\Tm_{\bU}\,\Gamma\,A && \defn (\gamma : \Gamma) \toind \El\,(A\,\gamma)
\end{alignat*}
\endgroup
$\emptycon_{\bU}$, $\blank\ext_{\bU}\blank$ and $\Id_{\bU}$ are defined using
the type formers in $\U$. As before, $\K_{\bU}$ is defined simply as a constant
function.  In $\bEl\,\ba : \bTy\,\bGamma$, we use the $\Id$ type in $\Tys$ in
morphisms and sections:
\begin{alignat*}{3}
  &\Con_{\bEl\,\ba}\,\ulGamma && \defn \El\,(\ba\,\ulGamma) \\
  &\Sub_{\bEl\,\ba}\,\Gamma\,\Delta\,\ulsigma && \defn \Id\,(\ba\,\ulsigma\,\Gamma)\,\Delta\\
  &\Ty_{\bEl\,\ba}\,\Gamma\,\ulA && \defn \El\,(\ba\,\ulA\,\Gamma)\\
  &\Tm_{\bEl\,\ba}\,\Gamma\,A\,\ult && \defn \Id\,(\ba\,\ult,\Gamma)\,A
\end{alignat*}

\subsubsection{Type formers in $\bU$}

For $\bs{\top}$, $\bs{\Sigma}$ and $\bId$ in $\bU$, we use $\top$, $\Sigma$ and
$\Id$ in $\U : \Tys$ in a straightforward way. For $\bPiinf\,\ms{Ix}\,\bb$, we
have the following:
\[
   \Con_{(\bPiinf\,\ms{Ix}\,\bb)}\,\ulGamma \defn (i : \ms{Ix}) \toinf (\bb\,i)\,\ulGamma
\]
Let us look at morphisms:
\begin{alignat*}{3}
  &\Sub_{(\bPiinf\,\ms{Ix}\,\bb)}\,\ulsigma : ((i : \ms{Ix}) \toinf (\bb\,i)\,\ulGamma)
    \toind \El\,((i : \ms{Ix}) \toinf (\bb\,i)\,\ulDelta)\\
  &\Sub_{(\bPiinf\,\ms{Ix}\,\bb)}\,\ulsigma \defn \lambda^{\ms{ind}}\,t.\,\lambdainf\,i.\,
    (\bb\,i)\,\ulsigma\,(t\,i)
\end{alignat*}
Here, we map an infinitary function to another one, which checks out just fine,
since $\toind$ allows such mapping. We have just enough higher-order functions
to complete this definition. The rest of $\bPiinf\,\ms{Ix}\,\bb$ follows evidently.

\subsubsection{$\bs{\Pi}$, $\bs{\Pie}$, $\bs{\top}$, $\bs{\Sigma}$}

In $\bs{\Pi\,a\,B}$, we use inductive functions in components:
\begin{alignat*}{3}
  &\Con_{(\bs{\Pi\,a\,B})}\,\ulGamma &&\defn (\gamma : \ba\,\ulGamma) \toind \Con_{\bB}\,(\ulGamma,\,\gamma)\\
  &\Sub_{(\bs{\Pi\,a\,B})}\,\Gamma\,\Delta\,\ulsigma &&\defn (\gamma : \ba\,\ulGamma) \toind \Sub_{\bB}\,(\Gamma\,\gamma)\,(\Delta\,(\ba\,\ulsigma\,\gamma))\,(\ulsigma,\,\refl)\\
  & ... &&
\end{alignat*}
In $\bs{\Pie}$, we use $\toe$. In $\bs{\top}$ and $\bs{\Sigma}$, we use $\top$ and $\Sigma$ in $\Tys$. This concludes the definition of the signature-based semantics.

\begin{mydefinition}[\textbf{Signature-based AMDS interpretation}]
For some $\Gamma : \Con$, we define the following by interpreting $\Gamma$ in
the signature-based flcwf model, then interpreting the result in the presheaf
model of 2LTT.
\begin{alignat*}{3}
  &\Gamma^A_{\ms{sig}} &&: \Ty\,\emptycon \\
  &\Gamma^M_{\ms{sig}} &&: \Ty\,(\emptycon\ext(\gamma_0 : \Gamma^A_{\ms{sig}})\ext(\gamma_1 : \Gamma^A_{\ms{sig}}))\\
  &\Gamma^D_{\ms{sig}} &&: \Ty\,(\emptycon\ext(\gamma : \Gamma^A_{\ms{sig}}))\\
  &\Gamma^S_{\ms{sig}} &&: \Ty\,(\emptycon\ext(\gamma : \Gamma^A_{\ms{sig}})\ext(\gamma^D : \Gamma^D_{\ms{sig}}))
\end{alignat*}
\end{mydefinition}

\subsubsection{Backporting to finitary signatures}

It is apparent from the previous section that the signature-based full flcwf
model requires at least $\top$, $\Sigma$ and $\Id$ in $\U$: in the definition of
$\bU$ in the model these are needed to define the family structure and the
finite limit structure.

Hence, if we want to only support structure in $\Tys$ corresponding to a theory
of finitary signatures, we need to drop all semantic components which rely on
the missing type formers. We have seen this kind of trimmed semantics in Section
\ref{sec:fqii-variations}. In particular, we still get a category of algebras
for each signature, since that can be modeled without $\top$, $\Sigma$ and
$\Id$.

\subsubsection{Application: colimits}

The signature-based semantics is often helpful when we want to construct new
signatures from old ones. We give an example application, in the construction of
colimits.

We would like to use left adjoints of substitutions to build colimits in
categories of algebras. For this, it is enough to build indexed coproducts and
binary coequalizers.

For some $\Gamma : \Con$, we get $\Gamma^A_{\ms{tos}} : \Ty\,\emptycon$. For
convenience we shall work with $\Gamma^A_{\ms{tos}}$ in the following, instead of $\Gamma$. First, we
construct $\ms{Ix}$-indexed coproducts in the category of $\Gamma$-algebras, by
taking the left adjoint of the following diagonal substitution:
\begin{alignat*}{3}
  &\ms{diag} : \Sub\,(\emptycon\ext(\gamma : \Gamma^A_{\ms{tos}}))\,(\emptycon\ext (f : (\ms{Ix} \toe \Gamma^A_{\ms{tos}})))\\
  &\ms{diag} \defn (f \mapsto \lambda^{\ms{ext}}\,i.\,\gamma)
\end{alignat*}
For coequalizers,we again take the left adjoint of a diagonal substitution, but
here we need to rely on internal morphisms in the signature:
\begin{alignat*}{3}
  &\ms{diag} : \Sub\,&&(\emptycon&&\ext(\gamma : \Gamma^A_{\ms{tos}}))\\
  &&&(\emptycon&&\ext (\gamma_0 : \Gamma^A_{\ms{tos}}) \ext (\gamma_1 : \Gamma^A_{\ms{tos}})\\
  &&&&&\ext (f : \Gamma^M_{\ms{tos}}[\gamma_0 \mapsto \gamma_0,\,\gamma_1 \mapsto \gamma_1])\\
  &&&&&\ext (g : \Gamma^M_{\ms{tos}}[\gamma_0 \mapsto \gamma_0,\,\gamma_1 \mapsto \gamma_1]))\\
  &\rlap{$\ms{diag} \defn (\gamma_0 \mapsto \gamma,\,\gamma_1 \mapsto \gamma,\,f \mapsto \id_{\ms{tos}}[\gamma \mapsto \gamma],\,
           g \mapsto \id_{\ms{tos}}[\gamma \mapsto \gamma]) $}
\end{alignat*}
Above, we use $\id_{\ms{tos}} : \Tm\,(\emptycon\ext(\gamma :
\Gamma^A_{\ms{tos}}))\,(\Gamma^M_{\ms{tos}}\,[\gamma_0 \mapsto \gamma,\,\gamma_1
  \mapsto \gamma])$, which also comes from the signature-based semantics.

Of course, if we want to be fully precise, we need to show that what we get is
equivalent to coproducts and coequalizers in the external sense. For this, we
would need the round-trip property of the signature-based semantics.

\section{Discussion of Semantics}

\subsubsection{Iso-fibrancy as a weak structure identity principle}

The flcwfs of algebras that we get from the infinitary semantics are exactly the
same as in the finitary case. However, semantic types are a bit more
interesting. The iso-fibrancy of types can be understood as a weaker version of
the \emph{structure identity principle} in homotopy type theory.

The structure identity principle says that isomorphism of algebras is equivalent
to equality of algebras. This is the same as saying that categories of algebras
are univalent \cite{univalent-categories}. Assuming a signature $\Gamma$ and
algebras $\gamma \simeq \gamma'$, we have $\gamma = \gamma'$. This equality is
respected by every construction in HoTT, which implies that for any HoTT type
family $F : \Gamma^A \to \ms{Type}$, we have a function $F\,\gamma \to
F\,\gamma'$.

We get a similar but weaker statement from the infinitary semantics: for $\sigma
: \gamma \simeq \gamma'$ and some ToS type $A : \Ty\,\Gamma$, we have a function
$\coe\,\sigma : A^A\,\gamma \to A^A\,\gamma'$. We also have
$\coh\,\sigma\,\alpha : \alpha \simeq_{\sigma} \coe\,\sigma\,\alpha$ for some
$\alpha : A^A\,\gamma$. So we can transport over isomorphisms, but not all
metatheoretic families can be transported, only those which arise as ToS types.

Of course, we can transport over multiple types, or telescopes of types too, by
iterated transport. For instance, given $A : \Ty\,\Gamma$, $B : \Ty\,(\Gamma
\ext A)$, $\alpha : A^A\,\gamma$ and $\beta : B^A\,(\gamma,\,\alpha)$, we can
transport $\alpha$ first, then transport $\beta$ by
$(\sigma,\,\coh\,\sigma\,\alpha)$.  Alternatively, if we have large $\Sigma$ types
in ToS, as $\Sigma : (A : \Ty\,\Gamma) \to \Ty\,(\Gamma \ext A) \to \Ty \Gamma$,
that makes iterated transport superfluous.

\subsubsection{Variations of semantics}

First, unlike in the finitary case, we have no opportunity to minimize
assumptions on the inner theory. Already when we compute algebras, we need inner
$\Pi$ for infinitary functions, inner $\top$ for $\top$, inner $\Sigma$ for
$\Sigma$ and inner $\blank\!=\!\blank$ for $\Id$. Note though that we still get
semantics in any LCCC (locally cartesian closed category), since we can build a
cwf with the required type formers from any LCCC
\cite{clairambault2014biequivalence}.

Second: can we add the ``large'' equality type, which includes sort equations,
back to infinitary signatures? We dropped sort equations in this chapter because
they are clearly not isofibrant. We can add them back into the mix though, at
the price of dropping components from the semantics of signatures. The reason for
having isofibrant types is that type formers in $\U$ preserve $\emptycon$ and
$\blank\!\ext\!\blank$ only up to isomorphism. If we drop all semantic
components which depend on $\emptycon$ and $\blank\!\ext\!\blank$, we can drop
isofibrancy too from the model, and everything works. In this case, we still get
a category of algebras, plus a notion of induction, but we cannot show that
initiality is equivalent to induction, as the proof of Theorem
\ref{thm:initiality-induction} depends on $\blank\!\ext\!\blank$.

\subsubsection{Model Constructions}

In this chapter we gain some expressive power in defining model constructions
using substitutions or terms. For starters, the construction of categories from
monoids works now:

\begin{myexample}
Let us have
$\ms{MonoidSig}$ as the signature for monoids, with $\ms{M} : \U$ as the carrier
set, $\blank\!\cdot\!\blank : \ms{M} \to \ms{M} \to \El\,\ms{M}$ as
multiplication and $\epsilon : \El\,\ms{M}$ as identity the element. We define
$\sigma : \Sub\,\ms{MonoidSig}\,\ms{CatSig}$ to contain $\ms{Obj} \defn \top$,
$\ms{Hom} \defn \lambda\,\_\,\_.\,\ms{M}$, $\id \defn \epsilon$ and
$\blank\!\circ\!\blank \defn \blank\!\cdot\!\blank$.
\end{myexample}

Many constructions in the literature which have been dubbed \emph{syntactic
models} \cite{next700} or \emph{syntactic translations} can be defined now in
the ToS, for the following reasons.
\begin{itemize}
\item
  Syntactic translations usually do not rely on models being actually syntactic:
  they do not use induction on \emph{target} theory syntax. A rare
  counterexample is our construction of recursors and eliminators for term
  models. These are perhaps syntactic in the sense that they prominently involve
  the syntax of some type theory, and they construct recursor/eliminator
  functions by induction on terms.
\item
  Syntactic translations rarely if ever involve higher-order constructions.
  Such would be interpreting $\Con$ with $(\Con \to \Con) \to \Con$, for a
  contrived example.
\end{itemize}

The gluing construction in Example \ref{ex:gluing} is already a fairly general
example that only requires the finitary ToS to define. That construction is more
in an ``indexed'' style, but now we can also do constructions in a more
``fibered'' style.

\begin{myexample}
\label{ex:param-translation}
We may consider a unary parametricity translation in the style of Bernardy,
Jansson, and Paterson \cite{bernardy12parametricity}, which makes use of the
small $\Sigma$-type in the theory of signature. We assume $\ms{TT} :
\Ty\,\emptycon$ as the signature for the theory, and $\ms{TT}^D :
\Ty\,(\emptycon\,\ext\,(M : \ms{TT}))$ as the signature for displayed models.
The translation can be typed as $\Tm\,(\emptycon\ext (M : \ms{TT}))\,\ms{TT}^D$:
we assume a model of the theory, and build a displayed model over the same
theory. Informally, when $M$ is initial, we get a translation which doubles each
context:
\[
  \sem{\Gamma \ext (a : A)} \equiv \sem{\Gamma} \ext (a : A) \ext (a^D : \sem{A}\,a)
\]
Formally, however, this is not well-typed because $A$ lives in $\Gamma$, not in
$\sem{\Gamma}$. Hence, in the definition of contexts in the displayed model, we
also include a substitution which projects out the ``base'' parts of contexts.
This can be used to weaken types in base contexts to total contexts.
\begin{alignat*}{3}
  &\Con : \Con_{M} \to \U\\
  &\Con\,\ulGamma \defn \Sigma\,(\ulGamma' : \Con_M)\,(\proj : \Sub_M\,\ulGamma'\,\ulGamma)
\end{alignat*}
This requires the small $\Sigma$-type in ToS. It is possible to rephrase the
construction without type formers in $\U$; again, Example \ref{ex:gluing} has unary
parametricity as a special case. However, the fibered version has the advantage
that contexts are translated to contexts, types to types, and terms to terms,
which makes it more convenient if we actually want to implement it as a program
translation. In contrast, the gluing definition of unary parametricity maps
contexts to types.
\end{myexample}

\section{Term Algebras}
\label{sec:inf-term-algebras}

We adapt now the previous term algebra construction to the infinitary
case. We again switch to the ETT setup with cumulative universes. We
assume Section \ref{sec:cumulative-ett} without any change. Also, we adapt
\ref{sec:ett-signatures} to infinitary signatures and semantics. All definitions
are the same, the only change is that the $\bM_{i,j}$ model is now the
isofibrant flcwf model, and we have the infinitary ToS.

\subsection{Term Algebra Construction}
\label{sec:iqii-term-algebra-construction}

The term algebra construction changes significantly. The reason is the
following. In the finitary case, the key property was that ``small types
evaluated in the term model are sets of terms''. Formally, we had for $a :
\Tm\,\Omega\,\U$ that $a^A\,(\Omega^T\,\id) \equiv \Tm\,\Omega\,(\El\,a)$.  This
is now weakened to an isomorphism, i.e.\ $a^A\,(\Omega^T\,\id) \simeq
\Tm\,\Omega\,(\El\,a)$.

This is again necessary because of the closure of $\U$ under type formers. For
example, $\top^A\,(\Omega^T\,\id) \equiv \top$, and $\Tm\,\Omega\,(\El\,\top)$
is merely isomorphic to $\top$. We assume $\Omega : \Sig_j$ for some $j$
level, and define $\blank^T$ by induction on $\syn_j$.
\begin{alignat*}{3}
  &\blank^T &&: (\Gamma : \Con)&&(\nu : \Sub\,\Omega\,\Gamma) \to \Gamma^A\\
  &\blank^T &&: (\sigma : \Sub\,\Gamma\,\Delta)&&(\nu : \Sub\,\Omega\,\Gamma) \to \Delta^T\,(\sigma \circ \nu) \simeq \sigma^A\,(\Gamma^T\,\nu)\\
  &\blank^T &&: (A : \Ty\,\Gamma)&&(\nu : \Sub\,\Omega\,\Gamma) \to \Tm\,\Omega\,(A[\nu])
  \to A^A\,(\Gamma^T\,\nu)\\
  &\blank^T &&: (t : \Tm\,\Gamma\,A)&&(\nu : \Sub\,\Omega\,\Gamma) \to A^T\,\nu\,(t[\nu]) \simeq_{\id} t^A\,(\Gamma^T\,\nu)
\end{alignat*}
In short, interpretations of substitutions and terms are weakened to
isomorphisms.  By $\simeq_{\id}$ we mean a displayed isomorphism of objects in
the semantic $A$ type (which is an flcwf isofibration); recall Definition
\ref{def:displayed-iso}.  The isomorphism is ``vertical'' since it lies over
$\id$.

The interpretation of the cwf is the same as before, but like in the isofibrant
semantics, we have to use explicit $\coe$ instead of silently transporting over
equalities. In the interpretations of substitutions and terms, we have to
explicitly compose isomorphisms and sometimes lift them using $\coh$. We give
some examples. The interpretation of context formers is the same as before:
\begin{alignat*}{3}
  &\emptycon^T\,\nu           &&\defn \tt\\
  &(\Gamma \ext A)^T(\nu,\,t) &&\defn (\Gamma^T\,\nu,\,A^T\,\nu\,t)
\end{alignat*}
Type substitution with $\sigma : \Sub\,\Gamma\,\Delta$ is
interpreted as coercion:
\begin{alignat*}{3}
  &(A[\sigma])^T : (\nu : \Sub\,\Omega\,\Gamma)(t : \Tm\,\Omega\,(A[\sigma][\nu]) \to A^A\,(\sigma^A\,(\Gamma^T\,\nu)) \\
  &(A[\sigma])^T\,\nu\,t \defn \coe\,(\sigma^T\,\nu)\,(A^T\,(\sigma\circ\nu)\,t)
\end{alignat*}
Composition of $\sigma : \Sub\,\Delta\,\Xi$ and $\delta : \Sub\,\Gamma\,\Delta$
is the following:
\begin{alignat*}{3}
  &(\sigma \circ \delta)^T : (\nu : \Sub\,\Omega\,\Gamma) \to \Delta^T\,(\sigma\circ\delta\circ\nu) \simeq \sigma^A\,(\delta^A\,(\Gamma^T\,\nu))\\
  &(\sigma \circ \delta)^T\,\nu \defn \sigma^M\,(\delta^T\,\nu) \circ \sigma^T\,(\delta\circ\nu)
\end{alignat*}
Above, we have
\begin{alignat*}{3}
& \delta^T\,\nu &&: \Xi^T\,(\delta\circ\nu) \simeq \delta^A\,(\Gamma^T\,\nu)\\
& \sigma^M\,(\delta^T\,\nu) &&: \sigma^A\,(\Xi^T\,(\delta\circ\nu)) \simeq \sigma^A\,(\delta^A\,(\Gamma^T\,\nu))\\
& \sigma^T\,(\delta\circ\nu) &&: \Delta^T\,(\sigma\circ\delta\circ\nu) \simeq \sigma^A\,(\Xi^T\,(\delta\circ \nu))
\end{alignat*}
Hence, the type of the composition in the definition checks out. We make use of
the fact that $\sigma^M$ sends an isomorphism in $\Gamma$ to an isomorphism in
$\Delta$.

Substitution extension is a somewhat more complicated case. We want to interpret
the extension of $\sigma : \Sub\,\Gamma\,\Delta$ with $t :
\Tm\,\Gamma\,(A[\sigma])$:
\[
      (\sigma,\,t)^T : (\nu : \Sub\,\Omega\,\Gamma)
  \to (\Delta \ext A)^T\,((\sigma,\,t)\circ\nu)\simeq(\sigma,\,t)^A\,(\Gamma^T\,\nu)
\]
The goal is an isomorphism in the semantic $\Gamma \ext A$ category, i.e.\ the
total category of $A$. Every isomorphism in $\Gamma \ext A$ arises as packing
together a $\Gamma$ isomorphism and a displayed $A$ isomorphism over it. We can compute
the type further:
\[
      (\sigma,\,t)^T : (\nu : \Sub\,\Omega\,\Gamma)
  \to (\Delta^T\,(\sigma \circ \nu),\,A^T\,(\sigma\circ\nu)\,(t[\nu])) \simeq (\sigma^A\,(\Gamma^T\,\nu),\,t^A\,(\Gamma^T\,\nu))
\]
We can exhibit $\sigma^T\,\nu : \Delta^T\,(\sigma \circ \nu) \simeq
\sigma^A\,(\Gamma^T\,\nu)$ as the base component of the goal isomorphism. Now we
need a displayed isomorphism over it. Following the pattern, we may try
$t^T\,\nu$:
\[
  t^T\,\nu : (A[\sigma])^T\,\nu\,(t[\nu]) \simeq_{\id} t^A\,(\Gamma^T\,\nu)
\]
Computing the type:
\[
  t^T\,\nu : \coe\,(\sigma^T\,\nu)\,(A^T\,(\sigma\circ\nu)\,(t[\nu])) \simeq_{\id} t^A\,(\Gamma^T\,\nu)
\]
So this is not quite what is needed; we want a displayed iso over $\sigma^T\,\nu$, but we have
something over $\id$. We can fix this using $\coh$:
\[
\coh\,(\sigma^T\,\nu)\,(A^T\,(\sigma\circ\nu)\,(t[\nu])) :
A^T\,(\sigma\circ\nu)\,(t[\nu]) \simeq_{\sigma^T\,\nu} \coe\,(\sigma^T\,\nu)\,(A^T\,(\sigma\circ\nu)\,(t[\nu]))
\]
The composition of $t^T\,\nu$ and the above now checks out:
\[
  (\sigma,\,t)^T\,\nu \defn (\sigma^T\,\nu,\,t^T\,\nu\circ \coh\,(\sigma^T\,\nu)\,(A^T\,(\sigma\circ\nu)\,(t[\nu])))
\]
We omit the rest of the cwf interpretation. It should be apparent that explicit
$\coe$ and $\coh$-handling is fairly technical. We note though that in a proof
assistant, the finitary and infinitary term model constructions would be of
similar difficulty, because there we cannot rely on equality reflection and
implicit transports to magically tidy up the formalization. In fact, even in the
finitary case it would be a good idea to structure the formalization around
coercions and coherences.

The high-level explanation for why the weakened constructions continue to work,
is the same as what we gave in the section on iso-fibrant semantics: we do
nothing which would violate stability under isomorphisms; additionally, because
our isofibrations are \emph{split}, coercion and coherence compute strictly on
identities and compositions, which ensures that conversion equations in the
syntax are respected. For example, functoriality of type substitution relies on
$\coe$ computation on identity and composition.

\subsubsection{Universe}
The universe is interpreted as follows.
\begin{alignat*}{3}
  &\U^T : (\nu : \Sub\,\Omega\,\Gamma) \to \Tm\,\Omega\,\U \to \Set_{j + 1} \\
  &\U^T\,\nu\,a \defn \Tm\,\Omega\,(\El\,a) \\
  &\\
  &(\El\,a)^T : (\nu : \Sub\,\Omega\,\Gamma)(t : \Tm\,\Omega\,(\El\,(a[\nu]))) \to a^A\,(\Gamma^T\,\nu)\\
  &(\El\,a)^T\,\nu\,t \defn (a^T\,\nu)\,t
\end{alignat*}
In the interpretation of $\El$, note that
\[
a^T\,\nu : \Tm\,\Omega\,(\El\,(a[\nu])) \simeq_{\id} a^A\,(\Gamma\,\nu)
\]
But this is an isomorphism in the semantic $\U$, which is the category of sets
in $\Set_{j + 1}$. So coercion along $a^T\,\nu$ is simply function application,
and we are justified in writing $(a^T\,\nu)\,t$.

For each type former in $\U$, we have to exhibit an isomorphism of sets in the
interpretation.

\subsubsection{$\bs{\top}$, $\bs{\Sigma}$}
We need
\[
  \top^T : (\nu : \Sub\,\Omega\,\Gamma) \to \U^T\,\nu\,(\top[\nu]) \simeq_{\id} \top^A\,(\Gamma^T\,\nu)
\]
The result type computes to $\Tm\,\Omega\,(\El\,\top) \simeq \top$, which is evident. For $\Sigma$, we
have to show
\[
   \Tm\,\Omega\,(\El\,(\Sigma\,(a[\nu])\,(b[\nu\circ\p,\,\q]))) \simeq ((\alpha : a^A\,(\Gamma^T\,\nu)) \times b^A\,(\Gamma^T\,\nu,\,\alpha))
\]
This follows from the induction hypotheses $a^T$ and $b^T$, which establish the
first and second components of the desired isomorphism.

\subsubsection{Identity}
For the identity type, we need
\[
  \Tm\,\Omega\,(\El\,(\Id\,(t[\nu])\,(u[\nu])) \simeq (t^A\,(\Gamma^T\,\nu) \equiv u^A\,(\Gamma^T\,\nu))
\]
This follows from $t^T\,\nu$, $u^T\,\nu$ and the specifying isomorphism of
$\Id$.

\subsubsection{Small external products}
This function type follows the same pattern. We define the isomorphism below
using induction hypotheses and the specifying isomorphism of $\Piinf$.
\[
\Tm\,\Omega\,(\El\,(\Piinf\,\mi{Ix}\,(\lambda\,i.\,(b\,i)[\nu])))
\simeq
((i : \mi{Ix}) \to (b\,i)^A\,(\Gamma\,\nu))
\]

\subsubsection{Internal products}
Inductive functions are interpreted using transport along $a^T\,\nu :
\Tm\,\Omega\,(\El\,(a[\nu])) \simeq a^A\,(\Gamma^T\,\nu)$:
\begin{alignat*}{3}
  &(\Pii\,a\,B)^T : (\nu : \Sub\,\Omega\,\Gamma)(t : \Tm\,\Omega\,(\Pii\,(a[\nu])\,(B[\nu\circ\p,\,\q])))\\
  & \hspace{3.5em}\to (\alpha : a^A\,(\Gamma^T\,\nu)) \to B^A\,(\Gamma^T\,\nu,\,\alpha)\\
  &(\Pii\,a\,B)^T\,\nu\,t \defn \lambda\,\alpha.\,
         B\,(\nu,\,(a^T\,\nu)^{-1}\,\alpha)\,(t\,((a^T\,\nu)^{-1}\,\alpha))
\end{alignat*}
\textbf{External products} are interpreted the same way as in the finitary case.

\subsection{Eliminator Construction}

We only present the eliminator construction in the following, since (unique)
recursors are derivable from this.

Compared to the finitary case, the eliminator construction does not change as
much as the term algebra construction. The reason is that although we have
weakened strict algebra equality to isomorphism, in the current construction we
only have to show equalities of substitutions and terms, which we do not need to
weaken (and they cannot be sensibly weakened anyway).

We assume $j$ and $k$ such that $ j + 1 \leq k$, and also $\Omega : \Sig_j$ and
$\omega^D : \Omega^D_k\,(\Omega^T\,\id)$. Hence, $\omega^D$ is a displayed
$\Omega$-algebra over the term algebra, and we aim to construct its section.
Note that we lift $\Omega^T\,\id : \Omega^A_{j+1}$ to level $k$ by cumulativity.
We define $\blank^E$ by induction on $\syn_j$.
\begin{alignat*}{3}
  &\blank^E &&: (\Gamma : \Con)&&(\nu : \Sub\,\Omega\,\Gamma) \to \Gamma^S\,(\nu^A\,(\Omega^T\,\id))\,(\nu^D\,\omega^D)\\
  &\blank^E &&: (\sigma : \Sub\,\Gamma\,\Delta)&&(\nu : \Sub\,\Omega\,\Gamma) \to \Delta^E\,(\sigma \circ \nu) \equiv \sigma^S\,(\Gamma^E\,\nu)\\
  &\blank^E &&: (A : \Ty\,\Gamma)&&(\nu : \Sub\,\Omega\,\Gamma)(t : \Tm\,\Omega\,(A[\nu]))
     \to A^S\,(t^A\,(\Omega^T\,\id))\,(t^D\,\omega^D)\,(\Gamma^E\,\nu)\\
  &\blank^E &&: (t : \Tm\,\Gamma\,A)&&(\nu : \Sub\,\Omega\,\Gamma) \to A^E\,\nu\,(t[\nu]) \equiv t^S\,(\Gamma^E\,\nu)
\end{alignat*}
This is so far exactly the same as in Section
\ref{sec:fqii-eliminator-construction}. The subsequent changes arise from
the need to transport along $\blank^T$ in definitions.

\subsubsection{Universe}
For the universe, we need
\begin{alignat*}{3}
  &\U^E : (\nu : \Sub\,\Omega\,\Gamma)(a : \Tm\,\Omega\,\U) \to (\alpha : a^A\,(\Omega^T\,\id)) \to a^D\,\omega^D\,\alpha
\end{alignat*}
Since we only have $a^T\,\id : a^A\,(\Omega^T\,\id) \simeq \Tm\,\Omega\,(\El\,a)$, the definition
becomes
\[
  \U^E\,\nu\,a\,t \defn (a^T\,\id\,t)^D\,\omega^D
\]
That this is well-typed, follows from
\begin{alignat*}{3}
  & ((a^T\,\id)\,t)^T\,\id &&: t \equiv ((a^T\,\id)\,t)^A\,(\Omega^T\,\id) \\
  & ((a^T\,\id)\,t)^D\,\omega^D &&: a^D\,\omega^D\,(((a^T\,\id)\,t)^A\,(\Omega^T\,\id))
\end{alignat*}
For $\El$, we need to show
\[
(\El\,a)^E : (\nu : \Sub\,\Omega\,\Gamma)(t : \Tm\,\Omega\,(\El\,(a[\nu]))) \to a^S\,(\Gamma^E\,\nu)\,(t^A\,(\Omega^T\,\id)) \equiv t^D\,\omega^D
\]
We have
\[
  t^T\,\id : (a[\nu])^T\,\id\,t \equiv t^A\,(\Omega^T\,\id)
\]
Moreover
\[
  a^E\,\nu : \U^E\,\nu\,(a[\nu]) \equiv a^S\,(\Gamma^E\,\nu)
\]
Hence
\[
  a^E\,\nu : (\lambda\,t.\,((a[\nu])^T\,\id\,t)^D\,\omega^D) \equiv a^S\,(\Gamma^E\,\nu)
\]
Applying both sides to $((a[\nu])^T\,\id)^{-1}\,t$, we have
\[
  ((a[\nu])^T\,\id\,(((a[\nu])^T\,\id)^{-1}\,t))^D\,\omega^D \equiv a^S\,(\Gamma^E\,\nu)\,(((a[\nu])^T\,\id)^{-1}\,t)
\]
This simplifies to
\[
  t^D\,\omega^D \equiv a^S\,(\Gamma^E\,\nu)\,(((a[\nu])^T\,\id)^{-1}\,t)
\]
By $(a^T\,\id\,t)^T\,\id : t \equiv (a^T\,\id\,t)^A\,(\Omega^T\,\id)$ this becomes:
\[
  t^D\,\omega^D \equiv a^S\,(\Gamma^E\,\nu)\,(((a[\nu])^T\,\id)^{-1}\,((a^T\,\id\,t)^A\,(\Omega^T\,\id)))
\]
Thus we have the required
\[
  t^D\,\omega^D \equiv a^S\,(\Gamma^E\,\nu)\,(\Omega^T\,\id)
  \]

\subsubsection{$\bs{\top}$, $\bs{\Sigma}$}

For $\top$, we need
\[
  \top^E : (\nu : \Sub\,\Omega\,\Gamma) \to \U^E\,\nu\,\top \equiv \top^S\,(\Gamma^E\,\nu)
\]
But this is clearly trivial, since $\top^S\,(\Gamma^E\,\nu) : \top \to \top$. Considering $\Sigma$:
\[
(\Sigma\,a\,b)^E : (\nu : \Sub\,\Omega\,\Gamma) \to
  \U^E\,\nu\,(\Sigma\,(a[\nu])\,(b[\nu\circ\p,\,\q])) \equiv (\Sigma\,a\,b)^S\,(\Gamma^E\,\nu)
\]
This case is a bit tedious. The sides above are functions, so appealing to function extensionality
we apply both sides to $(\alpha,\,\beta)$, where $\alpha : a^A\,(\nu^A\,(\Omega^T\,\id))$ and
$\beta : b^A\,(\nu^A\,(\Omega^T\,\id),\,\alpha)$. We also unfold some definitions:
\[
(((\Sigma\,(a[\nu])\,(b[\nu\circ\p,\,\q]))^T\,\id)^{-1}\,(\alpha,\,\beta))^D\,\omega^D \equiv
(a^S\,(\Gamma^E\,\nu)\,\alpha,\,b^S\,(\Gamma^E\,\nu,\,\refl)\,\beta)
\]
Unfolding the left side of this equation, we have
\[
  ((((a[\nu])^T\,\id)^{-1}\,\alpha)^D\,\omega^D,\,(((b[\nu\circ\p,\,\q])^T\,(\id,\,((a[\nu])^T\,\id)^{-1}\,\alpha))^{-1}\,\beta)^D\,\omega^D)
\]
Let us abbreviate $((a[\nu])^T\,\id)^{-1}\,\alpha : \Tm\,\Omega\,(\El\,(a[\nu]))$ as $\alpha'$:
\[
  (\alpha'^D\,\omega^D,\,(((b[\nu\circ\p,\,\q])^T\,(\id,\,\alpha'))^{-1}\,\beta)^D\,\omega^D)
\]
Hence, we need to show component-wise equality of pairs. The equality of first components follow
from the following:
\[
  a^E\,\nu : \U^E\,\nu\,(a[\nu]) \equiv a^S\,(\Gamma^E\,\nu)
\]
Unfolding definitions and applying both sides to $\alpha$, we get the equality of first components:
\[
  \alpha'^D\,\omega^D \equiv a^S\,(\Gamma^E\,\nu)\,\alpha
\]
Analogously, the equality of second components follows from
\[
b^E\,(\nu,\,\alpha') :
 (((b[\nu,\,\alpha'])^T\,\id)^{-1}\,\beta)^D\,\omega^D \equiv b^S\,(\Gamma^E\,\nu,\,\refl)\,\beta
\]
The right hand side is what we need, the left hand side though does not immediately match up.
Hence, it remains to show that
\[
(((b[\nu,\,\alpha'])^T\,\id)^{-1}\,\beta)^D\,\omega^D
\equiv
(((b[\nu\circ\p,\,\q])^T\,(\id,\,\alpha'))^{-1}\,\beta)^D\,\omega^D
\]
Thus, it suffices to show
\[
(b[\nu,\,\alpha'])^T\,\id \equiv (b[\nu\circ\p,\,\q])^T\,(\id,\,\alpha')
\]
This equation follows from a somewhat laborious unfolding of all involved definitions.
In particular, we use that for some $a : \Tm\,\Gamma\,\U$, we have
\[
  (a[\sigma])^T\,\nu \equiv a^M\,(\sigma^T\,\nu) \circ a^T\,(\sigma \circ \nu)
\]
which follows from the definition of $\blank^T$.

\subsubsection{Internal products}
In $\Pi$ we likewise transport along the domain isomorphism.
\begin{alignat*}{3}
 &(\Pi\,a\,B)^E : (\nu : \Sub\,\Omega\,\Gamma)(t : \Tm\,\Omega\,(\Pi\,(a[\nu])\,(B[\nu\circ\p,\,\q])))\\
 & \hspace{2em}\to (\alpha : \Tm\,\Omega\,(\El\,(a[\nu]))) \to B^S\,(t^A\,(\Omega^T\,\id)\,\alpha)\,(t^D\,\omega^D\,(\alpha^D\,\omega^D))\,(\Gamma^E\,\nu,\,\refl)\\
 &\rlap{$(\Pi\,a\,B)^E\,\nu\,t \defn \lambda\,\alpha.\,B^E\,(\nu,\,(a[\nu])^T\,\id\,\alpha)\,(t\,((a[\nu])^T\,\id\,\alpha))$}
\end{alignat*}
This is well-typed by the following:
\begin{alignat*}{3}
  & a^E\,\nu &&: ((a[\nu])^T\,\id\,\alpha)^D\,\omega^D \equiv a^S\,(\Gamma^E\,\nu)\,\alpha\\
  & ((a[\nu])^T\,\id\,\alpha)^T\,\id &&: \alpha \equiv ((a[\nu])^T\,\id\,\alpha)^A\,(\Omega^T\,\id)
\end{alignat*}

\subsubsection{$\bs{\Id}$, $\bs{\Piinf}$, $\bs{\Pie}$}

$\Id$ is trivial by UIP, and for $\Piinf$ and $\Pie$ we again do a straightforward recursion
under the indexing function.
\\\\
This concludes the definition of $\blank^E$. We again show the initiality of term algebras.

\begin{mydefinition}[\textbf{Eliminators}]
\label{def:iqiit-eliminator}
Assuming $\Omega : \Sig_j$, a $k$ level such that $k \geq j + 1$ and $\omega^D :
\Omega^D_{k}\,(\Omega^T\,\id)$, we have $\Omega^E\,\id : \Omega^S\,(\Omega^T\,\id)\,\omega^D$ as
the eliminator.
\end{mydefinition}

\begin{theorem}
  $\Omega^T\,\id : \Omega^A_{j+1}$ is initial when lifted to any $k \geq j + 1$ level.
\end{theorem}
\begin{proof}
  $\Omega^T\,\id$ supports elimination by Definition
  \ref{def:iqiit-eliminator}, and elimination is equivalent to initiality by
  Theorem \ref{thm:initiality-induction}.
\end{proof}

\section{Levitation and Bootstrapping}
\label{sec:iqii-levitation}

In this section we adapt the bootstrapping procedure from Section
\ref{sec:closed-levitation} to infinitary signatures.

\subsubsection{Bootstrapping for 2LTT semantics}

If we only want to write down signatures and get their 2LTT-based semantics, a
simplified bootstrapping suffices, which is essentially the same as in Section
\ref{sec:closed-levitation}. We write $\ToS_i : \Set_{i+1}$ for the type of
models where underlying sets are in $\Set_i$ and external indexing is over
$\Ty_0$.  We also have $\bM_{i} : \ToS_{i + 2}$ for the flcwf models where
underlying sets in algebras are in $\Set_i$ and external indexing is over types
in $\Ty_0$.

\begin{mydefinition} The type of \textbf{bootstrap signatures} is defined as follows:
\[
  \ms{BootSig} \defn (i : \ms{Level}) \to (M : \ToS_{i}) \to \Con_{M}
\]
These bootstrap signatures only allow external indexing by types in $\Ty_0$.  We
can write bootstrap signatures and interpret them in $\bM_i$, by applying them
to $\bM_i$.
\end{mydefinition}

\subsubsection{Bootstrapping for term algebras}

Now we reuse the ETT setting from Section \ref{sec:ett-signatures}. We have
$\ToS_{i,j} : \Set_{i+1\lub j+1}$ for the type of models where underlying sets
are in $\Set_i$ and $\Piinf$ and $\Pie$ abstract over $\Set_j$. We also have
$\bM_{i,j} : \ToS_{(i+1\lub j)+1,j}$ as the flcwf models, again with underlying
sets of algebras in $\Set_i$ and external indexing types in $\Set_j$.

\begin{mydefinition} The type of \textbf{bootstrap signatures} at level $j$ is defined as follows.
These may contain external indexing by types in $\Set_j$.
\[
  \ms{BootSig}_j \defn (i\,j : \ms{Level}) \to (M : \ToS_{i,j}) \to \Con_{M}
\]
\end{mydefinition}

\begin{mydefinition}[\textbf{Signature for ToS}]
We define $\ToSSig_j : \ms{BootSig_{j+1}}$ as the bootstrap signature for ToS, where the
described signatures may be indexed by types in $\Set_j$. Like in Section
\ref{sec:closed-levitation}, we use an internal notation. We present an excerpt.
\begingroup
\allowdisplaybreaks
\begin{alignat*}{3}
  & \Con       &&: \U\\
  & \Sub       &&: \Con \to \Con \to \U\\
  & \Ty        &&: \Con \to \U\\
  & \Tm        &&: (\Gamma : \Con) \to \Ty\,\Gamma \to \U\\
  & ...        &&\\
  & \ms{SigU}  &&: \{\Gamma : \Con\} \to \El\,(\Ty\,\Gamma)\\
  & \ms{SigEl} &&: \{\Gamma : \Con\} \to \Tm\,\Gamma\,\ms{SigU} \to \El\,(\Ty\,\Gamma)\\
  & \Piinf     &&: \{\Gamma : \Con\}(A : \Set_j) \toe (A \toinf \Tm\,\Gamma\,\ms{SigU}) \to \El\,(\Tm\,\Gamma\,\ms{SigU})\\
  & ...        &&
\end{alignat*}
Now the interpretation of $\ToSSig_j$ in $\bM_{i,j+1}$ yields the flcwf where
objects are elements of $\ToS_{i,j}$. Note the level bump: $\ToSSig_j$ is in
$\ms{BootSig_{j+1}}$, so we expend one level at each round of self-description.
We get the notion of ToS-induction from $\bM_{i,j+1}$, and we have $\ToS_{i,j}
\leq \ToS_{i+1,j}$ (by definition of $\ToS$ and the rules of subtyping), which
allows us to specify what it means for a model to support elimination into any
universe. Thus we recover all concepts that are used in the term algebra and
eliminator constructions.
\end{mydefinition}
\endgroup

\section{Related Work}
\label{sec:iqii-related-work}

This chapter is based on the publication ``Large and Infinitary Quotient
Inductive-Inductive Types'' \cite{iqiit}. We make the following changes:
\begin{itemize}
  \item We use 2LTT for the flcwf semantics, while the paper only used the cumulative
        ETT setting.
  \item We add the construction of left adjoints and the signature-based semantics in
        Sections \ref{sec:iqii-left-adjoints}-\ref{sec:signature-semantics}.
  \item We add small $\top$ and $\Sigma$ to the ToS, and also their large counterparts
        in Sections \ref{sec:iqii-left-adjoints}-\ref{sec:signature-semantics}.
\end{itemize}

Recall that we show that arbitrary substitutions have left adjoints. Moeneclaey
\cite{semicubical} describes sufficient conditions to have right adjoints as
well: given $t : \Tm\,(\emptycon \ext
\Gamma^A_{\ms{sig}})\,\Gamma^D_{\ms{sig}}$, we have the context
$\emptycon\ext\Gamma^A_{\ms{tos}}\ext\Gamma^S_{\ms{tos}}[\id,\,t]$, and then the
forgetful substitution from this context to $\emptycon\ext\Gamma^A_{\ms{tos}}$
has a right adjoint. The construction that we gave in
Example \ref{ex:param-translation} was given with such a $t$ term as well.

Fiore, Pitts and Steenkamp investigated infinitary QITs in \cite{inf-qits} and
\cite{DBLP:journals/corr/abs-2101-02994}. They introduced two signatures for QW
and QWI types, which generalize W-types and indexed W-types respectively. In the
latter work, they show that these types can be constructed using the WISC axiom
(weakly initial sets of covers).

Essentially algebraic theories generalize to the infinitary cases in a
straightforward way \cite{adamek1994locally}.

Specific examples of infinitary QIITs were introduced in \cite{hottbook}, as
QIITs for Cauchy real numbers, surreal numbers, and cumulative set hierarchies.
In \cite{partiality}, a partiality monad is specified as an infinitary QIITs.

\chapter{Higher Inductive-Inductive Signatures}
\label{chap:hiit}

So far we only considered semantics of signatures where equality constructors
are interpreted as proof-irrelevant equalities, i.e.\ those satisfying UIP. This
inspires the naming of \emph{quotient} inductive-inductive signatures. In
contrast, \emph{higher inductive-inductive} signatures are characterized by
having possibly proof-relevant and iterated equalities in algebras. The natural
setting of HIITs is homotopy type theory (HoTT) \cite{hottbook}, where higher
equalities can be manipulated and constructed in non-trivial ways. We might
think of HIITs as generalizations of QIITs, or alternatively, view QIITs as
set-truncated HIITs.

The theory of HII signatures is fairly similar to the theory of infinitary QII
signatures. The main difference is that the internal $\Id$ type does not support
equality reflection, nor UIP. In fact, infinitary QII signatures already allow
iterated $\Id$, and most HIITs that occur in the literature can be already
expressed using QII signatures. In contrast, the semantics of signatures changes
markedly: the semantic inner theory is now intensional, and $\Id$ is interpreted
as intensional inner equality. This may not seem that dramatic, but note that so
far we have made very heavy use of UIP and inner equality reflection in the
semantics, and now these are not available.

The more general semantics introduces significant complications. As a result, in
the following we shall restrict ourselves to the AMDS fragment of the semantics.
This is sufficient to compute what we mean by induction and initiality (which
has been called ``homotopy initiality'' in the context of HoTT
\cite{sojakova}).

Why do not we go further? The main reason is that the natural semantics is
actually in $(\omega,1)$-categories: we want $(\omega,1)$-categories of
algebras. This requires a different approach and toolset. In particular, in
\cite[Section~9]{hiits} we gave an example that a naive attempt to extend the AMDS
semantics of signatures with the notion of identity morphisms already fails.
The author of this thesis is not versed enough in higher category theory, so we
leave the exposition of the full semantics to future work.

We do note that a higher semantics has been developed by Capriotti and
Sattler. See \cite{capriotti2020higher} for an abstract; the bulk of the work
remains unpublished as of now. In short, Capriotti and Sattler define the ToS in
2LTT, and also use 2LTT to give a model where signatures are higher categories,
specified as complete Segal types. They show that categories of algebras have
finite limits and that initiality is equivalent to induction. Additionally, the
setup yields a structure identity principle for each signature. However,
reductions to simpler type formers are not discussed, nor possible term algebra
constructions. Both of these appear to be far more difficult than in the
quotient setting, and to the author's knowledge there are no concrete proposals
how to approach them.

\subsubsection{The necessity of 2LTT}

2LTT is firmly necessary in the specification of HIITs, and the ToS must live in
the outer layer. The reason is that there is no known way to sensibly
internalize the metatheory of type theories purely inside HoTT. This is the
problem of ``HoTT eating itself'' \cite{hott-eat-itself}. It is also closely
related to the problem of representing semisimplicial types in HoTT. If we can
construct semisimplicial types in an embedded type theory, and interpret that
into non-truncated HoTT types, that would indeed solve the problem. But so far
it has not been solved, or proven impossible to solve. A key original motivation
for 2LTT was precisely to allow construction of semisimplicial types
\cite{twolevel}.

We give a short summary of the problem; see \cite[Section~4]{hiits} for more
discussion. The goal is to have a notion of model of a dependent type theory in
HoTT, such that we have a standard model where contexts are HoTT types.

We may define the notion of model naively using types and equalities, by having
$\Con : \U$, $\Ty : \Con \to \U$, etc.\ and $\ms{idl} : \sigma \circ \id =
\sigma$. However, this does not yield a well-behaved notion of \emph{syntax}. If
we define the syntax as HIIT for the above notion (i.e.\ the initial model),
nothing forces the underlying types to be sets; the HIIT definition freely adds
a large number of non-trivial higher paths. Since the underlying types are not
sets, this syntax does not have decidable equality, by Hedberg's theorem
\cite{hedberg}. This is regardless of what type formers we include.

Alternatively, we may define the notion of model as having homotopy sets for
underlying types. The corresponding HIIT will be in fact a QIIT, where every
inductive sort is set-truncated. While this is better-behaved as syntax, we do
not get a standard model. Contexts in a model cannot be arbitrary types because
in HoTT, types (of a universe) do not form a h-set. In fact, not even h-sets form
a h-set; they form a h-groupoid. So we do not get any reasonable notion of standard
interpretation.

2LTT solves this issue in the following way: the embedded syntax is an outer
QIIT, and equations in the syntax are given as strict (outer) equalities. The
standard inner type model is now possible because in that model all equations
hold strictly, up to inner definitional equality. However, this implies that we
can \emph{only} define strict models; this leads to the following consideration.

\subsubsection{Strict vs.\ weak signatures}

We have an important choice in the semantics: homomorphisms (and sections) can
preserve structure strictly, i.e.\ up to outer equality, or weakly, up to inner
equality. This choice has an impact on the supported ToS features.
\begin{itemize}
\item
  With strict preservation, the semantics does not support an elimination
  rule for $\Id$. The problem is that $\Id$ is necessarily modeled as inner
  equality, but we cannot eliminate from that to outer types, and strict
  equality is an outer type.
\item
  With weak preservation, we do have elimination for $\Id$. However, the
  semantics does not support strict $\beta\eta$ rules in $\Id$, $\Sigma$,
  $\Piinf$ and $\Pi$.  In short, the problem is that $(\El\,a)^M$ and
  $(\El\,a)^S$ are defined as inner equality types, so we need to use inner path
  induction in the semantics of eliminators. This implies that $\beta\eta$-rules
  also hold only up to inner paths, but not definitionally. Thus, in the
  ``weak'' case, we may have $\beta\eta$ only up to internal $\Id$.
\end{itemize}
It makes sense to develop both semantics. Weak morphisms and sections are
useful because they can be defined purely in the inner theory (or in HoTT). Strict
morphisms and sections are useful if we want to specify type formers, since type
theories usually assume strict $\beta$-rules for recursors and eliminators. In
this chapter, we specify theories of signatures and semantics for both cases.

\subsubsection{Metatheory}

We work in 2LTT. We assume that $\Ty_0$ is closed under $\Pi$, $\Sigma$, $\top$
and intensional identity $\blank\!=\!\blank$. We assume the ``based'' path induction principle
\cite[Section~1.12.1]{hottbook}. Assuming $A : \Ty_0$, $x : A$ and $P : (y : A) \to x = y \to \Ty_0$,
we have
\begin{alignat*}{3}
  & \J_P      &&: P\,x\,\refl \to \{y : A\}(p : x = y) \to P\,y\,p\\
  & \J_P\beta &&: \J_P\,\mi{pr}\,\refl \equiv \mi{pr}
\end{alignat*}
The following operations are defined in the standard way
\cite[Section~2]{hottbook}.
\begin{itemize}
  \item Path inversion $\blank^{-1} : x = y \to y = x$.
  \item Path composition $\blank\!\sqcdot\!\blank : x = y \to y = z \to x = z$.
  \item Assuming $P : A \to \Ty_0$, we have transport $\tr_P : x = y \to P\,x \to P\,y$.
  \item Path lifting $\ap : (f : A \to B) \to x = y \to f\,x = f\,y$.
  \item Dependent path lifting $\apd : (f : (x : A) \to B\,x) \to (p : x = y) \to \tr_B\,p\,(f\,x) = f\,y$.
\end{itemize}

\section{Strict Signatures}

\begin{mydefinition}
A \textbf{model of strict ToS} is the same as a model of the theory of
infinitary QII signatures, with the following change: the $\Id$ type former in
$\U$ only supports $\refl$, but no elimination rule or reflection rule.
\end{mydefinition}

We assume that the syntax of ToS exists, and a signature is a context in the
syntax. We could use bootstrap signatures as well, without loss of generality,
as we will not use actual induction on signatures in the following, and we will
also not discuss fine-grained sizing or cumulativity of algebras.

\begin{myexample} The circle is one of the simplest higher inductive types \cite[Section~6.4]{hottbook}. The signature is the following.
\begin{alignat*}{3}
  &\ms{S}^1  &&: \U\\
  &\ms{base} &&: \El\,\ms{S}^1\\
  &\ms{loop} &&: \El\,(\Id\,\ms{base}\,\ms{base})
\end{alignat*}
Note that the circle signature is expressible as a QII signature, but in the QII
semantics the $\ms{loop}$ entry is made trivial by UIP.
\end{myexample}

\subsubsection{Non-examples}
From the HoTT book, all higher-inductive types are supported, except
\begin{itemize}
\item
  The torus \cite[Section~6.6]{hottbook}, since the specification contains $\Id$
  composition, which requires $\Id$ elimination.
\item
  The ``hubs-and-spokes'' HITs \cite[Section~6.7]{hottbook}.  This involves
  abstracting over some external $x : \ms{S}^1$ (a point of the circle), then
  referring to a ToS term which is computed by elimination on $x$. This is also
  not permitted in our setup because signature terms live in the outer theory
  of 2LTT, and external parameters are in $\Ty_0$.

  If instead signatures and external parameters lived in the same theory (like in
  our ETT setup for term algebra constructions of QIITs), this elimination would
  be possible. For HIITs, we cannot do that, since the inner theory cannot
  reasonably internalize the ToS.
\end{itemize}

\subsection{Semantics}

For each signature $\Gamma$, we wish to compute
\begin{alignat*}{3}
  &\Gamma^A &&: \Set\\
  &\Gamma^M &&: \Gamma^A \to \Gamma^A \to \Set\\
  &\Gamma^D &&: \Gamma^A \to \Set\\
  &\Gamma^S &&: (\gamma : \Gamma^A) \to \Gamma^D\,\gamma \to \Set
\end{alignat*}
corresponding respectively to algebras, morphisms, displayed algebras and
sections. Note that all of these return in $\Set$. Morphisms and sections in
particular are forced to return in $\Set$ because they may contain strict
equalities.

The AMDS interpretations can be found in Appendix \ref{app:hii-amds} in a
tabular manner, together with a listing of ToS components. We discuss these
in the following.

In \textbf{algebras} and \textbf{displayed algebras} there is no
complication; all equations hold in these (displayed) models strictly, and we
do not use equations from induction hypotheses anywhere.

In \textbf{morphisms}, note that all term formers returning in $\El$ specify a
strict equation. We write $\refl$ in their definition for brevity, which is
technically correct (by equality reflection), but the definitions may involve
using the strict equalities from induction hypotheses.  $\top^M\,\gamma^M :
\tt_0 \equiv \tt_0$ is trivial, but
\[  (\proj_1\,t)^M\,\gamma^M : a^M\,\gamma^M\,((\proj_1\,t)^A\,\gamma_0) \equiv (\proj_1\,t)^A\,\gamma_1\]
requires us to use
\[t^M\,\gamma^M : (a^M\,\gamma^M\,((\proj_1\,t)^A\,\gamma_0),\,b^M\,(\gamma^M,\,\refl)\,((\proj_2\,t)^A\,\gamma_1))
\equiv t^A\,\gamma_1 \]
Likewise we use $t^M\,\gamma^M$ in the equation for $(\proj_2\,t)^M$.

Also note that the definition for $(\Id\,t\,u)^M\,\gamma^M$ relies on $t^M$ and
$u^M$ for well-typing. The goal is
\begin{alignat*}{3}
  &(\Id\,t\,u)^M\,\gamma^M : (\Id\,t\,u)^A\,\gamma_0 \to (\Id\,t\,u)^A\,\gamma_1\\
  &(\Id\,t\,u)^M\,\gamma^M : t^A\,\gamma_0 = u^A\,\gamma_0 \to t^A\,\gamma_0 = u^A\,\gamma_1
\end{alignat*}
Assuming $p : t^A\,\gamma_0 = u^A\,\gamma_0$, we have $\ap\,(a^M\,\gamma^M)\,p :
a^M\,\gamma^M\,(t^A\,\gamma_0) = a^M\,\gamma^M\,(u^A\,\gamma_0)$, so we rewrite
the sides along $t^M\,\gamma^M : a^M\,\gamma^M\,(t^A\,\gamma_0) \equiv
t^A\,\gamma_1$ and $u^M\,\gamma^M$. The $\ap$ application must stay explicit in
the definition, since inner equalities can be proof-relevant.

We also demonstrate the failure of $\Id$ elimination. It is enough to show that
$\Id$ inversion fails. This would entail the following in the ToS:
\[
  \blank^{-1} : \Tm\,\Gamma\,(\El\,(\Id\,t\,u)) \to \Tm\,\Gamma\,(\El\,(\Id\,u\,t))
\]
In the $\blank^M$ interpretation, we would need to show
\begin{alignat*}{3}
  &(p^{-1})^M\,\gamma^M : \ap\,(a^M\,\gamma^M)\,((p^{-1})^A\,\gamma_0) \equiv ((p^{-1})^A\,\gamma_1)\\
  &(p^{-1})^M\,\gamma^M : \ap\,(a^M\,\gamma^M)\, \big((p^A\,\gamma_0)^{-1}\big) \equiv (p^A\,\gamma_1)^{-1}
\end{alignat*}
We have $p^M\,\gamma^M : \ap\,(a^M\,\gamma^M)\,(p^A\,\gamma_0) \equiv p^A\,\gamma_1$,
so we would need to show
\[
  \ap\,(a^M\,\gamma^M)\, \big((p^A\,\gamma_0)^{-1}\big) \equiv (\ap\,(a^M\,\gamma^M)\,(p^A\,\gamma_0))^{-1}
\]
This is not provable in 2LTT; it is false as a universal statement in the
initial model (syntax) of the inner theory. It holds in the empty context, where
both sides are necessarily equal to $\refl$ by canonicity, but not in arbitrary
contexts. It does hold as an inner equality, by induction on $p^A\,\gamma_0$.

\textbf{Sections} are a mostly mechanical generalization of morphisms, where the
codomain depends on the domain. Note that the $(\Id\,t\,u)^D$ definition is a
path-over-path, and accordingly we have $\apd$ instead of $\ap$ in $(\Id\,t\,u)^S$.

\begin{mydefinition} For some $\Gamma$ signature, notions of \textbf{initiality} and \textbf{induction} are
as follows.
\begin{alignat*}{3}
  &\ms{Initial}\,&&(\gamma : \Gamma^A) &&\defn (\gamma' : \Gamma^A) \to \ms{isContr}\,(\Gamma^M\,\gamma\,\gamma')\\
  &\ms{Inductive}\,&&(\gamma : \Gamma^A) &&\defn (\gamma^D : \Gamma^D\,\gamma) \to \Gamma^S\,\gamma\,\gamma^D
\end{alignat*}
This is the same as Definition \ref{def:induction-predicate}, except we do not have an
flcwf of algebras, so do not have properties that are evident in an flcwf, such
as Theorems \ref{thm:initiality-induction} and \ref{thm:flcwf-term-uniqueness}.
\end{mydefinition}

\begin{myexample}
For the circle signature $\ms{S^1Sig}$, we have the following (disregarding
the leading $\top$ components):
\begin{alignat*}{3}
  &\ms{S^1Sig^A} \equiv (S^1 : \Ty_0)
  \times (\mi{base} : S^1) \times (\mi{loop} : \mi{base} = \mi{base})
\end{alignat*}
\begin{alignat*}{3}
  &\rlap{$\ms{S^1Sig^D}\,(S^1,\,\mi{loop},\,\mi{base}) \equiv$}\\
  &        &&(S^{1D}        &&: S^1 \to \Ty_0)\\
  & \times\,&&(\mi{base^D} &&: S^{1D}\,\mi{base})\\
  & \times\,&&(\mi{loop^D} &&: \tr_{S^{1D}}\,\mi{loop}\,\mi{base^D} = \mi{base^D})\\
  & &&\\
  &\rlap{$\ms{S^1Sig^S}\,(S^1,\,\mi{loop},\,\mi{base})\,(S^{1D},\,\mi{loop^D},\,\mi{base^D}) \equiv$}\\
  &         &&(\mi{S^{1S}}  &&: (s : S^1) \to S^{1D}\,s)\\
  & \times\,&&(\mi{base^S} &&: \mi{S^{1S}}\,\mi{base} \equiv \mi{base^D})\\
  & \times\,&&(\mi{loop^S} &&: \apd\,\mi{S^{1S}}\,\mi{loop} \equiv \mi{loop^D})
\end{alignat*}
\end{myexample}
The computed induction principles are close to what we find in
\cite{hottbook}. The difference is that $\beta$-rules for path constructors are
strict, while in ibid.\ they are up to propositional equality. One reason for
choosing weak $\beta$-rules for paths is that we have $\ap$ and $\apd$ applications on the
left sides of such rules, and it is unconventional to definitionally specify the
behavior of operations which are derived from $\J$. In cubical type theories,
path $\beta$-rules are specified in a more primitive way, so strict computation is
more organic.

Currently, we have semantics in intensional inner theories, but it would be
possible to do the same in cubical inner theories. Intensional TT is clearly
much simpler, and has a wider variety of known models. On the other hand,
cubical type theories support strictly computing transports, so it is possible
that they would support stricter ToS $\beta$-rules in the case of the ``weak''
semantics. We leave this to possible future work.

\section{Weak Signatures}

\subsubsection{Metatheory}

On top of what we had so far in this chapter, we assume \emph{strong function
extensionality} in the inner theory: this means that for each $f,\,g : (a : A)
\to B\,a$, the following function is an equivalence.
\begin{alignat*}{3}
  &\ms{happly} : (f = g) \to ((a : A) \to f\,a = g\,a)\\
  &\ms{happly}\,p\,a \defn \ap\,(\lambda\,f.\,f\,a)\,p
\end{alignat*}
$\ms{funext}$ is obtained as the inverse of $\ms{happly}$. This definition,
unlike the simple assumption of $\ms{funext}$, is well-behaved in intensional
settings \cite[Section~2.9]{hottbook}.

Moreover, we assume two universes $\U_0$ and $\U_1$, such that $\U_0 \leq \U_1
\leq \Ty_0$. We use this to develop semantics which is entirely in the inner
theory: if algebra sorts are in $\U_0$, we need an $\U_1$ on top of that to
accommodate types of algebras.

%% \subsubsection{Theory of Signatures}

\begin{mydefinition}
A \textbf{model of weak ToS} consists of a base cwf (with $\Con$, $\Sub$, $\Ty$
and $\Tm$ returning in $\Set$) extended with certain type formers. We omit all
substitution rules in the following. As before, substitution rules are given
with strict equality. We list type formers below.
\begin{itemize}
\item
  A ``large'' identity type $\ID : \Tm\,\Gamma\,A \to \Tm\,\Gamma\,A \to
  \Ty\,\Gamma$, with the following rules:
  \begin{alignat*}{3}
  & \refl &&: \Tm\,\Gamma\,(\ID\,t\,t)\\
  & \J &&: \{t : \Tm\,\Gamma\,A\}
         (P : \Ty\,(\Gamma\ext(u : A)\ext(p : \ID\,t\,u)))\\
  & &&\to \Tm\,\Gamma\,(P[u \mapsto t,\,p \mapsto \refl])\\
  & &&\to \{u : \Tm\,\Gamma\,A\}
           (p : \Tm\,\Gamma\,(\ID\,t\,u))
  \to \Tm\,\Gamma\,(P[u \mapsto u,\,p \mapsto p])\\
  &\J\beta &&: \J\,b\,\mi{pr}\,\refl \equiv \mi{pr}
  \end{alignat*}

  \begin{notation}
    We may use a name binding notation in the induction motive for $\J$. For
    example, assuming $A : \Ty\,\Gamma$, $B : \Ty\,(\Gamma \ext A)$, $p :
    \Tm\,\Gamma\,(\ID\,t\,u)$ and $\mi{pt} : \Tm\,\Gamma\,(B[\id,\,t])$, we may
    define transport along $p$ as
    \[
      \J\,(x\,p.\,B[\id,\,x])\,\mi{pt}\,p : \Tm\,\Gamma\,(B[\id,\,u])
    \]
    where $x\,p.$ binds the term and path dependencies of the induction motive.
  \end{notation}

\item A universe $\U$ with decoding $\El$.
\item
  $\U$ is closed under a ``small'' identity type $\Id : \Tm\,\Gamma\,(\El\,a)
  \to \Tm\,\Gamma\,(\El\,a) \to \Tm\,\Gamma\,\U$, with elimination principle
  $\J$ targeting any type (not just types in $\U$!). The $\beta$-rule is
  specified with $\ID$.
  \begin{alignat*}{3}
  & \refl &&: \Tm\,\Gamma\,(\El\,(\Id\,t\,t))\\
  & \J &&: \{t : \Tm\,\Gamma\,(\El\,a)\}
         (P : \Ty\,(\Gamma\ext(u : \El\,a)\ext(p : \El\,(\Id\,t\,u))))\\
  & &&\to \Tm\,\Gamma\,(P[u \mapsto t,\,p \mapsto \refl])\\
  & &&\to \{u : \Tm\,\Gamma\,(\El\,a)\}
       (p : \Tm\,\Gamma\,(\El\,(\Id\,t\,u)))
  \to \Tm\,\Gamma\,(P[u \mapsto u,\,p \mapsto p])\\
  &\J\beta &&: \Tm\,\Gamma\,(\ID\,(\J\,b\,\mi{pr}\,\refl)\,\mi{pr})
  \end{alignat*}
\item
  $\U$ is also closed under $\top$, $\Sigma$, and $\Piinf$. All of these are
  specified with equivalences up to $\ID$. These are equivalences in the sense
  of HoTT \cite[Chapter~4]{hottbook}. There are several equivalent formulations
  of equivalence; we pick the bi-invertible definitions here. For $\top$, it is
  enough to have a simplified specification as $\top\!\eta :
  \Tm\,\Gamma\,(\ID\,t\,\tt)$. $\Sigma$ is specified as follows.
  \begin{alignat*}{3}
   &\ms{\blank\!,\!\blank} &&:
      (t : \Tm\,\Gamma\,(\El\,a)) \times \Tm\,\Gamma\,(\El\,(b[\id,\,t]))
    \to \Tm\,\Gamma\,(\El\,(\Sigma\,a\,b))\\
   &\ms{proj}         &&: \Tm\,\Gamma\,(\El\,(\Sigma\,a\,b))
    \to (t : \Tm\,\Gamma\,(\El\,a)) \times \Tm\,\Gamma\,(\El\,(b[\id,\,t]))\\
   &\ms{proj'}        &&: \Tm\,\Gamma\,(\El\,(\Sigma\,a\,b))
     \to (t : \Tm\,\Gamma\,(\El\,a)) \times \Tm\,\Gamma\,(\El\,(b[\id,\,t]))\\
   &\beta_1 &&: \Tm\,\Gamma\,(\ID\,(\proj_1\,(t,\,u))\,t)\\
   &\beta_2 &&: \Tm\,\Gamma\,(\ID\,((\J\,(x\,\_.\,(\El\,b)[\id,\,x])\,(\proj_2\,(t,\,u))\,\beta_1)\,u))\\
   &\eta    &&: \Tm\,\Gamma\,(\ID\,(\proj_1'\,t,\,\proj_2'\,t)\,t)
  \end{alignat*}
  We write $\proj_i$ and $\proj'_i$ for composing metatheoretic projections with
  ToS projections. The additional $\proj'$ component is required to get a
  bi-invertible equivalence. Also note that $\beta_2$ is only well-typed up to
  $\beta_1$, so we need to use a transport in the specification.

  $\Piinf : (\mi{Ix} : \U_0) \to (\mi{Ix} \to \Tm\,\Gamma\,\U) \to
  \Tm\,\Gamma\,\U$ is specified below.
  \begingroup
  \allowdisplaybreaks
  \begin{alignat*}{3}
    &\appinf &&: \Tm\,\Gamma\,(\El\,(\Piinf\,\mi{Ix}\,b)) \to ((i : \mi{Ix}) \to \Tm\,\Gamma\,(\El\,(b\,i)))\\
    &\laminf &&: ((i : \mi{Ix}) \to \Tm\,\Gamma\,(\El\,(b\,i))) \to \Tm\,\Gamma\,(\El\,(\Piinf\,\mi{Ix}\,b))\\
    &\laminfprime &&: ((i : \mi{Ix}) \to \Tm\,\Gamma\,(\El\,(b\,i))) \to \Tm\,\Gamma\,(\El\,(\Piinf\,\mi{Ix}\,b))\\
    &\ms{\beta} &&: \Tm\,\Gamma\,(\ID\,(\appinf\,(\laminf\,t)\,i)\,(t\,i))\\
    &\ms{\eta}  &&: \Tm\,\Gamma\,(\ID\,(\laminfprime\,(\appinf\,t))\,t)
  \end{alignat*}
  \endgroup

  Why have equivalences in the specification of models, would it be enough to
  have isomorphisms? We choose equivalences because they yield better-behaved
  models, and they do not make it any harder to construct models, since we can
  always construct the required equivalences from isomorphisms
  \cite[Chapter~4]{hottbook}.

  \item
  Internal product type $\Pi : (a : \Tm\,\Gamma\,\U) \to \Ty\,(\Gamma\ext
  \El\,a) \to \Ty\,\Gamma$, with the specifying equivalence given up to $\ID$,
  analogously as for $\Sigma$ and $\Piinf$:
  \[
    (\app,\,\lam,\,\lam') : \Tm\,\Gamma\,(\Pi\,a\,B) \simeq \Tm\,(\Gamma\ext \El\,a)\,B
  \]

  \item External product type $\Pie : (\mi{Ix} : \U_0) \to (\mi{Ix} \to
  \Ty\,\Gamma) \to \Ty\,\Gamma$, specified as a strict $\Set$ isomorphism:
  \[(\appe,\,\lame) : \Tm\,\Gamma\,(\Pie\,\mi{Ix}\,B) \simeq ((i : \mi{Ix}) \to \Tm\,\Gamma\,(B\,i))\]
\end{itemize}
\end{mydefinition}

\noindent To give a short summary of changes compared to strict signatures:
\begin{enumerate}
  \item Types are closed under an extra $\ID$ type former which has a strict $\beta$-rule.
  \item We can eliminate from $\Id$ to proper types, but with a weak $\beta$-rule.
  \item $\Sigma$ and $\Piinf$ support eliminators, but with weak $\beta$-rules.
\end{enumerate}

\begin{myexample}
The torus is now expressible thanks to path elimination in signatures. We define
$\blank\!\sqcdot\!\blank$ as path composition for $\Id$ in the evident way.
\begin{alignat*}{3}
  &\ms{T}^2  &&: \U\\
  &\ms{b}    &&: \El\,\ms{T}^2\\
  &\ms{p}    &&: \El\,(\Id\,\ms{b}\,\ms{b})\\
  &\ms{q}    &&: \El\,(\Id\,\ms{b}\,\ms{b})\\
  &\ms{t}    &&: \El\,(\Id\,(\p \sqcdot \q)\,(\q \sqcdot \p))
\end{alignat*}
We could also use $\ID$ instead of $\Id$ and get equivalent semantics.
\end{myexample}

\begin{myexample}
The $\ID$ type lets us express ``sort equivalences''. For example, a signature
for integers can be compactly written as follows \cite{hit-integers}:
\begin{alignat*}{3}
  &\ms{Int}  &&: \U\\
  &\zero     &&: \El\,\ms{Int}\\
  &\p        &&: \ID\,\ms{Int}\,\ms{Int}
\end{alignat*}
We get the $\suc$ constructor by coercing along $\p$, and predecessors by
coercing backwards.

Recall that in Chapter \ref{chap:iqiit} we dropped sort equations because of
their non-fibrancy in the semantics. In contrast, there is no issue with sort
equations here. Sort equations simply become inner paths between types in the
semantics; if we assume univalence in the inner theory, such paths are
equivalent to type equivalences. Hence, sort equations in HIITs can be viewed as
shorthands for sort equivalences. Without sort equations, it is still possible
to write equivalences in signatures, using any of the standard definitions
\cite[Chapter~4]{hottbook}.
\end{myexample}

\subsection{Semantics}

We do not repeat the tables for the strict ToS semantics in Appendix
\ref{app:hii-amds}, as much of it remains essentially the same in the weak
case. We consider the components of the model in order, highlighting relevant
changes and points of interest.

\begin{notation}
We may omit induction motives in $\tr$ and $\J$ in the following, as they will
often get excessively verbose. So we may write $\tr\,p\,\mi{px} : P\,y$ for $p :
x = y$ and $\mi{px} : P\,x$, and use $\J\,\mi{pr}\,p$ similarly.
\end{notation}

\subsubsection{Cwf}

A notable change here is that the entirety of the semantics is now in the inner
theory. This means that the interpretation functions of contexts and types all
return in $\U_1$, e.g.\ $\Gamma^A : \U_1$ and $\Gamma^M : \Gamma^A \to \Gamma^A
\to \U_1$. Accordingly, we use type formers in $\U_1$ to interpret structure in
the base cwf, e.g.\ $\top^A \defn \top$, where the $\top$ on the right is in
$\U_1$. The only change though is the move from $\Set$ to $\U_1$, all
definitions are essentially the same.

\subsubsection{$\bs{\ID}$}

The new $\ID$ type former is interpreted as pointwise equality of semantic
terms. We assume $t,\,u : \Tm\,\Gamma\,A$.
\begin{alignat*}{3}
  &(\ID\,t\,u)^A\,\gamma            &&\defn t^A\,\gamma = u^A\,\gamma \\
  &(\ID\,t\,u)^M\,p_0\,p_1\,\gamma^M &&\defn \tr\,p_1\,(\tr\,p_0\,(t^M\,\gamma^M)) = u^M\,\gamma^M\\
  &(\ID\,t\,u)^D\,p\,\gamma^D       &&\defn \tr_{(\lambda\,x.\,A^D\,x\,\gamma^D)}\,p\,(t^D\,\gamma^D) = u^D\,\gamma^D\\
  &(\ID\,t\,u)^S\,p\,p^D\,\gamma^S  &&\defn \tr\,p^D\,(\J\,(t^S\,\gamma^S)\,p) = u^S\,\gamma^S
\end{alignat*}
Above, we dropped induction motives in $\tr$ and $\J$ in $\blank^M$ and
$\blank^S$. For illustration, the more explicit definitions are:
\begin{alignat*}{3}
  &(\ID\,t\,u)^M\,p_0\,p_1\,\gamma^M \defn\\
  &\hspace{1em}\tr_{(\lambda\,x.\,A^M\,x\,(t^A\,\gamma_1)\,\gamma^M)}\,p_1\, (\tr_{(\lambda\,x.\,A^M\,(t^A\,\gamma_1)\,x\,\gamma^M)}\,p_0\,(t^M\,\gamma^M)) = u^M\,\gamma^M\\
  &(\ID\,t\,u)^S\,p\,p^D\,\gamma^S \defn\\
  &\hspace{1em}\tr_{(\lambda\,x.\,A^S\,x\,(u^A\,\gamma)\,\gamma^S)}\,p^D\\
  &\hspace{2em}(\J_{(\lambda\,y\,p.\,A^S\,y\,(\tr_{(\lambda\,x.\,A^D\,x\,\gamma^D)}\,p\,(t^D\,\gamma^D)))}\,(t^S\,\gamma^S)\,p) = u^S\,\gamma^S
\end{alignat*}
From now on, we shall generally avoid this amount of detail in motives.

$\refl$ is interpreted as pointwise $\refl$-s:
\begin{alignat*}{3}
  &\refl^A\,&&\_ \defn \refl\\
  &\refl^M\,&&\_ \defn \refl\\
  &\refl^D\,&&\_ \defn \refl\\
  &\refl^S\,&&\_ \defn \refl
\end{alignat*}
Let us look at $\J$ for $\ID$ now. It is helpful to temporarily consider a bundled
AMDS model instead of the four interpretation maps. Then, we have the following
equivalence up to $\blank\!=\!\blank$:
\[\Tm_{\ms{AMDS}}\,\Gamma\,(\ID_{\ms{AMDS}}\,t\,u) \simeq (t = u)\]
This follows from function extensionality and the characterization of
equivalence for inner $\Sigma$ \cite[Section~2.7]{hottbook}. Thus, semantic
$\ID$ is the same as equality of semantic terms. It follows that everything in
the inner theory respects $\ID$, so we can certainly define the semantic $\J$
for $\ID$.

The actual definition of $\J$ involves doing induction on all paths that
are available as induction hypotheses.
\begin{alignat*}{3}
  &(\J\,P\,\mi{pr}\,p)^A\,\gamma   &&\defn \J\,(\mi{pr}^A\,\gamma)\,(p^A\,\gamma)\\
  &(\J\,P\,\mi{pr}\,p)^M\,\gamma^M &&\defn
    \J\,(\J\,(\J\,(\mi{pr}^M\,\gamma^M)\,(p^A\,\gamma_1))\,(p^A\,\gamma_0))\,(p^M\,\gamma^M)\\
  &(\J\,P\,\mi{pr}\,p)^D\,\gamma^D &&\defn \J\,(\J\,(\mi{pr}^D\,\gamma^D)\,(p^A\,\gamma))\,(p^D\,\gamma^D)\\
  &(\J\,P\,\mi{pr}\,p)^S\,\gamma^S &&\defn
    \J\,(\J\,(\J\,(\mi{pr}^S\,\gamma^S)\,(p^A\,\gamma))\,(p^D\,\gamma^D))\,(p^S\,\gamma^S)
\end{alignat*}
The strict $\beta$-rule for $\J$ is supported, as the above definition computes
everywhere when $p$ is $\refl$.

\subsubsection{Universe}
We have the following changes. First, the interpretations of $\U$ now return in
$\U_0$:
\begin{alignat*}{3}
  & \U^A\,\gamma &&\defn \U_0\\
  & \U^D\,a\,\gamma^D &&\defn a \to \U_0
\end{alignat*}
Second, in $\El$, morphisms and sections are given by inner equality:
\begin{alignat*}{3}
  & (\El\,a)^M\,\alpha_0\,\alpha_1\,\gamma^M &&\defn a^M\,\gamma^M\,\alpha_0 = \alpha_1\\
  & (\El\,a)^M\,\alpha\,\alpha^D\,\gamma^D &&\defn a^S\,\gamma^S\,\alpha = \alpha^D
\end{alignat*}

\subsubsection{$\bs{\Id}$}

In this identity type, $\blank^A$ and $\blank^D$ are pointwise equality as
usual, and $\blank^M$ and $\blank^S$ complete squares of equalities. We assume
$t,\,u : \Tm\,\Gamma\,(\El\,a)$.
\begin{alignat*}{3}
  & (\Id\,t\,u)^A\,\gamma   &&\defn t^A\,\gamma = u^A\,\gamma\\
  & (\Id\,t\,u)^M\,\gamma^M &&\defn \lambda\,(p : t^A\,\gamma_0 = u^A\,\gamma_0).\,(t^M\,\gamma^M)^{-1} \sqcdot \ap\,(a^M\,\gamma^M)\,p \sqcdot u^M\,\gamma^M\\
  & (\Id\,t\,u)^D\,\gamma^D &&\defn \lambda\,(p : t^A\,\gamma = u^A\,\gamma).\,\tr_{(a^D\,\gamma^D)}\,(t^D\,\gamma^D) = u^D\,\gamma^D\\
  & (\Id\,t\,u)^S\,\gamma^S &&\defn \lambda\,(p : t^A\,\gamma = u^A\,\gamma).\,\ap\,(\tr_{(a^D\,\gamma^D)}\,p)\,(t^S\,\gamma^S)^{-1} \sqcdot \apd\,(a^S\,\gamma^S)\,p \sqcdot u^S\,\gamma^S
\end{alignat*}
We have $\refl^A\,\_ \defn \refl$ and $\refl^D\,\_ \defn \refl$. For $\refl^M\,\gamma^M$, the goal type is
\[
  (t^M\,\gamma^M)^{-1} \sqcdot t^M\,\gamma^M = \refl
\]
which is one of the groupoid laws for paths \cite[Section~2.1]{hottbook}. We
have a more dependent variant as goal type for $\refl^S\,\gamma^S$:
\[
  \ap\,(\lambda\,x.\,x)\,(t^S\,\gamma^S)^{-1} \sqcdot t^S\,\gamma^S = \refl
  \]
This again follows from groupoid laws and the functoriality of $\ap$.

It is still the case that
$\Tm_{\ms{AMDS}}\,\Gamma\,(\El_{\ms{AMDS}}\,(\Id_{\ms{AMDS}}\,t\,u)) \simeq (t =
u)$ up to $\blank\!=\!\blank$. Although $(\Id\,t\,u)^M$ and $(\Id\,t\,u)^S$ do
not express equality of $t$ and $u$, we do get the component-wise equalities if
we apply $\El$. We have that
\[
(\El\,(\Id\,t\,u))^M\,p_0\,p_1\,\gamma^M
\equiv ((t^M\,\gamma^M)^{-1} \sqcdot \ap\,(a^M\,\gamma^M)\,p_0 \sqcdot u^M\,\gamma^M = p_1)
\]
We can rearrange the definition to make it more apparent that this is an equality
of $t^M\,\gamma^M$ and $u^M\,\gamma^M$, which is well-typed up to $p_0$ and $p_1$.
\[
  \ap\,(a^M\,\gamma^M)\,p_0 \sqcdot u^M\,\gamma^M = t^M\,\gamma^M \sqcdot p_1
\]
Thus, we can again expect that $\J$ is definable for $\Id$. However, the actual
definitions get highly technical in the $\blank^M$ and $\blank^S$ cases, as we
have to repeatedly transport along higher paths to make certain eliminations
well-typed. We refer the reader to the Agda formalization \cite{ak-thesis-agda}
for these definitions. In the $\blank^A$ and $\blank^D$ cases, the definitions
are simple enough:
\begin{alignat*}{3}
  &(\J\,P\,\mi{pr}\,p)^A\,\gamma   &&\defn \J\,(\mi{pr}^A\,\gamma)\,(p^A\,\gamma)\\
  &(\J\,P\,\mi{pr}\,p)^D\,\gamma^D &&\defn \J\,(\J\,(\mi{pr}^D\,\gamma^D)\,(p^A\,\gamma))\,(p^D\,\gamma^D)
\end{alignat*}
Regarding the $\beta$-rule, note that $\refl^M$ and $\refl^S$ are not defined as
$\refl$, but rather by induction on $t^M\,\gamma^M$ and
$t^S\,\gamma^S$. Therefore, if we apply $\J$ to $\refl$, the $\blank^M$ and
$\blank^S$ components do not strictly compute.

\subsubsection{$\bs{\top}$}

$\top$ is unchanged. $\tt^M$ and $\tt^S$ could possibly change (since $\El$ has
changed, and $\tt : \Tm\,\Gamma\,(\El\,\top)$), but they are still definable
with $\refl$-s.

\subsubsection{$\bs{\Sigma}$}

Pairing and the projections change in $\Sigma$; now their $\blank^M$ and
$\blank^S$ cases return proof-relevant inner equalities. In pairing, we do path
induction on hypotheses:
\begin{alignat*}{3}
  &(t,\,u)^M\,\gamma^M &&\defn \J\,(\J\,\refl\,(t^M\,\gamma^M))\,(u^M\,\gamma^M)\\
  &(t,\,u)^S\,\gamma^S &&\defn \J\,(\J\,\refl\,(t^S\,\gamma^S))\,(u^S\,\gamma^S)
\end{alignat*}
In $\proj_1$, we use $\ap\,\proj_1$ on path hypotheses:
\begin{alignat*}{3}
  &(\proj_1\,t)^M\,\gamma^M &&\defn \ap\,\proj_1\,(t^M\,\gamma^M)\\
  &(\proj_1\,t)^S\,\gamma^S &&\defn \ap\,\proj_1\,(t^S\,\gamma^S)
\end{alignat*}
In $\proj_2$, the definitions could be given using $\apd\,\proj_2$, but the
result type does not immediately line up, so we can just do direct path induction.
\begin{alignat*}{3}
  &(\proj_2\,t)^M\,\gamma^M &&\defn \J\,\refl\,(t^M\,\gamma^M)\\
  &(\proj_2\,t)^S\,\gamma^S &&\defn \J\,\refl\,(t^S\,\gamma^S)
\end{alignat*}
$\proj_1'$ and $\proj_2'$ (required by the bi-invertible specification) are defined
the same way. We do not have strict $\beta\eta$-rules. For example:
\[ (\proj_1\,(t,\,u))^M\,\gamma^M \equiv \ap\,\proj_1\,(\J\,(\J\,\refl\,(t^M\,\gamma^M))\,(u^M\,\gamma^M)) \not\equiv t^M\,\gamma^M \]
We still get $(\proj_1\,(t,\,u))^M\,\gamma^M = t^M\,\gamma^M$ by path induction on
$t^M\,\gamma^M$ and $u^M\,\gamma^M$, and similarly in other cases, so $\Sigma$ in the ToS
does support the specifying equivalence.

\subsubsection{$\ms{\Piinf}$}
Again, the $\blank^M$ and $\blank^S$ cases change in term formers. Application
is given by $\happly$:
\begin{alignat*}{3}
  &(\appinf\,t\,i)^M\,\gamma^M &&\defn \happly\,(t^M\,\gamma^M)\,i\\
  &(\appinf\,t\,i)^S\,\gamma^S &&\defn \happly\,(t^S\,\gamma^S)\,i
\end{alignat*}
Abstraction is by $\funext$:
\begin{alignat*}{3}
  &(\laminf\,t)^M\,\gamma^M &&\defn \funext\,(\lambda\,i.\,(t\,i)^M\,\gamma^M)\\
  &(\laminf\,t)^S\,\gamma^S &&\defn \funext\,(\lambda\,i.\,(t\,i)^S\,\gamma^S)
\end{alignat*}
Thus, weak $\beta\eta$-rules for $\Piinf$ follow from strong function extensionality.

\subsubsection{$\ms{\Pi}$}
We need to use explicit path induction in $\app^M$ and $\app^S$:
\begin{alignat*}{5}
  & (\app\,t)^M\,(\gamma^M,\,\alpha^M) &&\defn \J\,(t^M\,\gamma^M\,\alpha_0)\,\alpha^M &&\hspace{1em}\text{where}\hspace{1em} \alpha^M &&: a^M\,\gamma^M\,\alpha_0 = \alpha_1\\
  & (\app\,t)^S\,(\gamma^S,\,\alpha^S) &&\defn \J\,(t^S\,\gamma^S\,\alpha)\,\alpha^S &&\hspace{1em}\text{where}\hspace{1em} \alpha^S &&: a^S\,\gamma^S\,\alpha = \alpha^D
\end{alignat*}
In contrast, $\lam$ does not change. $\beta\eta$-rules are given by replaying
the path inductions on $\app^M$ and $\app^S$.

\subsubsection{$\ms{\Pie}$}

The interpretation of $\ms{\Pie}$ is unchanged. This concludes the AMDS semantics of weak signatures.

\section{Discussion \& Related Work}

\subsection{Evaluation}

The main advantage of the signatures in the current chapter is their generality.
We cover almost every higher inductive definition in the literature, and do so
in a direct manner, with minimal encoding overhead.

It is also possible to mechanically check validity of signatures and compute
AMDS interpretations. The current author has written a Haskell program which
takes as input a weak HII signature, and outputs ADS interpretations as
well-formed Agda source code \cite{hiit-sig-program}. The syntax is a bit more
restricted than what we have in this chapter, and the program does not compute
morphisms; but it is clear that the deficiencies would be straightforward to
patch up.

On the other hand, we note that our semantics is in a minimal intensional
theory, a fragment of the ``book'' version of homotopy type theory. This setting
supports neither computational univalence nor computational higher inductive
types. If our goal is to add computationally adequate HIITs to a theory (and
eventually to its implementation), the current chapter is not immediately
applicable. As we mentioned in Section \ref{sec:implementation}, in a cubical
setting we would need to reformulate both signatures and semantics. However, the
current work should be still helpful as a guideline, and a provide a point of
comparison and validation.

\subsection{Related Work}
\label{sec:hii-related-work}

This chapter is based on ``A Syntax for Higher Inductive-Inductive Types''
\cite{hiit} and ``Signatures and Induction Principles for Higher
Inductive-Inductive Types'' \cite{hiits}, both by Ambrus Kaposi and the current
author. The latter is an extended journal version of the former. In this
chapter, we extend and refine these sources in the following ways.
\begin{itemize}
\item
    We use 2LTT. In the papers, we instead used a custom syntactic translation:
    the theory of signatures was an ad hoc mixture of the inner and outer
    theory, and the AMDS interpretations were syntactic translations targeting
    the inner theory. The setup turns out to be mostly the same as here; but
    2LTT brings a lot of clarity and convenience.
\item
    We add the strict/weak signature distinction. The papers only considered
    weak signatures and semantics.
\item
    We improve on the specification of signatures. The papers had a small $\Id$
    type with elimination only to $\U$, not to arbitrary types. The journal
    version also had a second identity type, but only for sort equations,
    i.e.\ it expressed only equality of inhabitants of $\U$.

    The small and large identity types in this chapter are more expressive; the
    weaker definitions in the paper were just oversights.

    The papers also omitted eliminators of type formers in weak signatures, and
    thus their $\beta\eta$ rules, and they did not have $\top$ or
    $\Sigma$. However, this was done mostly for the sake of brevity, as these
    extra features are not really used in any HIIT signature in the literature.
    It makes more sense to include the extras here, to match infinitary QII
    signatures as much as possible.
\end{itemize}

The homotopy type theory book \cite{hottbook} introduced numerous higher
inductive types and developed their use cases, but it did not give a theory of
signatures, nor discussed semantics.

Sojakova \cite{sojakova} specified a class of HITs called W-suspensions
(building on W-types), and proved the equivalence of induction and homotopy
initiality, working internally to an intensional type theory.

Lumsdaine and Shulman gave a general specification of models of type theories
supporting higher inductive types \cite{lumsdaineShulman}. They gave a more
semantic specification of algebras, as algebras of a cell monad, and
characterized the class of models which support initial algebras. They did not
cover indexed families or induction-induction.

Dybjer and Moeneclaey \cite{moeneclaey} gave signatures for class of finitary
HITs with up to 2-dimensional path constructors, and built semantics in
groupoids.

Coquand, Huber and Mörtberg \cite{cubicalhits} specified syntax for a cubical
type theory which supports several HITs (sphere, torus, suspensions,
truncations, pushouts) and built semantics in cubical sets.

Cavallo and Harper \cite{cubicalcomptt} specify HITs which support indexed
families and arbitrary higher paths, although not induction-induction. They
provide semantics in a PER (partial equivalence relation) realizability setting.

Cubical Agda \cite{cubicalagda} is the principal proof assistant which natively
supports computational univalence and HITs. Its implementation of pattern
matching, mutual inductive definitions, termination checking and strict
positivity checking yields of a large class of higher \emph{inductive-inductive}
types. However, there is no compact theory of signatures (valid specifications
fall out from positivity/termination checking) nor a categorical semantics.

\raggedbottom
\appendix

\chapter{AMDS interpretation of FQII signatures}
\label{app:fqii-amds}

This appendix supplements Chapter \ref{chap:fqiit}. It contains the AMDS
interpretation for finitary QII signatures. We omit substitution and
$\beta\eta$-rules. We also omit the $\Tm_0$ decoding operation of two-level type
theory.

\pagebreak

\subsubsection{Components of ToS (without substitution and $\beta\eta$-rules)}

\begin{alignat*}{3}
  &\Con &&: \Set\\
  &\Sub &&: \Con \to \Con \to \Set\\
  &\Ty  &&: \Con \to \Set\\
  &\Tm  &&: (\Gamma : \Con) \to \Ty\,\Gamma \to \Set \\
  &\emptycon &&: \Con\\
  &\epsilon &&: \Sub\,\Gamma\,\emptycon\\
  &\id &&: \Sub\,\Gamma\,\Gamma\\
  &\blank\!\circ\!\blank &&: \Sub\,\Delta\,\Xi \to \Sub\,\Gamma\,\Delta \to \Sub\,\Gamma\,\Xi\\
  &\blank\![\blank] &&: \Ty\,\Delta \to \Sub\,\Gamma\,\Delta \to \Ty\,\Gamma\\
  &\blank\![\blank] &&: \Tm\,\Delta\,A \to (\sigma : \Sub\,\Gamma\,\Delta) \to \Tm\,\Gamma\,(A[\sigma])\\
  &\p &&: \Sub\,(\Gamma\ext A)\,\Gamma\\
  &\q &&: \Tm\,(\Gamma\ext A)\,(A[\p])\\
  &(\blank\!,\!\blank) &&: (\sigma : \Sub\,\Gamma\,\Delta) \to \Tm\,\Gamma\,(A[\sigma]) \to \Sub\,\Gamma\,(\Delta\ext A)\\
  &\U &&: \Ty\,\Gamma\\
  &\El &&: \Tm\,\Gamma\,\U \to \Ty\,\Gamma\\
  &\Id &&: \Tm\,\Gamma\,A \to \Tm\,\Gamma\,A \to \Ty\,\Gamma\\
  &\refl &&: \Tm\,\Gamma\,(\Id\,t\,t)\\
  &\reflect &&: \Tm\,\Gamma\,(\Id\,t\,u) \to t \equiv u\\
  &\Pi &&: (a : \Tm\,\Gamma\,\U) \to \Ty\,(\Gamma\ext \El\,a) \to \Ty\,\Gamma\\
  &\app &&: \Tm\,\Gamma\,(\Pi\,a\,B) \to \Tm\,(\Gamma \ext \El\,a)\,B\\
  &\lam &&: \Tm\,(\Gamma \ext \El\,a)\,B \to \Tm\,\Gamma\,(\Pi\,a\,B)\\
  &\Pie &&: (\mi{Ix} : \Ty_0) \to (\mi{Ix} \to \Ty\,\Gamma) \to \Ty\,\Gamma\\
  &\appe &&: \Tm\,\Gamma\,(\Pie\,\mi{Ix}\,B) \to (i : \mi{Ix}) \to \Tm\,\Gamma\,(B\,i)\\
  &\lame &&: ((i : \mi{Ix}) \to \Tm\,\Gamma\,(B\,i)) \to \Tm\,\Gamma\,(\Pie\,\mi{Ix}\,B)
\end{alignat*}

\subsubsection{Algebras}

\begin{alignat*}{3}
  &\blank^A &&: \Con \to \Set\\
  &\blank^A &&: \Sub\,\Gamma\,\Delta \to \Gamma^A \to \Delta^A\\
  &\blank^A &&: \Ty\,\Gamma \to \Gamma^A \to \Set\\
  &\blank^A &&: \Tm\,\Gamma\,A \to (\gamma : \Gamma^A) \to A^A\,\gamma\\
  &\emptycon^A &&\defn \top\\
  &\epsilon^A\,\gamma &&\defn \tt\\
  &\id^A\,\gamma &&\defn \gamma\\
  &(\sigma \circ \delta)^A\,\gamma &&\defn \sigma^A\,(\delta^A\,\gamma)\\
  &(\Gamma \ext A)^A &&\defn (\gamma : \Gamma^A) \times A^A\,\gamma\\
  &(A[\sigma])^A\,\gamma &&\defn A^A\,(\sigma^A\,\gamma)\\
  &(t[\sigma])^A\,\gamma &&\defn t^A\,(\sigma^A\,\gamma)\\
  &\p^A\,(\gamma,\,\alpha) &&\defn \gamma\\
  &\q^A\,(\gamma,\,\alpha) &&\defn \alpha\\
  &(\sigma,\,t)^A\,\gamma &&\defn (\sigma^A\,\gamma,\,t^A\,\gamma)\\
  &\U^A\,\gamma &&\defn \Ty_0\\
  &(\El\,a)^A\,\gamma &&\defn a^A\,\gamma\\
  &(\Id\,t\,u)^A\,\gamma &&\defn t^A\,\gamma \equiv u^A\,\gamma\\
  &\refl^A\,\gamma &&\defn \refl : t^A\,\gamma \equiv t^A\,\gamma\\
  &(\reflect\,p)^A &&\defn \funext\,(\lambda\,\gamma.\,p^A\,\gamma)\\
  &(\Pi\,a\,B)^A\,\gamma &&\defn (\alpha : a^A\,\gamma) \to B^A\,(\gamma,\,\alpha)\\
  &(\app\,t)^A\,(\gamma,\,\alpha) &&\defn t^A\,\gamma\,\alpha\\
  &(\lam\,t)^A\,\gamma &&\defn \lambda\,\alpha.\,t^A\,(\gamma,\,\alpha)\\
  &(\Pie\,\mi{Ix}\,B)^A\,\gamma &&\defn (i : \mi{Ix}) \to (B\,i)^A\,\gamma\\
  &(\appe\,t\,i)^A\,\gamma &&\defn t^A\,\gamma\,i\\
  &(\lame\,t)^A\,\gamma &&\defn \lambda\,i.\,(t\,i)^A\,\gamma
\end{alignat*}

\subsubsection{Morphisms}

\begin{alignat*}{3}
  &\blank^M &&: (\Gamma : \Con) \to \Gamma^A \to \Gamma^A \to \Set\\
  &\blank^M &&: (\sigma : \Sub\,\Gamma\,\Delta) \to \Gamma^M\,\gamma_0\,\gamma_1 \to \Delta^M\,(\sigma^A\,\gamma_0)\,(\sigma^A\,\gamma_1)\\
  &\blank^M &&: (A : \Ty\,\Gamma) \to A^A\,\gamma_0 \to A^A\,\gamma_1 \to \Gamma^M\,\gamma_0\,\gamma_1 \to \Set\\
  &\blank^M &&: (t : \Tm\,\Gamma\,A) \to (\gamma^M : \Gamma^M\,\gamma_0\,\gamma_1) \to A^M\,(t^A\,\gamma_0)\,(t^A\,\gamma_1)\,\gamma^M\\
  &\emptycon^M\,\gamma_0\,\gamma_1 &&\defn \top \\
  &\epsilon^M\,\gamma^M &&\defn \tt\\
  &\id^M\,\gamma^M &&\defn \gamma^M\\
  &(\sigma \circ \delta)^M\,\gamma^M &&\defn \sigma^M\,(\delta^M\,\gamma^M)\\
  &(\Gamma \ext A)^M\,(\gamma_0,\,\alpha_0)\,(\gamma_1,\,\alpha_1) &&\defn
    (\gamma^M : \Gamma^M\,\gamma_0\,\gamma_1) \times A^M\,\alpha_0\,\alpha_1\,\gamma^M\\
  &(A[\sigma])^M\,\alpha_0\,\alpha_1\,\gamma^M &&\defn A^M\,\alpha_0\,\alpha_1\,(\sigma^M\,\gamma^M)\\
  &(t[\sigma])^M\,\gamma^M &&\defn t^M\,(\sigma^M\,\gamma^M)\\
  &\p^M\,(\gamma^M,\,\alpha^M) &&\defn \gamma^M\\
  &\q^M\,(\gamma^M,\,\alpha^M) &&\defn \alpha^M\\
  &(\sigma,\,t)^M\,\gamma^M &&\defn (\sigma^M\,\gamma^M,\,t^M\,\gamma^M)\\
  &\U^M\,a_0\,a_1\,\gamma^M &&\defn a_0 \to a_1\\
  &(\El\,a)^M\,\alpha_0\,\alpha_1\,\gamma^M &&\defn a^M\,\gamma^M\,\alpha_0 \equiv \alpha_1 \\
  &(\Id\,t\,u)^M\,p_0\,p_1\,\gamma^M &&\defn t^M\,\gamma^M \equiv u^M\,\gamma^M\\
  &\refl^M\,\gamma^M &&\defn \refl : t^M\,\gamma^M \equiv t^M\,\gamma^M\\
  &(\reflect\,p)^M &&\defn \funext\,(\lambda\,\gamma^M.\,p^M\,\gamma^M)\\
  &(\Pi\,a\,B)^M\,t_0\,t_1\,\gamma^M &&\defn (\alpha : a^A\,\gamma_0) \to B^M\,(t_0\,\alpha)\,(t_1\,(a^M\,\gamma^M\,\alpha))\,(\gamma^M,\,\refl)\\
  &(\app\,t)^M\,(\gamma^M,\,\alpha^M) &&\defn t^M\,\gamma^M\,\alpha_0\hspace{1em}\text{where}\hspace{1em} \alpha^M : a^M\,\gamma^M\,\alpha_0 \equiv \alpha_1\\
  &(\lam\,t)^M\,\gamma^M &&\defn \lambda\,\alpha.\,t^M\,(\gamma^M,\,\refl)\hspace{1em}\text{where}\hspace{1em} \refl : a^M\,\gamma^M\,\alpha \equiv a^M\,\gamma^M\,\alpha\\
  &(\Pie\,\mi{Ix}\,B)^M\,t_0\,t_1\,\gamma^M &&\defn (i : \mi{Ix}) \to (B\,i)^M\,(t_0\,i)\,(t_1\,i)\,\gamma^M\\
  &(\appe\,t\,i)^M\,\gamma^M &&\defn t^M\,\gamma^M\,i\\
  &(\lame\,t)^M\,\gamma^M &&\defn (t\,i)^M\,\gamma^M
\end{alignat*}

\subsubsection{Displayed algebras}

\begin{alignat*}{3}
  &\blank^D &&: (\Gamma : \Con) \to \Gamma^A \to \Set\\
  &\blank^D &&: (\sigma : \Sub\,\Gamma\,\Delta) \to \Gamma^D\,\gamma \to \Delta^D\,(\sigma^A\,\gamma) \\
  &\blank^D &&: (A : \Ty\,\Gamma) \to A^A\,\gamma \to \Gamma^D\,\gamma \to \Set\\
  &\blank^D &&: (t : \Tm\,\Gamma\,A) \to (\gamma^D : \Gamma^D\,\gamma) \to A^D\,(t^A\,\gamma)\,\gamma^D\\
  &\emptycon^D\,\gamma &&\defn \top\\
  &\epsilon^D\,\gamma^D &&\defn \tt\\
  &\id^D\,\gamma^D &&\defn \gamma^D\\
  &(\sigma \circ \delta)^D\,\gamma^D &&\defn \sigma^D\,(\delta^D\,\gamma^D)\\
  &(\Gamma \ext A)^D\,(\gamma,\,\alpha) &&\defn (\gamma^D : \Gamma^D\,\gamma) \times A^D\,\alpha\,\gamma^D\\
  &(A[\sigma])^D\,\alpha\,\gamma^D &&\defn A^D\,\alpha\,(\sigma^D\,\gamma^D)\\
  &(t[\sigma])^D\,\gamma^D &&\defn t^D\,(\sigma^D\,\gamma^D)\\
  &\p^D\,(\gamma^D,\,\alpha^D) &&\defn \gamma^D\\
  &\q^D\,(\gamma^D,\,\alpha^D) &&\defn \alpha^D\\
  &(\sigma,\,t)^D\,\gamma^D &&\defn (\sigma^D\,\gamma^D,\,t^D\,\gamma^D)\\
  &\U^D\,a\,\gamma^D &&\defn a \to \Ty_0\\
  &(\El\,a)^D\,t\,\gamma^D &&\defn a^D\,\gamma^D\,t\\
  &(\Id\,t\,u)^D\,\gamma^D &&\defn t^D\,\gamma^D \equiv u^D\,\gamma^D\\
  &\refl^D\,\gamma^D &&\defn \refl : t^D\,\gamma^D \equiv t^D\,\gamma^D\\
  &(\reflect\,p)^D &&\defn \funext\,(\lambda\,\gamma^D.\,p^D\,\gamma^D)\\
  &(\Pi\,a\,B)^D\,t\,\gamma^D &&\defn \{\alpha : a^A\,\gamma\}(\alpha^D : a^D\,\gamma^D\,\alpha)
    \to B^D\,(t\,\alpha)\,(\gamma^D,\,\alpha^D)\\
  &(\app\,t)^D\,(\gamma^D,\,\alpha^D) &&\defn t^D\,\gamma^D\,\alpha^D\\
  &(\lam\,t)^D\,\gamma^D &&\defn \lambda\,\{\alpha\}\,\alpha^D.\,t^D\,(\gamma^D,\,\alpha^D)\\
  &(\Pie\,\mi{Ix}\,B)^D\,t\,\gamma^D &&\defn (i : \mi{Ix}) \to (B\,i)^D\,(t\,i)\,\gamma^D\\
  &(\appe\,t\,i)^D\,\gamma^D &&\defn t^D\,\gamma^D\,i\\
  &(\lame\,t)^D\,\gamma &&\defn \lambda\,i.\,(t\,i)^D\,\gamma^D
\end{alignat*}

\subsubsection{Sections}

\begin{alignat*}{3}
  &\blank^S &&: (\Gamma : \Con) \to (\gamma : \Gamma^A) \to \Gamma^A\,\gamma \to \Set\\
  &\blank^S &&: (\sigma : \Sub\,\Gamma\,\Delta) \to \Gamma^S\,\gamma\,\gamma^D \to \Delta^S\,(\sigma^A\,\gamma)\,(\sigma^D\,\gamma^D)\\
  &\blank^S &&: (A : \Ty\,\Gamma) \to A^A\,\gamma \to A^D\,\gamma^D \to \Gamma^S\,\gamma\,\gamma^D \to \Set\\
  &\blank^S &&: (t : \Tm\,\Gamma\,A) \to (\gamma^S : \Gamma^S\,\gamma\,\gamma^D) \to A^S\,(t^A\,\gamma)\,(t^D\,\gamma^D)\,\gamma^S\\
  &\emptycon^S\,\gamma\,\gamma^D &&\defn \top \\
  &\epsilon^S\,\gamma^S &&\defn \tt\\
  &\id^S\,\gamma^S &&\defn \gamma^S\\
  &(\sigma \circ \delta)^S\,\gamma^S &&\defn \sigma^S\,(\delta^S\,\gamma^S)\\
  &(\Gamma \ext A)^S\,(\gamma,\,\alpha)\,(\gamma^D,\,\alpha^D) &&\defn
    (\gamma^S : \Gamma^S\,\gamma\,\gamma^D) \times A^S\,\alpha\,\alpha^D\,\gamma^S\\
  &(A[\sigma])^S\,\alpha\,\alpha^D\,\gamma^S &&\defn A^S\,\alpha\,\alpha^D\,(\sigma^S\,\gamma^S)\\
  &(t[\sigma])^S\,\gamma^S &&\defn t^S\,(\sigma^S\,\gamma^S)\\
  &\p^S\,(\gamma^S,\,\alpha^S) &&\defn \gamma^S\\
  &\q^S\,(\gamma^S,\,\alpha^S) &&\defn \alpha^S\\
  &(\sigma,\,t)^S\,\gamma^S &&\defn (\sigma^S\,\gamma^S,\,t^S\,\gamma^S)\\
  &\U^S\,a\,a^D\,\gamma^S &&\defn (\alpha : a) \to a^D\,\alpha\\
  &(\El\,a)^S\,\alpha\,\alpha^D\,\gamma^S &&\defn a^S\,\gamma^S\,\alpha \equiv \alpha^D \\
  &(\Id\,t\,u)^S\,p\,p^D\,\gamma^S &&\defn t^S\,\gamma^S \equiv u^S\,\gamma^S\\
  &\refl^S\,\gamma^S &&\defn \refl : t^S\,\gamma^S \equiv t^S\,\gamma^S\\
  &(\reflect\,p)^S &&\defn \funext\,(\lambda\,\gamma^S.\,p^S\,\gamma^S)\\
  &(\Pi\,a\,B)^S\,t\,t^D\,\gamma^S &&\defn (\alpha : a^A\,\gamma) \to B^S\,(t\,\alpha)\,(t^D\,(a^S\,\gamma^S\,\alpha))\,(\gamma^S,\,\refl)\\
  &(\app\,t)^S\,(\gamma^S,\,\alpha^S) &&\defn t^S\,\gamma^S\,\alpha\hspace{1em}\text{where}\hspace{1em} \alpha^S : a^S\,\gamma^S\,\alpha \equiv \alpha^D\\
  &(\lam\,t)^S\,\gamma^S &&\defn \lambda\,\alpha.\,t^S\,(\gamma^S,\,\refl)\hspace{1em}\text{where}\hspace{1em} \refl : a^S\,\gamma^S\,\alpha \equiv a^S\,\gamma^S\,\alpha\\
  &(\Pie\,\mi{Ix}\,B)^S\,t\,t^D\,\gamma^S &&\defn (i : \mi{Ix}) \to (B\,i)^S\,(t\,i)\,(t^D\,i)\,\gamma^S\\
  &(\appe\,t\,i)^S\,\gamma^S &&\defn t^S\,\gamma^S\,i\\
  &(\lame\,t)^S\,\gamma^S &&\defn (t\,i)^S\,\gamma^S
\end{alignat*}

\chapter{AMDS interpretation of strict HII signatures}
\label{app:hii-amds}

This appendix supplements Chapter \ref{chap:hiit}. It contains the AMDS
interpretation for strict HII signatures. We omit substitution and
$\beta\eta$-rules. We also omit the $\Tm_0$ decoding operation of two-level type
theory.

\pagebreak

\subsubsection{Components of ToS (without equations)}
\vspace{-0.5em}
\begin{alignat*}{3}
  &\Con &&: \Set\\
  &\Sub &&: \Con \to \Con \to \Set\\
  &\Ty  &&: \Con \to \Set\\
  &\Tm  &&: (\Gamma : \Con) \to \Ty\,\Gamma \to \Set \\
  &\emptycon &&: \Con\\
  &\epsilon &&: \Sub\,\Gamma\,\emptycon\\
  &\id &&: \Sub\,\Gamma\,\Gamma\\
  &\blank\!\circ\!\blank &&: \Sub\,\Delta\,\Xi \to \Sub\,\Gamma\,\Delta \to \Sub\,\Gamma\,\Xi\\
  &\blank\![\blank] &&: \Ty\,\Delta \to \Sub\,\Gamma\,\Delta \to \Ty\,\Gamma\\
  &\blank\![\blank] &&: \Tm\,\Delta\,A \to (\sigma : \Sub\,\Gamma\,\Delta) \to \Tm\,\Gamma\,(A[\sigma])\\
  &\p &&: \Sub\,(\Gamma\ext A)\,\Gamma\\
  &\q &&: \Tm\,(\Gamma\ext A)\,(A[\p])\\
  &(\blank\!,\!\blank) &&: (\sigma : \Sub\,\Gamma\,\Delta) \to \Tm\,\Gamma\,(A[\sigma]) \to \Sub\,\Gamma\,(\Delta\ext A)\\
  &\U &&: \Ty\,\Gamma\\
  &\El &&: \Tm\,\Gamma\,\U \to \Ty\,\Gamma\\
  &\top &&: \Tm\,\Gamma\,\U\\
  &\tt  &&: \Tm\,\Gamma\,(\El\,\top)\\
  &\Sigma &&: (a : \Tm\,\Gamma\,\U) \to \Tm\,(\Gamma \ext \El\,a)\,\U \to \Tm\,\Gamma\,\U\\
  &\proj_1 &&: \Tm\,\Gamma\,(\El\,(\Sigma\,a\,b)) \to \Tm\,\Gamma\,(\El\,a)\\
  &\proj_2 &&: (t : \Tm\,\Gamma\,(\El\,(\Sigma\,a\,b))) \to \Tm\,\Gamma\,(\El\,(b[\id,\,\proj_1\,t]))\\
  &(\blank\!,\!\blank) &&: (t : \Tm\,\Gamma\,(\El\,a)) \to \Tm\,\Gamma\,(\El\,(b[\id,\,t])) \to \Tm\,\Gamma\,(\El\,(\Sigma\,a\,b))\\
  &\Id &&: \Tm\,\Gamma\,(\El\,a) \to \Tm\,\Gamma\,(\El\,a) \to \Tm\,\Gamma\,\U\\
  &\refl &&: \Tm\,\Gamma\,(\El\,(\Id\,t\,t))\\
  &\Piinf &&: (\mi{Ix} : \Ty_0) \to (\mi{Ix} \to \Tm\,\Gamma\,\U) \to \Tm\,\Gamma\,\U\\
  &\appinf &&: \Tm\,\Gamma\,(\El\,(\Piinf\,\mi{Ix}\,b)) \to (i : \mi{Ix}) \to \Tm\,\Gamma\,(\El\,(b\,i))\\
  &\laminf &&: ((i : \mi{Ix}) \to \Tm\,\Gamma\,(\El\,(b\,i))) \to \Tm\,\Gamma\,(\El\,(\Piinf\,\mi{Ix}\,b))\\
  &\Pi &&: (a : \Tm\,\Gamma\,\U) \to \Ty\,(\Gamma\ext \El\,a) \to \Ty\,\Gamma\\
  &\app &&: \Tm\,\Gamma\,(\Pi\,a\,B) \to \Tm\,(\Gamma \ext \El\,a)\,B\\
  &\lam &&: \Tm\,(\Gamma \ext \El\,a)\,B \to \Tm\,\Gamma\,(\Pi\,a\,B)\\
  &\Pie &&: (\mi{Ix} : \Ty_0) \to (\mi{Ix} \to \Ty\,\Gamma) \to \Ty\,\Gamma\\
  &\appe &&: \Tm\,\Gamma\,(\Pie\,\mi{Ix}\,B) \to (i : \mi{Ix}) \to \Tm\,\Gamma\,(B\,i)\\
  &\lame &&: ((i : \mi{Ix}) \to \Tm\,\Gamma\,(B\,i)) \to \Tm\,\Gamma\,(\Pie\,\mi{Ix}\,B)
\end{alignat*}

\subsubsection{Algebras}
\vspace{-0.5em}
\begin{alignat*}{3}
  &\blank^A &&: \Con \to \Set\\
  &\blank^A &&: \Sub\,\Gamma\,\Delta \to \Gamma^A \to \Delta^A\\
  &\blank^A &&: \Ty\,\Gamma \to \Gamma^A \to \Set\\
  &\blank^A &&: \Tm\,\Gamma\,A \to (\gamma : \Gamma^A) \to A^A\,\gamma\\
  &\emptycon^A &&\defn \top\\
  &\epsilon^A\,\gamma &&\defn \tt\\
  &\id^A\,\gamma &&\defn \gamma\\
  &(\sigma \circ \delta)^A\,\gamma &&\defn \sigma^A\,(\delta^A\,\gamma)\\
  &(\Gamma \ext A)^A &&\defn (\gamma : \Gamma^A) \times A^A\,\gamma\\
  &(A[\sigma])^A\,\gamma &&\defn A^A\,(\sigma^A\,\gamma)\\
  &(t[\sigma])^A\,\gamma &&\defn t^A\,(\sigma^A\,\gamma)\\
  &\p^A\,(\gamma,\,\alpha) &&\defn \gamma\\
  &\q^A\,(\gamma,\,\alpha) &&\defn \alpha\\
  &(\sigma,\,t)^A\,\gamma &&\defn (\sigma^A\,\gamma,\,t^A\,\gamma)\\
  &\U^A\,\gamma &&\defn \Ty_0\\
  &(\El\,a)^A\,\gamma &&\defn a^A\,\gamma\\
  &\top^A\,\gamma &&\defn \top_0\\
  &\tt^A\,\gamma &&\defn \tt_0\\
  &(\Sigma\,a\,b)^A\,\gamma &&\defn (\alpha : a^A\,\gamma) \times_0 b^A\,(\gamma,\,\alpha)\\
  &(\proj_1\,t)^A\,\gamma &&\defn \proj_1\,(t^A\,\gamma)\\
  &(\proj_2\,t)^A\,\gamma &&\defn \proj_2\,(t^A\,\gamma)\\
  &(t,\,u)^A\,\gamma &&\defn (t^A\,\gamma,\,u^A\,\gamma)\\
  &(\Id\,t\,u)^A\,\gamma &&\defn t^A\,\gamma = u^A\,\gamma\\
  &\refl^A\,\gamma &&\defn \refl : t^A\,\gamma = t^A\,\gamma\\
  &(\Piinf\,\mi{Ix}\,b)^A\,\gamma &&\defn (i : \mi{Ix}) \to (b\,i)^A\,\gamma\\
  &(\appinf\,t\,i)^A\,\gamma &&\defn t^A\,\gamma\,i\\
  &(\laminf\,t)^A\,\gamma &&\defn \lambda\,i.\,(t\,i)^A\,\gamma\\
  &(\Pi\,a\,B)^A\,\gamma &&\defn (\alpha : a^A\,\gamma) \to B^A\,(\gamma,\,\alpha)\\
  &(\app\,t)^A\,(\gamma,\,\alpha) &&\defn t^A\,\gamma\,\alpha\\
  &(\lam\,t)^A\,\gamma &&\defn \lambda\,\alpha.\,t^A\,(\gamma,\,\alpha)\\
  &(\Pie\,\mi{Ix}\,B)^A\,\gamma &&\defn (i : \mi{Ix}) \to (B\,i)^A\,\gamma\\
  &(\appe\,t\,i)^A\,\gamma &&\defn t^A\,\gamma\,i\\
  &(\lame\,t)^A\,\gamma &&\defn \lambda\,i.\,(t\,i)^A\,\gamma
\end{alignat*}

\subsubsection{Morphisms}
\vspace{-0.5em}
\begin{alignat*}{3}
  &\blank^M &&: (\Gamma : \Con) \to \Gamma^A \to \Gamma^A \to \Set\\
  &\blank^M &&: (\sigma : \Sub\,\Gamma\,\Delta) \to \Gamma^M\,\gamma_0\,\gamma_1 \to \Delta^M\,(\sigma^A\,\gamma_0)\,(\sigma^A\,\gamma_1)\\
  &\blank^M &&: (A : \Ty\,\Gamma) \to A^A\,\gamma_0 \to A^A\,\gamma_1 \to \Gamma^M\,\gamma_0\,\gamma_1 \to \Set\\
  &\blank^M &&: (t : \Tm\,\Gamma\,A) \to (\gamma^M : \Gamma^M\,\gamma_0\,\gamma_1) \to A^M\,(t^A\,\gamma_0)\,(t^A\,\gamma_1)\,\gamma^M\\
  &\emptycon^M\,\gamma_0\,\gamma_1 &&\defn \top \\
  &\epsilon^M\,\gamma^M &&\defn \tt\\
  &\id^M\,\gamma^M &&\defn \gamma^M\\
  &(\sigma \circ \delta)^M\,\gamma^M &&\defn \sigma^M\,(\delta^M\,\gamma^M)\\
  &(\Gamma \ext A)^M\,(\gamma_0,\,\alpha_0)\,(\gamma_1,\,\alpha_1) &&\defn
    (\gamma^M : \Gamma^M\,\gamma_0\,\gamma_1) \times A^M\,\alpha_0\,\alpha_1\,\gamma^M\\
  &(A[\sigma])^M\,\alpha_0\,\alpha_1\,\gamma^M &&\defn A^M\,\alpha_0\,\alpha_1\,(\sigma^M\,\gamma^M)\\
  &(t[\sigma])^M\,\gamma^M &&\defn t^M\,(\sigma^M\,\gamma^M)\\
  &\p^M\,(\gamma^M,\,\alpha^M) &&\defn \gamma^M\\
  &\q^M\,(\gamma^M,\,\alpha^M) &&\defn \alpha^M\\
  &(\sigma,\,t)^M\,\gamma^M &&\defn (\sigma^M\,\gamma^M,\,t^M\,\gamma^M)\\
  &\U^M\,a_0\,a_1\,\gamma^M &&\defn a_0 \to a_1\\
  &(\El\,a)^M\,\alpha_0\,\alpha_1\,\gamma^M &&\defn a^M\,\gamma^M\,\alpha_0 \equiv \alpha_1 \\
  &\top^M\,\gamma^M &&\defn \lambda\,\_.\,\tt_0\\
  &\tt^M\,\gamma^M &&\defn \refl\\
  &(\Sigma\,a\,b)^M\,\gamma^M &&\defn \lambda\,(\alpha,\,\beta).\,(a^M\,\gamma^M\,\alpha,\,b^M\,(\gamma^M,\,\refl)\,\beta)\\
  &(\proj_1\,t)^M\,\gamma^M &&\defn \refl\\
  &(\proj_2\,t)^M\,\gamma^M &&\defn \refl\\
  &(t,\,u)^M\,\gamma^M &&\defn \refl \\
  &(\Id\,t\,u)^M\,\gamma^M &&\defn \lambda\,(p : t^A\,\gamma_0 = u^A\,\gamma_0).\,\ap\,(a^M\,\gamma^M)\,p\\
  &\refl^M\,\gamma^M &&\defn \refl\\
  &(\Piinf\,\mi{Ix}\,b)^M\,\gamma^M &&\defn \lambda\,t\,i.\,(b\,i)^M\,\gamma^M\,(t\,i)\\
  &(\appinf\,t\,i)^M\,\gamma^M &&\defn \refl \\
  &(\laminf\,t)^M\,\gamma^M &&\defn \refl \\
  &(\Pi\,a\,B)^M\,t_0\,t_1\,\gamma^M &&\defn (\alpha : a^A\,\gamma_0) \to B^M\,(t_0\,\alpha)\,(t_1\,(a^M\,\gamma^M\,\alpha))\,(\gamma^M,\,\refl)\\
  &(\app\,t)^M\,(\gamma^M,\,\alpha^M) &&\defn t^M\,\gamma^M\,\alpha_0\hspace{1em}\text{where}\hspace{1em} \alpha^M : a^M\,\gamma^M\,\alpha_0 \equiv \alpha_1\\
  &(\lam\,t)^M\,\gamma^M &&\defn \lambda\,\alpha.\,t^M\,(\gamma^M,\,\refl)\hspace{1em}\text{where}\hspace{1em} \refl : a^M\,\gamma^M\,\alpha \equiv a^M\,\gamma^M\,\alpha\\
  &(\Pie\,\mi{Ix}\,B)^M\,t_0\,t_1\,\gamma^M &&\defn (i : \mi{Ix}) \to (B\,i)^M\,(t_0\,i)\,(t_1\,i)\,\gamma^M\\
  &(\appe\,t\,i)^M\,\gamma^M &&\defn t^M\,\gamma^M\,i\\
  &(\lame\,t)^M\,\gamma^M &&\defn (t\,i)^M\,\gamma^M
\end{alignat*}

\subsubsection{Displayed algebras}
\vspace{-0.5em}
\begin{alignat*}{3}
  &\blank^D &&: (\Gamma : \Con) \to \Gamma^A \to \Set\\
  &\blank^D &&: (\sigma : \Sub\,\Gamma\,\Delta) \to \Gamma^D\,\gamma \to \Delta^D\,(\sigma^A\,\gamma) \\
  &\blank^D &&: (A : \Ty\,\Gamma) \to A^A\,\gamma \to \Gamma^D\,\gamma \to \Set\\
  &\blank^D &&: (t : \Tm\,\Gamma\,A) \to (\gamma^D : \Gamma^D\,\gamma) \to A^D\,(t^A\,\gamma)\,\gamma^D\\
  &\emptycon^D\,\gamma &&\defn \top\\
  &\epsilon^D\,\gamma^D &&\defn \tt\\
  &\id^D\,\gamma^D &&\defn \gamma^D\\
  &(\sigma \circ \delta)^D\,\gamma^D &&\defn \sigma^D\,(\delta^D\,\gamma^D)\\
  &(\Gamma \ext A)^D\,(\gamma,\,\alpha) &&\defn (\gamma^D : \Gamma^D\,\gamma) \times A^D\,\alpha\,\gamma^D\\
  &(A[\sigma])^D\,\alpha\,\gamma^D &&\defn A^D\,\alpha\,(\sigma^D\,\gamma^D)\\
  &(t[\sigma])^D\,\gamma^D &&\defn t^D\,(\sigma^D\,\gamma^D)\\
  &\p^D\,(\gamma^D,\,\alpha^D) &&\defn \gamma^D\\
  &\q^D\,(\gamma^D,\,\alpha^D) &&\defn \alpha^D\\
  &(\sigma,\,t)^D\,\gamma^D &&\defn (\sigma^D\,\gamma^D,\,t^D\,\gamma^D)\\
  &\U^D\,a\,\gamma^D &&\defn a \to \Ty_0\\
  &(\El\,a)^D\,t\,\gamma^D &&\defn a^D\,\gamma^D\,t\\
  &\top^D\,\gamma^D &&\defn \lambda\,\_.\,\top_0\\
  &\tt^D\,\gamma^D &&\defn \tt_0\\
  &(\Sigma\,a\,b)^D\,\gamma^D &&\defn \lambda\,(\alpha,\,\beta).\,(\alpha^D : a^D\,\gamma^D\,\alpha)\times_0 b^D\,(\gamma^D,\,\alpha^D)\,\beta\\
  &(\proj_1\,t)^D\,\gamma^D &&\defn \proj_1\,(t^D\,\gamma^D)\\
  &(\proj_2\,t)^D\,\gamma^D &&\defn \proj_2\,(t^D\,\gamma^D)\\
  &(t,\,u)^D\,\gamma^D &&\defn (t^D\,\gamma^D,\,u^D\,\gamma^D)\\
  &(\Id\,t\,u)^D\,\gamma^D &&\defn \lambda\,(p : t^A\,\gamma = u^A\,\gamma).\,\tr_{(a^D\,\gamma^D)}\,p\,(t^D\,\gamma^D) = u^D\,\gamma^D\\
  &\refl^D\,\gamma^D &&\defn \refl : t^D\,\gamma^D = t^D\,\gamma^D\\
  &(\Piinf\,\mi{Ix}\,b)^D\,\gamma^D &&\defn \lambda\,t.\,(i : \mi{Ix}) \to (b\,i)^D\,\gamma^D\,(t\,i)\\
  &(\appinf\,t\,i)^D\,\gamma^D &&\defn t^D\,\gamma^D\,i\\
  &(\laminf\,t)^D\,\gamma^D &&\defn \lambda\,i.\,(t\,i)^D\,\gamma^D\\
  &(\Pi\,a\,B)^D\,t\,\gamma^D &&\defn \{\alpha : a^A\,\gamma\}(\alpha^D : a^D\,\gamma^D\,\alpha)
    \to B^D\,(t\,\alpha)\,(\gamma^D,\,\alpha^D)\\
  &(\app\,t)^D\,(\gamma^D,\,\alpha^D) &&\defn t^D\,\gamma^D\,\alpha^D\\
  &(\lam\,t)^D\,\gamma^D &&\defn \lambda\,\{\alpha\}\,\alpha^D.\,t^D\,(\gamma^D,\,\alpha^D)\\
  &(\Pie\,\mi{Ix}\,B)^D\,t\,\gamma^D &&\defn (i : \mi{Ix}) \to (B\,i)^D\,(t\,i)\,\gamma^D\\
  &(\appe\,t\,i)^D\,\gamma^D &&\defn t^D\,\gamma^D\,i\\
  &(\lame\,t)^D\,\gamma &&\defn \lambda\,i.\,(t\,i)^D\,\gamma^D
\end{alignat*}

\subsubsection{Sections}
\vspace{-0.5em}
\begin{alignat*}{3}
  &\blank^S &&: (\Gamma : \Con) \to (\gamma : \Gamma^A) \to \Gamma^A\,\gamma \to \Set\\
  &\blank^S &&: (\sigma : \Sub\,\Gamma\,\Delta) \to \Gamma^S\,\gamma\,\gamma^D \to \Delta^S\,(\sigma^A\,\gamma)\,(\sigma^D\,\gamma^D)\\
  &\blank^S &&: (A : \Ty\,\Gamma) \to (\alpha : A^A\,\gamma) \to A^D\,\alpha\,\gamma^D \to \Gamma^S\,\gamma\,\gamma^D \to \Set\\
  &\blank^S &&: (t : \Tm\,\Gamma\,A) \to (\gamma^S : \Gamma^S\,\gamma\,\gamma^D) \to A^S\,(t^A\,\gamma)\,(t^D\,\gamma^D)\,\gamma^S\\
  &\emptycon^S\,\gamma\,\gamma^D &&\defn \top \\
  &\epsilon^S\,\gamma^S &&\defn \tt\\
  &\id^S\,\gamma^S &&\defn \gamma^S\\
  &(\sigma \circ \delta)^S\,\gamma^S &&\defn \sigma^S\,(\delta^S\,\gamma^S)\\
  &(\Gamma \ext A)^S\,(\gamma,\,\alpha)\,(\gamma^D,\,\alpha^D) &&\defn
    (\gamma^S : \Gamma^S\,\gamma\,\gamma^D) \times A^S\,\alpha\,\alpha^D\,\gamma^S\\
  &(A[\sigma])^S\,\alpha\,\alpha^D\,\gamma^S &&\defn A^S\,\alpha\,\alpha^D\,(\sigma^S\,\gamma^S)\\
  &(t[\sigma])^S\,\gamma^S &&\defn t^S\,(\sigma^S\,\gamma^S)\\
  &\p^S\,(\gamma^S,\,\alpha^S) &&\defn \gamma^S\\
  &\q^S\,(\gamma^S,\,\alpha^S) &&\defn \alpha^S\\
  &(\sigma,\,t)^S\,\gamma^S &&\defn (\sigma^S\,\gamma^S,\,t^S\,\gamma^S)\\
  &\U^S\,a\,a^D\,\gamma^S &&\defn (\alpha : a) \to a^D\,\alpha\\
  &(\El\,a)^S\,\alpha\,\alpha^D\,\gamma^S &&\defn a^S\,\gamma^S\,\alpha \equiv \alpha^D \\
  &\top^S\,\gamma^S &&\defn \lambda\,\_.\,\tt_0\\
  &\tt^S\,\gamma^S &&\defn \refl\\
  &(\Sigma\,a\,b)^S\,\gamma^S &&\defn \lambda\,(\alpha,\,\beta).\,(a^S\,\gamma^S\,\alpha,\,b^S\,(\gamma^S,\,\refl)\,\beta)\\
  &(\proj_1\,t)^S\,\gamma^S &&\defn \refl\\
  &(\proj_2\,t)^S\,\gamma^S &&\defn \refl\\
  &(t,\,u)^S\,\gamma^S &&\defn \refl \\
  &(\Id\,t\,u)^S\,\gamma^S &&\defn \lambda\,(p : t^A\,\gamma = u^A\,\gamma).\,\apd\,(a^S\,\gamma^S)\,p\\
  &\refl^S\,\gamma^S &&\defn \refl\\
  &(\Piinf\,\mi{Ix}\,b)^S\,\gamma^S &&\defn \lambda\,t\,i.\,(b\,i)^S\,\gamma^S\,(t\,i)\\
  &(\appinf\,t\,i)^S\,\gamma^S &&\defn \refl \\
  &(\laminf\,t)^S\,\gamma^S &&\defn \refl \\
  &(\Pi\,a\,B)^S\,t\,t^D\,\gamma^S &&\defn (\alpha : a^A\,\gamma) \to B^S\,(t\,\alpha)\,(t^D\,(a^S\,\gamma^S\,\alpha))\,(\gamma^S,\,\refl)\\
  &(\app\,t)^S\,(\gamma^S,\,\alpha^S) &&\defn t^S\,\gamma^S\,\alpha\hspace{1em}\text{where}\hspace{1em} \alpha^S : a^S\,\gamma^S\,\alpha \equiv \alpha^D\\
  &(\lam\,t)^S\,\gamma^S &&\defn \lambda\,\alpha.\,t^S\,(\gamma^S,\,\refl)\hspace{1em}\text{where}\hspace{1em} \refl : a^S\,\gamma^S\,\alpha \equiv a^S\,\gamma^S\,\alpha\\
  &(\Pie\,\mi{Ix}\,B)^S\,t\,t^D\,\gamma^S &&\defn (i : \mi{Ix}) \to (B\,i)^S\,(t\,i)\,(t^D\,i)\,\gamma^S\\
  &(\appe\,t\,i)^S\,\gamma^S &&\defn t^S\,\gamma^S\,i\\
  &(\lame\,t)^S\,\gamma^S &&\defn (t\,i)^S\,\gamma^S
\end{alignat*}

\backmatter
\bibliography{references}

\pagebreak
\section*{Summary}

This thesis develops the usage of certain type theories as specification
languages for algebraic theories and inductive types. We observe that the
expressive power of dependent type theories proves useful in the specification
of more complicated algebraic theories. In the thesis, we describe three type
theories where each typing context can be viewed as an algebraic signature,
specifying sorts, operations and equations. These signatures are useful in
broader mathematical contexts, but we are also concerned with potential
implementation in proof assistants.

In \textbf{Chapter \ref{chap:2ltt}}, we describe a way to use
two-level type theory \cite{twolevel} as a metalanguage for developing semantics
of algebraic signatures. This makes it possible to work in a concise internal
notation of a type theory, and at the same time build semantics internally to
arbitrary structured categories. For example, the signature for natural number
objects can be interpreted in any category with finite products.

In \textbf{Chapter \ref{chap:fqiit}}, we describe finitary quotient
inductive-inductive (FQII) signatures. Most type theories themselves can be
specified with FQII signatures. We build a structured category of algebras for
each signature, where equivalence of initiality and induction can be shown. We
additionally present term algebra constructions, constructions of left adjoint
functors of signature morphisms, and we describe a way to use self-describing
signatures to minimize necessary metatheoretic assumptions.

In \textbf{Chapter \ref{chap:iqiit}}, we describe infinitary quotient
inductive-inductive signatures. These allow specification of infinitely
branching trees as initial algebras. We adapt the semantics from the previous
chapter. We also revisit term models, left adjoints of signature morphisms and
self-description of signatures. We also describe how to build semantics of
signatures internally to the theory of signatures itself, which yields numerous
ways to build new signatures from existing ones.

In \textbf{Chapter \ref{chap:hiit}}, we describe higher inductive-inductive
signatures. These differ from previous semantics mostly in that their intended
semantics is in homotopy type theory \cite{hottbook}, and allows
higher-dimensional equalities. In this more general setting we only consider
enough semantics to compute notions of initiality and induction for each signature.

\pagebreak
\section*{Summary in Hungarian - Magyar összefoglaló}

A tézis fő célja az, hogy kidolgozza bizonyos típuselméletek használatát
algebrai elméletek és induktív típusok leírásához. Meglátásunk szerint a függő
típuselméletek kifejezőereje nagyban elősegíti a tömör és általános
specifikációkat. A tézisben három típuselméletet írunk le, amelyekben a
típuskörnyezeteket értelmezzük algebrai szignatúraként, ami felsorolja egy
algebrai elmélet szortjait, műveleteit és egyenleteit. A eredményeink
felhasználhatók általánosabb matematikai kontextusban, viszont az is célunk,
hogy elősegítsük az esetleges pratikus implementációt
tételbizonyító-rendszerekben.

A \textbf{harmadik fejezetben} kifejtjük, hogy a kétszintű típuselmélet \cite{twolevel}
hogyan használható metanyelvként az algebrai szignatúrák szemantikájához. Ez
lehetővé teszi, hogy a szemantikát általánosan adjuk meg, internálisan tetszőleges
strukturált kategóriákban, és ugyanakkor tömör típuselméleti nyelvben
dolgozzunk. Például a természetes szám objektumok szignatúrája értelmezhető
tetszőleges olyan kategóriában, ami rendelkezik véges szorzatokkal.

A \textbf{negyedik fejezetben} leírjuk a véges aritású kvóciens
induktív-induktív (FQII) szignatúrák elméletét. A legtöbb típuselmélet maga is
leírható FQII szignatúrával. Minden szignatúrához megadjuk az algebrák egy
strukturált kategóriáját, ahol az inicialitás és az indukció ekvivalenciája
belátható. Továbbá, bemutatunk term algebra konstrukciókat, bal adjungált
funktorok konstrukcióját szignatúra-morfizmusokhoz, és bemutatjuk, hogy az
önmaguk elméletét specifikáló szignatúrák segítségével hogyan minimalizálhatók a
szükséges metaelméleti feltételezések.

Az \textbf{ötödik fejezetben} leírjuk a végtelen aritású kvóciens
induktív-induktív szignatúrák elméletét, amivel végtelenül elágazó fa
struktúrákat is le tudjunk írni az iniciális algebrákban. Adaptáljuk a korábbi
term algebra konstrukciót, a bal adjungált funktorok konstrukcióját és az
önmaguk elméletét specifikáló szignatúrák használatát. Továbbá, megadjuk a
szignatúrák szemantikáját internálisan a szignatúrák elméletének a
szintaxisában, amelynek segítségével sokféleképpen építhetünk új szignatúrákat.

A \textbf{hatodik fejezetben} leírjuk a magasabb induktív-induktív
szignatúrákat. Ezek elsősorban a szemantikában különbözek a korábbi
szignatúráktól: a metanyelv most a homotópia típuselmélet \cite{hottbook}, és
lehetőség van magasabb dimenziós egyenlőségek megadására. Itt csak annyi
szemantikát adunk meg, amiből az inicialitás és indukció fogalmai kiszámolhatók
minden szignatúrához.

\end{document}